\documentclass[a4paper]{article}
\pdfoutput=1

\usepackage[english]{babel}
\usepackage[utf8]{inputenc}
\usepackage[T1]{fontenc}
\usepackage{amsfonts}
\usepackage{units}
\usepackage[labelfont=bf]{caption}
\usepackage[sorting=nyt,style=numeric,doi=false,url=false,isbn=false]{biblatex}
\usepackage{csquotes}
\renewbibmacro{in:}{}
\usepackage{filecontents}
\usepackage{tikz}
\usepackage{graphicx}
\usepackage{todonotes}
\usepackage{hyperref}
\usepackage{xcolor}
\usepackage{euscript}
\usepackage{mathdots}
\usepackage{float}
\usepackage{subfigure}

\usetikzlibrary{decorations.markings}

\usepackage[a4paper,top=3cm,bottom=3cm,left=3cm,right=3cm,marginparwidth=1.75cm]{geometry}

\usepackage{mathtools}
\usepackage{enumitem}
\usepackage{amsmath}
\usepackage{bbm}
\usepackage{amsthm}
\usepackage{thmtools}
\usepackage{hhline}
\usepackage{cleveref}
\usepackage{amssymb,bm,color}
\usetikzlibrary{decorations.pathreplacing}
\usetikzlibrary{circuits.logic.US,circuits.logic.IEC,fit}
\usetikzlibrary{shapes, plotmarks, positioning, matrix}
\usetikzlibrary{calc}

\newtheorem{theorem}{Theorem}[section]
\newtheorem{Lemma}[theorem]{Lemma}
\newtheorem{Def}[theorem]{Definition}
\newtheorem{Prop}[theorem]{Proposition}
\newtheorem{Cor}[theorem]{Corollary}
\newtheorem{Remark}[theorem]{Remark}

\numberwithin{equation}{section}

\DeclareMathOperator{\sign}{sign}

\DeclarePairedDelimiter{\ceil}{\lceil}{\rceil}
\DeclarePairedDelimiter\floor{\lfloor}{\rfloor}

\newcommand{\overDot}[1]{
	\scalebox{#1}{
    	\tikz{\draw[fill] (1, 0) circle [radius=0.025];}
    }
}

\DeclareMathOperator*{\sumDot}{%
\mathchoice%
  {\ooalign{\phantom{$\displaystyle\sum$}\cr\hidewidth\raisebox{6.5\height}{$\mkern-6mu\overDot{1}$}\hidewidth\cr%
                                  \hidewidth$\displaystyle\sum$}}
  {\ooalign{$\textstyle\sum$\cr%
                                \hidewidth\raisebox{6\height}{$\mkern-6mu\overDot{0.8}$}\hidewidth\cr}}
  {\ooalign{\raisebox{0\height}{\scalebox{.6}{$\scriptstyle\sum$}}\cr%
                                \hidewidth\raisebox{1.6\height}{$\mkern7.5mu\overDot{0.2}$}\hidewidth\cr
                                \hidewidth\raisebox{-0.2\height}{$\mkern7.5mu\overDot{0.2}$}\hidewidth\cr}}
  {\ooalign{\raisebox{.2\height}{\scalebox{.6}{$\scriptstyle\sum$}}\cr%
                                \hidewidth\raisebox{2.2\height}{$\mkern7.5mu\overDot{0.2}$}\hidewidth\cr
                                \hidewidth\raisebox{0.4\height}{$\mkern7.5mu\overDot{0.2}$}\hidewidth\cr}}
}

\DeclareMathAlphabet{\mathpzc}{OT1}{pzc}{m}{it}

\newcommand\smallO{
  \mathchoice
    {{\scriptstyle\mathcal{O}}}
    {{\scriptstyle\mathcal{O}}}
    {{\scriptscriptstyle\mathcal{O}}}
    {\scalebox{.7}{$\scriptscriptstyle\mathcal{O}$}}
  }

\newcommand{\mysquare}[1]{\tikz{\node[draw=#1,fill=#1,rectangle,minimum
width=0.15cm,minimum height=0.15cm,inner sep=0pt] at (0,0) {};}}

\newcommand{\mytriangle}[1]{\tikz{\node[draw=#1,fill=#1,isosceles
triangle,isosceles triangle stretches,shape border rotate=90,minimum
width=0.15cm,minimum height=0.15cm,inner sep=0pt] at (0,0) {};}}

\newcommand\marksymbol[2]{\tikz[#2,scale=1]\pgfuseplotmark{#1};}

\DeclareSymbolFont{Letters} {U}{zeur}{m}{n}
\DeclareMathSymbol{\polygamma}    {\mathalpha}{Letters}{"1D}

\title{Stationary Stochastic Higher Spin Six Vertex Model and $q$-Whittaker measure}

\author{Takashi Imamura \thanks{Department of mathematics and informatics, Chiba University, E-mail: imamura@math.s.chiba-u.ac.jp
} ,
Matteo Mucciconi\thanks{Department of physics, Tokyo Institute of Technology, E-mail: matteomucciconi@gmail.com} , Tomohiro Sasamoto \thanks{
Department of physics, Tokyo Institute of Technology, E-mail: sasamoto@phys.titech.ac.jp}}
\date{}

\begin{document}

\maketitle

\begin{abstract}
In this paper we consider the Higher Spin Six Vertex Model on the lattice $\mathbb{Z}_{\geq 2} \times \mathbb{Z}_{\geq 1}$. We first identify a family of translation invariant measures and subsequently we study the one point distribution of the height function for the model with certain random boundary conditions. Exact formulas we obtain prove to be useful in order to establish the asymptotic of the height distribution in the long space-time limit for the stationary Higher Spin Six Vertex Model. In particular, along the characteristic line we recover Baik-Rains fluctuations with size of characteristic exponent $1/3$. We also consider some of the main degenerations of the Higher Spin Six Vertex Model and we adapt our analysis to the relevant cases of the $q$-Hahn particle process and of the Exponential Jump Model.
\end{abstract}

\tableofcontents
\section{Introduction}
\subsection{Background}

During the last two decades, the study of one dimensional integrable systems, related to random growth of interfaces or to particle transport, 
has produced a number of fundamental results. This wave of interest was surely fostered by breakthroughs like that of Johansson \cite{Johansson2000}, which around year 2000 successfully provided an exact description of the current of particles in the Totally Asymmetric Simple Exclusion Process (TASEP).  Methods used by Johansson, which were drawing inspiration from combinatorics and random matrix
theory, soon proved to be parallel to a more algebraic framework in terms of free fermions\cite{Okounkov2001InfiniteWedge}, leading to the definition of the Schur processes in \cite{Okounkov2003SchurProcesses}.
These last are probability measures weighting sequences of partitions of integers expressed in terms of Schur functions, a class of special symmetric functions more commonly used in representation theory. The richness of techniques and possibilities given by the intersection of so many apparently distinct fields gathered immediately the attention of the community of mathematicians and physicists and gave rise to a new field of its own that today bears the name of Integrable Probability \cite{BorodinGorin2016IntegrableProbability}. Over the years, purely determinantal processes like the TASEP or the Schur processes have seen a number of generalizations and methods introduced in 
\cite{Johansson2000,Okounkov2003SchurProcesses} have been extended and applied to these more general models, which are not
necessarily free fermionic.

Conceptually relevant deformations of the TASEP are exclusion processes like the $q$-TASEP \cite{BorodinCorwin2014Mac}, the $q$-Hahn TASEP \cite{PovolotskyQHahn} or also the long standing Asymmetric Simple Exclusion Process (ASEP) \cite{Spitzer1970}. The first of these models was introduced first by Borodin and Corwin as a marginal projection of the Macdonald processes \cite{BorodinCorwin2014Mac}, which as of today, happen also to be the among the richest generalization of the Schur processes. On the other hand, models like the $q$-Hahn TASEP or the ASEP had not been immediately identified as particular cases of general integrable models, like the Macdonald processes, and the study of their properties was carried out by different authors employing different techniques \cite{TracyWidom2009ASEP}, \cite{bcs2014}, \cite{Corwin2015qHahn}, \cite{BorCor2013qTaseps}. 

A unifying picture was offered by Corwin and Petrov in \cite{CorwinPetrov2016HSVM} using the language of vertex models. Here, authors, taking advantage of recent developments on algebraic theories concerning the Yang-Baxter equation \cite{Mangazeev2014}, introduced the Higher Spin Six Vertex Model, which they used to construct a random dynamics of particles on the lattice where the update rules of position of particles at each time were given in terms of what are usually called stochastic $\mathcal{R}$-matrices. These are operators, which we denote with the symbol $\mathsf{L}$, solving the celebrated Yang-Baxter equation \cite{JimboYBE} 
\begin{equation} \label{YBE}
    \mathsf{L}^{(1,2)} \mathsf{L}^{(1,3)} \mathsf{L}^{(2,3)} = \mathsf{L}^{(1,3)} \mathsf{L}^{(2,3)}
    \mathsf{L}^{(1,2)},
\end{equation}
which also satisfy the property of having positive entries and sum-to-one condition for rows. In \eqref{YBE}, the stochastic $\mathcal{R}$-matrices depend on a number of parameters, by specializing which one can degenerate the Higher Spin Six Vertex Model to previously mentioned models including the $q$-TASEP, the $q$-Hahn TASEP or the ASEP. 
In \cite{BorodinPetrov2016HS6VM} a description of the Higher Spin Six Vertex Model complementary to that of \cite{CorwinPetrov2016HSVM} was offered using a language closer to that of the Schur processes. Here, the model was studied through a family of symmetric rational functions whose properties descended from the commutation relation \eqref{YBE} and that are in fact multi-parameter generalizations of the Schur functions.

On top of offering a unified theory embracing the majority of methods used in previously studied models, spanning from asymmetric exclusion processes to random partitions, the Higher Spin Six Vertex Model also admits as particular cases new integrable systems, including inhomogeneous traffic models as the Exponential Jump Model \cite{BorodinPetrov2017ExpoJump} or the Hall-Littlewood Push-TASEP \cite{Ghosal2017}.

The idea of utilizing the Yang-Baxter integrability to produce and solve stochastic particle systems can be traced back to the early work \cite{GwaSpohn1992S6VM} by Gwa and Spohn. There authors interpreted a particular degeneration of the Six Vertex Model, of which the Higher Spin Six Vertex Model is a generalization, as a cellular automata in the same fashion as in \cite{CorwinPetrov2016HSVM} and they were able to compute the roughening exponent of a random interface associated with the current of particles. More recently, following the example of \cite{CorwinPetrov2016HSVM}, the formalism of vertex models has shown other promising applications producing dynamical version of the Higher Spin Six Vertex Model \cite{Aggarwal2018DynamicalHSVM} or multi-species integrable particle processes in the very recent work \cite{BorodinWheeler2018Coloured}.

\subsection{KPZ universality, integrability and initial conditions} \label{section: Kpz universality, integrability and initial conditions}

What drives the field of Integrable Probability is the broader context of the KPZ universality class \cite{Corwin2011KPZ}. This is often vaguely defined as a class of random processes describing the stochastic evolution of interfaces that, in the long time limit, possess a characteristic 3:2:1 scaling. That is to say that the profile of the random surfaces in question grows linearly in the time of the system $t$, while the range of spatial correlations and the size of fluctuations around its expected shape scale asymptotically as $t^{2/3}$ and $t^{1/3}$.

The principal example of a model exhibiting such properties is the KPZ equation itself, that is the simplest stochastic differential equation describing the random profile of an interface where both relaxation and lateral growth are allowed. It was first introduced in \cite{KPZ1986}, where authors after deducing its characteristic scaling, conjectured that its "nontrivial relaxation patterns" must be shared by a large class of growth processes. During the last 30 years extensive work has been done in order to understand the properties and the boundaries of this universality class, including also experimental confirmations \cite{TakeuchiSano2010Universal,TakeuchiSano2012Evidence} of the predictions of \cite{KPZ1986}.

The first result of this sort is the one obtained in \cite{Johansson2000}, which predicts that, under narrow wedge initial conditions (step initial conditions for the TASEP), the limiting fluctuations of the height function obey the GUE Tracy-Widom distribution. Confirmations of the fact that this limiting law is in fact universal for the class are given in a number of other papers (see \cite{FerrariVeto2015TWLimit}, \cite{OrrPetrov2016}, \cite{BorCor2013qTaseps}, \cite{Corwin2015qHahn} and references therein) 
and in particular this result is established for the solution of the KPZ equation in \cite{SaSp2010}, \cite{ACQ2010}.

For flat initial conditions the limiting fluctuations of the height profile are believed to be ruled by the GOE Tracy-Widom distribution and this again is based on results on the polynuclear growth model \cite{BaikRains2001Symmetrized,PrSp2000Universal} and on the TASEP \cite{Sa05}. The validation of this conjecture for other models and in particular for the KPZ equation has proven to be rather troublesome, although some steps forward were made recently in \cite{OQR2016} for the particular case of the ASEP with alternating initial data.

From the point of view of non-equilibrium statistical physics or stochastic interacting particle systems, arguably the most relevant class of initial conditions one can consider is represented by the stationary ones and this is indeed the case we pursue in this paper. For continuous models these can be regarded as Brownian motions and the height function possesses in the long time limit two different regimes. Around the characteristic line of the Burgers equation associated with the dynamics, that one can understand as the direction of the growth, the one point distribution of the height is believed to be governed by the Baik-Rains distribution $F_0$ (see Definition \ref{Baik-Rains definition}). Alternatively, when we move away from the characteristic line the size of fluctuations coming from the stationary initial data overwhelms any possible nontrivial behavior produced by the random dynamics and the height performs a gaussian process. The special law $F_0$ was introduced first in \cite{BaikRains2000Png}, where authors identified it as limiting distribution of the height of a polynuclear growth model with critical boundary conditions. An analogous result was obtained for the TASEP in \cite{FerrariSpohn2006sTASEP}. More recently, using certain limiting properties of the $q$-TASEP, Baik-Rains fluctuations were established for the solution of the KPZ equation first in \cite{BoCoFeVe2015} and then in \cite{ImaSasStationaryoConnellYor}. 

Our main result is the confirmation of the KPZ scaling theory for the stationary Higher Spin Six Vertex Model (that we define in the Section \ref{subs: the Model}) and hence for the hierarchy of models it generalizes. Prior to our work, in \cite{Aggarwal2016FluctuationsASEP}, the stationary Six Vertex Model was studied by Aggarwal, who was also able to solve the long standing problem of characterizing the asymptotic fluctuations of the current in the stationary ASEP. 

Very much related to this paper is the work by two of the authors on the stationary continuous time $q$-TASEP \cite{q-TASEPtheta}, which can in fact be considered as a part one of a two parts effort. 

\subsection{The Model} \label{subs: the Model}
In this subsection we give a definition of the Higher Spin Six Vertex Model \cite{CorwinPetrov2016HSVM}. This represents a generalization of the Six Vertex Model, a classical integrable system in statistical physics first introduced by Pauling in 1935 \cite{PaulingIceTypeModel}. Although the Six Vertex Model served originally to study geometric configurations of water molecules in an ice layer, it is today often presented in literature along with other exactly solvable models describing different physical phenomena, such as one dimensional quantum spin chains (see \cite{baxter2007exactly}). The reason behind this association is that such models all share the same integrability structure, that pivots around the notion of Yang-Baxter equation. In particular such algebraic structures allow for general constructions, leading for example to higher spin generalizations of the Heisenberg XXZ model \cite{KirillovReshetkin1987XXZ1}. We refer to \cite{quantumGroupsIn2DPhysics} for an extended review of these results.

We will consider the Stochastic Higher Spin Six Vertex Model, presented here as an ensemble of directed paths in a quadrant of the two dimensional lattice $\mathbb{Z} \times \mathbb{Z}$. Define the lattice $\Lambda_{1,0}$ and its boundary $\partial \Lambda_{1,0}$ as the sets
\begin{equation*}
    \Lambda_{1,0} = \left( \mathbb{Z}_{\geq 1} \times \mathbb{Z}_{\geq 0} \right) \setminus (1,0) \qquad \text{and} \qquad \partial \Lambda_{1,0} = \left( \mathbb{Z}_{\geq 2} \times \{0\} \right) \cup \left( \{1\} \times \mathbb{Z}_{\geq 1} \right).
\end{equation*}
We see that $\Lambda_{1,0}$ is the union of the quadrant $\mathbb{Z}_{\geq 2} \times \mathbb{Z}_{\geq 1}$ with the set of its nearest neighbor vertices $\partial \Lambda_{1,0}$ and we refer to the interior $\mathring{\Lambda}_{1,0} = \Lambda_{1,0} \setminus \partial \Lambda_{1,0} $ as the bulk of the lattice. We use the symbol $\mathfrak{P}(\Lambda_{1,0})$ to denote the set of up right directed paths in $\Lambda_{1,0}$. That is the generic element $\mathfrak{p}$ of $\mathfrak{P}(\Lambda_{1,0})$ is a collection of up right directed paths emanating from the boundary $\partial \Lambda_{1,0}$, as those represented in Figure \ref{figure paths}. 
 
\begin{figure}[t]
\centering
\includegraphics[scale=1.3]{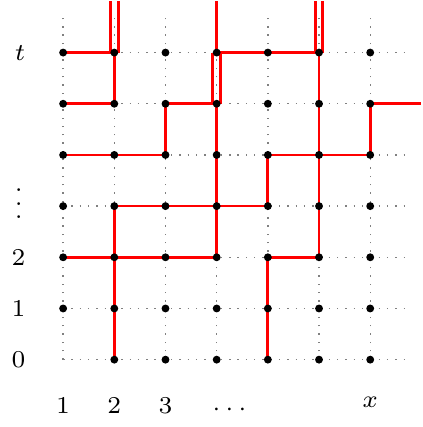}
\caption{\small{A possible configuration of up right directed paths in the lattice $\Lambda_{1,0}$.}}
\label{figure paths}
\end{figure}

A natural way to encode the information contained in the single configuration $\mathfrak{p}$ is to record how many times each edge of the lattice is shared by its paths. For this we introduce the collection of occupancy numbers
\begin{gather}
\mathsf{m}_x^t = \text{number of paths exiting the vertex $(x,t)$ in the upward direction},\label{m rv}\\
\mathsf{j}_x^t = \text{number of paths exiting the vertex $(x,t)$ in the rightward direction}, \label{j rv}
\end{gather}
of which a graphical representation is given in Figure \ref{one vertex Figure}. We think at the specific $\mathfrak{p}$ as a realization of sequences \eqref{m rv}, \eqref{j rv}, that we express with the notation $\{m_x^t,j_x^t\}_{(x,t)\in \Lambda_{1,0}}$. Since paths only generate at $\partial \Lambda_{1,0}$, quantities $\mathsf{m}^0_2, \mathsf{m}^0_3, \dots, \mathsf{j}^1_1, \mathsf{j}^2_1, \dots$ describe the boundary conditions of configurations $\mathfrak{p}$, while in the bulk, around the generic vertex $(x,t)$, we necessarily have the conservation law
\begin{equation} \label{conservation law}
\mathsf{m}_x^{t-1} + \mathsf{j}_{x-1}^t = \mathsf{m}_x^{t} + \mathsf{j}_{x}^t.
\end{equation}

\begin{figure}[t]
    \centering
\includegraphics{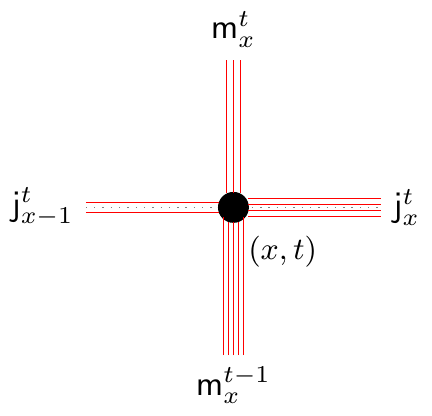}
\caption{\small{An arrangement of paths across the vertex $(x,t)$. Random variables $\mathsf{m}_x^{t-1}, \mathsf{j}_{x-1}^t$ represent paths entering respectively from below and from the left. Random variables $\mathsf{m}_x^{t}, \mathsf{j}_{x}^t$ describe the number of paths exiting the vertex respectively in the upward direction and to the right.}}
\label{one vertex Figure}
\end{figure}
We associate now, to each vertex $(x,t)$ in $\mathring{\Lambda}_{1,0}$ a \emph{stochastic weight} $L_{(x,t)}$. This is a non-negative valued function
\begin{equation} \label{stochastic vertex weight}
    L_{(x,t)}(m_x^{t-1}, j_{x-1}^t|\ m_x^t, j_x^t)
\end{equation}
of the vertex configuration, that is zero when the occupancy numbers do not fulfill \eqref{conservation law} and that satisfies, for any fixed $m_x^{t-1},j_{x-1}^t$, the sum-to-one condition
\begin{equation} \label{eq: sum to one}
    \sum_{m_x^t, j_x^t \geq 0} L_{(x,t)}(m_x^{t-1}, j_{x-1}^t|\ m_x^t, j_x^t) =1.
\end{equation}
An \emph{up right directed path stochastic vertex model} on $\Lambda_{1,0}$ is a probability measure $\mathcal{P}$ on the set $\mathfrak{P}(\Lambda_{1,0})$ where, for all vertices $(x,t)$ in the bulk $\mathring{\Lambda}_{1,0}$, the joint law of $\mathsf{m}_x^t, \mathsf{j}_x^t$, conditioned to $\mathsf{m}_x^{t-1}, \mathsf{j}_{x-1}^t$ is written as
\begin{equation} \label{eq: vertex model rule}
\mathcal{P}(\mathsf{m}_x^t=i', \mathsf{j}_x^t=j' |\ \mathsf{m}_x^{t-1}=i, \mathsf{j}_{x-1}^t=j)= L_{(x,t)}(i,j|\ i',j'),
\end{equation}
for all 4-tuples of non-negative integers $i,j,i',j'$. When this is the case, we interpret \eqref{stochastic vertex weight} as the probabilities ruling how paths propagate in the lattice in the up-right direction while crossing single vertices. To complete the definition of the measure $\mathcal{P}$ we need to assign a probability law to boundary random variables $\mathsf{m}^0_2, \mathsf{m}^0_3, \dots, \mathsf{j}^1_1, \mathsf{j}^2_1, \dots$. For the sake of this paper we will always consider them as mutually independent random variables that are a.s. finite and we can denote their joint law with the symbol $\mathcal{P}_B$. It is rather clear that, once we specify $\mathcal{P}_B$, the \emph{vertex model rule} \eqref{eq: vertex model rule} uniquely defines the measure $\mathcal{P}$ on any bounded subset of $\Lambda_{1,0}$. Through this observation, by making use of standard techniques of measure theory we could at this point prove that, given stochastic vertex weights $\{ L_{(x,t)} \}_{(x,t) \in \mathring{\Lambda}_{1,0} }$ and boundary conditions $\mathcal{P}_B$, the measure $\mathcal{P}$ is uniquely defined on the full set $\mathfrak{P}(\Lambda_{1,0})$. This is a consequence of the fact that paths are up-right directed and the distribution of $\mathsf{m}_x^t,\mathsf{j}_x^t$ only depends on $\mathsf{m}_x^{t-1},\mathsf{j}_{x-1}^t$ and therefore one can view the generic configuration $\mathfrak{p}$ as result of a markovian propagation as done in \cite{CorwinPetrov2016HSVM}.

In this paper we focus on the particular class of vertex weights $\mathsf{L}_{\xi_x u_t, s_x}$, defined in Table \ref{weights table}. Collectively, the family $\{ \mathsf{L}_{\xi_x u_t, s_x} \}_{(x,t) \in \mathring{\Lambda}_{1,0}}$ depends on a number $q$ and on the sets of values 
$$
\Xi=(\xi_2, \xi_3, \dots), \qquad \mathbf{S}=(s_2, s_3, \dots), \qquad \mathbf{U}=(u_1, u_2, \dots),
$$
which are called respectively \emph{inhomogeneity}, \emph{spin} and \emph{spectral} parameters. Unless otherwise specified we will always assume that 
\begin{equation}\label{HS6VM parameters}
0\leq q <1, \qquad
0 \leq s_x < 1 , \qquad \xi_x > 0, \qquad u_t < 0, \qquad \text{for all }x,t.
\end{equation}
Under condition \eqref{HS6VM parameters} we can easily verify that $\mathsf{L}_{\xi_x u_t, s_x}$ are positive quantities fulfilling the sum-to-one condition \eqref{eq: sum to one} and hence we regard them as bona fide stochastic vertex weights. We refer to the directed path stochastic vertex model with choice $L_{(x,t)}=\mathsf{L}_{\xi_x u_t, s_x}$ as the \emph{Stochastic Higher Spin Six Vertex Model}. 

From Table \ref{weights table} we see that the only vertex configurations having positive probability are those where different paths do not share any of the horizontal edges of the lattice. This limitation can be removed by means of a procedure called \emph{fusion}, that consists in collapsing together a number of different rows of vertices and that we review in Section \ref{subsection fused dynamics}. By fusing together a column of $J$ vertices one can construct the vertex weight $\mathsf{L}^{(J)}_{\xi_x u_t, s_x}$, which in this case takes a rather complicated form, stated below in \eqref{fused weights L}. We refer to the directed path stochastic vertex model with choice of weights $L_{(x,t)}=\mathsf{L}^{(J)}_{\xi_x u_t, s_x}$ as the \emph{fused Stochastic Higher Spin Six Vertex Model}.

\begin{table}[t]
  \centering
  \begin{tabular}{l||c|c|c|r}
     & \includegraphics{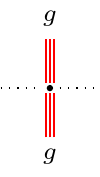}  &  \includegraphics{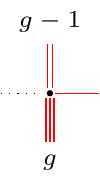} & 
     \includegraphics{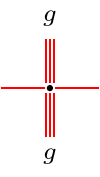} & \includegraphics{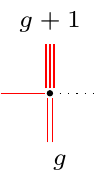}\\
    \hhline{-----}
    $\mathsf{L}_{\xi_x u_t,s_x}$ & $\frac{1-q^{g} \xi_x s_x u_t}{1-s_x \xi_x u_t}$ & $\frac{-s_x \xi_x u_t+ q^{g} \xi_x s_x u_t}{1-s_x \xi_x u_t}$ & $\frac{- s_x \xi_x u_t+ s_x^2q^{g}}{1-s_x \xi_x u_t}$ & $\frac{1- s_x^2q^{g}}{1-s_x \xi_x u_t}$\\
  \end{tabular}
  \caption{\small{In the top row we see all acceptable configurations of paths entering and exiting a vertex; below we reported the corresponding stochastic weights $\mathsf{L}_{\xi_x u_t,s_x}(m, j |\  m', j')$. Spin and inhomogeneity parameters $s$ and $\xi$ will depend on the $x$ coordinate of a vertex, whereas spectral parameters $u$ will depend on the $t$ coordinate.}}
  \label{weights table}
\end{table}

Boundary conditions for the model we study in this paper are given in terms of a special family of probability distribution, that we call $q$-\emph{negative binomial}. A random variable $X$ is said to be $q$-negative binomial if its probability mass function is expressed as
\begin{equation} \label{q negative binomial}
    \mathbb{P}(X=n) = p^n \frac{(b;q)_n}{(q;q)_n} \frac{(p;q)_\infty}{(p b ; q)_\infty}, \qquad \text{for all } n \in \mathbb{Z}_{\geq 0},
\end{equation}
for some parameters $p,b$ and we use the notation $X \sim q$NB$(b,p)$. In case $p,b$ belong to the interval $[0,1)$, then $X$ is supported on $\mathbb{Z}_{\geq 0}$, whereas if $p<0$ and $b=q^{-L}$ for some positive integer $L$, then $X$ only takes values on the set $\{0, \dots, L\}$. We define the \emph{double sided} $q$-\emph{negative binomial Higher Spin Six Vertex Model} as the (fused) Higher Spin Six Vertex Model where boundary random variables $\mathsf{m}_2^0, \mathsf{m}_3^0, \dots, \mathsf{j}_1^1, \mathsf{j}_1^2, \dots$ are independently distributed with laws
\begin{equation} \label{bc two sid q nb}
    \mathsf{m}_x^0 \sim q \text{\text{NB}}(s_x^2 , v/(\xi_x s_x)), \qquad \mathsf{j}_1^t \sim q \text{\text{NB}}(q^{-J} , q^J u_t \mathpzc{d} ),
\end{equation}
for parameters $\mathpzc{d},v$ satisfying $\mathpzc{d}>0$ and 
\begin{equation} \label{condition v}
0 \leq v < \inf_x \{ \xi_x s_x \}.
\end{equation} 
A special case of conditions \eqref{bc two sid q nb} will be given setting $v= \mathpzc{d}$ and we will refer to the model with this choice as the \emph{stationary Higher Spin Six Vertex Model}.

Exact properties of the model will be described by means of the \emph{height function} $\mathcal{H}$, an observable that we define as
\begin{equation}\label{height correct}
\mathcal{H}(x,t)= - \mathsf{m}_2^0 - \dots - \mathsf{m}_{x}^0 + \mathsf{j}_{x}^{1} + \cdots + \mathsf{j}_{x}^{t}.
\end{equation}
Pictorially we might look at paths in the generic configuration $\mathfrak{p}$ as the contours of an irregular staircase with steps one unit length tall, that we climb down as we move in the down-right direction. By centering the value of the height of the staircase $\mathcal{H}(1,0)=0$, we see that \eqref{height correct} describes the vertical displacement that we encounter going from $(1,0)$ to $(x,t)$ (see Figure \ref{fig: height centered}). The main results of this this paper is the exact description of the one point distribution of $\mathcal{H}$ under a restricted class of $q$-negative binomial boundary conditions, that includes the stationary case.

\begin{figure}[t]
\centering
\includegraphics[scale=1.3]{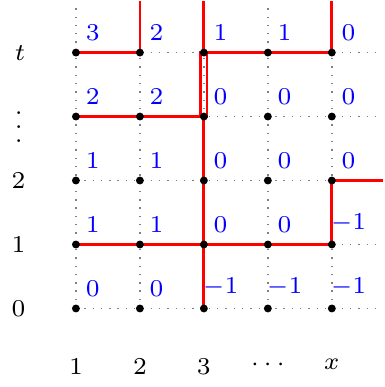}
\caption{\small{A possible set of up right paths in the Higher Spin Six Vertex Model. The numbers in blue reported next to each vertex are the values of the height function $\mathcal{H}$ defined in \eqref{height correct}.}}
\label{fig: height centered}
\end{figure}

\subsection{Methods}
In this paper we study the double sided $q$-negative binomial Higher Spin Six Vertex Model by expressing the distribution of the height function $\mathcal{H}$ in terms of certain $q$-Whittaker measures. These measures arise from the formalism of the Macdonald processes and are known to describe the joint law of a class of dynamics for the $q$-TASEP \cite{BorodinCorwin2014Mac}, \cite{MatveevPetrov2015}. Under a more restricted set of boundary conditions analogies between the Higher Spin Six Vertex Model and the $q$-TASEPs were studied first in \cite{OrrPetrov2016}, where authors provided a coupling between the height function and the position of tagged particles along \emph{time-like paths}, that are up right paths in $\Lambda_{0,1}$ with the vertical direction read as the time for the $q$-TASEP dynamics. Here we extend part of Orr and Petrov's argument \cite{OrrPetrov2016} to the case of $q$-negative binomial boundary conditions.

Random boundary conditions we will consider for our model will be produced through a fusion procedure. In the case of the Stochastic Six Vertex Model, such techniques were adopted first in \cite{Aggarwal2016FluctuationsASEP} to generate independent Bernoulli random entries from the horizontal axis. In this regard we find that the $q$-negative binomial boundary conditions represent the natural generalization of the independent Bernoulli ones for a higher spin version of the model.

Information describing the probability distribution of $\mathcal{H}$ are encoded in the $q$-Laplace transform (see Appendix \ref{appendix qfunctions}) 
\begin{equation} \label{eq: q Laplace H}
    \mathbb{E}_{\text{HS}(v, \mathpzc{d})}\left( \frac{1}{(\zeta q^{\mathcal{H}(x,t)};q)_\infty} \right),
\end{equation}
where the subscript HS$(v, \mathpzc{d})$ refers to boundary conditions \eqref{bc two sid q nb}. Ever since the introduction of the $q$-Whittaker processes in \cite{BorodinCorwin2014Mac}, the study of $q$-Laplace transforms such as \eqref{eq: q Laplace H} has proven to be successful in order to derive rigorous asymptotic analysis in a number of examples, especially for cases corresponding to step boundary conditions $(v=0)$ \cite{FerrariVeto2015TWLimit}, \cite{Barraquand2015qTASEP}. In the language of the $q$-TASEP, determinantal structures for the $q$-Laplace transform of the probability density of a tagged particle $y_x$ were obtained in \cite{q-TASEPtheta} employing elliptic analogs of the Cauchy determinants after considering the expansion
\begin{equation*}
     \mathbb{E} \left( \frac{1}{(\zeta q^{y_x +x};q)_\infty} \right) = \sum_{l \in \mathbb{Z}} \mathbb{P} (y_x + x =l)  \frac{1}{(\zeta q^{l};q)_\infty}.
\end{equation*}
This is the strategy we follow here, bringing computations and asymptotic analysis of \cite{q-TASEPtheta} to the more general setting of the Higher Spin Six Vertex Model.

A different approach and also a more established one in the context of Integrable Probability, would have been that of studying the $q$-Laplace transform \eqref{eq: q Laplace H} by means of a $q$-moments expansion as done in \cite{Aggarwal2016FluctuationsASEP}, \cite{BoCoFeVe2015}. Although in the case of random boundary conditions such type of expansion would be ill posed as the height $\mathcal{H}$ assumes also negative values and its $q$-moments diverge, one can still make sense of it after employing certain analytic continuation in parameters governing the initial measure. 

An interesting finding is that, in terms of exact determinantal formulas, these two strategies produce similar yet different results. In particular the Cauchy determinants approach offers a Fredholm determinant representation for \eqref{eq: q Laplace H}, where the kernel has finite rank and it admits a biorthogonal expansion reminiscent of those found in random matrix theory for the study of gaussian ensembles. Due to the finiteness of the rank of the kernel, formulas we obtain are rather easy to manipulate in concrete examples, when for instance one wants to compute the $q$-Laplace transform \eqref{eq: q Laplace H} at a vertex $(x,t)$ reasonably close to the origin. On the other hand, through a $q$-moments approach one would obtain a representation of \eqref{eq: q Laplace H} in terms of a Fredholm determinant of an infinite rank operator. Such procedure, which often requires a certain amount of guesswork for the choice of the kernel, has nevertheless shown its advantages too as expressions one gets are amenable to rigorous asymptotics with relatively weak assumptions on parameters defining the system. The question on how to move from one representation to the other is indeed an interesting one and it remains open, although we plan to address this issue in a forthcoming paper \cite{ImaMucSasASEPqTASEP}.

In order to establish Baik-Rains asymptotic fluctuations for the height function we employ techniques closer to works on the TASEP with deterministic initial conditions \cite{BFS2007}, than to more recent ones \cite{FerrariVeto2015TWLimit}, \cite{Barraquand2015qTASEP} on more general models. The Baik-Rains limit, compared to GOE or GUE Tracy-Widom limits often requires an extra amount of care, conceptually because the procedure involves the exchange of a limit and of a derivative sign. Throughout Section \ref{subsection the baik rains limit} we take care of such technical difficulties by a detailed analysis of the remainder terms in the asymptotic limit that presents some novel aspects.

\subsection{Results} \label{section results}
Our first result is a characterization of the stationary Higher Spin Six Vertex Model, which we recall was defined above as the model with double sided $q$-negative binomial boundary conditions \eqref{bc two sid q nb} with parameters $v=\mathpzc{d}$. We find that the probability measures with these particular choices of boundary conditions are the only one to satisfy a certain translation invariance that we call \emph{Burke's property}.
\begin{Def} \label{def: Burke property}
We say that a probability measure $\mathcal{P}$ on $\mathfrak{P}(\Lambda_{1,0})$ satisfies the \emph{Burke's property} if there exist families $\{ P^{(x)} \}_{x \geq 2 }, \{ \tilde{P}^{(t)} \}_{t \geq 1 }$ of probability distributions such that, for all $(x',t')\in \mathbb{Z}_{\geq 1} \times \mathbb{Z}_{\geq 0}$ we have, independently 
\begin{gather*}
    \mathsf{m}_{x'+k}^{t'} \sim P^{(x'+k)}, \qquad \text{for all } k \geq 1, \\
    \mathsf{j}_{x'}^{t'+k} \sim \tilde{P}^{(t'+k)}, \qquad \text{for all } k \geq 1.
\end{gather*}
\end{Def}
In words, the Burke's property states that, for any choice of a vertex $(x',t')$, the Higher Spin Six Vertex Model on the shifted lattice $\Lambda_{x',t'}=(\mathbb{Z}_{\geq x'} \times \mathbb{Z}_{\geq t'}) \setminus (x',t')$, obtained as a marginal process of the model on $\Lambda_{1,0}$, possesses boundary conditions that are always described by the same family of probability laws 
$ P^{(x)} , \tilde{P}^{(t)}$ after appropriately shifting indices $x,t$ (see Figure \ref{figure Burke's property}).

\begin{figure}[t]
\centering
\includegraphics[scale=1.3]{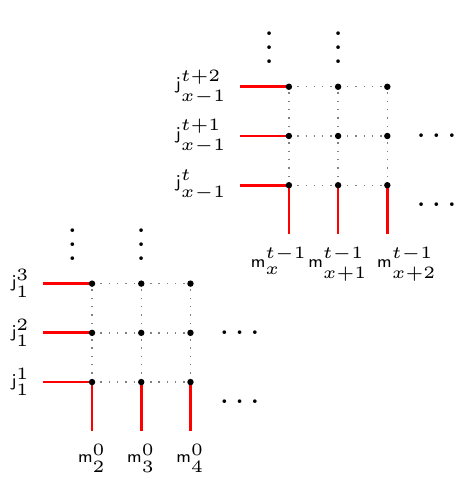}
\caption{
\small{An illustration of the Burke's property.}
} \label{figure Burke's property}
\end{figure}

\begin{Prop} \label{prop translation invariant}
The Higher Spin Six Vertex Model on the lattice $\Lambda_{1,0}$ satisfies the Burke's property if and only if boundary conditions are taken as
\begin{equation} \label{bc Burke}
\mathsf{m}^{0}_{x} \sim q\text{\emph{NB}}( s^2_x , \mathpzc{d}/(\xi_x s_x) ), \qquad \mathsf{j}_{1}^{t} \sim q\text{\emph{NB}}( q^{-J} , q^{J} \mathpzc{d} u_t ),
\end{equation}
independently of each other, for all $x\geq 2, t\geq 1$, where $\mathpzc{d}$ is a parameter that meets the condition
\begin{equation} \label{bound curvy d}
0 \leq \mathpzc{d} < \inf_{x} \{ \xi_x s_x  \}.
\end{equation}
\end{Prop}

Result of Proposition \ref{prop translation invariant} can be compared to the well known characterization of translation invariant measures for a general class zero range processes on $\mathbb{Z}$ obtained in \cite{andjel1982}, that states that these come in the form of factorized measures. In our case we define a notion of translation invariance for factorized measures on vertex models inhomogeneous both in the spatial and time coordinates and subsequently we describe the entire family of measures satisfying such properties.

The next result we present offers a Fredholm determinant representation of the $q$-Laplace transform \eqref{eq: q Laplace H} in a model with double sided $q$-negative binomial boundary conditions with parameters $v < \mathpzc{d}$. In the following we refer to a $q$-Poisson random variable with parameter $p$, in short $q$Poi($p$), as a $q$-negative binomial \eqref{q negative binomial} with parameters $0<p<1$ and $b=0$. For the sake of the following analytical statements we will assume that parameters $\Xi, \mathbf{S}$ are placed in such a way that 
\begin{equation} \label{eq: placements Xi S}
q \sup_{i} \{ \xi_i s_i \} < \mathpzc{d} < \inf_{i } \{ \xi_i s_i \} \leq \sup_{i } \{ \xi_i s_i \} < \inf_{i } \{ \xi_i / s_i \}.
\end{equation}

\begin{theorem} \label{theorem: q laplace shifted introduction}
Consider a double sided $q$-negative binomial Higher Spin Six Vertex Model on $\Lambda_{1,0}$ with parameters $v, \mathpzc{d}, \Xi, \mathbf{S}$ satisfying \eqref{eq: placements Xi S} and $v < \mathpzc{d}$. Also set $\mathsf{m}$ to be an independent $q$-Poisson random variable of parameter $v/\mathpzc{d}$. 
Then we have
\begin{equation} \label{fredholm determinant introduction}
\mathbb{E}_{\emph{HS}(v, \mathpzc{d}) \otimes \mathsf{m}} \left( \frac{1}{(\zeta q^{\mathcal{H}(x,t)-\mathsf{m}};q)_\infty}  \right) = \det ( \mathbf{1} - f K)_{l^2(\mathbb{Z})}.
\end{equation}
The kernel $f K$ on the right hand side is given in \eqref{f}, \eqref{kernel} and it is finite dimensional.
\end{theorem}

In order to employ the statement of Theorem \ref{theorem: q laplace shifted introduction} for the study of the stationary model one needs to remove from expression \eqref{fredholm determinant introduction} the contribution of the $q$-Poisson random variable $\mathsf{m}$, that becomes a.s. infinite in the limit $v \to \mathpzc{d}$. In Section \ref{section regularizations} it is shown how such decoupling procedure provides us with determinantal formulas which describe the height function $\mathcal{H}$ in the case of $q$-negative binomial boundary conditions with parameters $v,\mathpzc{d}$ satisfying $qv < \mathpzc{d}< v/q$.
This is in fact the most general range of boundary conditions we will state exact results for.
\begin{theorem} \label{theorem q laplace double sided}
Consider the double sided $q$-negative binomial Higher Spin Six Vertex Model with parameters $v, \mathpzc{d}, \Xi, \mathbf{S}$ satisfying \eqref{eq: placements Xi S} and
\begin{equation} \label{condition v d analytical continuation}
qv < \mathpzc{d} <v/q. 
\end{equation}
Then we have
\begin{equation} \label{eq: q laplace H expansion}
    \mathbb{E}_{\emph{HS}(v,\mathpzc{d})}\left( \frac{1}{(\zeta q^{\mathcal{H}(x,t)};q)_\infty} \right) = \frac{1}{(q v /\mathpzc{d};q)_\infty} \sum_{k \geq 0} \frac{(-1)^k q^{\binom{k}{2}}}{(q;q)_k} \left( \frac{v}{\mathpzc{d}} \right)^k V_{x;v,\mathpzc{d}}(\zeta q^{-k}),
\end{equation}
where the function $V_{x;v, \mathpzc{d}}$ is defined as
\begin{equation*}
    V_{x;v, \mathpzc{d}}(\zeta) = \frac{1}{1 - v/\mathpzc{d}} \det (\mathbf{1} - f K)_{l^2(\mathbb{Z})}.
\end{equation*}
\end{theorem}

Remarkably, expression \eqref{eq: q laplace H expansion} is amenable to rigorous asymptotic analysis and in particular we pursue the case when the model is stationary. By taking the limit $v \to \mathpzc{d}$, the expression of function $V_{x;v, \mathpzc{d}}$ takes a rather complicated form, stated below in \eqref{V} and we devote Section \ref{subsection the baik rains limit} to establish its behavior in the large $x$ limit. As already mentioned in Section \ref{section: Kpz universality, integrability and initial conditions}, when the measure is stationary, the characteristic 3:2:1 scaling of the model is only observed along a specific direction, which is usually referred to as the \emph{characteristic line}. The scaling of the height function along this line is conjectured to be universal and it is described by the KPZ scaling theory \cite{SpohnScaling}, that we explain briefly in Section \ref{section: kpz scaling}.

For the stationary Higher Spin Six Vertex Model we now want to give the exact expression of scaling parameters defining the characteristic line and the expected behavior of the height function $\mathcal{H}$. We make use of $q$-polygamma type functions $\polygamma_k$ defined in Appendix \ref{appendix qfunctions}. For non-negative integers $k$ consider the functions
\begin{equation} \label{eq: a_k h_k}
    a_k(\mathpzc{d})=\polygamma_k(q^{J} u \mathpzc{d}) - \polygamma_k( u \mathpzc{d}), \qquad h_k(\mathpzc{d})=\frac{1}{x} \sum_{y=2}^x \left( \polygamma_k( \mathpzc{d}/(\xi_y s_y)) - \polygamma_k(\mathpzc{d} s_y / \xi_y)\right)
\end{equation}
and, depending on the parameter $\mathpzc{d}$, define the quantities
\begin{equation} \label{eq: k0 eta0 gamma0}
    \kappa_0 = \frac{h_1(\mathpzc{d})}{a_1(\mathpzc{d})}, \qquad \eta_0 = \kappa_0 a_0(\mathpzc{d}) - h_0(\mathpzc{d}), \qquad \gamma = -\left(\frac{1}{2}( \kappa_0 a_2(\mathpzc{d}) - h_2(\mathpzc{d}) )   \right)^{1/3}.
\end{equation}
We assume that the functions $h_k$ always converge in the large $x$ limit and we refer to the curve $(x, \kappa_0 x)$ as the \emph{characteristic line} of the stationary Higher Spin Six Vertex Model.
For random growth models usually the characteristic line is expressed as a function of the time $t$, rather than of the coordinate $x$, but in our case, since the system exhibits spatial inhomogeneities we find more natural to adopt the notation $(x, \kappa_0 x)$. The parameter $\eta_0$ multiplied by $x$ is readily understood as the expectation $\mathbb{E}(\mathcal{H}(x, \kappa_0 x))$, whereas $\gamma$ will be used to describe the size of the characteristic fluctuations of $\mathcal{H}$ around $\eta_0$. By slightly perturbing quantities $\kappa_0, \eta_0$ we can analyze the asymptotic behavior of $\mathcal{H}$ in a region of size $x^{2/3}$ of the characteristic line. For this we extend the definitions given in  \eqref{eq: k0 eta0 gamma0} setting
\begin{equation}\label{k perturbed}
    \kappa_\varpi = \kappa_0 + \frac{h_2 a_1 - h_1 a_2}{a_1^2} \frac{\varpi}{\gamma x^{1/3}} 
\end{equation}
\begin{equation}\label{eta perturbed}
    \eta_\varpi = \eta_0 + \frac{a_0 (h_2 a_1 - h_1 a_2) }{a_1^2}
 \frac{\varpi}{\gamma x^{1/3}} + \frac{h_2 a_1 - h_1 a_2}{a_1} \frac{\varpi^2}{\gamma^2 x^{2/3}},
\end{equation}
where $\varpi$ is a real number parameterizing the displacement from the characteristic line and functions $a_k, h_k$ are evaluated at $\mathpzc{d}$. We come now to state our main result.
\begin{theorem} \label{theorem: baik rains limit H}
Consider the stationary Higher Spin Six Vertex Model with parameters $q, \mathpzc{d}, \Xi, \mathbf{S}$ fulfilling conditions stated in Definition \ref{conditions on parameters}. Then, for any real numbers $\varpi, r$ we have
\begin{equation} \label{eq: baik rains limit intro}
    \lim_{x\to \infty} \mathbb{P}_{\emph{HS}(\mathpzc{d},\mathpzc{d})} \left( \frac{\mathcal{H}(x, \kappa_\varpi x ) - \eta_\varpi x}{\gamma x^{1/3}} < -r \right) = F_\varpi (r),
\end{equation}
where $F_\varpi (r)$ is the Baik-Rains distribution presented in Definition \ref{Baik-Rains definition}.
\end{theorem}

Assumptions on parameters made in the statement of Theorem \ref{theorem: baik rains limit H} are technical and they substantially require $q$ to be sufficiently close to zero. These arise while establishing the steep descent property of integration contours of the integral kernel $K$ (see Appendix \ref{appendix contours}). Such conditions can be considerably weakened employing certain determinant preserving transformations of the kernel that involve deforming integration contours to regions containing poles of the integrand function. As such procedures are rather technical, we postpone their description to a future work   \cite{ImaMucSasASEPqTASEP} and for the sake of this paper we stick to the small $q$ assumption.

Techniques used in the proof of Theorem \ref{theorem: baik rains limit H} can be employed also to establish Tracy-Widom asymptotic fluctuations of the height function $\mathcal{H}$ when the model has step Bernoulli boundary conditions ($v=0$). This result was already proved in \cite{OrrPetrov2016} using a certain matching between $q$-Whittaker measures and Schur measures.

Additional results we obtain are stated in Section \ref{section specializations} and they are adaptations of Fredholm determinant formulas \eqref{fredholm determinant introduction}, \eqref{eq: q laplace H expansion} and of the universal limit \eqref{eq: baik rains limit intro} to two of the main degenerations of the Higher Spin Six Vertex Models, the $q$-Hahn TASEP and the Exponential Jump Model.
\subsection{Outline of the paper}
In Section \ref{section HS} we describe some further properties of the Higher Spin Six Vertex Model, that were left out in Section \ref{subs: the Model}. Especially we recall nested contour integral formulas for $q$-moments of the model with step boundary conditions. In Section \ref{section qwhittaker} we recall the definition and main properties of the $q$-Whittaker process. In Section \ref{About the initial conditions} we proof the Burke's property of the stationary Higher Spin Six Vertex Model and we establish its integrability. In Section \ref{section matching and fredholm} we employ elliptic determinant computations from \cite{q-TASEPtheta} in order to compute the $q$-Laplace transform of the probability mass function of the height function in the case of double sided $q$-negative binomial boundary conditions. In Section \ref{section time asymptotics} we specialize determinantal expression obtained in Section \ref{section matching and fredholm} to the stationary model and we compute the asymptotics of the one point distribution of the height function along the critical line. Finally, in Section \ref{section specializations} we consider the main degenerations of the Higher Spin Six Vertex Model and and we establish determinantal formulas and Baik-Rains fluctuations for these models.
\subsection{Acknowledgements}
M.M. is very grateful to Patrik Ferrari and Alexander Garbali for helpful discussions. The work of T.S. is supported by JSPS KAKENHI Grant Numbers JP15K05203, JP16H06338, JP18H01141, JP18H03672. The work of T.I. is supported by JSPS KAKENHI Grant Number
JP16K05192.

\section{Stochastic Higher Spin Six Vertex Model} 
\label{section HS}

In this section we give a review on the Higher Spin Six Vertex Model. We take the chance to fix some notations and recall major results which will be used throughout the rest of the paper.
\subsection{Directed paths picture} \label{subsection directed paths picture}
A description of the Higher Spin Six Vertex Model as an up right directed path ensemble in the lattice $\Lambda_{1,0}$ was given in Section \ref{subs: the Model}. The choice of the set $\Lambda_{1,0}$ was made only to keep out notation consistent with that introduced in previous works \cite{CorwinPetrov2016HSVM},\cite{BorodinPetrov2016HS6VM} and we could extend the notion of the model to the generic lattice $\Lambda_{x',t'}$ with boundary $\partial \Lambda_{x',t'}$, defined as
\begin{equation*}
    \Lambda_{x',t'} = \left( \mathbb{Z}_{\geq x'} \times \mathbb{Z}_{\geq t'} \right) \setminus (x',t') \qquad \text{and} \qquad \partial \Lambda_{x',t'} = \left( \mathbb{Z}_{\geq x'+1} \times \{t'\} \right) \cup \left( \{x'\} \times \mathbb{Z}_{\geq t'+1} \right),
\end{equation*}
for a generic $(x',t')$ in $\mathbb{Z}\times \mathbb{Z}$. When this is the case boundary conditions are given specifying the laws of $\mathsf{m}_{x'+1}^{t'}, \mathsf{m}_{x'+2}^{t'}, \dots, \mathsf{j}_{x'}^{t'+1}, \mathsf{j}_{x'}^{t'+2}, \dots$ and the measure depends on parameters $q, \{ u_{t} \}_{t>t'}, \{\xi_{x} , s_{x} \}_{x > x'}$. In Section \ref{subs: the Model} such parameters were assumed to satisfy condition \eqref{HS6VM parameters}, so to guarantee the stochasticity of vertex weights $\mathsf{L}_{\xi_{x} u_{t}, s_{x}}$ of Table \ref{weights table} and the same assumption is made now. Although in this paper we will not investigate range of parameters different than \eqref{HS6VM parameters}, we want to point out that there exist also different conditions that would make all $\mathsf{L}_{\xi_{x} u_{t}, s_{x}}$ non-negative quantities. The full list of stochasticity conditions is given imposing for all $x,t$ one of the following:
\begin{enumerate}

\item $0<q<1$, $-1<s_x< 1$ and $s_x \xi_x u_t <0$, \label{HS6VM parameters list}

\item $0<q<1$, $q^{-G}=s_x^2 \leq s_x \xi_x u_t$  with $G \in \mathbb{Z}_{>0}$ and $g =0,1,\dots, G$, \label{6VM parameters}

\item $-1<q<0$, $q^{-1} \leq \xi_x s_x u_t \leq 0$ and $\xi_x s_x u_t \leq s_x^2 \leq \min \big( 1, q^{-1} \xi_x s_x u_t \big)$,

\item $q>1$, $0 \leq s_x \xi_x u_t \leq s_x^2= q^{-G}$, with $G\in \mathbb{Z}_{>0}$ and $g =0,1,\dots, G$,  \label{q^-1 tasep conditions}
\end{enumerate}
where the integer $g$ appearing in \ref{6VM parameters}, \ref{q^-1 tasep conditions} is the number of vertical path entering the vertex as in Table \ref{weights table}. Choice \ref{HS6VM parameters list} corresponds to \eqref{HS6VM parameters}, where in \eqref{HS6VM parameters}, with no loss of generality, we fixed the signs of $u_t,\xi_x,s_x$, as the weights $\mathsf{L}$ only depend on the product $\xi_x s_x u_t$ and $s_x^2$. Choice \ref{6VM parameters} produces the Six Vertex Model and the only path configurations with positive measure are those where both horizontal and vertical edges are crossed at most by one path. A list of stochasticity conditions analogous to that presented above appeared in \cite{CorwinPetrov2016HSVM}, where authors used a slightly different notation.

In Section \ref{subs: the Model} we described the double sided $q$-negative binomial boundary conditions, introduced in \eqref{bc two sid q nb} for the model in the lattice $\Lambda_{1,0}$. Clearly the same definition can be adapted also for a Higher Spin Six Vertex Model in $\Lambda_{x',t'}$. A very special case of boundary conditions is obtained when we set $J=1, v=0$ and $\mathpzc{d}=\infty$, obtaining
\begin{equation} \label{step boundary bc}
\mathsf{j}_{x'}^t=1 \quad \text{a.s.}, \qquad \mathsf{m}_x^{t'}=0 \quad \text{a.s.},
\end{equation}
for all $t>t'$, $x>x'$. This is to say that from the horizontal axis no path originates and at each vertex $(x',t)$ exactly one path enters the system. We refer to \eqref{step boundary bc} as \emph{step boundary conditions}, as they are an analogous version of the step initial conditions for simple exclusion processes on the infinite lattice, where vertical segments of paths $ \mathsf{m}_x^t $ are interpreted as gaps between consecutive particles when the time of the system is $t$. (more in Section \ref{subsection observables HS6VM}).

Other relevant choices of boundary conditions are given setting $v=0$, but leaving $\mathpzc{d}$ as a finite quantity. In this case too, paths can only enter the system from the vertical axis and they do so randomly with $q$-negative binomial distribution of parameters $(q^{-J}, q^J \mathpzc{d} u_t).$ We refer to these as \emph{step q-negative binomial boundary conditions} and when $J=1$ we also use the name \emph{step Bernoulli boundary conditions}, considered in \cite{Aggarwal2016FluctuationsASEP}, \cite{OrrPetrov2016}.

So far we considered a model where at each vertex the weight $\mathsf{L}$ was depending both on the $x$ and $t$ coordinate. A slight simplification is given by choosing parameters $u_t, s_x$ to be constant numbers $u, s$ for all $x,t$ and to set the inhomogeneity parameters $\xi_x = 1$ for all $x$. This specialization takes the name of \emph{Homogeneous Higher Spin Six Vertex Model} and it was considered in the original paper \cite{CorwinPetrov2016HSVM}.

\subsection{One line dynamical picture} \label{subsection one line dynamical picture} 
In this Section we focus on the Higher Spin Six Vertex Model with $J=1$. A possible alternative to presenting it as a static ensemble of directed paths is to interpret the arrangements of paths along each row of vertices as a dynamical process. As in Section \ref{subs: the Model}, denote with $m_x^t$ the number of paths exiting vertex $(x,t)$ from above in a particular realization of the model. The information contained in the sequence $\{ m_x^t \}_x$ can be encoded in the symbol
\begin{equation} \label{lambda partition}
\lambda(t) = \prod_{x} x^{m_x^t} .
\end{equation}
When $\sum_{x} m_x^t$ is a finite number, or equivalently, when only finitely many paths populate the region with ordinate $t' \leq t$, then $\lambda(t)$ can be thought as a signature written in multiplicative notation.

If $\lambda(t+1)$ is the symbol generated by the path configuration on vertical edges joining vertices $(x,t+1)$ and $(x,t+2)$ for $x \geq 1$, we say that $\lambda(t)$ transitions to $\lambda(t+1)$ and we want, at least formally, to describe the probability of such transition to take place. For this assume first that the probability of the event $\{\mathsf{j}_{x}^t = 1, \text{ eventually for }x \gg 0 \}$ is zero, which is to say that paths traveling on horizontal lines will almost surely turn upward. To ensure this condition we can take parameters $\xi, s, u$ such that
\begin{equation} \label{no horizontal lines}
\sup_{ x, t}\mathsf{L}_{\xi_x u_t, s_x} (0, 1|\ 0, 1) < 1.
\end{equation}
In case we consider a model in the lattice $\Lambda_{0,0}$, conservation law \eqref{conservation law} implies that, when we specify the number $j_0^{t+1}$ of paths emanating the boundary vertex $(0,t+1)$, there exist at most one choice of $\{j_x^{t+1}\}$ such that $\lambda(t)$ transitions to $\lambda(t+1)$ and the probability of such transition is formally given by
\begin{equation} \label{stochastic matrix X}
\mathfrak{X}_{u_{t+1}}(j_0^{t+1 }; \lambda(t) \rightarrow \lambda(t+1)) = \mathbb{P}(\mathsf{j}_0^{t+1} = j_0^{t+1}) \prod_{x \geq 1} \mathsf{L}_{\xi_x u_{t+1}, s_x} (m_x^{t}, j_{x-1}^{t+1}|\ m_x^{t+1}, j_x^{t+1}).
\end{equation}
Operator $\mathfrak{X}$ takes the name of \emph{transfer operator} and one can possibly define it rigorously through an inverse limit procedure of its action on finite path configurations as done in \cite{CorwinPetrov2016HSVM}, Definition 2.6. For example, when $\lambda(t)$ and $\lambda(t+1)$ describe configurations of a finite number of paths, expression \eqref{stochastic matrix X} is well posed, as the infinite product on weights $\mathsf{L}$ on the right hand side contains almost surely only finitely many factors different from $\mathsf{L}_{\xi u, s} (0,1|\ 0,1)=1$.

\begin{figure}[ht]
\centering
\includegraphics{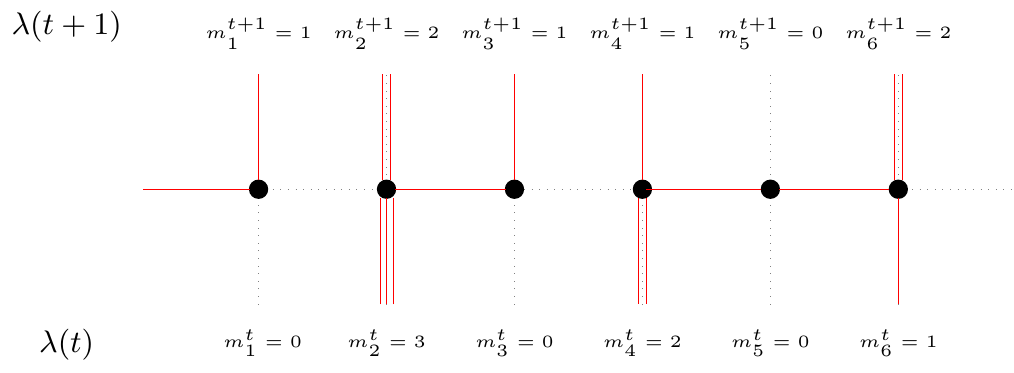}
\caption{\small{
As the paths cross the horizontal line of vertices the signature $\lambda(t)=1^0 2^3 3^0 4^2 5^0 6^1$ transitions to $\lambda(t+1)=1^12^23^14^15^06^2$.}}
\label{transition single line Picture}
\end{figure}

We can possibly remove the dependence of $\lambda(t+1)$ from the boundary value $j_0^{t+1}$. To do so we simply need to set $\mathsf{m}_0^1=\infty$ a.s., which is to say that each vertex $(1,t)$ is vertically crossed by infinitely many paths. From Table \ref{weights table} we see that, assuming \eqref{HS6VM parameters}, when $g=\infty$, we have
\begin{gather}
\mathsf{L}_{\xi_1 u_t, s_1}(\infty,0|\ \infty,0) = \mathsf{L}_{\xi_1 u_t, s_1}(\infty,1|\ \infty,0) = \frac{1}{1-\xi_1 s_1 u_t} \label{limit L J=1 g=inf 1}\\
\mathsf{L}_{\xi_1 u_t, s_1}(\infty,0|\ \infty,1) = \mathsf{L}_{\xi_1 u_t, s_1},(\infty,1|\ \infty,1) = \frac{-\xi_1 s_1 u_t}{1-\xi_1 s_1 u_t} \label{limit L J=1 g=inf 2}.
\end{gather}
This means that the choice $\mathsf{m}_0^1= \infty$ a.s. implies that random variables $\{ \mathsf{j}_1^t \}_{t \geq 1}$ become mutually independent Bernoulli distributed as
\begin{equation} \label{j bernoulli}
\mathsf{j}_1^t \sim \text{Ber}\left( \frac{-\xi_1 s_1 u_t}{1-\xi_1 s_1 u_t} \right) \qquad \text{for all }t \geq 1.
\end{equation}
By looking at configurations of paths in the resticted lattice $\Lambda_{1,0}$, \eqref{j bernoulli} can be regarded as a boundary condition, so that setting $\mathsf{m}_x^0=0$ a.s. for each $x \geq 2$ we produce the step Bernoulli boundary conditions, considered in above (see Figure \ref{figure step Bernoulli}.a).  

\subsection{Fused transfer operator \texorpdfstring{$\mathfrak{X}^{(J)}$}{TEXT}}  \label{subsection fused dynamics}
In Section \ref{subsection one line dynamical picture} we considered the unfused Higher Spin Six Vertex Model and each horizontal edge of the lattice could be crossed by no more than one path. We see now how it is possible to exploit combinatorial properties of weights $\mathsf{L}_{\xi_x u_t, s_x}$ in order to take away this restriction. The strategy consists of collapsing together multiple horizontal lines of vertices.

Suppose we aim to allow up to $J$ paths to travel an edge horizontally. For a given probability distribution $P$ on $\{ 0,1 \}^J$, consider the quantity
\begin{equation} \label{J column}
\sum_{\substack{\mathbf{h}, \mathbf{h}' \in \{0,1\}^J \\ 
|\mathbf{h}|=j, |\mathbf{h}'|=j',
}} P(\mathbf{h}) \prod_{k=1}^J \mathsf{L}_{\xi_x u_k, s_x}(i_{k-1}, h_k| \ i_k, h_k'),
\end{equation}
where 
\begin{equation} \label{conservation law i j}
i_k=i_{k-1} + h_k- h_k', \qquad i_0=i,\qquad i_J=i' \qquad \text{and} \qquad i+j=i'+j'.
\end{equation}
Naturally, \eqref{J column} is the probability that in a column of $J$ vertices $i$ paths enter from below, $i'$ exit from above and, independently on their arrangement $j$ of them enter from the left and $j'$ of them exit from the right. We ask under what conditions on $P$, expression \eqref{J column} can be written in the form
\begin{equation} \label{J column collapsed}
\tilde{P}(j) \mathsf{L}^{(J)}_{\xi_x u, s_x}(i, j| \ i', j'),
\end{equation}
for some probability distribution $\tilde{P}$ on $\{1, \dots, J\}$ and some weight $\mathsf{L}^{(J)}_{\xi_x u, s_x}$. A possible answer is essentially contained in the following
\begin{Def}
A probability distribution $P$ on $\{0,1\}^J$ is said to be $q$-\emph{exchangeable} if it is of the form
\begin{equation} \label{q exchang}
P(\mathbf{\emph{h}})=\tilde{P}(|\mathbf{\emph{h}}|) \frac{q^{\sum_{k=1}^Jh_k(k-1)}}{Z_J(|\mathbf{\emph{h}}|)},
\end{equation}
where $\tilde{P}$ is a probability distribution on $\{1, \dots, J\}$, $|\mathbf{\emph{h}}|=h_1+\cdots+h_J$ and 
$$
Z_J(j)=q^{\frac{j(j-1)}{2}}\frac{(q;q)_J}{(q;q)_j (q;q)_{J-j}}.
$$
\end{Def}
In the previous definition we made use of the common notation of $q$-Pochhammer symbol $(x,q)_n$, whose definition is recalled in Appendix \ref{appendix qfunctions}.
For the next result we set the parameters $(u_1, \dots, u_J)=(u, qu, \dots, q^{J-1}u)$.
\begin{Prop} \label{prop BP q exch} [\cite{BorodinPetrov2016HS6VM}, Proposition 5.4]
Fix non-negative integers $i,i'$ and let $P$ be a $q$-exchangeable probability distribution on $\{0,1\}^J$. Then also
\begin{equation} \label{q exchang P'}
P'(\mathbf{\emph{h}}')=\sum_{\mathbf{\emph{h}}\in \{0,1\}^J} P(\mathbf{\emph{h}}) \prod_{k=1}^J \mathsf{L}_{\xi_x q^{k-1}u, s_x}(i_{k-1}, h_k| \ i_k, h_k')
\end{equation}
is $q$-exchangeable. Here numbers $i_k$ are defined as in \eqref{conservation law i j}.
\end{Prop}
A way to rephrase result of Proposition \ref{prop BP q exch} is to say that, in expression \eqref{q exchang P'}, assuming $|\mathbf{h}'|=j'$ and assuming that the probability distribution $P$ has the form \eqref{q exchang}, then we can write
\begin{equation} \label{ p tilde L^J}
P'(\mathbf{h}') = \tilde{P}(i' + j' - i) \mathsf{L}^{(J)}_{\xi_x u, s_x}(i, i'+j'-i | \ i', j') \frac{q^{\sum_{k=1}^J h_k'(k-1)}}{Z_J(j')}
\end{equation}
where the exact form of the probability weight $\mathsf{L}^{(J)}_{\xi_x u, s_x}$ can be also computed and it is given below in \eqref{fused weights L}. This can be easily exploited to collapse together $J$ different rows of vertices. Assume that in the leftmost column paths enter with $q$-exchangeable distribution $P$ (see Figure \ref{figure collapse J}), then
\begin{equation} \label{fusion J rows}
\begin{split}
\sum_{ \substack{ \mathbf{h} \in \{0,1\}^J \\ |\mathbf{h}|=j_0 } } \sum_{ \nu_1 , \nu_2, \dots, \nu_{J-1} } & P(\mathbf{h}) \mathfrak{X}_u (h_1 ; \lambda \to \nu_1) \cdots \mathfrak{X}_{q^{J-1}u} (h_J ; \nu_{J-1} \to \lambda')\\
&= \tilde{P}(j_0) \prod_{x \geq 1} \mathsf{L}^{(J)}_{\xi_x u, s_x}(m_x, j_{x-1}| \ m_x', j_x),
\end{split}
\end{equation}
where $\lambda = 1^{m_1} 2^{m_2} \cdots$ and $\lambda'=1^{m_1'} 2^{m_2'} \cdots $ are symbols indicating configurations entering and exiting respectively the bottom and the top row. In analogy with expression \eqref{stochastic matrix X}, we formally define the fused transfer operator
\begin{equation} \label{fused transfer matrix}
\mathfrak{X}^{(J)}_{u_{t+1}} ( j_0^{t+1}; \lambda(t) \to \lambda(t+1) ) = \mathbb{P}(\mathsf{j}^{t+1}_0 = j_0^{t+1}) \prod_{x \geq 1} \mathsf{L}^{(J)}_{\xi_x u_{t+1},s_x} (m_x^t, j_{x-1}^{t+1} | \ m_x^{t+1} , j_x^{t+1}),
\end{equation}
where again, numbers $j_x^{t+1}$ satisfy the conservation law
\begin{equation} \label{conservation law J}
m_x^t + j_{x-1}^{t+1} = m_x^{t+1} + j_{x}^{t+1}.
\end{equation}
Here, with a little abuse of notation, we assume that $\mathsf{j}_0^{t+1}$ is a random variable taking values in the set $\{0,1, \dots, J\}$. It is clear that, thanks to \eqref{conservation law J}, once we specify $\{ j_0^t \}_{t\geq 1}$ and $\{m_x^t\}_{x,t \geq 1}$, we automatically obtain quantities $\{j_{x}^t\}_{x,t \geq 1}$, which we interpret as occupancy numbers of collapsed horizontal edges. In this way, definition of random variables $\mathsf{m}_x^t, \mathsf{j}_x^t$ given in \eqref{m rv}, \eqref{j rv} has been extended to include the case where multiple paths can share horizontal edges.

The closed expression of weights $\mathsf{L}^{(J)}$
(\cite{BorodinPetrov2016HS6VM}, Formula 5.6 or \cite{CorwinPetrov2016HSVM}, Theorem 3.15) is considerably more involved than that of the $J=1$ case presented in Table \ref{weights table} and it is
\begin{equation}\label{fused weights L}
\begin{split}
\mathsf{L}^{(J)}_{u, s}(i_1, j_1| \ i_2, j_2)=& \mathbbm{1}_{i_1+j_1=i_2+j_2}\frac{(-1)^{i_1}q^{\frac{1}{2}i_1(i_1+2j_1-1)}u^{i_1}s^{j_1+j_2-i_2}(us^{-1};q)_{j_2-i_1}}{(q;q)_{i_2}(su;q)_{i_2+j_2}(q^{J+1-j_1};q)_{j_1-j_2}}\\
& \qquad \qquad \times \setlength\arraycolsep{1pt}
{}_4 \bar{\phi}_3\left(\begin{matrix}q^{-i_2},& q^{-i_1}, & suq^{J}, & qs/u &\\&s^2,&q^{1+j_2-i_1},
&q^{J+1-i_2-j_2} &\end{matrix} \Big| q,\ q\right).
\end{split}
\end{equation}
Here the function $
{}_4 \bar{\phi}_3$ is a particular instance of the regularized $q$-hypergeometric series defined in Appendix \ref{appendix qfunctions}, \eqref{regularized q hypergeometric series}.

\begin{figure}[ht]
\centering

\includegraphics{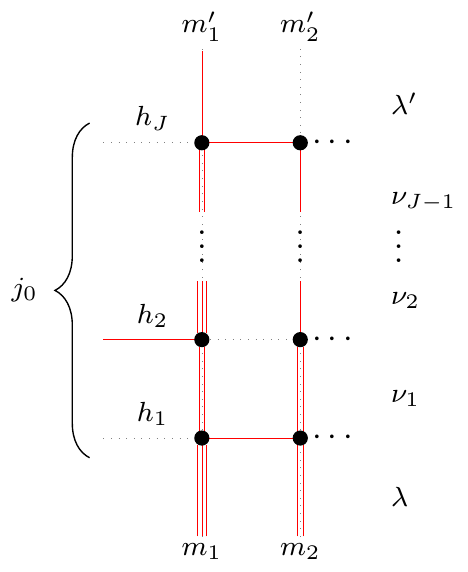}

\caption{\small{A schematic representation of the fusion of $J$ rows reported in \eqref{fusion J rows}. Here from the leftmost $J$ vertices the total number of entering paths is $j_0=h_1 + h_2+ \cdots + h_J$. Symbols $\lambda, \lambda'$ indicate the initial and final configurations, whereas intermediate ones are described by symbols $\nu_1, \dots, \nu_{J-1}$.}} \label{figure collapse J}
\end{figure}

In expression \eqref{fused weights L}, we notice the rational dependence of $\mathsf{L}^{(J)}_{u, s}$ on $q^J$, so that one can provide an analytic continuation in this parameter. Substituting $q^J$ with a generic complex number not belonging to the set  $q^{\mathbb{Z}}$ , we see that the fused weights are well defined for each choice of $j_1, j_2$, so that paths of the Higher Spin Six Vertex Model no more undergo any limitation as far as number of horizontal edges they can simultaneously cross. 

As explained in the last paragraph of subsection \ref{subsection one line dynamical picture}, we can decouple boundary random variables $\{\mathsf{j}_0^t\}_t$ and $\{ \lambda(t) \}_t$ by setting $\mathsf{m}_1^0=\infty$ a.s.. This is still true after the fusion of rows procedure and what we obtain is a fused version of the step Bernoulli boundary conditions for the Higher Spin Six Vertex Model on the restricted lattice $\Lambda_{1,0}$. In this case, to express the probability distribution of random variables $\mathsf{j}_1^t$ we need the following 

\begin{Prop} \label{prop sum bernoulli geom prog}
Consider $Y_1, \dots Y_J$ independent Bernoulli random variables respectively of mean $p/(1+p), qp/(1+qp), \dots , q^{J-1}p/(1 + q^{J-1}p)$, with $p \in \mathbb{R}_{>0}$. Then, defining $X_J=Y_1 + \cdots  + Y_J$ we have $X_J \sim q$\emph{NB}$(q^{-J}, -q^{J}p)$.
\end{Prop}

\begin{proof}
For any $k$ in the set $\{0,1, \dots, ,J\}$, we have
\begin{equation} \label{prob sum bernoulli}
    \mathbb{P}(X_J =k) = \frac{p^k}{(-p;q)_J
    }  \sum_{\substack{ \mathbf{\emph{h}} \in\{0,1\}^J:\\ |\mathbf{\emph{h}}|=k}} q^{\sum_{i=1}^J(i-1)h_i}.
\end{equation}
The sum involving powers of $q$ in the right hand side of \eqref{prob sum bernoulli} is easily expressed as 
\begin{equation} \label{sum power q}
    \sum_{\substack{ \mathbf{\emph{h}} \in\{0,1\}^J:\\ |\mathbf{\emph{h}}|=k}} q^{\sum_{i=1}^J(i-1)h_i} = (-1)^k q^{Jk} \frac{(q^{-J};q)_k}{(q;q)_k},
\end{equation}
as a result of the two different notable expansions for the $q$-Pochhammer symbol
\begin{equation*}
    (z;q)_J = \sum_{k=0}^J z^k (-1)^k \sum_{\substack{ \mathbf{\emph{h}} \in\{0,1\}^J:\\ |\mathbf{\emph{h}}|=k}} q^{\sum_{i=1}^J(i-1)h_i} \qquad \text{and} \qquad (z;q)_J = \sum_{k=0}^J (z q^J)^k \frac{(q^{-J};q)_k}{(q;q)_k}.
\end{equation*}
Combining \eqref{prob sum bernoulli} and \eqref{sum power q} we conclude the proof.
\end{proof}

As we already made clear, the fusion of rows procedure consists in taking the Higher Spin Six Vertex Model as defined in subsections \ref{subsection directed paths picture}, \ref{subsection one line dynamical picture}, specializing spectral parameters in geometric progressions of ratio $q$ and tracing out over configurations of paths sharing the same number of occupied horizontal edges at a each column of vertices. When we do so, after the choice $\mathsf{m}_1^0 = \infty$ a.s., recalling \eqref{j bernoulli} and utilizing result of Proposition \ref{prop sum bernoulli geom prog}, we obtain
\begin{equation} \label{j fused bernoulli}
\mathbb{P}( \mathsf{j}_1^t = k) = (q^J s_1 \xi_1 u_t)^k \frac{(q^{-J};q)_k}{(q;q)_k} \frac{(q^J \xi_1 s_1 u_t;q)_\infty}{(\xi_1 s_1 u_t; q)_\infty}.
\end{equation}
Therefore we can conclude that in the Higher Spin Six Vertex Model, after a fusion of $J$ rows, when in the leftmost column of vertices flow infinitely many paths, $\mathsf{j}_1^1, \mathsf{j}_1^2, \mathsf{j}_1^3, \dots$, become mutually independent $q$-negative binomial random variables respectively with parameters $(q^{-J}, q^J s_1 \xi_1 u_t)$ for $t = 1, 2, \dots$. 
The special case when $\mathsf{m}_0^x =0$ a.s. for $x \geq 2$ generates the step $q$-negative binomial boundary conditions introduced in Section \ref{subsection directed paths picture}  (see Figure \ref{figure step Bernoulli}.b).

\begin{figure}[t]
\centering
\begin{minipage}{.45 \textwidth}
\includegraphics[scale=1.2]{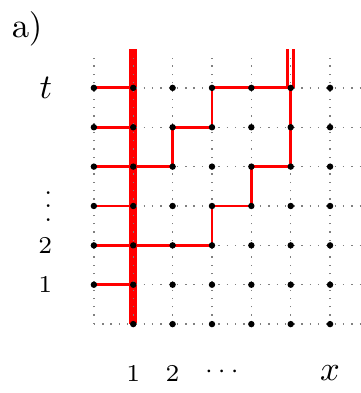}
\end{minipage}
\begin{minipage}{.45 \textwidth}
\includegraphics[scale=1.2]{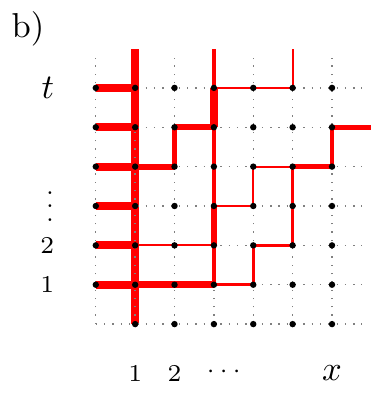}
\end{minipage}
\caption{\small{a) A possible configuration of paths in the Higher Spin Six Vertex Model on $\Lambda_{0,0}$ and $J=1$. Here boundary conditions are takes as $\mathsf{m}_1^0=\infty$ and $\mathsf{j}_0^t=1$ $\mathsf{m}_x^0=0$ a.s. for all $x,t>0$. Such choice produces the step Bernoulli boundary conditions in $\Lambda_{1,0}$.
\\
\qquad \qquad b) A possible configuration of paths in the Higher Spin Six Vertex Model on $\Lambda_{0,0}$ and $J>1$. The thickness of red traits indicates multiple occupations at edges. Here boundary conditions are takes as $\mathsf{m}_1^0=\infty$ and $\mathsf{j}_0^t=J$ $\mathsf{m}_x^0=0$ a.s. for all $x,t>0$. Such choice produces the step $q$-negative binomial boundary conditions in $\Lambda_{1,0}$.
}
}
\label{figure step Bernoulli}
\end{figure}

\vline

We close this subsection remarking that the history of expression \eqref{fused weights L} is actually longer than how it might seem from reading this brief overview of results. More complicated expressions for a quantity analogous to the transition matrix $\mathsf{L}^{(J)}$ had been known in the context of quantum integrable systems for almost three decades since \cite{KirillovReshetkin1987XXZ1}. Relatively compact expressions such as that presented in \eqref{fused weights L} became available only in more recent times after the work \cite{Mangazeev2014}. A detailed probabilistic derivation of the stochastic weight $\mathsf{L}^{(J)}$ can be found in \cite{CorwinPetrov2016HSVM}.

\subsection{Observables in the Higher Spin Six Vertex Model with step boundary conditions} \label{subsection observables HS6VM}
After discussing the fusion procedure in subsection \ref{subsection fused dynamics} we come back to the unfused model, with $J = 1$, defined in the lattice $\Lambda_{0,0}$. Recall that with step boundary conditions \eqref{step boundary bc}, no path enters from the $x$ axis and each vertex on the $t$ axis has a path entering to its left. When this is the case $\lambda(t)$ is at each level associated with a signature in the set
\begin{equation*}
\text{Sign}_t^{>0} = \{ \nu = (\nu_1 \geq \nu_2 \geq \cdots \geq \nu_t > 0) |\ \nu_i \in \mathbb{Z}_{>0} \}.
\end{equation*}
A relevant quantity for which we possess exact formulas is given in the following

\begin{Def} Consider the Higher Spin Six Vertex Model on $\Lambda_{0,0}$ with step boundary conditions. We refer to the quantity
\begin{equation} \label{height definition}
\mathfrak{h}(x,t)=\sum_{y \geq x}\mathsf{m}_y^t=\#\{j|\ \lambda_j(t) \geq x \}
\end{equation}
as the \emph{height} function $\mathfrak{h}$ at the vertex $(x,t)$ (for an example see Figure \ref{fig: height step}).
\end{Def}

\begin{figure}[ht]
\centering
\includegraphics[scale=1.2]{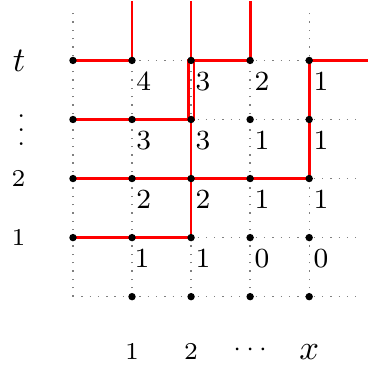}
\caption{\small{A possible set of up right paths in the Higher Spin Six Vertex Model with step boundary conditions. The numbers reported next to each vertex are the values of the height function $\mathfrak{h}$ defined in \eqref{height definition}.}}
\label{fig: height step}
\end{figure}

The two different height functions $\mathfrak{h}$ and $\mathcal{H}$, defined respectively in \eqref{height definition} and \eqref{height correct} are clearly related quantities. In particular they are connected by the trivial relation
\begin{equation*}
\mathcal{H}(x,t) = \mathfrak{h}(x+1,t).
\end{equation*} 
We like to keep their notation distinct as $\mathfrak{h}$ only refers to a model with step boundary conditions, while the definition of $\mathcal{H}$ also makes sense when paths emanate randomly from the horizontal axis.

In \cite{CorwinPetrov2016HSVM} the height function is thought of as the position of a specific particle evolving in a certain totally asymmetric exclusion process, or as the current of a totally asymmetric zero range process. Authors derive a closed expression for the multi point $q$-moments in case of step initial conditions (for the exclusion process), exploiting Markov self duality of the Higher Spin Six Vertex Model. The same result is achieved in \cite{BorodinPetrov2016HS6VM} following a rather algebraic approach.

What follows is a nested contour integral expression for the single point $q$-moments of the height function of the Higher Spin Six Vertex Model with step boundary conditions.

\begin{Prop}[\cite{BorodinPetrov2016HS6VM}, Proposition 9.5] \label{height q moments} Assume conditions\footnote{In \cite{BorodinPetrov2016HS6VM} authors consider parameters $s_i, u_j$ which have opposite sign compared to our choice \eqref{HS6VM parameters}. This is just a convention and hence not a problem, as the stochastic weights $\mathsf{L}$ depend on $s_iu_j$ and $s_i^2$.}
\eqref{HS6VM parameters}.
Let $u_i \neq q u_j$ for all $i,j$ and let products $\xi_i s_i$ be strictly positive for all $i$. Then for all $l \in \mathbb{Z}_{\geq 0}$ and $x \in  \mathbb{Z}_{\geq 0}$ we have
\begin{equation}\label{HS moments}
\begin{split}
\mathbb{E}\left( q^{l\mathfrak{h}(x+1,t)} \right)=\frac{q^{\frac{l(l-1)}{2}}}{(2\pi i)^l} \oint_{\overline{\gamma_1}[\overline{\mathbf{U}}|1]} \cdots \oint_{\overline{\gamma_l}[\overline{\mathbf{U}}|l]} \prod_{1\leq A < B \leq l} \frac{z_A - z_B}{z_A - qz_B}\\
\prod_{i=1}^l\left( \prod_{j=1}^{x}\frac{\xi_j - s_jz_i}{\xi_j - s_j^{-1}z_i} \prod_{j=1}^t\frac{1-qu_jz_i}{1-u_jz_i} \frac{dz_i}{z_i}\right)
\end{split}
\end{equation}
where the variable $z_i$ is integrated along the path $\overline{\gamma_i}[\overline{\mathbf{U}}|i]=\gamma_i[\overline{\mathbf{U}}] \cup r^iC_0$ where $r>q^{-1}$, $C_0$ is a small contour around 0 and $\overline{\gamma_i}[\overline{\mathbf{U}}]$ encircles the set $\{ u_1^{-1}, \dots, u_t^{-1} \}$ , $q^{-1}\overline{\gamma_{i-1}}[\overline{\mathbf{U}}]$ and no other singularities. Moreover $q\overline{\gamma_i}[\overline{\mathbf{U}}]$ doesn't intersect any $r^jC_0$ and $r^lC_0$ doesn't contain any point of the set $\{\xi_j s_j \}_j$ (Figure \ref{intgration contours HS moments}).
\end{Prop}

\begin{figure}[ht]
\centering
\includegraphics{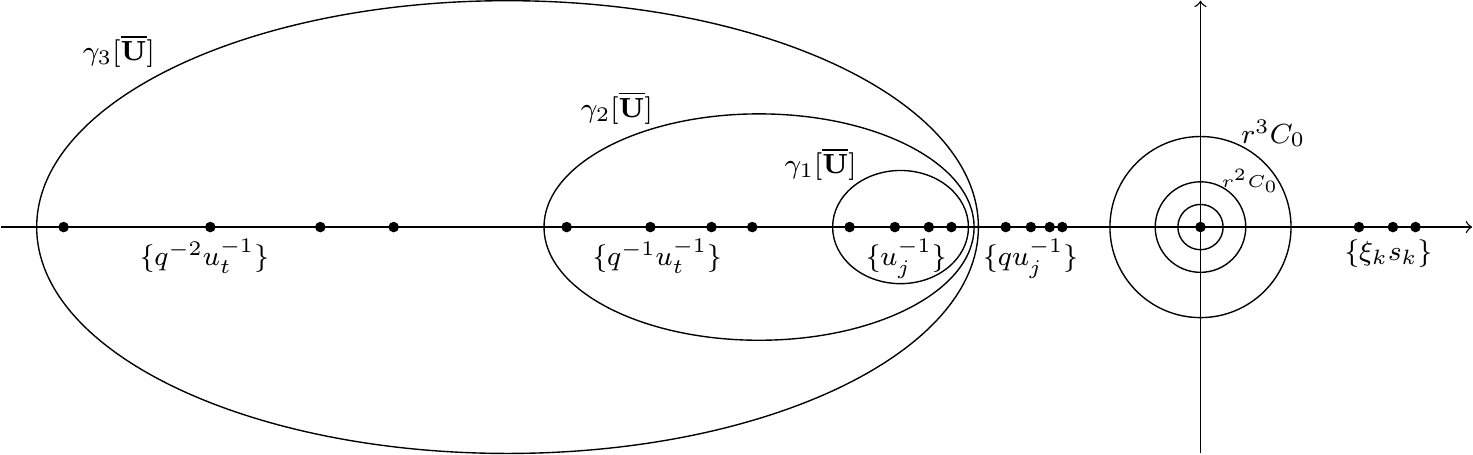}
\caption{\small{An example of nested integration contours for \eqref{HS moments} in the case $l=3$.}}
\end{figure} \label{intgration contours HS moments}

The next Corollary is a reformulation of Proposition \ref{height q moments} with a different choice of integration contours in \eqref{HS moments}.

\begin{Cor} \label{analytic continuation Borodin moments}
Assume conditions \eqref{HS6VM parameters} and let
\begin{equation*}
q \sup_{i} \{ \xi_i s_i \} < \inf_{i } \{ \xi_i s_i \}.
\end{equation*}
Then we have
\begin{equation}\label{moment HS Matteo}
\begin{split}
\mathbb{E} \left( q^{l\mathfrak{h}(x+1,n)} \right)=\frac{(-1)^lq^{\frac{l(l-1)}{2}}}{(2\pi \mathrm{i})^l} \oint_{\overline{C}[\Xi \mathbf{S}|1]} \cdots  \oint_{\overline{C}[\Xi \mathbf{S}|l]} \prod_{1\leq A < B \leq l} \frac{z_A - z_B}{z_A - qz_B}\\
\prod_{i=1}^l\Big( \prod_{j=1}^{x}\frac{\xi_js_j - s_j^2z_i}{\xi_js_j - z_i} \prod_{j=1}^n\frac{1-qu_jz_i}{1-u_jz_i} \frac{dz_i}{z_i}\Big)
\end{split}
\end{equation}
where the integration contour of $z_i$ is $\overline{C}[\Xi \mathbf{S}|i]$ and is the disjoint union of two curves $C[\Xi \mathbf{S}|i]$ and $r^{i-1}\partial D$. Here $C[\Xi \mathbf{S}|i]$ is a counterclockwise contour encircling the set $\{ \xi_i s_i \}_i$ and $qC[\Xi \mathbf{S}|i+1]$, but not 0 or any number $u_i^{-1}$. The disk $D$ is centered at 0 and it is sufficiently large to contain every $C[\Xi \mathbf{S}|i]$ and the set $\{ u_i^{-1}\}_i$.The coefficient $r$ is bigger than $q^{-1}$ and $r^{i-1}\partial D$ is clockwise oriented. A visualization of such contours is given in Figure \ref{integration contours + infinity figure}.
\end{Cor}

\begin{figure}[h!]
\centering
\includegraphics{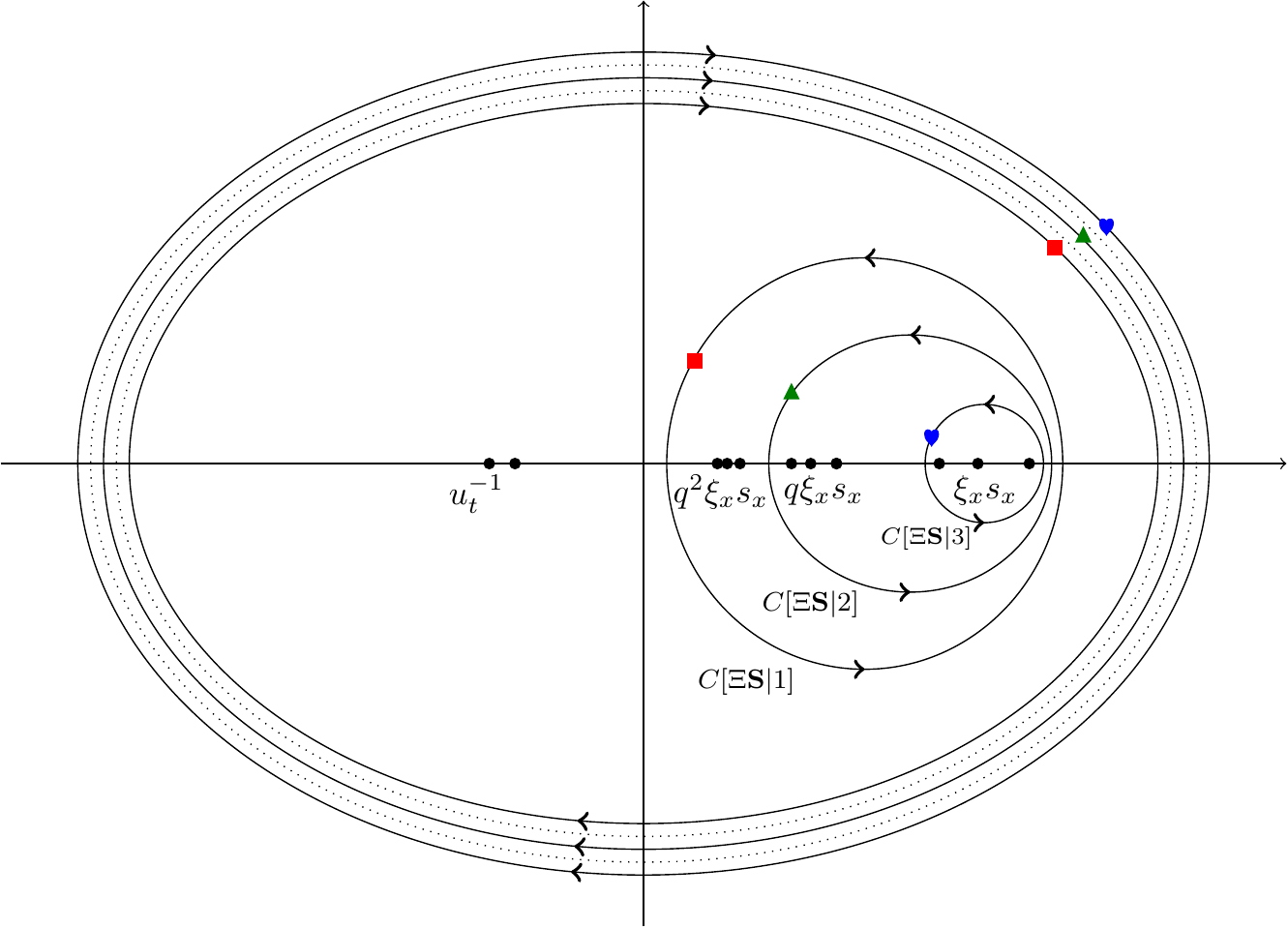}
\caption[]{\small An example of nested integration contours in \eqref{moment HS Matteo} for the case $l=3$. We see the integration contours for $z_1, z_2, z_3$ labeled respectively with \mysquare{red}, \mytriangle{green!50!black}, \marksymbol{heart}{blue}. Dotted lines between the external clockwise oriented contours $\partial D, r \partial D, r^2 \partial D$ are the shifted paths $q^{-1} \partial D, q^{-2} \partial D$. } \label{integration contours + infinity figure}
\end{figure} 

\begin{proof}
We use inductively the residue theorem to turn the integrations around contours $\overline{\gamma_i}[\overline{\mathbf{U}}|i]$ in \eqref{HS moments} into integrations around $\overline{C}[\Xi\mathbf{S}|i]$. Let's start with the most external contour $\overline{\gamma_l}[\overline{\mathbf{U}}|l]$. We see that the integrand in the rhs of \eqref{HS moments} has in the variable $z_l$ poles at
$$
q^{-1}z_1, \dots, q^{-1}z_{l-1}, \xi_1 s_1, \dots, \xi_x s_x, u_1^{-1}, \dots, u_t^{-1}, 0, \infty,
$$
so that since $\overline{\gamma_l}[\overline{\mathbf{U}}|l]$ only leaves outside $\xi_1 s_1, \dots, \xi_x s_x$ and $\infty$ and we have 
$$
\oint_{\overline{\gamma_l}[\overline{\mathbf{U}}|l]}=-\oint_{\overline{C}[\Xi\mathbf{S}|l]}.
$$
We move now to the integration in $z_{l-1}$. Here poles are at
$$
q z_l,q^{-1}z_1, \dots, q^{-1}z_{l-2}, \xi_1 s_1, \dots, \xi_x s_x, u_1, \dots, u_t, 0, \infty
$$
and $\overline{\gamma_{l-1}}[\overline{\mathbf{U}}|l-1]$ leaves out $\infty, \xi_1 s_1, \dots, \xi_x s_x$ and the shifted contours $qC[\Xi\mathbf{S}|l], qr^{l-1}\partial D$ where $q z_l$ lies. Since $C[\Xi \mathbf{S}|l-1]$ contains $\xi_1 s_1, \dots, \xi_x s_x,$ $qC[\Xi\mathbf{S}|l]$ and no other pole  and $r^{l-2}\partial D$ encircles $qr^{l-1}\partial D$ and $\infty$ we also get
$$
\oint_{\overline{\gamma_{l-1}}[\overline{\mathbf{U}}|l-1]}=-\oint_{\overline{C}[\Xi\mathbf{S}|l-1]}.
$$
By repeating the same procedure for all the remaining integration variables we are done.
\end{proof}

Following a consolidated approach developed in \cite{bcs2014,BorodinCorwin2014Mac}, in the particular case of step initial conditions, by using the result of Proposition \ref{height q moments} or Corollary \ref{analytic continuation Borodin moments}, one should eventually be able to obtain a determinantal expression for the quantity
$$
\mathbb{E}_{} \left( \frac{1}{(\zeta q^{\mathfrak{h}(x+1,t)};q)_\infty} \right),
$$
which is known to be the $q$-Laplace transform of the probability mass function of $\mathfrak{h}(x+1,t)$. This type of result has proven to be fruitful (see \cite{FerrariVeto2015TWLimit}) when it comes to the study of asymptotics of $\mathfrak{h}(x+1,t)$ as $x$ and $t$ go to infinity.

We remark that results like those of Proposition 
\ref{height q moments} or Corollary \ref{analytic continuation Borodin moments} essentially hold only for step boundary conditions, which at the current state are the only boundary conditions exhibiting a nice enough underlying algebraic structure to derive exact formulas for observables.
It is nonetheless possible, after a suitable choice of parameters $\Xi, \mathbf{S}, \mathbf{U}$, to produce other initial conditions, as those considered in \cite{OrrPetrov2016} or \cite{Aggarwal2016FluctuationsASEP}. In particular we will use a similar approach to that developed in \cite{Aggarwal2016FluctuationsASEP} to produce and study the stationary model.

\section{\texorpdfstring{$q$-Whittaker processes}.} \label{section qwhittaker}
In this section we give a brief review on $q$-Whittaker processes and present main results which will be used in the remainder of the paper.
\subsection{\texorpdfstring{Macdonald Processes and $q$-Whittaker processes}.} \label{subsection Mac and qWhittaker}
The $q$-Whittaker processes and the $q$-Whittaker measure have been first introduced in \cite{BorodinCorwin2014Mac} as particular cases of the more general Macdonald processes. This is a family of measures on the Gelfand-Tsetlin\footnote{the notation $\mathbb{GT}^{\geq 0}$ refers to Gelfand-Tsetlin cones of partitions. One can define the same object with generic signatures instead of partitions. We refer to the two sided Gelfand-Tsetlin cone with the notation $\mathbb{GT}$ as in \eqref{two sided GT cone}.} cone $\mathbb{GT}^{\geq 0}_n$ (Figure \ref{GT cone}), the set of sequences of partitions of integers
$$
\emptyset \prec \lambda^{(1)} \prec \lambda^{(2)} \prec \dots \prec \lambda^{(n)},
$$
where every $\lambda^{(i)}$ is an element of 
$$
\text{Part}^{\geq 0}_i=\{ \mu = (\mu_1 \geq \cdots \geq \mu_i \geq 0) |\ \mu_j \in \mathbb{Z}_{\geq 0} \}.
$$
The interlacing relation between two partitions $\lambda$, $\mu$, denoted with  $\mu \prec \lambda$, means that 
$$
\lambda_{k+1} \leq \mu_k \leq \lambda_k \qquad \qquad \text{for each $k$}.
$$
The Macdonald process is defined as the measure
$$
\mathbb{M}\big( \emptyset \prec \lambda^{(1)} \prec \dots \prec \lambda^{(n)} \big) = \frac{1}{\Pi( \mathbf{a}, \rho)} \prod_{i=1}^n P_{\lambda^{(i)}/\lambda^{(i-1)}}(a_i;q,t) Q_{\lambda^{(n)}}(\rho;q,t),
$$
where $P$ and $Q$ are Macdonald functions (\cite{macdonald1998symmetric}, Chapter IV) and
$\Pi$ is a normalization constant and its value can be expressed through the known Cauchy sums for symmetric functions.
Here, $q,t$ are parameters in [0,1), $\mathbf{a}=(a_1, \dots a_n)$ denotes a set of numerical values at which Macdonald polynomials are evaluated, whereas $\rho$ can be a generic complex algebra homomorphism on the algebra of symmetric functions. In this sense the quantity $\Pi(\mathbf{a}, \rho)$
can be thought as the generating function of functions $Q_\lambda(\rho;q,t)$.
\begin{figure}[ht]
\centering
\includegraphics{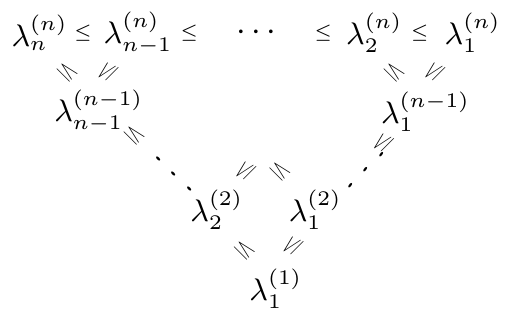}
	\caption{\small{ A triangular array in the Gelfand-Tsetlin cone $\mathbb{GT}^{\geq 0}_n$.}}\label{GT cone}
\end{figure}

The Macdonald measure is a particular case of the Macdonald process, obtained by projecting the measure $\mathbb{M}$ on the last partition $\lambda^{(n)}$. The branching rule
\begin{equation*}
\sum_{\varkappa} P_{\varkappa/\mu}(a_1, \dots, a_{k-1} ;q,t) P_{\lambda / \varkappa}(a_k;q,t) = P_{\lambda / \mu} (a_1, \dots, a_k;q,t),
\end{equation*}
of Macdonald functions allows us to write 
\begin{equation*}
\begin{split}
\mathbb{M}(\lambda) = \sum_{\lambda^{(1)}, \dots , \lambda^{(n-1)} } \mathbb{M} \big( \emptyset \prec \lambda^{(1)} \prec \dots \prec \lambda^{(n-1)} \prec \lambda \big) 
=\frac{1}{\Pi(\mathbf{a}, \rho)} P_{\lambda}(\mathbf{a};q,t) Q_{\lambda}(\rho;q,t).
\end{split}
\end{equation*}
The definition itself of the Macdonald functions depends on two parameters $q,t$ and letting these parameters vary one can obtain numerous other families of symmetric functions such as the Schur polynomials or the Hall-Littlewood functions. The $q$-Whittaker functions arise when we set $t=0$. An exact expression for the one variable skew $q$-Whittaker polynomials $P_{\lambda / \mu}$, which we denote dropping the explicit dependence on $q$, is
$$
P_{\lambda/\mu}(a)= \prod_{i=1}^N a^{\lambda_i} \prod_{i=1}^{N-1}a^{-\mu_i} \binom{\lambda_i-\lambda_{i+1}}{\lambda_i - \mu_i}_q,
$$
and the generic $q$-Whittaker polynomial is
$$
P_{\lambda}(\mathbf{a})=\sum_{\substack{\lambda^{(k)}_i:1\leq i\leq k\leq N-1 \\ \lambda^{(k+1)}_{i+1} \leq \lambda^{(k)}_i \leq \lambda^{(k+1)}_i}} \prod_{j=1}^N P_{\lambda^{(j)}/\lambda^{(j-1)}}(a_j),
$$
where $\lambda^{(n)}$ is meant to be $\lambda$. In case the Macdonald measure is considered with $t=0$ we use the notation
\begin{equation} \label{q whittaker measure}
    \mathbb{W}(\lambda) = \frac{1}{\Pi (\mathbf{a} , \rho)} P_\lambda (\mathbf{a}) Q_\lambda (\rho).
\end{equation}
Macdonald functions possess an important orthogonality property with respect to the so called \emph{torus scalar product}. For $t=0$, it is defined as the $n$-fold integral
$$
\langle f, g \rangle_n= \int_{\mathbb{T}^n}\prod_{j=1}^n\frac{dz_j}{z_j}f(\mathbf{z}^{-1})g(\mathbf{z})m_n^q(\mathbf{z}),
$$
where $\mathbb{T}=\{z\in \mathbb{C} : |z| = 1\}$ is the complex circle and
\begin{equation} \label{q-Sklyanin measure}
m_n^q(\mathbf{z})=\frac{1}{(2 \pi \mathrm{i})^n n!} \prod_{1 \leq i \neq j \leq n} (z_i/ z_j;q)_{\infty}
\end{equation}
is the $q$-Sklyanin weight. For the $q$-Whittaker polynomials this relation reads as
$$
\langle P_\lambda, P_\mu \rangle_n= c_{\lambda} \delta_{\lambda, \mu} 
$$
and $c_{\lambda}$ is some non zero constant. This property can be exploited to give an indirect definition of the dual $q$-Whittaker function $Q$ as
\begin{equation}\label{definition Q}
\begin{split}
Q_{\lambda}(\rho)&=\frac{1}{c_\lambda}\langle P_{\lambda}, \sum_{\mu} P_{\mu}Q_{\mu}(\rho) \rangle_n \\
&=\frac{1}{c_\lambda} \langle P_\lambda, \Pi(\bullet, \rho) \rangle_n.
\end{split}
\end{equation}
An additional remarkable structural feature of Macdonald polynomials is the shifting property,
$$
P_{\lambda+r^n}(\mathbf{x};q,t)= (x_1 \cdots x_n)^r P_\lambda(\mathbf{x};q,t),
$$
where $r$ is a non-negative integer and $\lambda + r^n = (\lambda_1 + r \geq \dots \geq \lambda_n +r )$. Combining this shifting invariance with \eqref{definition Q}, we are allowed to extend the notion of $q$-Whittaker functions $P_\lambda$ and $Q_\lambda$ to the set of signatures
$$
\text{Sign}_i = \{ \lambda = (\lambda_1 \geq \cdots \geq \lambda_i ) |\ \lambda_j \in \mathbb{Z} \}.
$$
This observation is used in \cite{q-TASEPtheta} to extend the definition of the $q$-Whittaker processes on the two sided Gelfand-Tsetlin cone 
\begin{equation} \label{two sided GT cone}
\mathbb{GT}_n = \{ \{\lambda^{(k)}_i\}_{i=1,\dots, k}^{k=1, \dots n} \in \mathbb{Z}^{\binom{n+1}{2}} |\ \lambda^{(k)}_{i+1} \leq \lambda^{(k-1)}_i \leq \lambda^{(k)}_i \}.
\end{equation}
\subsection{\texorpdfstring{$q$-moments of the corner coordinate}.}
Macdonald polynomials can be constructed as eigenfunctions of the Macdonald operator $D_n$ (\cite{macdonald1998symmetric}, Sect. VI.3-4), whose action on the generic symmetric function $F$ is
$$
D_n F(\mathbf{x}) = \sum_{j=1}^n \prod_{\substack{ i=1\\i \neq j}}^n \frac{x_i-t x_j}{x_i-x_j} F(x_1, \dots, x_{j-1}, qx_j, x_{j+1}, \dots x_n)
$$
and the eigenvalue relative to the function $P_{\lambda}$ is
$$
\sum_{j=1}^n q^{\lambda_j}t^{n-j}.
$$
As such eigenvalues are all distinct, the function $P_{\lambda}$ is uniquely identified imposing the condition
$$
P_{\lambda}(\mathbf{x};q,t) = m_\lambda (\mathbf{x}) + \sum_{\mu < \lambda} C_{\lambda, \mu} m_{\mu}(\mathbf{x}),
$$
where $m_\lambda$ is the monomial symmetric function, the summation in the right hand side is taken over all partitions $\mu < \lambda$ in the lexicographic order and $C_{\lambda, \mu}$ are constants. In the simple case where the Macdonald operator acts on a product function $F(\mathbf{x})=f(x_1) \cdots f(x_n)$, it can be written in the integral form
$$
\frac{D_n F(\mathbf{x})}{F(\mathbf{x})}=\frac{(t-1)^{-1}}{2 \pi \mathrm{i}} \oint \prod_{j=1}^n\frac{x_j-tz}{x_j-z}\frac{f(qz)}{f(z)} \frac{dz}{z},
$$
being the integration contour a path encircling $x_1, \dots, x_n$ and no other singularity. This last formula has been used in \cite{BorodinCorwin2014Mac} to determine the $q$-moments of the observable $\lambda^{(n)}_n$ in the $q$-Whittaker process. We report this result in the next
\begin{Prop}[\cite{BorodinCorwin2014Mac}, Proposition 3.1.5] \label{q-moments q-Whittaker Prop}
Let $\rho$ be the specialization of symmetric functions for which the normalization constant $\Pi$ of the $q$-Whittaker measure \eqref{q whittaker measure} assumes the form
\begin{equation*}
\Pi(\mathbf{a};\rho)= \prod_{i=1}^n e^{\tau a_i}\prod_{i\geq 1} \frac{(1+\beta_i a_j)}{(\alpha_i a_j;q)_\infty},
\end{equation*}
for some sets of non-negative real numbers $\tau, \{ \alpha_i \}_{i\geq 1}, \{ \beta_i \}_{i \geq 1}$ satisfying
$$
\sum_{i\geq 1} (\alpha_i + \beta_i) < \infty, \qquad \qquad \sup_{i,j} |\alpha_i a_j| < 1.
$$
Then, for any non-negative integer $k$ we have
\begin{equation} \label{q-moments q-Whittaker}
\begin{split}
\mathbb{E}_{\mathbb{W}} \left( q^{k \lambda^{(n)}_n} \right) =\frac{(-1)^k q^{\binom{k}{2}}}{(2 \pi \mathrm{i})^k} \oint_{C[\mathbf{a}|k]} \cdots \oint_{C[\mathbf{a}|1]} & \prod_{1\leq i < j \leq k} \frac{z_i - z_j}{z_i - q z_j} \prod_{j=1}^k \left( \prod_{m=1}^n \frac{a_m}{a_m -z_j} \right.\\
& \left. \qquad \times  \prod_{i \geq 1} (1-\alpha_i z_j) \frac{1+q \beta_i z_j}{1+\beta_i z_j} \frac{dz_j}{z_j} \right),
\end{split}
\end{equation}
where $C[\mathbf{a}|j]$ is the integration contour for the complex variable $z_j$ and contains $a_1, \dots, a_n$, each shifted contour $q C[\mathbf{a}|l]$ for $l>j$ and no other pole of the integrand.
\end{Prop}
\begin{figure}
\centering
\includegraphics{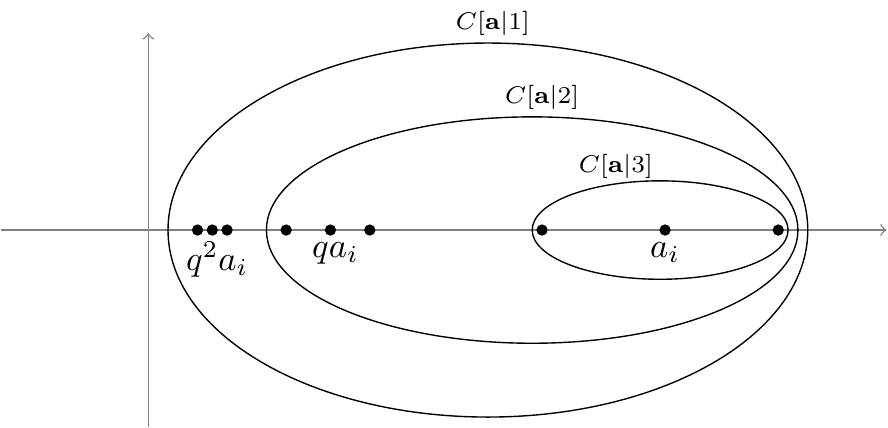}
\caption{\small{A possible choice of integration contrours for \eqref{q-moments q-Whittaker} in the case $k=3$.}}
\end{figure}
Naturally, under some suitable assumption on the growth of the $q$-moments, they completely determine the distribution of $\lambda^{(n)}_n$ as they are generated by the $q$-Laplace transform (see \eqref{q-Laplace Appendix})
$$
 \mathbb{E}_{\mathbb{W}} \left( \frac{1}{(\zeta q^{\lambda^{(n)}_n};q)_\infty} \right).
$$
This is the case for expressions \eqref{q-moments q-Whittaker}, stated only for step initial conditions, where all $q^{k \lambda_n^{(n)}}$ are positive quantities bounded above by 1. 

So far we presented the exact expression of the $q$-moments of two different observables of two different stochastic processes. Despite the difference between the Higher Spin Six Vertex Model and the $q$-Whittaker processes, expression of the $q$-moments respectively of the height function $\mathfrak{h}$ and of the corner coordinate $\lambda^{(n)}_n$, given in \eqref{moment HS Matteo} and \eqref{q-moments q-Whittaker} indeed show similarities. This fact was noticed before in \cite{OrrPetrov2016}, where authors described the correspondence between these two models under step Bernoulli boundary conditions for the Higher Spin Six Vertex Model. The matching reported in the following Proposition traces that of Theorem 4.11 of \cite{OrrPetrov2016}.
\begin{Prop} \label{matching proposition}
Set $\Xi, \mathbf{S}, \mathbf{U}$ to be as in Corollary \ref{analytic continuation Borodin moments}. Moreover let the following bounds
\begin{equation} \label{bounds q-Whittaker}
\sup_{i} \{ \xi_is_i \}<1,\qquad \sup_i \{ s_i / \xi_i \}<1, \qquad \qquad \sup_{i,j}|\xi_i s_i u_j|<1,
\end{equation}
hold. Then, we have
\begin{equation}\label{matching}
\mathbb{E}_{\text{\emph{HS}}} \left(q^{l\mathfrak{h}(x+1,n)}\right)=  \mathbb{E}_{\mathbb{W}_{\mathbf{\Xi}, \mathbf{S}, \mathbf{U}}}\left( \left( q^{\lambda_x} + q^n\prod_{j=1}^x s_j^2 \right)^l \right),
\end{equation}
where the $q$-Whittaker measure in the right hand side is given by
\begin{equation} \label{q-Whittaker}
\mathbb{W_{\mathbf{\Xi}, \mathbf{S}, \mathbf{U}}}(\lambda)=\prod_{i=1}^x\frac{\prod_{j=1}^x(s_is_j\xi_i / \xi_j;q)_{\infty}}{\prod_{j=1}^n(1-u_j\xi_is_i)}  P_{\lambda}(\xi_1 s_1, \dots, \xi_x s_x)Q_{\lambda}(\underbrace{s_1/\xi_1, \dots , s_x/ \xi_x^{-1}}_{\alpha-\text{specializations}},\underbrace{-u_1, \dots, -u_n}_{\beta-\text{specializations}}).
\end{equation}
\end{Prop}

In \eqref{q-Whittaker} we made use of the common terminology of $\alpha$-specializations and $\beta$-specializations, justified by the fact that the $q$-Whittaker measure $\mathbb{W}_{\Xi,\mathbf{S},\mathbf{U}}$ is obtained setting in Proposition \ref{q-moments q-Whittaker Prop} $\alpha_i = s_i/\xi_i$ and $\beta_j=-u_j$.

\begin{proof}
From \eqref{moment HS Matteo} we start to integrate the singularity at $\infty$ from the most external contour down to the internal ones. We have
\begin{equation}\label{equivalence moments}
\begin{split}
\mathbb{E}_{\text{HS}} \left(q^{l\mathfrak{h}(x+1,n)}\right)&=\frac{(-1)^l q^{\frac{l(l-1)}{2}}}{(2\pi \mathrm{i})^l} \oint_{\overline{C}[\Xi\mathbf{S}|1]} \cdots \oint_{ \overline{C}[\Xi\mathbf{S}|l] } \prod_{1\leq A < B \leq l} \frac{z_A - z_B}{z_A - qz_B}
\prod_{i=1}^l\left( \prod_{j=1}^{x}\frac{\xi_j - s_jz_i}{\xi_j - s_j^{-1}z_i} \prod_{j=1}^n\frac{1-qu_jz_i}{1-u_jz_i} \frac{dz_i}{z_i}\right)\\
&=\frac{(-1)^l q^{\frac{l(l-1)}{2}}}{(2\pi \mathrm{i})^l} \oint_{\overline{C}[\Xi\mathbf{S}|1]} \cdots \oint_{ C[\Xi\mathbf{S}|l] } \prod_{1\leq A < B \leq l} \frac{z_A - z_B}{z_A - qz_B}
\prod_{i=1}^l\left( \prod_{j=1}^{x}\frac{\xi_j - s_jz_i}{\xi_j - s_j^{-1}z_i} \prod_{j=1}^n\frac{1-qu_jz_i}{1-u_jz_i} \frac{dz_i}{z_i}\right)\\
&\quad+q^n\prod_{j=1}^x s_j^2 \mathbb{E}_{\text{HS}} \left(q^{(l-1)\mathfrak{h}(x+1,n)}\right)
\\ & \qquad \qquad\vdots\\
&= \sum_{k=0}^l \binom{l}{k} \left(q^n\prod_{j=1}^x s_j^2 \right)^{l-k}\frac{(-1)^k q^{\frac{k(k-1)}{2}}}{(2\pi \mathrm{i})^k} \oint_{C[\Xi\mathbf{S}|1]} \cdots \oint_{ C[\Xi\mathbf{S}|k] } \prod_{1\leq A < B \leq k} \frac{z_A - z_B}{z_A - qz_B}\\
& \qquad \qquad \qquad \qquad \times\prod_{i=1}^k\left( \prod_{j=1}^x\frac{\xi_j s_j}{\xi_j s_j - z_i}(1- s_j\xi_j^{-1}z_i) \prod_{j=1}^n\frac{1-qu_jz_i}{1-u_jz_i} \frac{dz_i}{z_i}\right)\\ \\
&= \sum_{k=0}^l \binom{l}{k} \left(q^n\prod_{j=1}^x s_j^2 \right)^{l-k} \mathbb{E}_{\mathbb{W_{\mathbf{\Xi}, \mathbf{S}, \mathbf{U}}}}
 \left( q^{k\lambda_x} \right), 
\end{split}
\end{equation}
where the contour $C[\Xi\mathbf{S}|i]$ has been defined in Proposition \ref{analytic continuation Borodin moments}. In the last  equality we used the fact that, assuming the bounds \eqref{bounds q-Whittaker}, the $q$-moments of $\lambda_x$ in the $q$-Whittaker measure $\mathbb{W}_{\Xi, \mathbf{S}, \mathbf{U}}$ defined as in \eqref{q-Whittaker} are written as
\begin{equation*}
\begin{split}
\mathbb{E}_{\mathbb{W}_{\Xi, \mathbf{S}, \mathbf{U}}} \left ( q^{k\lambda_x} \right ) =\frac{(-1)^k q^{\frac{k(k-1)}{2}}}{(2\pi \mathrm{i})^k} \oint_{ C[\Xi \mathbf{S}|1] } \cdots \oint_{ C[\Xi \mathbf{S}|k] }& \prod_{1\leq A < B \leq k} \frac{z_A - z_B}{z_A - qz_B}\\ &\times\prod_{i=1}^k\left( \prod_{j=1}^x\frac{\xi_j s_j}{\xi_j s_j - z_i}(1- s_j\xi_j^{-1}z_i) \prod_{j=1}^n\frac{1-qu_jz_i}{1-u_jz_i} \frac{dz_i}{z_i}\right),
\end{split}
\end{equation*}
in accord with \eqref{q-moments q-Whittaker}.
By bringing the binomial sum inside the expectation we obtain the result.
\end{proof} 

\subsection{\texorpdfstring{Explicit distribution of $\lambda^{(n)}_n$}.} By making use of \eqref{definition Q} and additional combinatorical properties of $q$-Whittaker functions, in \cite{q-TASEPtheta} it is shown to be possible to express in a compact form the probability
distribution of $\lambda^{(n)}_n$ in a two sided $q$-Whittaker process, which has been briefly defined at the end of section \ref{subsection Mac and qWhittaker}.
We have the following
\begin{Prop}[\cite{q-TASEPtheta}, Proposition 3.14] \label{prop integral qtasep}
Let $\overline{\rho}$ be a specialization of two sided $q$-Whittaker functions, such that 
$$
\Pi(\mathbf{a};\overline{\rho})=\prod_{i=1}^n e^{\tau a_i} \prod_{i,j=1}^n \frac{1}{(\gamma_i/a_j ; q)_{\infty}},
$$
for parameters satisfying $|\gamma_i| < |a_j|$. Then, for any integer $l$, we have
\begin{equation}\label{probability 2sided qWhittaker}
\mathbb{P}_{q\emph{W}}(\lambda^{(n)}_n = l) = (q;q)_\infty^{n-1} \int_{\mathbb{T}^n}\Big(\frac{A}{Z}\Big)^l m_n^q(\mathbf{z}) \frac{\Pi(\mathbf{z}; \overline{\rho})}{\Pi(\mathbf{a}; \overline{\rho})} \frac{(A/Z;q)_\infty}{\prod_{i,j=1}^n (a_i/z_j;q)_\infty}\prod_{j=1}^n\frac{dz_j}{z_j},
\end{equation}
where $A=a_1\cdots a_n$, $Z=z_1 \cdots z_n$ and the integration is performed over the $n$ dimensional torus $\mathbb{T}^n$.
\end{Prop}
With very little changes in the proof of this last proposition one can allow the normalization constant $\Pi$ to be of a more general form. 
\begin{Prop} \label{probability 2sided q-Whittaker proposition}
Expression \eqref{probability 2sided qWhittaker} also holds for $\overline{\rho}$ being a specialization of two sided $q$-Whittaker functions such that 
$$
\Pi(\mathbf{a};\overline{\rho})=\prod_{i=1}^n e^{a_i \tau} \prod_{j\geq 1} \frac{(1+\beta_j a_i)}{(\alpha_j a_i;q)_\infty} \prod_{j=1}^n\frac{1}{(\gamma_j / a_i;q)_\infty},
$$
for non-negative real parameters $\tau,\{\alpha_i\},\{\beta_i\},\{\gamma_i\}$ satisfying
$$
\sum_{j \geq 1} \left( \alpha_j + \beta_j \right) < \infty, \qquad \qquad |\alpha_i a_j|<1, \qquad \qquad |\gamma_i| < |a_j|.
$$
\end{Prop}

\section{Boundary conditions} \label{About the initial conditions}
The aim of this section is twofold. First, in Subsection \ref{subs: Burke property}
we prove Proposition \ref{prop translation invariant}, which characterizes the family of probability measures satisfying a certain translational symmetry, which we denoted as \emph{Burke's property} (see Definition \ref{def: Burke property}). By means of this property we define the \emph{full plane Stationary Higher Spin Six Vertex Model} and this is done in Proposition \ref{prop extension HS6VM}. Subsequently, in Subsection \ref{subsection integrable initial}, we give a description of a family of boundary conditions which one can construct from the step one \eqref{step boundary bc} and that will be suitable to study the model in the stationary case.

\subsection{Burke's property in the Higher Spin Six Vertex Model} \label{subs: Burke property}

We like to start this Subsection by giving the proof of Proposition \ref{prop translation invariant}. First we show our results for the simpler case where the model has unfused rows, corresponding to the choice $J=1$. Subsequently we extend our proof to the general $J$ case. In particular, when $J=1$, each horizontal edge is crossed by either zero or one path and therefore random variables $\mathsf{j}_x^t$ are Bernoulli distributed. We also recall the sequential update mechanism produced by the transfer operator $\mathfrak{X}_{u_t}$, which has been described in Section 2 (and in \cite{BorodinPetrov2016HS6VM}, Section 6.4.2). Let $\lambda(t-1) = 2^{m_2^{t-1}} 3^{m_3^{t-1}} \cdots $ be a configuration of paths entering the row of vertices with ordinate $t$ and assume that, conditionally to the value of $\lambda(t-1)$ and $\mathsf{j}_1^t$, random variables $ \mathsf{m}_2^t , \dots , \mathsf{m}_{x-1}^t , \mathsf{j}_2^t , \dots , \mathsf{j}_{x-1}^t $ assumed respectively the values $ m_2^t, \dots, m_{x-1}^t , j_2^t, \dots, j_{x-1}^t $. Then we have
\begin{equation} \label{update rule X}
\mathbb{P} \left( \mathsf{m}_x^t = m_x^t | \lambda(t-1) , \{m_y^t , j_y^t \}_{y<x} \right) = \mathsf{L}_{u_t \xi_x, s_x} (m_x^{t-1} , j_{x-1}^t  |\ m_x^t, j_x^t), 
\end{equation}
where the definition of $\mathsf{L}$ is given in Table \ref{weights table} and at the boundaries (that is for $x=2$ or $t=1$), the law of $\{\mathsf{m}^0_x\}_{x \geq 2}, \{ \mathsf{j}_1^t \}_{t \geq 1}$ is assumed to be known. This update is called sequential since it can be regarded as a sequence of moves propagating from the leftmost vertex to the right.

The update produced by the fused transfer operator $\mathfrak{X}_{u_t}^{(J)}$ naturally follows the same rule, with weights $\mathsf{L}$ being replaced by weights $\mathsf{L}^{(J)}$ given in \eqref{fused weights L}.

\begin{Lemma}\label{density recurrence}
Assume that the Higher Spin Six Vertex Model, with $J=1$, satisfies the Burke's property. Then, setting
\begin{equation} \label{p_t and pi_Mx}
    p_t= \mathbb{P}(\mathsf{j}_x^t = 1) \qquad \text{and} \qquad \pi_{M,x} = \mathbb{P}(\mathsf{m}_x^t = M)
\end{equation}
we have
\begin{equation}\label{pi recurrence}
\begin{split}
    \pi_{M,x}=&\pi_{M-1,x} p_{t} \frac{1-s_x^2q^{M-1}}{1-s_x\xi_x u_t}\\
    &+\pi_{M,x} \left[ (1-p_{t}) \frac{1-s_x \xi_x u_t q^M}{1-s_x \xi_x u_t}+ p_{t}\frac{-s_x \xi_x u_t + s_x^2q^M}{1-s_x \xi_x u_t} \right] \\
    &+\pi_{M+1,x}(1-p_{t}) \frac{-s_x \xi_x u_t +  s_x \xi_x u_t q^{M+1} }{1-s_x \xi_x u_t}
\end{split}
\end{equation}
for each $M \geq 0$, with $\pi_{-1,x}=0$.
\end{Lemma}
\begin{proof}
From the definition \eqref{update rule X} of the update rule given by $\mathfrak{X}_{u_t}$ we have
\begin{equation} \label{eq: prob m_t}
\begin{split}
\mathbb{P}(\mathsf{m}_x^t=M) = & \mathbb{P}(\mathsf{m}_x^{t-1}=M-1, \mathsf{j}_{x-1}^t=1) \mathsf{L}_{u_t \xi_x, s_x} (M-1, 1 |\ M,0) \\
& + \mathbb{P}(\mathsf{m}_x^{t-1}=M, \mathsf{j}_{x-1}^t=0) \mathsf{L}_{u_t \xi_x, s_x} (M, 0 |\ M, 0) \\
& + \mathbb{P}(\mathsf{m}_x^{t-1}=M, \mathsf{j}_{x-1}^t=1)  \mathsf{L}_{u_t \xi_x, s_x} (M, 1 |\ M, 1) \\
& +\mathbb{P}(\mathsf{m}_x^{t-1}=M+1, \mathsf{j}_{x-1}^t=0) \mathsf{L}_{u_t \xi_x, s_x} (M+1, 0 |\ M, 1),
\end{split}
\end{equation}
which becomes \eqref{pi recurrence} when we substitute the definition of weights $\mathsf{L}$ given in Table \ref{weights table} and use the Burke's property to express the joint law of $\mathsf{m}_x^{t-1}$ and $\mathsf{j}_{x-1}^t$ through quantities $\pi_{M,x}, p_t$.
\end{proof}

An interesting feature of the recurrence relations \eqref{pi recurrence} is that it admits an exact solution in terms of the Al Salam-Chihara polynomials (\cite{IsmailClassicalAndQuantum}).

\begin{Lemma} \label{Lemma local density J=1}
Assume that the Higher Spin Six Vertex Model, with $J=1$, satisfies the Burke's property. Then, for each $(x,t) \in \Lambda_{1,0}$ random variables $\mathsf{m}_x^t , \mathsf{j}_x^t$ have laws
\begin{equation} \label{law m and j J=1}
    \mathsf{m}^{t}_{x} \sim q\text{\emph{NB}}( s^2_x , \mathpzc{d}/(\xi_x s_x) ), \qquad \mathsf{j}_{x}^{t} \sim \text{\emph{Ber}}( - \mathpzc{d} u_t /( 1 - \mathpzc{d} u_t) ),
\end{equation}
where $\mathpzc{d}$ is a parameter independent of $x$ or $t$ that satisfies \eqref{bound curvy d}.
\end{Lemma}
\begin{proof}
We use results of Lemma \ref{density recurrence}. We claim that 
\begin{equation} \label{stationary density}
\pi_{M,x} = \left( \frac{p_t}{-s_x \xi_x u_t (1-p_t)} \right)^M \frac{(s_x^2;q)_M}{(q;q)_M} \frac{(\frac{-s_xp_t}{(1-p_t)\xi_x u_t};q)_{\infty}}{(\frac{p_t}{-s_x \xi_x (1-p_t) u_t};q)_{\infty}}
\end{equation}
is solution of the recurrence \eqref{pi recurrence}. Such expression for $\pi_{M,x}$ is relatively simple, so that plugging it into \eqref{pi recurrence} one could easily verify that indeed our claim holds. The assumption that the probability measure satisfies the Burke's property implies that values of $\pi_{M,x}$ cannot depend on the $t$ coordinate and therefore $t$ dependent quantities $u_t, p_t$ must satisfy the relation 

\begin{equation} \label{definition curvy d}
\frac{p_t}{-u_t ( 1 - p_t )} = \mathpzc{d},
\end{equation}
for some parameter $\mathpzc{d}$ independent on $x$ or $t$, which necessarily has to meet condition \eqref{bound curvy d} as well. By inverting \eqref{definition curvy d}
and recalling the definition of $p_t$ given in \eqref{p_t and pi_Mx} we complete the proof of \eqref{law m and j J=1}.

This checking style argument might not be the most elegant, so we now quickly show how this solution was obtained. Setting $s_x=s, \xi_x u_t = u, p_t = p$ and defining the auxiliary sequence $f_M$ as 
\begin{equation*}
\pi_M=\beta^{-M}\frac{(s^2;q)_M}{(q;q)_M}f_M,
\end{equation*}
the recurrence \eqref{pi recurrence} becomes
\begin{equation} \label{recurrence fM}
(1-q^M)\frac{p\beta^2}{(1-p)s u }f_{M-1}+\big[\beta\frac{s u-p-s u p }{(1-p)s u }-\beta\frac{s u -s u p -s^2p}{(1-p)s u }q^M\big]f_M -(1-s^2q^M)f_{M+1}=0.
\end{equation}
This last expression has to be compared with the general recurrence relation 
\begin{equation} \label{Al Salam-Chihara recurrence}
-t_1^2(1-q^n)g_{n-1}+t_1\big[ z+1/z -(t_1+t_2)q^n \big]g_n-(1-t_1t_2q^n)g_{n+1}=0,
\end{equation}
with initial conditions $p_{-1}=0, p_0=1$, which is known to be satisfied by the Al Salam-Chihara polynomials (\cite{IsmailClassicalAndQuantum}, (15.1.6))
\begin{equation*}
g_n(z; t_1, t_2|q)=
\setlength\arraycolsep{1pt}
{}_3 \phi_2\left(\begin{matrix}q^{-n},&t_1z, &t_1/z & &\\&t_1t_2,&0
& &\end{matrix} \Big| q,\ q\right).
\end{equation*}
Equating term by term \eqref{recurrence fM} and \eqref{Al Salam-Chihara recurrence} we get, in terms of variables $z, t_1, t_2, \beta$, the second order system
\begin{equation}\label{system Al Salam Chihara}
\begin{cases}
   (z+1/z)t_1=\beta \frac{s u -p-s u p}{(1-p)s u },\\
   t_1^2=-\beta^2\frac{p}{(1-p)s u },\\
   t_1^2+t_1 t_2=\beta\frac{s u -s u -s^2p}{(1-p)s u },\\
   t_1 t_2 = s^2,
  \end{cases}
\end{equation}
whose solution is
\begin{equation*}
t_1=-s\ \sqrt[]{\frac{-sp}{(1-p) u }},\qquad\qquad
t_2=-s\ \sqrt[]{\frac{(1-p) u }{-sp}},\qquad\qquad z=t_2, \qquad\qquad \beta=s^2.
\end{equation*}
With these choices of values the Al Salam-Chihara polynomial assumes a simple form, so that
\begin{equation*}
      f_M = g_M(t_2; t_1, t_2|q)= \setlength\arraycolsep{1pt}
{}_2 \phi_1\left(\begin{matrix}q^{-M}, &-sp/( u - u p) & &\\&0
& &\end{matrix} \Big| q,\ q\right) = \left( \frac{-sp}{ u (1-p)} \right)^M,
\end{equation*}
where in the last equality we used the $q$-analog Chu-Vandermonde identity (see Appendix \ref{appendix qfunctions}).
Using the definition  of $f_M$ we finally obtain \eqref{stationary density}.
\end{proof}

Result of Lemma \ref{Lemma local density J=1} suffices to prove the "only if" part of the statement of Proposition \ref{prop translation invariant} in the particular case of a model with $J=1$. The next two Lemmas address the "if" part of Proposition \ref{prop translation invariant}.

\begin{Lemma} \label{lemma bc propagation horiz}
Consider the Higher Spin Six Vertex Model on $\Lambda_{1,0}$, with $J=1$ and boundary conditions
\begin{equation} \label{boundary conditions burke J=1}
\mathsf{m}^{0}_{x} \sim q\text{\emph{NB}}( s^2_x, \mathpzc{d}/(\xi_x s_x) ), \qquad \mathsf{j}_{1}^{t} \sim \text{\emph{Ber}}( - \mathpzc{d} u_t /( 1 - \mathpzc{d} u_t) ),
\end{equation}
where $\mathsf{m}_2^0 , \mathsf{m}_3^0, \mathsf{m}_4^0, \dots, \mathsf{j}_1^1, \mathsf{j}_1^2, \mathsf{j}_1^3, \dots$ are independent random variables. Then for all $x \geq 2$ the sequence $\mathsf{m}_x^0 , \mathsf{m}_{x+1}^0, \mathsf{m}_{x+2}^0, \dots, \mathsf{j}_{x-1}^1, \mathsf{j}_{x-1}^2, \mathsf{j}_{x-1}^3, \dots$ is a family of independent random variables and for each $t\geq1$ we have $\mathsf{j}_{x-1}^t \sim \text{\emph{Ber}}( - \mathpzc{d} u_t /( 1 - \mathpzc{d} u_t) )$.
\end{Lemma}

\begin{proof}
We start observing that, due to the choice of boundary conditions and due to the fact that paths propagate in the lattice in the up right direction the family of random variables $\mathsf{m}_x^0, \mathsf{m}_{x+1}^0, \mathsf{m}_{x+2}^0,\dots$ is always independent of the family $\mathsf{j}_{x-1}^1, \mathsf{j}_{x-1}^2, \mathsf{j}_{x-1}^3,\dots$ and hence we only need to show that
$\mathsf{j}_{x-1}^1, \mathsf{j}_{x-1}^2, \mathsf{j}_{x-1}^3,\dots$ are mutually independent and that their distributions follow the law described in the statement of Lemma \ref{lemma bc propagation horiz}. We prove this claim for the $x=3$ case, as the general $x$ case would simply follow by induction procedure. This means that, for all $t\geq 1$ and for all choices of $(j_1, \dots, j_t) \in \{0,1\}^t$, we need to show that

\begin{equation} \label{j's are independent}
\mathbb{P}\left( \mathsf{j}_2^1=j_1, \dots, \mathsf{j}_2^t=j_t  \right) = \prod_{k=1}^t \frac{( - \mathpzc{d} u_k )^{j_k} } { 1- \mathpzc{d} u_k }.
\end{equation}
To do so we follow a rather algebraic approach. Introduce the $2 \times 2$ matrices
\begin{equation*}
\mathcal{A}_k = \frac{1}{1-s_2 \xi_2 u_k} \begin{pmatrix}
  1 & -s_2 \xi_2 u_k \\
  1 & -s_2 \xi_2 u_k \\ 
 \end{pmatrix}, \qquad 
 \mathcal{B}_k = \frac{1}{1-s_2 \xi_2 u_k} \begin{pmatrix}
  -s_2 \xi_2 u_k & s_2 \xi_2 u_k \\
  -s_2^2 & s_2^2 \\ 
 \end{pmatrix},
\end{equation*}
through which we can express the stochastic weight $\mathsf{L}$, using the classical bra-ket notation\footnote{the numbering of rows and column starts from zero rather than from one}, as
\begin{equation*}
\mathsf{L}_{\xi_2 u_k, s_2}(m,j | \ m+j-j', j') = \langle e_j| \left( \mathcal{A}_k  +q^m \mathcal{B}_k \right) | e_{j'} \rangle,
\end{equation*}
for each $j,j'=0,1$. It is not hard to convince oneself that it is possible to describe the weight of any admissible configuration of paths around a column of two vertices as 

\begin{equation*}
\begin{split}
&\mathsf{L}_{ \{\xi_2 u_1, \xi_2 u_k \},s_2} \left(
\raisebox{-23pt}{\includegraphics{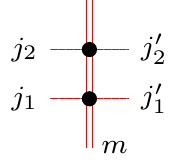}} \right) = \sum_{l_1,l_2 \geq 0} \mathsf{L}_{\xi_2 u_1, s_2}(m,j_1|\ l_1, j_1') \mathsf{L}_{\xi_2 u_2, s_2}(l_1,j_1|\ l_2, j_2') \\
&= \langle e_{j_1}| \otimes \langle e_{j_2}| \left[ \left( \mathcal{A}_1 + q^m \mathcal{B}_1 \right) \otimes \mathcal{A}_2 + q^m 
\left( \begin{smallmatrix}
1&0\\0&q
\end{smallmatrix} \right)  \left( \mathcal{A}_1 + q^m \mathcal{B}_1 \right) \left( \begin{smallmatrix}
1&0\\0&1/q
\end{smallmatrix} \right) \otimes \mathcal{B}_2 \right] |e_{j_1'}\rangle \otimes |e_{j_2'}\rangle .
\end{split}
\end{equation*}
More in general, defining the sequence
\begin{equation} \label{recursion T}
\begin{cases}
T_k^{(m)} = T_{k-1}^{(m)} \otimes \mathcal{A}_k  + q^m 
\left( \begin{smallmatrix}
1&0\\0&q
\end{smallmatrix} \right)^{\otimes (k-1)}  T_{k-1}^{(m)} \left( \begin{smallmatrix}
1&0\\0&1/q
\end{smallmatrix} \right)^{\otimes (k-1)} \otimes \mathcal{B}_k\\
T_1^{(m)} = \mathcal{A}_1  +q^m \mathcal{B}_1,
\end{cases}
\end{equation}
one can show that, for all admissible configurations of paths around a column of $k$ vertices, we have

\begin{equation*}
\mathsf{L}_{\{\xi_x u_1, \dots, \xi_x u_k \},s_x} \left(
	\raisebox{-30pt}{\includegraphics{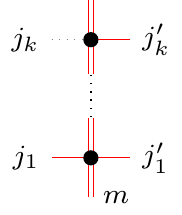}} \right) = ( \langle e_{j_1}| \otimes \cdots \otimes \langle e_{j_k}| ) \cdot T_k^{(m)} \cdot (  | e_{j_1'} \rangle \otimes \cdots \otimes | e_{j_k'} \rangle ).
\end{equation*}
Define also the vector 
\begin{equation*}
\langle \mathbf{v}_k | = \frac{1}{1- \mathpzc{d} u_k}  \langle e_0 | + \frac{ - \mathpzc{d} u_k }{1- \mathpzc{d} u_k}  \langle e_1 | ,
\end{equation*}
so that
\begin{equation*}
\mathbb{P}\left( \mathsf{j}_1^1=j_1, \dots, \mathsf{j}_1^t=j_t  \right) =  \langle \mathbf{v}_1 | \otimes \cdots \otimes \langle \mathbf{v}_t | \cdot | e_{j_1} \rangle \otimes \cdots \otimes | e_{j_t} \rangle.
\end{equation*}
Adopting this matrix notation we can translate equality \eqref{j's are independent} into the eigenrelation
\begin{equation} \label{eigenrelation T}
\langle \mathbf{v}_1 | \otimes \cdots \otimes \langle \mathbf{v}_t | \cdot \sum_{m \geq 0} \pi_m^{(2)} T_t^{(m)}  =  \langle \mathbf{v}_1 | \otimes \cdots \otimes \langle \mathbf{v}_t|, 
\end{equation}
where we used the shorthand
\begin{equation} \label{pi factor probability}
    \pi_m^{(k)} = \mathbb{P}(\mathsf{m}_k^0= m) = \left( \frac{ \mathpzc{d} }{ \xi_x s_x } \right)^m \frac{(s_k^2;q)_m}{(q;q)_m} \frac{( \mathpzc{d} / ( \xi_k s_k ) ; q)_\infty}{( \mathpzc{d} s_k / \xi_k ;q)_\infty}.
\end{equation}
We prove \eqref{eigenrelation T} by induction. When $t=1$, we have
\begin{equation} \label{eigenrelation T t=1}
\langle \mathbf{v}_1 | \cdot \sum_{m \geq 0} \pi_m^{(2)} \left( \mathcal{A}_1 + q^m \mathcal{B}_1 \right) = \langle \mathbf{v}_1 | \cdot \left(\mathcal{A}_1 + \frac{1 -  \mathpzc{d} / (s_2 \xi_2 ) }{ 1 - \mathpzc{d} s_2 / \xi_2 } \mathcal{B}_1 \right) = \langle \mathbf{v}_1 |,
\end{equation}
where the summation with respect to $m$ was performed using the expression \eqref{pi factor probability} for $\pi_\bullet^{(2)}$ and the $q$-binomial theorem \eqref{q binomial summation} and the second equality follows by direct inspection of the matrix product. We now assume that \eqref{eigenrelation T} is true for $t-1$ and from this we would like to show that the $t$ case follows. Using recursion \eqref{recursion T}, we write
\begin{equation}\label{T induction}
\begin{split}
\langle \mathbf{v}_1 | \otimes \cdots \otimes \langle \mathbf{v}_t | \cdot \sum_{m \geq 0} \pi_m^{(2)} T_t^{(m)} = \langle \mathbf{v}_1 | \otimes \cdots \otimes \langle \mathbf{v}_t | & \cdot \left( \sum_{m \geq 0} \pi_m^{(2)} T_{t-1}^{(m)} \otimes \mathcal{A}_t \right. \\
& + \left. \sum_{m \geq 0} \pi_m^{(2)} q^m 
\left( \begin{smallmatrix}
1&0\\0&q
\end{smallmatrix} \right)^{\otimes (t-1)}  T_{t-1}^{(m)} \left( \begin{smallmatrix}
1&0\\0&1/q
\end{smallmatrix} \right)^{\otimes (t-1)} \otimes \mathcal{B}_t \right)
\end{split}
\end{equation}
and we see that the more complicated term to analyze is the second addend in the right hand side, as the first one becomes
\begin{equation} \label{T induction rhs 1}
\langle \mathbf{v}_1 | \otimes \cdots \otimes \langle \mathbf{v}_{t-1} | \otimes  \left( \langle \mathbf{v}_t | \cdot \mathcal{A}_t \right), 
\end{equation}
using the inductive hypothesis and the Kronecker rule for multiplication of tensor products of matrices and vectors. From the computation of $q$ moments of the probability measure $\pi_\bullet^{(2)}$ it is easy to see that
\begin{equation*}
\sum_{m \geq 0} q^{n m} \pi_m^{(2)} =  \frac{1 -  \mathpzc{d} / (s_2 \xi_2) }{ 1 - \mathpzc{d} s_2 / \xi_2 } \sum_{m\geq 0} q^{(n-1)m} {\pi^{(2)}_m}_{\big|_{ \mathpzc{d} \to q \mathpzc{d} }},
\end{equation*}
which implies the identity
\begin{equation} \label{T shift identity}
\sum_{m \geq 0} \pi_m^{(2)} q^m T_{t-1}^{(m)} = \frac{1 -  \mathpzc{d} / (s_2 \xi_2) }{ 1 - \mathpzc{d} s_2 / \xi_2 } \sum_{m \geq 0} {\pi_m^{(2)}}_{\big|_{ \mathpzc{d} \to q \mathpzc{d} }} T_{t-1}^{(m)},
\end{equation}
where the subscript ${\big|_{ \mathpzc{d} \to q \mathpzc{d} }}$ in the previous two equations denotes that in $\pi_m^{(2)}$ we substitute every $\mathpzc{d}$ with $ q \mathpzc{d}$.
Using \eqref{T shift identity} and the inductive hypothesis again, we can evaluate the second addend in the right hand side of \eqref{T induction} as
\begin{multline} \label{T induction rhs 2}
\langle \mathbf{v}_1 | \otimes \cdots \otimes \langle \mathbf{v}_{t-1} | \cdot \left( \begin{smallmatrix}
1&0\\0&q
\end{smallmatrix} \right)^{\otimes (t-1)} \left(  \frac{1 -  \mathpzc{d} / (s_2 \xi_2 ) }{ 1 - \mathpzc{d} s_2 / \xi_2 } \sum_{m \geq 0} {\pi_m^{(2)}}_{\big|_{ \mathpzc{d} \to q \mathpzc{d} }} T_{t-1}^{(m)} \right) \left( \begin{smallmatrix}
1&0\\0&1/q
\end{smallmatrix} \right)^{\otimes (t-1)} \otimes \left( \mathbf{v}_t \cdot \mathcal{B}_t \right) \\
 = \langle \mathbf{v}_1 | \otimes \cdots \otimes \langle \mathbf{v}_{t-1} | \otimes \left( \frac{1 - \mathpzc{d} / (s_2 \xi_2) }{ 1 - \mathpzc{d} s_2 / \xi_2 } \langle \textbf{v}_t | \cdot \mathcal{B}_t \right).
\end{multline}
So far we were able to transform the right hand side of \eqref{T induction} in the sum of \eqref{T induction rhs 1} and of the right hand side of \eqref{T induction rhs 2}. Employing identity \eqref{eigenrelation T t=1} we recover \eqref{eigenrelation T} for the general $t$ case, which completes the proof.
\end{proof}

Statement of Lemma \ref{lemma bc propagation horiz} implies a certain propagation of the boundary conditions \eqref{boundary conditions burke J=1} in the horizontal direction. The next Lemma addresses their propagation in the vertical direction.

\begin{Lemma} \label{lemma bc propagation vert}
Consider the Higher Spin Six Vertex Model on $\Lambda_{1,0}$ with $J=1$ and boundary conditions as in Lemma \ref{lemma bc propagation horiz}. Then, for all $t \geq 1$ the sequence $\mathsf{m}_2^{t-1}, \mathsf{m}_3^{t-1}, \mathsf{m}_4^{t-1}, \dots, \mathsf{j}_1^{t}, \mathsf{j}_1^{t+1}, \mathsf{j}_1^{t+2}, \dots$ is a family of independent random variables and for each $x \geq 2$ we have $\mathsf{m}_x^{t-1} \sim q\text{\emph{NB}}( s^2_x, \mathpzc{d}/(\xi_x s_x) )$.
\end{Lemma}

\begin{proof}
By similar argument as in the proof of Lemma \ref{lemma bc propagation horiz} the two families of random variables $\mathsf{m}_2^{t-1}, \mathsf{m}_3^{t-1}, \mathsf{m}_4^{t-1}, \dots$ and $\mathsf{j}_1^{t}, \mathsf{j}_1^{t+1}, \mathsf{j}_1^{t+2}, \dots$ are always independent and in order to prove Lemma \ref{lemma bc propagation vert} it is sufficient to show that, for all $x \geq 2$ and for all $(m_2, \dots, m_x) \in \mathbb{Z}_{\geq 0}^{x-1}$ we have

\begin{equation}\label{vert propagation}
    \mathbb{P}(\mathsf{m}_2^1=m_2, \dots , \mathsf{m}_x^1=m_x ) = \prod_{k=2}^x \pi_{m_k}^{(k)},
\end{equation}
where the quantities $\pi_m^{(k)}$'s were introduced in \eqref{pi factor probability}. Also in this case we follow a rather algebraic approach. First introduce the operator
\begin{equation*}
    \mathcal{U}^{(k)} = \sum_{m, m' \geq 0} | e_m \rangle \begin{pmatrix}
  \mathsf{L}_{\xi_k u, s_k}(m,0|\ m',0) &  \mathsf{L}_{\xi_k u, s_k}(m,0|\ m',1) \\
  \mathsf{L}_{\xi_k u, s_k}(m,1|\ m',0) & \mathsf{L}_{\xi_k u, s_k}(m,1|\ m',1) \\ 
 \end{pmatrix} \langle e_{m'} |
\end{equation*}
and vectors
\begin{equation*}
    \langle \mathbf{w}_k | = \sum_{m \geq 0} \pi_m^{(k)} \langle e_k | , \qquad \mathbf{v}^T= \frac{1}{1 - \mathpzc{d} u } \left(\! \!
    \begin{array}{c}
      1 \\
      - \mathpzc{d} u
    \end{array}
  \! \! \right).
\end{equation*}
Notice that operator $\mathcal{U}^{(k)}$ is the one vertex analog of the transfer operator $\mathfrak{X}_u$ and that the distribution of $\mathsf{m}_2^{1}, \dots, \mathsf{m}_x^{1}$ is given by
\begin{equation} \label{ver propag prob operator form}
    \mathbb{P}(\mathsf{m}_2^1=m_2, \dots, \mathsf{m}_x^1 = m_x) = \mathbf{v} \cdot \left( \bigotimes_{k=2}^x \langle \mathbf{w}_k | \right) \cdot \left( \bigotimes_{k=2}^x  \mathcal{U}^{(k)} \right) \cdot \left( \bigotimes_{k=2}^x | e_{m_k} \rangle \right) \cdot \left(\! \!
    \begin{array}{c}
      1 \\
      1
    \end{array}
  \! \! \right).
\end{equation}
By direct inspection we easily see that 
\begin{equation} \label{ver propag w U}
    \langle \mathbf{w}_k | \cdot \mathcal{U}^{(k)} = \sum_{m' \geq 0} \left( \hat{\mathcal{A}}_k + q^{m'} \hat{\mathcal{B}}_k \right) \pi_{m_k}^{(k)} \langle e_{m'} | ,
\end{equation}
where matrices $\hat{\mathcal{A}}_k, \hat{\mathcal{B}}_k$ are given by
\begin{equation*}
\hat{\mathcal{A}}_k = \frac{1}{1-s_k \xi_k u} \begin{pmatrix}
  1 & - \mathpzc{d} u \\
  \xi_k s_k /\mathpzc{d} & - \xi_k s_k u \\ 
 \end{pmatrix}, \qquad 
 \hat{\mathcal{B}}_k = \frac{1}{1-s_k \xi_k u} \begin{pmatrix}
  - \xi_k s_k u & \mathpzc{d} u s_k^2 \\
  -\xi_k s_k / \mathpzc{d} & s_k^2 \\ 
 \end{pmatrix}.
\end{equation*}
By means of \eqref{ver propag w U} and of the Kronecker rule for multiplication of tensor products we see that the right hand side of \eqref{ver propag prob operator form} reduces to 
\begin{equation} \label{rhs ver propag prob}
    \mathbf{v} \cdot \prod_{k=2}^x \left( \hat{\mathcal{A}}_k + q^{m_k} \hat{\mathcal{B}}_k \right) \cdot \left(\! \!
    \begin{array}{c}
      1 \\
      1
    \end{array}
  \! \! \right) \times \prod_{k=2}^x \pi_{m_k}^{(k)},
\end{equation}
which would prove \eqref{vert propagation} in case the scalar product $\mathbf{v} \cdot \prod_{k=2}^x ( \hat{\mathcal{A}}_k + q^{m'} \hat{\mathcal{B}}_k ) \cdot \binom{1}{1}$ is equal to one. This is a consequence of certain product identities involving matrices $\hat{\mathcal{A}}_k, \hat{\mathcal{B}}_k$. In particular, for all $k, k'$ we have 
\begin{equation*}
    \hat{\mathcal{A}}_k \hat{\mathcal{A}}_{k'} = \hat{\mathcal{A}}_k, \qquad \hat{\mathcal{A}}_k \hat{\mathcal{B}}_{k'} = 0, \qquad \hat{\mathcal{B}}_k \hat{\mathcal{A}}_{k'} = c_{k, k'} \hat{\mathcal{C}}, \qquad \hat{\mathcal{B}}_k \hat{\mathcal{B}}_{k'} = d_{k,k'}\hat{\mathcal{B}}_{k'},
\end{equation*}
where $c_{k,k'}, d_{k,k'}$ are constants depending on parameters $s,\xi,u$ and
\begin{equation*}
    \hat{\mathcal{C}} = \begin{pmatrix}
  1 & - \mathpzc{d} u \\
  1 / (\mathpzc{d} u) & - 1 \\ 
 \end{pmatrix}.
\end{equation*}
Since $\mathbf{v} \cdot \hat{\mathcal{C}} \cdot \binom{1}{1} = \mathbf{v} \cdot \hat{\mathcal{B}}_k \cdot \binom{1}{1} = 0 $ for all $k$, we can now compute \eqref{rhs ver propag prob} as
\begin{equation*}
    \eqref{rhs ver propag prob} = \mathbf{v} \cdot \hat{\mathcal{A}}_2 \cdot \left(\! \!
    \begin{array}{c}
      1 \\
      1
    \end{array}
  \! \! \right) \times \prod_{k=2}^x \pi_{m_k}^{(k)} = \prod_{k=2}^x \pi_{m_k}^{(k)},
\end{equation*}
which proves \eqref{vert propagation}.
\end{proof}

We can now summarize results obtained so far in this Subsection and extend them to the slightly more general setting of the model with fused rows.

\begin{proof}[Proof of Proposition \ref{prop translation invariant}]
The case $J=1$ of Proposition \ref{prop translation invariant} is obtained combining Lemmas \ref{Lemma local density J=1}, \ref{lemma bc propagation horiz}, \ref{lemma bc propagation vert}. To prove the general $J$ case we need to show that $\mathsf{j}_{x-1}^1, \mathsf{j}_{x-1}^2, \mathsf{j}_{x-1}^3, \dots$ are independent $q$-negative binomial random variables respectively of parameters $(q^{-J} , q^J u_t \mathpzc{d})$ for $t= 1,2,3, \dots$, even when $J \neq 1$. Recall that the fusion of rows is obtained by collapsing together $J$ rows of vertices with spectral parameters taken in geometric progression of ratio $q$ (see Proposition \ref{prop BP q exch}). Since, as a result of Lemma \ref{lemma bc propagation horiz}, for each $x\geq 2$, random variables $\mathsf{j}_{x-1}^1, \mathsf{j}_{x-1}^{2}, \mathsf{j}_{x-1}^{3}, \dots $ are independently distributed, the proof reduces to show that, in the unfused model
\begin{equation*}
\mathbb{P} \left( \mathsf{j}_{x-1}^1 + \cdots, \mathsf{j}_{x-1}^{J} = k \right) = \left( q^J \mathpzc{d} u \right)^k \frac{(q^{-J},q)_k}{(q;q)_k} \frac{(q^J \mathpzc{d} u ;q)_\infty}{( \mathpzc{d} u;q)_\infty}, \qquad \text{for all }k\geq 0,
\end{equation*}
when $u_{1}= u,  u_2 = q u, \dots, u_{J} = q^{J-1} u$, which is the statement of Proposition \ref{prop sum bernoulli geom prog}. This concludes the proof.
\end{proof}

Employing result of Proposition \ref{prop translation invariant} we are able to extend the Higher Spin Six Vertex Model to the full lattice $\mathbb{Z} \times \mathbb{Z}$.

\begin{Prop} \label{prop extension HS6VM}
Take paramaters
\begin{equation*}
\mathbf{U}= (\dots, u_{-1}, u_0, u_1, \cdots), \qquad \mathbf{S}= (\dots, s_{-1}, s_0, s_1, \dots), \qquad \Xi = (\dots, \xi_{-1}, \xi_0, \xi_1, \dots)
\end{equation*}
fulfilling conditions \eqref{HS6VM parameters}. Take also a parameter $\mathpzc{d}$ to fulfill condition \eqref{bound curvy d}. Then there exists a probability measure on the set of directed up right paths on $\mathbb{Z} \times \mathbb{Z}$, such that, for each choice of $(x,t)$,
\begin{equation*}
\mathbb{P}\left( \mathsf{m}_x^t=m', \mathsf{j}_x^t = j' |\ \mathsf{m}_x^{t-1}=m, \mathsf{j}_{x-1}^t = j \right) = \mathsf{L}^{(J)}_{u_t \xi_x, s_x} \left( m,j|\ m', j' \right)
\end{equation*}
and $\mathsf{j}_{x-1}^{t}, \mathsf{j}_{x-1}^{t+1}, \mathsf{j}_{x-1}^{t+2}, \dots, \mathsf{m}_{x}^{t-1}, \mathsf{m}_{x+1}^{t-1}, \mathsf{m}_{x+2}^{t-1}, \dots$ are independent random variables distributed as
\begin{equation}\label{eq: boundary cond extension}
\mathsf{m}_{x+k}^{t-1} \sim q\text{\emph{NB}}( s^2_{x+k}, \mathpzc{d}/(\xi_{x+k} s_{x+k})), \qquad \mathsf{j}_{x-1}^{t+k} \sim q\text{\emph{NB}}( q^{-J}, q^{J} \mathpzc{d} u_{t+k}),
\end{equation}
for each $k$. In other words it is possible to define the Higher Spin Six Vertex Model on the lattice $\mathbb{Z} \times \mathbb{Z}$ in such a way that it satisfies the Burke's property and such that, for each choice of vertex $(x,t)$, paths entering the restricted lattice $\Lambda_{x-1,t-1}$ are distributed as in \eqref{eq: boundary cond extension}. We refer to this model as the \emph{full plane Stationary Higher Spin Six Vertex Model}.
\end{Prop}

\begin{proof}
The procedure we follow to prove this extension result is fairly standard and it has been utilized, for the translation invariant Six Vertex Model case in \cite{Aggarwal2016FluctuationsASEP}, Appendix A.

We call $\mathbb{P}_{N}$ the probability measure of the Higher Spin Six Vertex Model defined on the lattice $\Lambda_{-N,-N}$ with boundary conditions given by
\begin{equation*}
\mathsf{m}_{-N + k}^{-N} \sim q\text{NB}(s^2_{-N+k}, \mathpzc{d}/(\xi_{-N+k} s_{-N+k})), \qquad \mathsf{j}_{-N}^{-N + k} \sim q\text{NB}( q^{-J}, q^{J} \mathpzc{d} u_{-N+k}).
\end{equation*}
Let $E_N$ be an event involving only configurations of paths in $\Lambda_{-N,-N}$. It is clear, from Proposition \ref{prop translation invariant}, that, for each $N^* >N$, 
\begin{equation*}
\mathbb{P}_{N^*}(E_N) = \mathbb{P}_{N}(E_N)
\end{equation*}
and therefore $\mathbb{P}_{N^*}$ extends $\mathbb{P}_N$. By using the Caratheodory's extension theorem we can take the limit $N^* \to \infty $ and deduce the existence of a measure $\mathbb{P}_{\infty}$ to finally define the translation invariant Higher Spin Six Vertex Model on $\mathbb{Z}^2$.
\end{proof}

We close this Subsection explaining the reason behind the use of the terminology "Burke's property" of Definition \ref{def: Burke property}. Here we rephrase a generalization by Ferrari and Fontes of a theorem by Burke \cite{Burke,FerrariFontes1994NetOutput}, stated for queuing systems in a language more familiar to us.

\begin{theorem}[Burke] Let $\{y_x\}_{x\in \mathbb{Z}}$ be a totally asymmetric simple exclusion process where, at time $t=0$, $y_1 = -1$ a.s. and consecutive particles are spaced independently with geometric distribution of parameter $\mathpzc{d}$. Then, the distribution of gaps is stationary in time and the marginal distribution of $y_1$ is that of a Poisson process with rate $1 - \mathpzc{d}$.
\end{theorem}

In the Higher Spin Six Vertex Model on the full plane $\mathbb{Z} \times \mathbb{Z}$ defined by Proposition \ref{prop extension HS6VM}, regarding the generic occupancy number $\mathsf{m}_x^t$ as the gap between the $(x-1)$-th and the $x$-th particle of a process $\{ y_x \}_{x \in \mathbb{Z} }$ at time $t$, we obtain a discrete time generalization of the totally asymmetric simple exclusion process \cite{CorwinPetrov2016HSVM}. A consequence of Proposition \ref{prop extension HS6VM} is that, in this generalized model, when at time $t=0$ we set $y_1=-1$ a.s. and consecutive particles are independently spaced with $q$-negative binomial distribution of parameters $( s_x^2, \mathpzc{d}/(\xi_x, s_x))$, then the marginal process $y_1$ is equivalent, in distribution, to a sequence of independent jumps with $q$-negative binomial distribution of parameters $( q^{-J}, q^J \mathpzc{d} u_t)$. This analogy should justify our choice of words. 

A concept analogous to the Burke's property stated in Definition \ref{def: Burke property} appeared already in literature in the context of random polymers (\cite{balazs2006},\cite{seppalainen2012}). In particular, in \cite{seppalainen2012} the author considers the log-Gamma directed polymers model with random external sources. This sort of model is known to be described by the so called $\alpha$-Whittaker processes \cite{BorodinCorwin2014Mac}, of which the $q$-Whittaker processes presented in Section \ref{section qwhittaker} represent a "quantized" generalization. Although we do not describe here relations between the model studied in \cite{seppalainen2012} and the Stochastic Higher Spin Six Vertex Model, we will say that the role played by random external sources in polymer models is analogous to that played by boundary conditions in the Higher Spin Six Vertex Model or to that played by random initial conditions for totally asymmetric simple exclusion processes.

\subsection{Exactly solvable boundary conditions} \label{subsection integrable initial}

In the previous Subsection we characterized the family of Higher Spin Six Vertex Models satisfying the Burke's property. Here we explain how the study of the model with double sided $q$-negative binomial boundary conditions is accessible by properly specializing parameters $\Xi,\mathbf{S}, \mathbf{U}$ starting from a Higher Spin Six Vertex Model with step boundary conditions.

For later purpose we now introduce the quantities
\begin{equation} \label{weight l}
\boldsymbol\ell^{(i)}_{\wp,v}(M;\overline{M}) =  \Big( \frac{v}{\xi_i s_i} \Big)^M \frac{(s_i^2 ,\wp q^{\overline{M}};q)_M}{( v \wp q^{\overline{M}} s_i / \xi_i,q;q)_M} \frac{( v \wp q^{\overline{M}} s_i /\xi_i  ,v/(\xi_i s_i);q)_\infty}{( v \wp q^{\overline{M}} /(\xi_i s_i ) ,v s_i /\xi_i ;q)_\infty} ,
\end{equation}
which are families of probability mass functions (in $M$), provided $\overline{M}$ is a non-negative integer and parameters $\wp, v$ satisfy one of the two conditions
\begin{equation} \label{condition wp v 1}
    v \text{ as in \eqref{condition v}} \qquad \text{and} \qquad \wp < 1,
\end{equation}
or
\begin{equation}\label{condition wp v 2}
    v < 0 \qquad \text{and} \qquad \wp = q^{-K}, \qquad \text{for } K \in \mathbb{Z}_{\geq \overline{M}}.
\end{equation}
In case $v, \wp$ are taken according to \eqref{condition wp v 2}, $\boldsymbol\ell^{(i)}_{\wp,v}(\bullet;\overline{M})$ is supported on the set $\{0, K - \overline{M} \}$, whereas when they are taken as in \eqref{condition wp v 1}, expression \eqref{weight l} takes positive values  for each $M \in \mathbb{Z}_{\geq 0}$. In both cases the sum-to-one condition in $M$ is guaranteed by the $q$-Gauss summations \eqref{q Gauss sum}.

The next definition is rather technical and aims to describe the most general set of boundary conditions we will cover in the remaining part of the paper.

\begin{Def} \label{def: bc general}
Consider a random variable $\mathsf{m}\sim q\emph{NB}(\wp,v/\mathpzc{d})$, with $\mathpzc{d} > \max(0,v)$. We denote with the symbol $\mathbb{P}_{\wp, v, \mathpzc{d}}$ the probability measure of a coupling of $\mathsf{m}$ with a Higher Spin Six Vertex Model on $\Lambda_{1,0}$, in which $\mathsf{j}_1^t \sim q\emph{NB}(q^{-J}, q^{J} u_t \mathpzc{d})$ for $t = 1, 2, 3 \dots$ are independent random variables and $\mathsf{m}_2^0, \mathsf{m}_3^0, \mathsf{m}_4^0, \dots$, conditionally to $\mathsf{m}$, have law
\begin{equation} \label{eq: bc general}
    \mathbb{P}_{\wp, v, \mathpzc{d}} (\mathsf{m}_2^0=m_2, \dots, \mathsf{m}_x^0=m_x | \mathsf{m} = m_1) = \prod_{i=2}^x \boldsymbol\ell^{(i)}_{\wp,v} \left( m_i;\sum_{j=1}^{i-1}m_j \right).
\end{equation}
For this particular model we introduce the shifted height function
\begin{equation} \label{height H bar}
\overline{\mathcal{H}}(x,t) =  \mathcal{H}(x,t) -\mathsf{m}.
\end{equation}
When parameters $\wp, v$ are taken as in \eqref{condition wp v 1} and $\wp = 0$, $\mathsf{m}$ and the Higher Spin Six Vertex Model become independent processes and in this case we use the notation $\mathbb{P}_{0, v, \mathpzc{d}} = \mathbb{P}_{\text{\emph{HS}}(v,\mathpzc{d}) \otimes \mathsf{m} }$.
\end{Def}

Indeed the probability measure $\mathbb{P}_{\wp, v, \mathpzc{d}}$ introduced in Definition \ref{def: bc general} represents a generalization of the double sided $q$-negative binomial Higher Spin Six Vertex Model. The reason why the choice $\wp=0$ decouples $\mathsf{m}, \mathsf{m}_2^0, \mathsf{m}_3^0, \dots$ comes from the exact expression \eqref{weight l} of weights $\boldsymbol\ell^{(i)}_{\wp,v}$. In fact, for any $k$, the law of $\mathsf{m}_k^0$ depends on the outcome of $\mathsf{m}_2^0, \mathsf{m}_3^0, \dots, \mathsf{m}_{k-1}^0$ and $\mathsf{m}$ only when the factor $\wp q^{\overline{M}}$ is different than zero. By setting $\wp=0$ we see that $\mathsf{m}_2^0, \mathsf{m}_3^0, \dots$ become independent $q$-negative binomials of parameters respectively $(s_k^2, v/(\xi_k s_k))$ for $k=2,3,\dots$, whereas $\mathsf{m}$ becomes a $q$-Poisson random variable of parameter $v/\mathpzc{d}$ independent of the rest of the process. The reason why we consider a coupling between the Higher Spin Six Vertex Model and the random variable $\mathsf{m}$ becomes clear with the construction we present next. We claim in fact that the measure $\mathbb{P}_{\wp, v, \mathpzc{d}}$ is obtained as a marginal process of a certain specialization of the Higher Spin Six Vertex Model with step boundary conditions. We give the following 
\begin{Prop} \label{prop: bc general}
Let $K \in \mathbb{Z}_{\geq 1}$ and consider the Higher Spin Six Vertex Model on the lattice $\Lambda_{0,-K}$ with step boundary conditions (that is $\mathsf{j}_0^t=1$ a.s. for $t \geq -K+1$ and $\mathsf{m}_x^{-K}=0$ a.s. for $x \geq 1$).  Spectral parameters are taken as 
\begin{equation} \label{parameter initial conditions}
\mathbf{U}=(q/v, q^2/v, \dots, q^K/v ) \cup (u_1,u_2,\dots),
\end{equation}
where $v<0$ and at $x=1$, $\xi_1, s_1$ are given setting
\begin{equation} \label{first s xi}
s_1=1/N, \qquad \xi_1 = \mathpzc{d} N \qquad \text{and taking the limit }N \to \infty.
\end{equation}
Then, the marginal process on the lattice $\Lambda_{1,0}$ is described by the law $\mathbb{P}_{\wp, v, \mathpzc{d}}$, presented in Definition \ref{def: bc general} with parameters $v, \wp$ as is \eqref{condition wp v 2}.
\end{Prop}

The proof of Proposition \ref{prop: bc general} boils down to finding the following simplified expression for the fused vertex weight \eqref{fused weights L}.

\begin{Lemma} \label{prop fused weight l simplified}
We have
\begin{equation} \label{fused weight L simplified}
\mathsf{L}_{\xi qv^{-1} ,s}^{(K)}(0, j_1|\ i_2, j_1-i_2) = \left( \frac{v}{\xi  s} \right)^{i_2} \frac{(s^2,q^{-j_1};q)_{i_2}}{( v q^{-j_1} s / \xi ,q;q)_{i_2} } \frac{ ( v q^{-j_1}s / \xi, v / (\xi  s) ,q)_\infty  }{ ( v q^{-j_1} / (\xi s ) , v s / \xi  ;q)_\infty } .
\end{equation}
\end{Lemma}
\begin{proof}
From the exact expression of the weight $\mathsf{L}^{(K)}$ in \eqref{fused weights L}, setting the number $i_1$ of paths entering the vertex from below to zero, its rather complicated formula simplifies to
\begin{equation*}
\mathsf{L}_{\xi qv^{-1} ,s}^{(K)}(0, j_1|\ i_2, j_1-i_2)=\frac{(s^2;q)_{i_2}}{(q;q)_{i_2}}\frac{(q^{1+j_1-i_2};q)_{i_2}s^{2(j_1-i_2)}(q \xi /( v s ) ;q)_{j_1-i_2}}{(q\xi s /v ;q)_{j_1}}.
\end{equation*}
By multiplying and dividing this last expression by 
$$
\prod_{l=0}^{i_2-1}(s^2-\xi s q^{j_1 -i_2 +l+1 }/v)
$$
and taking out of the product all factors depending only on $j_1$ we get a term proportional to
\begin{equation} \label{phi21 factor}
\left( \frac{v}{\xi s} \right)^{i_2} \frac{(s^2;q)_{i_2} (q^{-j_1};q)_{i_2}}{(q;q)_{i_2} (v q^{-j_1}s / \xi;q)_{i_2}}.
\end{equation}
Result \eqref{fused weight L simplified} is obtained normalizing \eqref{phi21 factor} so that its sum over all $i_2$ is one and this is done by means of the $q$-Gauss summation \eqref{q Gauss sum}.
\end{proof}

\begin{proof}[Proof of Proposition \ref{prop: bc general}]
First we observe that choice \eqref{first s xi} generates $q$-negative binomial random entries in the vertical boundary of the lattice $\Lambda_{1,0}$. By substituting the values of $\xi_1, s_1$ in the definition of transition probabilities $\mathsf{L}$, we have
\begin{equation}\label{limit L J=1}
\lim_{N \to \infty} \mathsf{L}_{u_t \mathpzc{d}N,1/N}(m,1|\ m,1) = \frac{ - u_t \mathpzc{d}}{ 1- u_t \mathpzc{d} },
\end{equation}
which also implies, using Proposition \ref{prop sum bernoulli geom prog},
\begin{equation} \label{limit LJ}
\lim_{N \to \infty} \mathsf{L}_{u_t \mathpzc{d}N,1/N}^{(J)}(m,J|\ m+J-l,l) = \left( q^J \mathpzc{d} u_t \right)^l \frac{(q^{-J},q)_l}{(q;q)_l} \frac{(q^J \mathpzc{d} u_t ;q)_\infty}{( \mathpzc{d} u_t;q)_\infty}.
\end{equation}
This procedure of obtaining independent random entries at column $\{(2,t)\}_{t \geq 1}$ is alternative to that presented in Section \ref{section HS}, where in \eqref{limit L J=1 g=inf 1}, \eqref{limit L J=1 g=inf 2} a result analogous to \eqref{limit L J=1} was achieved setting $\mathsf{m}_1^0 = \infty$ a.s.. Here the value of $\mathsf{m}_1^0$ depends on the process on the strip $\mathbb{Z}_{\geq 1} \times \{ -K+1, \dots, 0 \}$ and it is in general not infinite. From \eqref{parameter initial conditions} we see that the first $K$ spectral parameters (those related to non-positive ordinates $t$) are in geometric progression of ratio $q$ and therefore we can use the notion of fused transfer operator $\mathfrak{X}_{q/v}^{(K)}$, formally given by \eqref{fused transfer matrix}, to calculate the probability
\begin{equation} \label{probability first x vertices}
\begin{split}
\mathbb{P}\left(\mathsf{m}_1^0=m_1, \dots ,\mathsf{m}_x^0=m_x  \right)&=\sum_{m_{x+1},m_{x+2},\dots} \mathfrak{X}^{(K)}_{q/v}(K,\emptyset \to 1^{m_1}2^{m_2}\cdots)\\
&=\mathsf{L}_{q \xi_1 v^{-1},s_1}^{(K)}(0, K|\ m_1, j_1) \prod_{i=2}^x \mathsf{L}_{q \xi_i v^{-1},s_i}^{(K)}(0,j_{i-1}|\ m_i, j_i),
\end{split}
\end{equation}
where $j_i=K-m_1 -\dots -m_i$ for $i=1, \dots, x$. In the last expression we took account of the boundary conditions and we let no path enter the axis $\mathbb{Z}_{\geq 1} \times \{ -K+1 \}$ from below and exactly $K$ paths entered the region $\mathbb{Z}_{\geq 1} \times \{ - K +1 , \dots, 0\}$ from the leftmost column of vertices. All factors in the right hand side of \eqref{probability first x vertices} are of the form
$$
\mathsf{L}^{(K)}_{\xi q v^{-1},s} (0,K-\overline{M},M,K-\overline{M}-M),
$$
for some integers $M, \overline{M}$, so that using result of Lemma \ref{fused weight L simplified} and expression of weights $\boldsymbol\ell^{(i)}_{\wp, v}$ we obtain 
\begin{equation*}
    \mathbb{P} (\mathsf{m}_1^0=m_1, \mathsf{m}_2^0=m_2, \dots, \mathsf{m}_x^0=m_x) = \prod_{i=1}^x \boldsymbol\ell^{(i)}_{\wp,v} \left( m_i;\sum_{j=1}^{i-1}m_j \right),
\end{equation*}
which completes the proof, after identifying $\mathsf{m}$ with the random variable $\mathsf{m}_1^0$.
\end{proof}
Result of Proposition \ref{prop: bc general} opens the door to study the measure $\mathbb{P}_{\wp, v, \mathpzc{d}}$, at least when $\wp=q^{-K}$, using integral formulas for $q$-moments presented in Section \ref{subsection observables HS6VM}. We recall that results like \eqref{height q moments} are available only for the particular choice of step boundary conditions and following construction presented in Proposition \ref{prop: bc general} they are extended  to boundary conditions given by Definition \ref{def: bc general}. We remark that at this stage we are not yet ready to study the Higher Spin Six Vertex Model in the case of double sided $q$-negative binomial boundary conditions, but only when the distribution of $\mathsf{m}_2^0, \mathsf{m}_3^0, \dots$ is of the form \eqref{eq: bc general} in which probability weights $\boldsymbol\ell^{(i)}_{\wp,v}$ are considered with $\wp, v$ as in \eqref{condition wp v 2}. We devote the remaining part of this Subsection to extend integrability results of the measure $\mathbb{P}_{\wp, v, \mathpzc{d}}$ also to the region of parameters $\wp, v$ in \eqref{condition wp v 1}. The strategy we follow is an analytic continuation of the probability distribution of the shifted height function $\overline{\mathcal{H}}$. 

Following the construction provided in Proposition \ref{prop: bc general}, we recover the equality 
\begin{equation}\label{current H}
\mathfrak{h}(x+1,t)-K\stackrel{\mathcal{D}}{=}\overline{\mathcal{H}}(x,t),
\end{equation}
for all meaningful $x,t$, where the left hand side refers to a Higher Spin Six Vertex Model on $\Lambda_{0,-K}$ with step boundary conditions and employing relation \eqref{current H} we write the one point probability distribution of $\overline{\mathcal{H}}$ using that of $\mathfrak{h}$. Following techniques analogous to those used in \cite{BCPS2015SpectalTheory}, \cite{Aggarwal2016FluctuationsASEP}, we now provide a description of the probability mass function of $\overline{\mathcal{H}}(x,t)$ when the probability measure is considered both with choices of parameters \eqref{condition wp v 1} or \eqref{condition wp v 2}.

\begin{Prop} \label{analytical extension probability proposition}
Consider the probability measure $\mathbb{P}_{\wp, v, \mathpzc{d}}$ introduced in Definition \ref{def: bc general} and assume that parameters $\Xi, \mathbf{S,U}$ satisfy \eqref{HS6VM parameters} and the additional bounds
\begin{equation*}
    \sup_{i}\{ \xi_i s_i s_j /\xi_j \} < 1, \qquad \sup_i \{ s_i/\xi_i \} <1, \qquad |\wp| < |1/v| \times \inf_i \{ \xi_i s_i\}.  
\end{equation*}
Then, we have
\begin{equation}\label{analytical extension probability}
\begin{split}
&\mathbb{P}_{\wp, v, \mathpzc{d}} \left(\overline{\mathcal{H}}(x,t)=l \right)\\
&=(q;q)^{x-1} \int_{\mathbb{T}^x} \prod_{j=1}^x \frac{dz_j}{z_j} m^x_q(\mathbf{z})\frac{\tilde{\Pi}(\mathbf{z}, \Xi^{-1} \mathbf{S}, \hat{\mathbf{U}})}{\tilde{\Pi}(\Xi \mathbf{S}, \Xi^{-1} \mathbf{S}, \hat{\mathbf{U}})} \frac{(\frac{\Xi \mathbf{S}}{Z};q)_\infty}{\prod_{i,j=1}^x(\frac{\xi_i s_i}{z_j};q)_{\infty}} \left( \frac{\Xi \mathbf{S}}{Z} \right)^l \prod_{j=1}^x \frac{(\wp\frac{v}{z_j};q)_\infty (\frac{v}{\xi_j s_j};q)_\infty}{(\wp\frac{v}{\xi_j s_j};q)_\infty (\frac{v}{z_j};q)_\infty},
\end{split}
\end{equation}
where $\Xi S/Z=\prod_{i=1}^x \xi_i s_i/z_i$, $m_q^x$ is the $q$-Sklyanin measure \eqref{q-Sklyanin measure} and the the factor $\tilde{\Pi}$ is given by
\begin{equation} \label{Pi tilde}
\tilde{\Pi}(\emph{z}, \Xi^{-1} \mathbf{S}, \hat{\mathbf{U}})=\prod_{j=1}^x\left(\prod_{i=2}^x ( z_js_i/\xi_i;q)^{-1}_\infty\prod_{i=1}^t(z_ju_i;q)_J\right).
\end{equation}
\end{Prop}

The proof of Proposition \ref{analytical extension probability proposition} makes use of the matching of $q$-moments between the height in the Higher Spin Six Vertex Model and the corner coordinate in a $q$-Whittaker process stated in Proposition \ref{matching proposition} and based on a result of \cite{OrrPetrov2016}. We have the following 

\begin{Lemma}
Consider the probability measure $\mathbb{P}_{\wp,v,\mathpzc{d}}$ as in Proposition \ref{analytical extension probability proposition} with parameters $\wp,v$ as in \eqref{condition wp v 2}. Then, we have
\begin{equation} \label{matching useful}
\mathbb{E}_{q^{-K},v,\mathpzc{d}}\left( q^{l (\overline{\mathcal{H}} (x,t) +K) } \right) = \mathbb{E}_{\mathbb{W}_{\Xi,\mathbf{S},\mathbf{U}}} \left( q^{l \lambda_x} \right),
\end{equation}
where the right hand side refers to the $q$-Whittaker measure \eqref{q-Whittaker} specialized as in \eqref{parameter initial conditions}, \eqref{first s xi}.
\end{Lemma}
\begin{proof}
We know, from Proposition \ref{prop: bc general} that the probability measure $\mathbb{P}_{q^{-K},v,\mathpzc{d}}$ is obtained as a marginal process from a Higher Spin Six Vertex Model on $\Lambda_{0,-K}$ with step boundary conditions and parameters specialized as \eqref{parameter initial conditions}, \eqref{first s xi}. From this equivalence of models relation \eqref{current H} follows and we see that \eqref{matching useful} is obtained as a corollary of Proposition \ref{matching proposition}, since choice \eqref{first s xi} annihilates the term $q^n \prod_{j=1}^x s_j^2$ in \eqref{matching}.
\end{proof}

\begin{Lemma} \label{lemma probability height}
Consider the probability measure $\mathbb{P}_{\wp,v,\mathpzc{d}}$ as in Proposition \ref{analytical extension probability proposition} with parameters $\wp, v$ as in \eqref{condition wp v 2}. Then, we have
\begin{equation}\label{probability height}
\mathbb{P}_{ q^{-K},v,\mathpzc{d} }\left( \overline{\mathcal{H}}(x,t) +K = l \right)=(q;q)_{\infty}^{x-1}\int_{\mathbb{T}^x}\prod_{j=1}^x \frac{dz_j}{z_j}m_q^{x}(\mathbf{z}) \frac{\Pi(\mathbf{z};\Xi^{-1} \mathbf{S},\hat{\mathbf{U}})}{\Pi(\Xi \mathbf{S}; \Xi^{-1}\mathbf{S},\hat{\mathbf{U}})} \frac{(\Xi S/Z;q)_\infty}{\prod_{i,j=1}^x(\xi_j s_j/z_i;q)_\infty}\left(\frac{\Xi S}{Z} \right)^l,
\end{equation}
where $\Xi S/Z=\prod_{i=1}^x \xi_i s_i/z_i$, $m_q^x$ is the $q$-Sklyanin measure \eqref{q-Sklyanin measure} and the factor $\Pi$ in the integrand is given by
\begin{equation} \label{partition function}
\Pi(\mathbf{z};\Xi^{-1} \mathbf{S},\hat{\mathbf{U}})=\prod_{j=1}^x\left(\prod_{i=2}^x (z_js_i/\xi_i;q)^{-1}_\infty\prod_{i=1}^t(z_ju_i)_J ( q z_j/v;q)_K\right).
\end{equation}
\end{Lemma}
\begin{proof}
The matching of $q$-moments reported in Lemma \ref{lemma probability height} implies that the $q$-Laplace transforms $\mathbb{E}_{q^{-K},v,\mathpzc{d}}(1/(\zeta q^{\overline{\mathcal{H}}(x,t) +K } ;q)_\infty)$ and $\mathbb{E}_{\mathbb{W}_{\Xi, \mathbf{S,U}} }(1/(\zeta q^{\lambda_x };q)_\infty)$ are equal. This implies that
\begin{equation} \label{H bar = lambda}
    \mathbb{P}_{q^{-K},v,\mathpzc{d}}\left( \overline{\mathcal{H}}(x,t) + K = l\right) = \mathbb{W}_{\Xi, \mathbf{S,U}} (\lambda_x = l)
\end{equation}
for all $l$ in $\mathbb{Z}$, from which \eqref{probability height} follows after specializing \eqref{probability 2sided qWhittaker} according to \eqref{parameter initial conditions}, \eqref{first s xi}. 
\end{proof}

\begin{proof}[Proof of Proposition \ref{analytical extension probability proposition}] Lemma \ref{lemma probability height} established the claim of Proposition \ref{analytical extension probability proposition} for the choice \eqref{condition wp v 2} of parameters $\wp,v$, so in order to conclude our argument we need to extend such result to the region \eqref{condition wp v 1} as well. We will show that both sides of \eqref{analytical extension probability} are analytic functions of the variable $v$ in a neighborhood of zero. Expanding in Taylor series the equality \eqref{analytical extension probability} can be written as
\begin{equation} \label{expansion in v}
\sum_{n\geq 0} P_n(\wp) v^n = \sum_{n \geq 0} R_n(\wp)v^n.
\end{equation}
where the radius of convergence of both series depends on the magnitude of $\wp$ and it is given by conditions 
\begin{equation}\label{constraint analyticity}
\max_i \left| \frac{v\wp}{ \xi_i s_i} \right| < 1, \qquad \qquad  \max_i \left|\frac{v}{\xi_i s_i} \right|<1.
\end{equation}
In particular, for any compact set $\mathcal{C} \subset \mathbb{C}$, there exists a small enough neighborhood of $v=0$ such that both sides of \eqref{analytical extension probability} are well defined for all $\wp \in \mathcal{C}$. For all $n$, we will prove that $P_n$ and $R_n$ are polynomials in the variable $\wp$. We can therefore set 
\begin{equation*}
d_n = \max( \deg(P_n), \deg(R_n) )
\end{equation*}
and take $v$ small enough so that the Taylor expansions in \eqref{expansion in v} hold for all $\wp$ in a disk of radius greater than $q^{ - d_n - 1 }$ and centered at the origin. As a result of Lemma \ref{lemma probability height}, $P_n, R_n$ assume the same value when $\wp = q^{-1}, \dots, q^{-d_n -1}$, since for these particular choices \eqref{analytical extension probability} holds. This means that $P_n - R_n$ is a polynomial of degree $d_n$ with $d_n +1$ zeros and hence $P_n$ and $R_n$ are the same function. Since $n$ is generic we can conclude that all Taylor coefficients of the expansion of left and right hand side of \eqref{expansion in v} coincide and this concludes our argument.

We come now to verify the claim that all expressions we deal with are analytic in $v$ and that $P_n, R_n$ are polynomials. We treat separately the left and the right hand side of \eqref{analytical extension probability}.

\emph{lhs of \eqref{analytical extension probability} }: first we write down the probability of the event $\{ \overline{\mathcal{H}}(x,t) = l \}$ as
\begin{equation*}
    \mathbb{P}_{\wp, v, \mathpzc{d}} \left( \overline{\mathcal{H}} (x,t)=l  \right) =
    \sum_{M_1, \dots, M_x \geq 0} \prod_{i=1}^x \boldsymbol\ell^{(i)}_{\wp,v}\left( M_i ; \sum_{j=1}^{i-1}M_j \right) \mathbb{P}_{M_1, \dots, M_x}\left(\overline{\mathcal{H}}(x,t)=l \right),
\end{equation*}
where the families of weights
$\boldsymbol\ell^{(i)}_{\wp,v}$ have been defined in \eqref{weight l} and the notation $\mathbb{P}_{M_1, \dots, M_x}(E)$ is a shorthand for $\mathbb{P}_{\wp,v,\mathpzc{d}}(E|\mathsf{m} = M_1, \dots \mathsf{m}_x^0 = M_x )$ for any event $E$. Naturally the probabilities $\mathbb{P}_{M_1, \dots, M_x}$ do not depend either on $v$ or $\wp$ as these are probabilities of events in the Higher Spin Six Vertex Model with deterministic boundary conditions given by occupation numbers $M_1, \dots, M_x$ and $v,\wp$ only pertain factors $\boldsymbol\ell^{(i)}_{\wp,v}$. Set a small positive number $\epsilon$ such that
\begin{equation*}
\left|\frac{v}{\xi_i s_i} \right|<1-\epsilon, \qquad \qquad \left| \frac{\wp v s_i}{\xi_i} \right| < 1-\epsilon, \qquad \text{for all }i=1, \dots, x.
\end{equation*}
With these conditions it is easy to see, from the definition \eqref{weight l} of the weight $\boldsymbol\ell^{(i)}_{\wp,v}$ that a bound as 
\begin{equation*}
\left|\boldsymbol\ell^{(i)}_{\wp,v}\big( M ; \overline{M} \big) \right| < \left| \frac{v}{\xi_i s_i} \right|^M \frac{(-|s_i|^2, -|\wp q^{\overline{M}}|,q)_\infty}{(|v \wp q^{\overline{M}}s_i/\xi_i|,q;q)_\infty} \frac{(-|v \wp q^{\overline{M}}s_i/\xi_i|, -|v/(\xi_i s_i)|; q)_\infty }{(|v \wp q^{\overline{M}}/(\xi_i s_i)| , |v s_i / \xi_i| ; q)_\infty } < C' \left| \frac{v}{\xi_i s_i} \right|^{M_i}
\end{equation*}
holds for a constant $C'=C'(\wp, \epsilon, \xi_i, s_i)$. The lhs of \eqref{analytical extension probability} is therefore an absolutely convergent series of analytic functions in $v$ and the analyticity in the region \eqref{constraint analyticity} follows.
The Taylor expansion at $v=0$ of the weight $\boldsymbol\ell^{(i)}_{\wp,v}$ is easily seen to be of the form
\begin{equation}\label{expansion weight l}
\sum_{n \geq M} \tilde{p}_n^{(i)}(\wp,M,\overline{M}) v^n,
\end{equation}
where $\tilde{p}^{(i)}_n(\wp, M, \overline{M})$ is a polynomial in $\wp$. Alternatively we can rewrite the right hand side of \eqref{expansion weight l} as 
\begin{equation*}
\sum_{n\geq 0}p^{(i)}_{n-M}(\wp, M, \overline{M}) v^n,
\end{equation*}
where $p_L^{(i)}(\wp, M, \overline{M})$ is again a polynomial in $\wp$, which is zero when $L<0$. With this notation the lhs of \eqref{analytical extension probability} admits the expansion
\begin{equation*}
\begin{split}
&\sum_{M_1, \dots, M_x \geq 0}\prod_{i=1}^x \left( \sum_{n_i\geq 0}p^{(i)}_{n_i-M_i}(\wp, M_i, \sum_{j=1}^{i-1}M_j) v^{n_i} \right) \mathbb{P}_{M_1, \dots, M_x}\left(\overline{\mathcal{H}}(x,t)=l \right)\\
&=\sum_{M_1, \dots, M_x \geq 0} \left[ \sum_{n \geq 0} \left( \sum_{n_1+\cdots+n_x=n} \prod_{i=1}^x p^{(i)}_{n_i-M_i}(\wp, M_i, \sum_{j=1}^{i-1}M_j)  \right) v^n \right]\mathbb{P}_{M_1, \dots, M_x}\left(\overline{\mathcal{H}}(x,t)=l \right)\\
&=\sum_{n \geq 0}\left( \sum_{\substack{ n_1 + \cdots + n_x = n \\ M_1, \dots, M_x \geq 0}} \mathbb{P}_{M_1, \dots, M_x}\big(\overline{\mathcal{H}}(x,t)=l \big) \prod_{i=1}^x p^{(i)}_{n_i-M_i}(\wp, M_i, \sum_{j=1}^{i-1}M_j) \right) v^n,
\end{split}
\end{equation*}
where one can see that in the last equality the coefficient of $v^n$ is a polynomial in $\wp$ as the generic summation in the $M_i$ terminates when $n_i - M_i <0$.

\emph{rhs of \eqref{analytical extension probability} }: the analyticity in the variable $v$ is evident, so we look at the Taylor expansion around zero. Using the $q$-binomial theorem \eqref{q binomial summation}, the term in the integrand depending on $v$ and $\wp$ can be expanded as
\begin{equation} \label{expansion proof analytical}
\frac{(\frac{\wp v}{z_j},\frac{v}{ \xi_j s_j};q)_\infty}{(\frac{\wp v}{ \xi_j s_j},\frac{v}{ z_j};q)_\infty}= \sum_{n\geq 0} v^n \left( \frac{1}{z_j^n} \sum_{m=0}^n \frac{(\frac{\xi_j s_j}{z_j};q)_m (\frac{z_j}{\xi_j s_j};q)_{n-m}}{(q;q)_m (q;q)_{n-m}} \left( \frac{\wp z_j}{\xi_j s_j} \right)^m \right).
\end{equation}
By means of simple inequalities as
\begin{equation*}
\left| 1 - q^k \frac{\xi_j s_j}{z_j} \right| \leq 1 + \left| \frac{\xi_j s_j}{z_j} \right|, \qquad \left| 1 - q^k \frac{z_j}{\xi_j s_j} \right| \leq 1 + \left| \frac{z_j}{\xi_j s_j} \right|, \qquad |1-q^k| \geq 1-q  
\end{equation*}
we see that the coefficient of $v^n$ in the right hand side of \eqref{expansion proof analytical} is a polynomial in $\wp$ and can be bounded by $nC''^n$, where $C''$ does not depend on $v$. By taking $v$ sufficiently small we can bring the summation outside of the integral in the rhs of \eqref{analytical extension probability} and obtain a summation of the form
$$
\sum_{n \geq 0} R_n(\wp)v^n
$$
with $R_n(\wp)$ polynomials as promised.

\end{proof}

The statement of Proposition \ref{analytical extension probability proposition} establishes the exact solvability of the coupled measure $\mathbb{P}_{\wp, v,\mathpzc{d}}$. Naturally, $\mathbb{P}_{\wp, v,\mathpzc{d}}$ isn't a particularly interesting object per se, but rather its specialization $\wp = 0$, which describes a double sided $q$-negative binomial Higher Spin Six Vertex Model coupled with an independent random variable $\mathsf{m} \sim q$Poi$(v/\mathpzc{d})$. Unfortunately, due to the presence of $\mathsf{m}$, the measure $\mathbb{P}_{0, v,\mathpzc{d}} = \mathbb{P}_{\text{HS}(v, \mathpzc{d}) \otimes \mathsf{m} }$ is only well defined when $v < \mathpzc{d}$ and this condition prevents us to study the stationary model directly from $\mathbb{P}_{\text{HS}(v, \mathpzc{d}) \otimes \mathsf{m} }$. In order to consider the case $v=\mathpzc{d}$, in Section \ref{section matching and fredholm} we will decouple the Higher Spin Six Vertex Model from $\mathsf{m}$, expressing the probability distribution of $\mathcal{H}$ rather than that of $\overline{\mathcal{H}}$. This will allow us to consider a different family of double sided $q$-negative binomial boundary conditions, where parameters $v, \mathpzc{d}$ will be subjected to the less stringent bound \eqref{condition v d analytical continuation} as explained below in the proof of Theorem \ref{theorem q laplace double sided}.

\section{Fredholm determinant formulas for double sided \texorpdfstring{$q$}{q}-negative binomial boundary conditions} \label{section matching and fredholm}
The main content of this section consists in the proofs of Theorems \ref{theorem: q laplace shifted introduction}, \ref{theorem q laplace double sided} presented in the Introduction. We do this by first considering the coupled model $\mathbb{P}_{\wp, v, \mathpzc{d}}$, which in Section \ref{subsection integrable initial} was proven to be integrable, and then considering its degeneration $\wp=0$. 

\subsection{Fredholm determinants in the coupled model \texorpdfstring{$\mathbb{P}_{\wp,v,\mathpzc{d}}$}{P(p,v,d)}} \label{subsection analytic continuation}

In this section we give a Fredholm determinant expression for the $q$-Laplace transform of the probability mass function of the shifted height function $\overline{\mathcal{H}}$ defined in \eqref{height H bar}. Results given in Proposition \ref{fredholm determinant} hold for the coupled measure $\mathbb{P}_{\wp, v, \mathpzc{d}}$, for a general coupling parameter $\wp$. In Section 5.2 we will consider the meaningful choice $\wp=0$ and hence the model with double sided $q$-negative binomial boundary conditions. The proof of Proposition \ref{fredholm determinant} is based on calculations involving an elliptic version of the Cauchy determinant that  were developed in a previous work by two of the authors \cite{q-TASEPtheta} and it is therefore omitted.
\begin{Prop} \label{fredholm determinant}
Assume conditions on parameters $ \Xi, \mathbf{S,U}, \wp, v$ \eqref{eq: placements Xi S}, \eqref{condition wp v 1}, take $v < \mathpzc{d}$ and $\zeta \in \mathbb{C} \setminus q^{\mathbb{Z}}$. Then we have
\begin{equation}\label{first fredholm determinant}
\begin{split}
\mathbb{E}_{\wp, v, \mathpzc{d}} \left( \frac{1}{(\zeta q^{\overline{\mathcal{H}}(x,t)};q)_{\infty}} \right) = \det( \mathbf{1}-fK)_{l^2(\mathbb{Z})},
\end{split}
\end{equation}
where
\begin{equation}\label{f}
f(n)= \frac{1}{1-q^{n}/\zeta}, 
\end{equation}
\begin{equation}\label{kernel}
K(n,m)=\sum_{l=1}^{x-1}\phi_l(m) \psi_l(n) + (\mathpzc{d} - v) \Phi_x(m)\Psi_x(n),
\end{equation}
\begin{equation}\label{phi}
\phi_l(n)=\tau(n) \int_D \frac{dw}{2 \pi \mathrm{i}} \frac{1}{w^{x+n-l+1}} \prod_{k=1}^{l} \frac{1}{(w-\xi_{k+1} s_{k+1} )} \frac{( qv / w ;q)_\infty }{F(w)},
\end{equation}
\begin{equation}\label{psi}
\psi_l(n)=\frac{\xi_{l+1}s_{l+1}} {\tau(n)} \int_C \frac{dz}{2 \pi \mathrm{i}} z^{n+x-l-1} \prod_{k=2}^l(z-\xi_k s_k) \frac{F(z)}{( qv / z ;q)_\infty },
\end{equation}
\begin{equation}\label{Phi}
\Phi_x(n)=\tau(n)\int_D \frac{dw}{2 \pi \mathrm{i}}\frac{1}{w^{n+1}} \frac{1}{w- \mathpzc{d} } \prod_{k=2}^x \frac{1}{w- \xi_k s_k} \frac{( qv / w ;q)_\infty}{F(w)},
\end{equation}
\begin{equation}\label{Psi}
\Psi_x(n)=\frac{1}{\tau(n)} \int_C \frac{dz}{2 \pi \mathrm{i}} \frac{z^{n-1}}{( v / z ;q)_\infty} \prod_{k=2}^x (z-\xi_k s_k) F(z).
\end{equation}
The contour $D$ encircles $\{\mathpzc{d},\xi_2 s_2, \dots ,\xi_x s_x\}$ and no other singularity, whereas $C$ contains 0 and $vq^k$, for any $k$ in $\mathbb{Z}_{\geq 0}$. Finally, $\tau(n)$ is taken to be
\begin{equation} \label{tau}
\tau(n)=\begin{cases}
    b^n       & \quad \text{if } n \geq 0,\\
    c^n  & \quad \text{if } n < 0,\\
  \end{cases}
\end{equation}
with 
\begin{equation*}
v < b < \mathpzc{d} \leq \inf_{i \geq 2} \{ \xi_i s_i \} \leq \sup_{i \geq 2} \{ \xi_i s_i \} < c < \inf_{i\geq 2} \{ \xi_i/s_i \} ,
\end{equation*}
and 
\begin{equation} \label{F}
F(z)=( v \wp / z;q)_{\infty} \prod_{j=1}^t(zu_j;q)_J  ( qz / \mathpzc{d} ;q)_\infty \prod_{k=2}^x \frac{( qz / ( \xi_k s_k ) ;q)_\infty}{(z s_k / \xi_k ;q)_\infty}.
\end{equation}
\end{Prop}
\begin{proof}
From Proposition \ref{analytical extension probability proposition} we can apply, with minor changes the same argument of \cite{q-TASEPtheta}, Theorem 4.3. More specifically, using the notation used in \cite{q-TASEPtheta}, we need to set

$$
a_k=\begin{cases}
s_{k+1}\xi_{k+1}, \qquad &\text{if } k\neq x\\
\mathpzc{d}, \qquad &\text{if } k=x 
\end{cases}, \qquad \qquad \alpha_k=\delta_{k,x}v,
$$
and substitute $e^{zt}$ with the expression
$$( v\wp / z;q)_{\infty}\frac{\prod_{k=1}^t(zu_k;q)_J}{\prod_{k=1}^x (z s_k / \xi_k ;q)_\infty},
$$
as a result of considering a more general specialization of the $q$-Whittaker measure as that presented in Proposition \ref{probability 2sided q-Whittaker proposition}.
\end{proof}
An alternative expression for the kernel $K$ is given defining an auxiliary kernel $A$ as
\begin{equation} \label{discrete Airy}
A(n,m) = \sum_{l=1}^{x-1} \phi_l(n) \psi_l(m)
\end{equation}
of which we report the explicit form. 
\begin{Prop}[Double integral kernel] \label{prop double integral kernel}
The discrete kernel $A$ admits the following expression
\begin{equation} \label{double integral}
A(n,m)= \frac{\tau(n)}{\tau(m)} \frac{1}{(2\pi \mathrm{i})^2} \int_D dw \int_C dz \frac{z^{m}}{w^{n+1}} \prod_{j=1}^t \left( \frac{(u_jz;q)_J}{(u_jw;q)_J} \right) \prod_{k=2}^x \left( \frac{ ( z / (\xi_k s_k );q)_\infty (w s_k / \xi_k;q)_\infty}{(w / (\xi_k s_k ) ;q)_\infty (z s_k / \xi_k ;q)_\infty} \right) \frac{(q v/w;q)_\infty}{(q v / z;q)_\infty} \frac{1}{z-w}.
\end{equation}
\end{Prop}
\begin{proof}
All it takes to show \eqref{double integral} is to perform the summation
$$
\sum_{l=1}^{x-1} \phi_l(n) \psi_l(m),
$$
using the rather tricky identity
\begin{equation*}
\frac{1}{z-w} \left[ \frac{w^{x-1}}{z^{x-1}}  \prod_{k=2}^x \frac{z-a_k}{w-a_k} -1 \right] = \sum_{l=1}^{x-1} \frac{a_{l+1}}{w-a_{l+1}}  \frac{w^{l-1}}{z^l} \prod_{k=2}^l\frac{z-a_k}{w-a_k},
\end{equation*}
which can be proven by induction. We see that the addend $(z-w)^{-1}$ in the left hand side doesn't give any contribution to the integral as integrating over the variable $z$ it only leaves an integral in $w$ over a path containing no singularities.
\end{proof}

Before moving to the analysis of the stationary case we show that the Fredholm determinant expression \eqref{first fredholm determinant} makes sense. We need the following
\begin{Prop} \label{trace class}
The kernel $fK$ defined by \cref{f,kernel,phi,psi,Phi,Psi} is trace class.
\end{Prop}
\begin{proof}
From $fK$ being a finite sum of products of operators of rank one, it is enough to show that each one of these operators is of Hilbert-Schmidt class. This is essentially proven in Appendix \ref{appendix  bounds}. In fact the generic terms $f(n) \phi_l (n) \psi_l(m)$ and $f(n) \Phi_x(n) \Psi_x(m)$ are bounded in absolute value by quantities exponentially small in $|n|+|m|$, thanks to Proposition \ref{useful inequalities}, and therefore the double summation
\begin{equation*}
\sum_{n,m \in \mathbb{Z}} |f(n)K(n,m)|^2
\end{equation*}
is indeed convergent.
\end{proof}

We close this subsection by offering the proof of Theorem \ref{theorem: q laplace shifted introduction}, presented in the Section \ref{section results}.
\begin{proof}[proof of Theorem \ref{theorem: q laplace shifted introduction}] This is a trivial consequence of Proposition \ref{fredholm determinant}, after setting $\wp=0$.
\end{proof}

\subsection{The double sided \texorpdfstring{$q$}{q}-negative binomial case and the stationary specialization} \label{section regularizations}
In this Section we give a proof of Theorem \ref{theorem q laplace double sided}, that characterizes the probability distribution of the height function $\mathcal{H}$ in the Higher Spin Six Vertex Model with double sided $q$-negative binomial boundary conditions. Our starting point is the Fredholm determinant formula stated in Theorem 
\ref{theorem: q laplace shifted introduction}. Removing the effect of the independent $q$-Poisson random variable $\mathsf{m}$ from expression \eqref{fredholm determinant introduction} we find that determinantal expressions we obtain are well defined in the region \eqref{condition v d analytical continuation}. In Corollary \ref{Corollary stationary q laplace} we specialize the result of Theorem \ref{theorem q laplace double sided} to the relevant case of the stationary Higher Spin Six Vertex Model.

Before we begin the proof of Theorem \ref{theorem q laplace double sided}, we like to state some regularity properties of the integral kernel $fA$ defined by \eqref{f},\eqref{discrete Airy} that hold for parameters $v, \mathpzc{d}$ in \eqref{condition v d analytical continuation}.

\begin{Prop} \label{lemma invertible}
Let $\zeta < 0$ and take $f,A$ as in \eqref{f}, \eqref{discrete Airy}. Then $\mathbf{1}-fA$ is an invertible operator. Moreover both $fA$ and $(\mathbf{1} - fA)^{-1}$ are well defined, bounded operators in the region \eqref{condition v d analytical continuation}.
\end{Prop}
In the proof of Proposition \ref{lemma invertible} we use the following biorthogonality property. 

\begin{Lemma} \label{lemma scalar product}
Let $v,\mathpzc{d}$ be such that $v<\mathpzc{d}$ or \eqref{condition v d analytical continuation} holds. Then we have
\begin{equation} \label{ scalar phi psi}
\sum_{n \in \mathbb{Z}} \phi_l(n) \psi_m(n)=\delta_{l,m}.
\end{equation}
\end{Lemma}
\begin{proof}
This is a simple consequence of the contour integral expressions \eqref{phi}, \eqref{psi}. The generic term in the summation in the left hand side of \eqref{ scalar phi psi} is
\begin{equation*}
\phi_l(n) \psi_m(n) = \xi_{m+1} s_{m+1} \int_D \int_C \frac{dw dz}{(2 \pi \mathrm{i})^2} \left( \frac{z}{w} \right)^n \frac{w^{l-1}}{z^{m+1}} \frac{\prod_{k=2}^m (z-\xi_k s_k)}{\prod_{k=2}^{l+1} (w-\xi_k s_k)}   \frac{G(w)}{G(z)},
\end{equation*}
where, in the function $G$, we gathered together the factors independent on $n$ or $l,m$ as
$$
G(u)=\frac{( qv / u ;q)_\infty}{u^x F(u)}.
$$
We see that in order to take the summation over all integers inside the integrals we need $|z/w|$ to be suitably defined depending on the positivity of $n$ itself.

Consider the contour $\tilde{C}$ being a circle of radius $r$ such that $\max_i |\xi_i s_i| < r < \min_i |\xi_i / s_i|$. We can write
\begin{equation*}
\begin{split}
\sum_{n \in \mathbb{Z}} \phi_l(n) \psi_m(n) &= \xi_{m+1} s_{m+1} \left [ \sum_{n \geq 0} \int_D \int_C \frac{dw dz}{(2 \pi \mathrm{i})^2} \left( \frac{z}{w} \right)^n \frac{w^{l-1}}{z^{m+1}} \frac{\prod_{k=2}^m (z-\xi_k s_k)}{\prod_{k=2}^{l+1} (w-\xi_k s_k)}   \frac{G(w)}{G(z)} \right.\\
& \left. \qquad \qquad \qquad+ \sum_{n < 0} \int_D \int_{\tilde{C}} \frac{dw dz}{(2 \pi \mathrm{i})^2} \left( \frac{z}{w} \right)^n \frac{w^{l-1}}{z^{m+1}} \frac{\prod_{k=2}^m (z-\xi_k s_k)}{\prod_{k=2}^{l+1} (w-\xi_k s_k)}   \frac{G(w)}{G(z)} \right]\\
&= \xi_{m+1} s_{m+1} \left[ \int_D \left( \int_{\tilde{C}} - \int_C \right) \frac{dw dz}{(2 \pi \mathrm{i})^2} \frac{1}{z-w} \frac{w^l}{z^{m+1}} \frac{\prod_{k=2}^m (z-\xi_k s_k)}{\prod_{k=2}^{l+1} (w-\xi_k s_k)}   \frac{G(w)}{G(z)} \right]\\
&=\xi_{m+1} s_{m+1} \left[ \int_D \int_{\tilde{D}} \frac{dw dz}{(2 \pi \mathrm{i})^2} \frac{1}{z-w} \frac{w^l}{z^{m+1}} \frac{\prod_{k=2}^m (z-\xi_k s_k)}{\prod_{k=2}^{l+1} (w-\xi_k s_k)}   \frac{G(w)}{G(z)} \right],
\end{split}
\end{equation*}
where, for the last equality, we deformed $\tilde{C}$ into an union or the two contours $C$ and $\tilde{D}$, with the latter being a curve encircling $D$ and no other singularity for the $z$ variable. Performing the $z$ integral we get
\begin{equation*}
\xi_{m+1} s_{m+1} \int_D \frac{dw}{2 \pi \mathrm{i}} w^{l-m-1} \left( \mathbbm{1}_{m\leq l} \prod_{k=m}^{l+1} \frac{1}{w-\xi_k s_k} + \mathbbm{1}_{m > l} \prod_{k=l+1}^m (w-\xi_k s_k)\right). \end{equation*}
Naturally, when $m>l$ there is no pole and the integral vanishes. On the other hand, if $m \leq l$ we can evaluate the residue at infinity and obtain the result.
\end{proof}

\begin{proof}[Proof of Proposition \ref{lemma invertible}]
We want to show that $\lVert fA \rVert_{l^2(\mathbb{Z})} <1$,
so to define $(\mathbf{1} -fA)^{-1}$ through the geometric series 
$$
\sum_{k \leq 0} (fA)^k.
$$
The biorthogonality relation \eqref{ scalar phi psi} implies $\lVert A \rVert =1$ since $A^2 = A$.

On the other hand, consider the following bound for the function $f$, defined in \eqref{f},
\begin{equation*}
f(n) \leq (1-\varepsilon) f(n+n_0) + \varepsilon  P_{n_1}(n),
\end{equation*}
with $P_{n_1}(n)= \mathbbm{1}_{n\geq n_1}$ and $\varepsilon, n_0 ,n_1$, suitably chosen. In particular, taking $n_0$ large enough and $\varepsilon$ small we can let $n_1$ be an arbitrary big number. Moreover, the fact that $f$ is a diagonal operator implies the simple estimate
\begin{equation*}
\lVert f A \rVert  \leq (1-\varepsilon) \Vert f \Vert \Vert A \Vert + \varepsilon \Vert P_{n_1} A \Vert \leq 1-\varepsilon + \varepsilon \Vert P_{n_1} A \Vert.
\end{equation*}
Let's consider now an element $\eta$ in the unitary sphere of $l^2(\mathbb{Z})$. By simply using the definition of the kernel $A$ and the Schwartz inequality we have
\begin{equation*}
\Vert P_{n_1}A \eta \Vert  \leq \sum_{l=1}^{x-1} \left( \sum_{n \in \mathbb{Z}} \left| \mathbbm{1}_{n\geq n_1} \phi_l(n)  \sum_{m \in \mathbb{Z}} \psi_l(m) \eta_m \right|^2 \right)^{1/2}
\leq \sum_{l=1}^{x-1} \Vert\psi_l \Vert \left( \sum_{n \geq n_1} |\phi_l(n) |^2 \right)^{1/2},
\end{equation*}
which, thanks to the bound \eqref{geometric bound phi}, can be shown to be geometrically small in $n_1$.

Finally, we remark that for sequences $\phi_l$ or $\psi_l$ the bounds \eqref{geometric bound phi} hold true also in the region \eqref{condition v d analytical continuation}, so $A$ is analytic in this domain, so are its powers and so is $(\mathbf{1} - fA)^{-1}$ as one can show the geometrical decay of derivatives of $(fA)^N$ as well. 
\end{proof}

The following lemma offers a tool to decouple a generic process from the contribution of an independent $q$-Poisson random variable.

\begin{Lemma} \label{lemma shifting argument}
Let $\mathsf{m} \sim q$Poi(p). Then, for any bounded function $B$, we have
\begin{equation} \label{eq: shifting}
    B(z)=\frac{1}{(p;q)_\infty} \sum_{k\geq 0}\frac{(-1)^kq^{\binom{k}{2}}}{(q;q)_k} p^k \mathbb{E}_{\mathsf{m}} (B(z- \mathsf{m} -k)).
\end{equation}
\end{Lemma}
\begin{proof}
To verify identity \eqref{eq: shifting} we simply open up the average in the right hand side with respect to $\mathsf{m}$, as 
\begin{equation*}
    \mathbb{E}_{\mathsf{m}} (B(z- \mathsf{m} -k)) = \sum_{l \geq 0} p^l \frac{(p;q)_\infty}{(q;q)_l}B(z-l-k).
\end{equation*}
We can now rearrange the double summation in the indices $k,l$, naming $L=l+k$, as
\begin{equation*}
    \text{rhs of \eqref{eq: shifting}} = \sum_{L \geq 0} B(z -L) \sum_{k=0}^{L} \frac{(-1)^k q^{\binom{k}{2}}}{(q;q)_k (q;q)_{L-k}},
\end{equation*}
which completes the proof, after recognizing, in the right hand side, the $q$-Pochhammer expansion \eqref{q Pochammer expansion} (with $z=1$) that is one for $L=0$ and zero otherwise. 
\end{proof}

In the remaining part of the paper we will use the following decomposition of terms $\Phi_x, \Psi_x$:
\begin{gather}
\Phi_x(n) = \Phi_x^{(1)}(n) + \Phi_x^{(2)}(n),\label{Phi1 + Phi2 }\\
\Psi_x(n) = \Psi_x^{(1)}(n) + \Psi_x^{(2)}(n), \label{Psi1 + Psi2}
\end{gather}
obtained separating from the integration \eqref{Phi} (resp. \eqref{Psi}) the contribution of pole $w= \mathpzc{d} $ (resp. $z= v$) from that of other poles. The exact expressions are
\begin{equation}\label{Phi 1}
\Phi_x^{(1)}(n) = \frac{\tau(n)}{\mathpzc{d}^{n+1}} \prod_{k=2}^x \frac{1}{\mathpzc{d} - \xi_k s_k} \frac{( qv / \mathpzc{d};q)_\infty}{F(\mathpzc{d})},
\end{equation}
\begin{equation}\label{Phi 2}
\Phi_x^{(2)}(n)=\tau(n) \int_{D_1} \frac{dw}{2 \pi \mathrm{i}} \frac{1}{w^{n+1}} \frac{1}{w-\mathpzc{d}} \prod_{k=2}^x \frac{1}{w-\xi_k s_k} \frac{( qv / w ;q)_\infty}{F(w)},
\end{equation}
\begin{equation}\label{Psi 1}
\Psi_x^{(1)}(n)=\frac{v^n}{\tau(n)} \prod_{k=2}^x (v-\xi_k s_k) \frac{F(v)}{(q;q)_\infty},
\end{equation}
\begin{equation}\label{Psi 2}
\Psi_x^{(2)}(n)=\frac{1}{\tau(n)} \int_{C_1} \frac{dz}{2 \pi \mathrm{i}} \frac{z^{n-1}}{( v / z ;q)_\infty} \prod_{k=2}^x (z-\xi_k s_k) F(z),     
\end{equation}
where $F$ was given in \eqref{F}, contour $D_1$ contains $\{\xi_i s_i\}_{i\geq 2}$ and no other singularity and $C_1$ contains $\{ q^{k}v \}_{k \geq 1}$ and no other singularity.

\begin{proof}[Proof of Theorem \ref{theorem q laplace double sided}]
Using Lemma \ref{lemma shifting argument} setting $B:z \to \mathbb{E}_{\text{HS}(v,\mathpzc{d})}(1/(\zeta q^{\mathcal{H}+z};q)_\infty)$ and expressing the $q$-Laplace transform $\mathbb{E}_{\text{HS}(v,\mathpzc{d})}(1/(\zeta q^{\mathcal{H}-\mathsf{m}};q)_\infty)$ as in \eqref{fredholm determinant introduction} we obtain formula \eqref{eq: q laplace H expansion}. Therefore we only need to show that expression \eqref{eq: q laplace H expansion} is well posed in the region \eqref{condition v d analytical continuation}.

Using basic properties of Fredholm determinants, along with the regularity of the kernel $fA$ proved in Proposition \ref{lemma invertible}, we can write
\begin{equation*}
    \det(\mathbf{1} - fK) = \det\left( \mathbf{1} - f A - (\mathpzc{d} - v) f \Phi_x \Psi_x \right) = \det(\mathbf{1} - fA) \det\left(\mathbf{1} - (\mathpzc{d} - v) (\varrho f \Phi_x)\Psi_x \right),
\end{equation*}
where we called $\varrho=(\mathbf{1} - fA)^{-1}$. This allows us to express $V_{x;v,\mathpzc{d}}(\zeta)$ as
\begin{equation}\label{expansion determinant}
V_{x;v,\mathpzc{d}}(\zeta) =\det(\mathbf{1} - fA)\frac{1}{1- v / \mathpzc{d} } \left( 1 - (\mathpzc{d} - v)\sum_{n \in \mathbb{Z}} (\varrho f \Phi_x) (n) \Psi_x(n)\right),
\end{equation}
We turn our attention to the term
\begin{equation} \label{eq: difficult term}
\sum_{n \in \mathbb{Z}} (\varrho f \Phi_x) (n) \Psi_x(n) = \sum_{n\in \mathbb{Z}}f(n)\Phi_x(n) \Psi_x(n) + \sum_{n \in \mathbb{Z}}(fA\varrho\Phi_x)(n) \Psi_x(n),
\end{equation}
that we rewrite, using \cref{Phi 1,Phi 2,Psi 1,Psi 2} as
\begin{equation} \label{eq: difficult term decomposition}
\eqref{eq: difficult term}=\sum_{n \in \mathbb{Z}}f(n)\Phi_x^{(1)}(n) \Psi_x^{(1)}(n) + \sum_{\substack{i,j=1,2\\ (i,j)\neq (1,1)}} \sum_{n \in \mathbb{Z}} f(n)\Phi_x^{(i)}(n) \Psi_x^{(j)}(n) + \sum_{n \in \mathbb{Z}}(fA\varrho f\Phi_x)(n) \Psi_x(n). 
\end{equation}
We easily see that, thanks to bounds stated in Appendix \ref{appendix  bounds}, the second and the third terms in the right hand side of \eqref{eq: difficult term decomposition} are geometrically convergent summations in the region \eqref{condition v d analytical continuation}. On the other hand, the generic term of the summation in the first addend of the right hand side of \eqref{eq: difficult term decomposition} takes the form
\begin{equation*}
\frac{1}{1-q^n/\zeta} \frac{1}{\mathpzc{d}} \left( \frac{v}{\mathpzc{d}} \right)^n \prod_{k=2}^x \frac{v - \xi_k s_k}{\mathpzc{d} - \xi_k s_k} \frac{(qv / \mathpzc{d};q)_\infty}{(q;q)_\infty}\frac{ F(v)}{F(\mathpzc{d})}.
\end{equation*}
We can perform the summation over $n$ through the Ramanujan $\setlength\arraycolsep{1pt}
{}_1 \psi_1$  formula (\cite{IsmailClassicalAndQuantum}, Theorem 12.3.1) as
\begin{equation*}
\sum_{n \in \mathbb{Z}} \frac{1}{1-q^n/\zeta} \left( \frac{v}{ \mathpzc{d}} \right)^n = \frac{(v / ( \mathpzc{d} \zeta ), q\zeta \mathpzc{d} / v, q,q ;q)_\infty }{ ( v /  \mathpzc{d}, q \mathpzc{d} / v, 1/\zeta, q\zeta;q)_\infty }.
\end{equation*}
which itself leads to the expression
\begin{equation} \label{eq: Phi1 Psi1 sum}
    \sum_{n \in \mathbb{Z}}f(n)\Phi_x^{(1)}(n) \Psi_x^{(1)}(n) = \frac{1}{\mathpzc{d}} \frac{( v / ( \mathpzc{d} \zeta ) , q\zeta \mathpzc{d} / v, q  ;q)_\infty }{ ( v / \mathpzc{d} , 1/\zeta, q \zeta;q)_\infty } \frac{F(v)}{F(\mathpzc{d})} \prod_{k=2}^x \frac{v - \xi_k s_k}{\mathpzc{d} - \xi_k s_k}.
\end{equation}
Combining the explicit formula \eqref{eq: Phi1 Psi1 sum} with \eqref{eq: difficult term decomposition} and \eqref{expansion determinant} we see that $V_{x;v, \mathpzc{d}}$ does not in fact present any singularity in $v=\mathpzc{d}$ by virtue of the Taylor expansion
\begin{equation} \label{star 1}
\text{rhs of \eqref{eq: Phi1 Psi1 sum}}
=\frac{1}{\mathpzc{d}- v}  + \frac{1}{\mathpzc{d}} \left( \polygamma_0(1/\zeta) - \polygamma_0(q \zeta) +2\polygamma_0(q ) +x h_0(\mathpzc{d}) -\sum_{j=1}^ta_0(\mathpzc{d};j)\right)  + \mathcal{O}\left( \mathpzc{d}-v \right),
\end{equation}
where we used the $q$-polygamma type functions $\polygamma_k$ defined in Appendix \ref{appendix qfunctions}, their combinations $a_0,h_0$ presented in \eqref{eq: a_k h_k} and the notation $a_0(\mathpzc{d};j)$ stresses the dependence on the spectral parameters $u_j$ as
\begin{equation*}
    a_k(\mathpzc{d};j) = \polygamma_k(q^J u_j \mathpzc{d}) - \polygamma_k(u_j \mathpzc{d}).
\end{equation*}
Therefore $V_{x;v, \mathpzc{d}}$ is an analytic function of both parameters $v, \mathpzc{d}$ in the region \eqref{condition v d analytical continuation} and this concludes the proof.
\end{proof}
Calculations performed during the proof of Theorem \ref{theorem q laplace double sided} can be exploited to obtain determinantal formulas for for the stationary Higher Spin Six Vertex Model.
\begin{Cor} \label{Corollary stationary q laplace}
Consider stationary Higher Spin Six Vertex Model with parameters $\mathpzc{d}, \Xi, \mathbf{S}$ as in \eqref{eq: placements Xi S}. Then we have
\begin{equation} \label{stationary q laplace 2}
 \mathbb{E}_{\emph{HS}(\mathpzc{d}, \mathpzc{d})} \left ( \frac{1}{(\zeta q^{\mathcal{H}(x,t)};q)_\infty} \right ) = \frac{1}{(q;q)_\infty} \sum_{k\geq 0} \frac{(-1)^k q^{\binom{k}{2}}}{(q;q)_k}q^k \left( V_x(\zeta q^{-k})- V_x(\zeta q^{-k -1})\right) ,
\end{equation}
with the function $V_x=V_{x;\mathpzc{d},\mathpzc{d}}$ being
\begin{equation} \label{V}
\begin{split}
V_x(\zeta) =  \det \left( \mathbf{1} - fA \right) \Bigg(& - \polygamma_0(1/\zeta) + \polygamma_0(q\zeta) - 2 \polygamma_0(q) - x h_0(\mathpzc{d}) + \sum_{j=1}^t a_0(\mathpzc{d};j) \\
&- \mathpzc{d} \sum_{\substack{i,j=1,2\\ (i,j)\neq (1,1)}} \sum_{n \in \mathbb{Z}} f(n)\Phi_x^{(i)}(n) \Psi_x^{(j)}(n) - \mathpzc{d} \sum_{n \in \mathbb{Z}}(fA\varrho f\Phi_x)(n) \Psi_x(n) \Bigg).
\end{split}
\end{equation}
\end{Cor}
\begin{proof}
All it takes to show \eqref{stationary q laplace 2} is to take the limit $v \to \mathpzc{d}$ of both sides of equality \eqref{fredholm determinant introduction}. 
We exchange the limit sign in $v \to \mathpzc{d}$ and the summation sign in the right hand side of \eqref{fredholm determinant introduction} and this can be done as, in the proof of Theorem \ref{theorem q laplace double sided}, the function $V_{x;v,\mathpzc{d}}$ was shown to be uniformly bounded in a neighborhood of $v=\mathpzc{d}$ and its limit $V_{x;v,\mathpzc{d}} \to V_{x;\mathpzc{d},\mathpzc{d}}$ is readily computed using \eqref{star 1}.
This is enough to prove that
\begin{equation} \label{dummy limit summation}
\lim_{ v \to \mathpzc{d} } \mathbb{E}_{\text{HS}(v,\mathpzc{d})} \left ( \frac{1}{(\zeta q^{\mathcal{H}(x,t)};q)_\infty} \right ) = \frac{1}{(q;q)_\infty} \sum_{k\geq 0} \frac{(-1)^k q^{\binom{k}{2}}}{(q;q)_k}V_x(\zeta q^{-k}).
\end{equation}
Finally, substituting $V_x(\zeta q^{k})$ with $q^{k}V_x(\zeta q^{-k}) + (1-q^k)V_x(\zeta q^{-k})$ and rearranging the summation in the right hand side of \eqref{dummy limit summation} we obtain \eqref{stationary q laplace 2}.
\end{proof}

\section{Asymptotics along the critical line} \label{section time asymptotics}

In this section we discuss the time asymptotics of the stationary Higher Spin Six Vertex Model. First we give some details on the general conjecture concerning scaling limits of models in the Kardar-Parisi-Zhang universality class and subsequently we confirm these conjectures in this particular case.

\subsection{The KPZ scaling for the Higher Spin Six Vertex Model} \label{section: kpz scaling}

Before we enter the discussion it is appropriate to recall the definition of the Baik-Rains distribution \cite{BaikRains2000Png}, which will ultimately describe long time fluctuations of the stationary height function $\mathcal{H}$ under a suitable scaling. Rather than showing the original definition given by authors in \cite{BaikRains2000Png}, formula (2.16), we present an equivalent expression first found in \cite{ImamuraSasamoto2013sKPZ} (see also \cite{FerrariSpohn2006sTASEP}).

The building blocks for the construction of the Baik-Rains distributions are given in the following

\begin{Def}[Airy function and Airy kernel]
The Airy function $\emph{Ai}$ is given by
\begin{equation*}
\emph{Ai}(\nu)=\frac{1}{2 \pi \mathrm{i}} \int_{e^{-\frac{\mathrm{i}}{3}\pi} \infty }^{e^{\frac{\mathrm{i}}{3}\pi} \infty} \exp\left\{ \frac{z^3}{3} - z\nu \right\}dz,
\end{equation*}
where the integration contour is any open complex curve having the half lines $\{ R e^{\frac{\mathrm{i}}{3}\pi} | R \geq 0 \}$ and  $\{ R e^{-\frac{\mathrm{i}}{3}\pi} | R \geq 0 \}$ as asymptotes. 

The Airy kernel\footnote{Sometimes the equivalent expression $K_{\text{Airy}}(\nu, \theta)=\int_0^\infty \text{Ai}(\lambda + \nu) \text{Ai}(\lambda + \theta) d \lambda$ is found in literature.} $K_{\text{\emph{Airy}}}$ is defined as 
\begin{equation} \label{Airy kernel}
K_{\text{\emph{Airy}}}(\nu, \theta) = \int_{e^{-\frac{2}{3}\pi \mathrm{i}}\infty}^{e^{\frac{2}{3}\pi \mathrm{i}} \infty} \frac{dz}{2 \pi \mathrm{i}} \int_{e^{\frac{\pi}{3} \mathrm{i}} \infty }^{e^{-\frac{\pi}{3} \mathrm{i} }\infty } \frac{dw}{2 \pi \mathrm{i}}  \frac{1}{z-w} \exp \left\{ \frac{w^3}{3} - \frac{z^3}{3} -w \nu + z \theta  \right\},
\end{equation}
where again integration contours are non intersecting complex curves whose asymptotes are half lines $\{ R e^{ \pm  \frac{\mathrm{i}}{3}\pi} | R \geq 0 \}$ for $w$ and $\{ R e^{ \mp \mathrm{i}\frac{2}{3}\pi} | R \geq 0 \}$ for $z$.
\end{Def}

We come to the next

\begin{Def}[Baik-Rains distribution] \label{Baik-Rains definition}
Let $\varpi \in \mathbb{R}$ and define the one parameter family of functions
\begin{equation} \label{baik rains chi}
\begin{split}
\chi_\varpi (r) &=F_2(r) \Bigg(  r- \varpi^2 - \sum_{\substack{ i,j =1,2 \\ (i,j)\neq (1,1)}} \int_r^{\infty} \Upsilon_{-\varpi}^{(i)}(\nu) \Upsilon_{\varpi}^{(j)}(\nu) d \nu  \\
&\qquad \qquad \qquad  - \int_r^{\infty} d \nu \Upsilon_\varpi (\nu) \int_r^\infty d \lambda_1 \int_r^{\infty} d\lambda_2 \varrho_{\text{\emph{Airy}};r}(\nu, \lambda_1) K_{\emph{Airy}}(\lambda_1, \lambda_2) \Upsilon_{-\varpi}(\lambda_2) \Bigg)
\end{split}
\end{equation}
where terms $F_2, \varrho_{\emph{Airy};r}, \Upsilon_\varpi^{(i)}, \Upsilon_\varpi$ are given by:
\begin{itemize}

\item $F_2(r)$ is the GUE Tracy-Widom distribution \emph{(\cite{TracyWidom1994})}
 \begin{equation} \label{tracy widom}
 F_2(r)=\det\left( \mathbf{1} - \mathbbm{1}_{[r, \infty)}K_{\emph{Airy}} \right)_{\mathcal{L}^2(\mathbb{R})};
 \end{equation}
 
\item $\varrho_{\emph{Airy};r}(\nu, \lambda)$ is the kernel
\begin{equation} \label{traci widom rho}
\varrho_{\emph{Airy};r}(\nu, \lambda) = \left( \mathbf{1} - \mathbbm{1}_{[r, \infty)} K_{\emph{Airy}} \right)^{-1}(\nu, \lambda);
\end{equation}

\item auxiliary functions $\Upsilon_\varpi^{(1)}, \Upsilon_\varpi^{(2)}$ are
\begin{equation} \label{Upsilon12}
\Upsilon_\varpi^{(1)} (\nu) = e^{\frac{\varpi^3}{3} - \nu \varpi}, \qquad \qquad
\Upsilon_\varpi^{(2)}(\nu) = \int_{e^{-\frac{2}{3} \mathrm{i} \pi }\infty}^{e^{\frac{2}{3} \mathrm{i} \pi }\infty} \frac{d \omega}{2 \pi \mathrm{i}} \frac{\exp\{- \frac{\omega^3}{3} + \omega \nu\}}{\omega + \varpi},
\end{equation}
where the integration contour passes to the left of $-\varpi$;
\item lastly 
\begin{equation} \label{Upsilon}
\Upsilon_\varpi(\nu) = \Upsilon_\varpi^{(1)}(\nu)+ \Upsilon_\varpi^{(2)}(\nu).
\end{equation}

\end{itemize}

The Baik-Rains distribution $F_\varpi$ is 

\begin{equation} \label{F_0}
F_\varpi(r)= \frac{\partial}{\partial r} \chi_\varpi (r).
\end{equation}

\end{Def}

We now possess all the ingredients to give a brief explanation of the KPZ scaling theory, which gives a precise conjecture to describe stationary (asymptotic) fluctuations of the height function of models in the KPZ universality class (\cite{SpohnScaling}).

We start considering a properly rescaled version of our model, where, for convenience we interpret the vertical spatial coordinate as a time direction and where we regard space and time as continuous parameters. In this case the height function $\mathcal{H}$, defined for the Higher Spin Six Vertex Model in \eqref{height correct}, still contains every information on the random dynamics thanks to relations
\begin{gather}
\mathcal{H}(x, t) - \mathcal{H}(x+dx,t) = \text{\# of paths in $[x, x+dx]$ at time $t$}, \label{H scaling argument 1}\\
\mathcal{H}(x, t+dt) - \mathcal{H}(x,t) = \text{\# of paths crossing $x$ during the time interval $[t, t+dt]$} \label{H scaling argument 2}.
\end{gather}
For the sake of argument assume that the average of space and time infinitesimal increments of $\mathcal{H}$ are regular enough to define the deterministic density $\rho$ and current $j$ as
\begin{equation*}
\mathbb{E}\left( \mathcal{H}(x+dx,t)- \mathcal{H}(x, t) \right) \approx - \rho(x, t) dx, \qquad
\mathbb{E}\left( \mathcal{H}(x,t+dt)- \mathcal{H}(x, t) \right) \approx j(x,t) dt.
\end{equation*}
The system is autonomous, or, in other words, its evolution depends on space and time only implicitly, therefore the current $j$ must only be a function of $\rho$ and the continuity equation linking these quantities reads
\begin{equation} \label{continuity equation}
\partial_t \rho(x,t) + \partial_x j(\rho(x , t )) = 0.
\end{equation}
At this stage the height $\mathcal{H}$ remains defined, through \eqref{H scaling argument 1}, \eqref{H scaling argument 2}, only up to a global constant and to remove this ambiguity we fix its value at the reference point $x=0, t=0$ to be $\mathcal{H}(0,0)=0$. With this choice the average profile of $\mathcal{H}$ at the generic space time point $(x,t)$ can be expressed as
\begin{equation} \label{eta continuous}
\eta(x,t) = - \int_0^x \rho(y,0) dy + \int_0^t j(\rho(x,s)) ds
\end{equation}
and the study of fluctuations of the height is, by definition, the study of the random quantity 
\begin{equation} \label{fluctuations}
\mathcal{H}(x,t) - \eta(x,t).
\end{equation}

Assume now that the system has reached its steady state, or equivalently assume that at time zero the measure is stationary. Qualitatively, the randomness of \eqref{fluctuations} is affected by two different contributions. One is coming from the stochastic evolution of the system and the other is given by initial conditions. For growth processes in the KPZ universality class, when initial conditions are deterministic and sufficiently regular, fluctuations in the long time scale are expected to present with size of order $t^{1/3}$. This conjecture goes back to the seminal paper \cite{KPZ1986}, where authors argued such property to hold for the solution of the one dimensional KPZ equation itself. On the other hand, from our knowledge of the stationary measure of the Higher Spin Six Vertex Model displayed in Proposition \ref{prop translation invariant}, we certainly expect fluctuations in the space direction to have size of order $x^{1/2}$, as a result of the Central Limit Theorem applied to independent occupation numbers $\mathsf{m}_x^0$ at each site. This means that the information we have about $\mathcal{H}$, which is the choice $\mathcal{H}(0,0)=0$, will be transported by the random dynamics along the direction of growth of the surface and along this line we can observe the emergence of the $1/3$ exponent. Along all other lines, the distribution of \eqref{fluctuations} will be affected very little by the process, and asymptotic fluctuations remain of gaussian nature. An earlier evidence of this last fact was found in \cite{FerrariFontes1994ASEP} (see also \cite{BFP2010sTASEP}, Appendix D).

The direction along which nontrivial fluctuations are observed is given by the characteristic line of partial differential equation \eqref{continuity equation}. This is the curve $(x_t, t)$, where $x_t$ is set to be the solution of the differential equation 
\begin{equation}\label{characteristic line}
\begin{cases}
\dot{x}_t = j'(\rho(x_t, t)),\\
x_0=0.
\end{cases}
\end{equation}
When the system is in its stationary state, equation \eqref{characteristic line} loses its dependence on time and the $\dot{x_t}$ is only function of the stationary density profile $\rho_{\text{st}}$. In case the model does not present space inhomogeneities the characteristic curve is simply the line $(j'(\rho_{\text{st}})t, t)$, but when the stationary density is not constant, this is no more true. We use the explicit parametrization $(x, t_x)$, rather than $(x_t, t)$ and by integration of \eqref{characteristic line}, we obtain
\begin{equation}\label{charateristic line time}
t_x=\int_0^x \frac{dy}{j'(\rho_{\text{st}}(y))}.
\end{equation}
To reiterate what we just explained, consider diverging $x$ and $t$. Assume first that
$$
|t- t_x| = \mathcal{O}(x^{2/3 + \delta}),
$$
for some $\delta > 0$. Then, asymptotically \eqref{fluctuations} obeys the gaussian distributions and its size becomes of order $x^{1/2}$. When on the other hand, $(x,t)$ is taken in the vicinity of the characteristic curve, say
$$
|t- t_x| = \mathcal{O}(x^{2/3}),
$$
then fluctuations become of size $x^{1/3}$ and their law is described by the Baik-Rains distribution.

We can be more precise. Take at first $t= t_x$. The convergence result, in this case, is 
\begin{equation*}
\frac{\mathcal{H}(x, t_x) - \eta ( x, t_x)}{\gamma x^{1/3}} \xrightarrow[x \to \infty]{\mathcal{D}} F_0,
\end{equation*}
where, calling $\sigma_y^2 dy$ the variance of the number of paths lying at time $0$ in the infinitesimal segment $[y,y+dy]$ and its mean
\begin{equation*}
\overline{\sigma^2} = \lim_{x\to \infty} \frac{1}{x} \int_0^x \sigma_y^2 dy,
\end{equation*}
then the constant $\gamma$ is given by
\begin{equation} \label{gamma scaling theory}
\gamma^3 = - \lim_{x \to \infty}\frac{1}{2}j''(\rho_{\text{st}}) (\overline{\sigma^2})^2 \frac{t_x}{x}.
\end{equation}
The explicit parametrization of a fan of size $x^{2/3}$ around the characteristic line can be still expressed in terms of macroscopic quantities. Consider a perturbation of $t_x$ of the form
\begin{equation} \label{time scaling theory}
t_{x, \varpi} = t_x - \varpi \frac{\overline{\sigma}^2 j''(\rho_{\text{st}})}{\gamma j'(\rho_{\text{st}})^2} x^{2/3}.
\end{equation}
with $\varpi$ being a real number. The resulting effect on the expression of $\eta$ reads, up to order $x^{1/3}$, as
\begin{equation} \label{height average scaling theory}
\eta_{x,\varpi} = \eta(x,t_x) - \varpi \frac{\overline{\sigma}^2 j''(\rho_{\text{st}}) j(\rho_{\text{st}}) }{\gamma j'(\rho_{\text{st}})^2} x^{2/3} -\frac{1}{2} \varpi^2 \frac{ ( \overline{\sigma^2} )^2 j''(\rho_{\text{st}}) }{ \gamma^2 j'(\rho_{\text{st}}) }x^{1/3}.
\end{equation}
In this case, the convergence result for fluctuations along the line $(x, t_{x,\varpi})$ becomes 
\begin{equation} \label{Baik Rains conjecture}
\frac{\mathcal{H}(x, t_{x,\varpi}) - \eta_{x,\varpi}}{\gamma x^{1/3}}\xrightarrow[x \to \infty]{\mathcal{D}} F_\varpi.
\end{equation}

The same kind of results are conjectured for discrete time systems, where the characteristic curves can be again explicitely expressed through relation \eqref{charateristic line time}. In the next section we will establish result \eqref{Baik Rains conjecture} for the stationary Higher Spin Six Vertex Model. For this model, the scaling parameters $\kappa_\varpi,\eta_\varpi,\gamma$ were defined in \cref{eq: k0 eta0 gamma0,k perturbed,eta perturbed} and it is a simple exercise to verify that they match with expressions given in \cref{gamma scaling theory,time scaling theory,height average scaling theory}.

\subsection{The Baik-Rains limit} \label{subsection the baik rains limit}

This Section is entirely devoted to the proof of Theorem \ref{theorem: baik rains limit H}, that characterizes the asymptotic fluctuations of the height function in the stationary Higher Spin Six Vertex Model. Our strategy relies on taking the large $x$ limit, after an appropriate scaling of parameters, of expression \eqref{stationary q laplace 2}. Throughout the proof we will assume that the model presents only spatial inhomogeities and hence the spectral parameters $\mathbf{U}$ are taken as
\begin{equation} \label{U choice}
\mathbf{U}= (u, u, u , \dots).
\end{equation}
This simply implies that the transfer matrix $\mathfrak{X}^{(J)}$ stays the same at each time step. 

Before we  move  to the actual core of  the  proof of Theorem \ref{theorem: baik rains limit H}, we like to fix the notation for the main quantities we plan to use. First, we introduce the functions
\begin{equation*}
a_{-1}(z)=\log(z u; q)_J \qquad \text{and} \qquad h_{-1}(z) = \frac{1}{x} \sum_{y=2}^x \log \left( \frac{ ( z s_y /\xi_y ; q )_\infty }{ ( z / (\xi_y s_y) ; q )_\infty } \right),
\end{equation*}
so that, combined with $a_k,h_k$ defined in \eqref{eq: a_k h_k}, they satisfy the properties
\begin{equation*}
z \frac{d}{dz} a_k(z) = a_{k+1}(z) \qquad \text{and} \qquad z \frac{d}{dz} a_k(z) = a_{k+1}(z),
\end{equation*}
for all $k \geq -1$. Define also the function
\begin{equation} \label{g}
g(z) = -\eta \log(z) + \kappa a_{-1}(z) -h_{-1}(z),
\end{equation}
where here and in the rest of the Section, for the sake of a cleaner notation, we set
\begin{equation*}
\eta=\eta_{\varpi}, \qquad \text{and} \qquad \kappa=\kappa_{\varpi},
\end{equation*}
dropping the explicit dependence on the real number $\varpi$ from $\eta_\varpi, \kappa_\varpi$ introduced in \cref{k perturbed,eta perturbed}. A crucial property of the function $g$ is reported next.
\begin{Prop} \label{prop g properties}
With the choice $\varpi=0$, the function $g$, has a double critical point in $\mathpzc{d}$, that is $g'(\mathpzc{d}) = g''(\mathpzc{d}) = 0$. When $\varpi \neq 0$, there exist a point $\varsigma=\varsigma(\varpi)$ such that
\begin{equation} \label{derivatives g}
g'(\varsigma) = g''(\varsigma) = \mathcal{O}(1/x), \qquad g'''(\varsigma) = -2 \frac{\gamma^3}{\varsigma^3} + \mathcal{O}(x^{1/3})
\end{equation}
and $g''(\varsigma)<0$ for $x$ large enough. Moreover, such $\varsigma$ admits the expansion
\begin{equation} \label{varsigma}
\varsigma = \mathpzc{d} \left( 1 + \frac{\varpi}{\gamma x^{1/3}} + \frac{1}{2}  \frac{\varpi^2}{\gamma^2 x^{2/3}} \left( 1+ \frac{a_1^2 h_3 - a_1 a_2 h_2 + 2 a_2^2 h_1 - a_1 a_3 h_1}{a_1(a_2 h_1 - a_1 h_2)} \right)\right) + \smallO (1/x^{2/3}),
\end{equation}
where $a_k,h_k$ are as in \eqref{eq: a_k h_k}.
\end{Prop}
\begin{proof}
Equalities reported in \eqref{derivatives g} can be verified by direct inspection making use of the approximate form of $\varsigma$ \eqref{varsigma}. Therefore the only thing we are left to prove is that $\gamma$ is a positive quantity. From expression \eqref{eq: k0 eta0 gamma0} we write
\begin{equation} \label{gamma expansion}
\gamma^3 = \frac{1}{2} \left( h_2(\mathpzc{d}) - \frac{h_1(\mathpzc{d})}{a_1(\mathpzc{d})}a_2(\mathpzc{d}) \right).
\end{equation}
Functions $a_1, a_2$ have the explicit expressions
\begin{equation*}
a_1(\mathpzc{d}) = \sum_{j=0}^{J-1} \frac{-\mathpzc{d} u q^j }{(1 - \mathpzc{d} u q^j )^2}, \qquad a_2(\mathpzc{d}) = \sum_{j=0}^{J-1} \frac{-\mathpzc{d} u q^j (1+\mathpzc{d} u q^j) }{(1 - \mathpzc{d} u q^j )^3},
\end{equation*}
that can be recovered using the form \eqref{real polygamma} to compute $\polygamma_1, \polygamma_2$. On the other hand, expressing the $\polygamma_k$'s in $h_1,h_2$ using \eqref{digamma psi} we can write, after some algebraic manipulations, the right hand side of \eqref{gamma expansion} as
\begin{equation*}
\frac{1}{2} \frac{1}{x} \sum_{y=2}^x \sum_{k \geq 1} \left( \frac{\mathpzc{d}}{\xi_y s_y} \right)^k \frac{k(1-s_y^{2k})}{1-q^k}\sum_{j=0}^{J-1} \frac{-u \mathpzc{d} q^j }{(1 - u \mathpzc{d} q^j)^2} \left( k - \frac{1 + u \mathpzc{d} q^j}{ 1 - u \mathpzc{d} q^j } \right),
\end{equation*}
that is a sum of positive terms since $u < 0$.
\end{proof}

The function $g$ becomes useful when it comes to rewrite the kernel $f(n) A(n,m)$ and quantities $\Phi_x^{(i)}(n), \Psi_x^{(j)}(n),\Phi_x(n),\Psi_x(n)$, appearing in the statement of Corollary \ref{Corollary stationary q laplace}, after a suitable scaling of parameters. We use the following expressions for the integers $n,m$, in terms of new variables $\nu,\theta$:
\begin{equation} \label{scaling n_m}
n=n_{\nu}= -\eta x + \nu \gamma x^{1/3}, \qquad m=m_\theta= -\eta x + \theta\gamma x^{1/3}.
\end{equation}
Here $\nu, \theta$ belong to the set of rescaled integers
\begin{equation*}
\dot{\mathbb{Z}} = \{ \nu\in \mathbb{R} |-\nu x+\nu\gamma x^{1/3} \in \mathbb{Z} \}
\end{equation*}
and we use the symbol $\sumDot$ to denote a summation where the index ranges over $\dot{\mathbb{Z}}$ rather than $\mathbb{Z}$.

We also set the expression for the time variable $t$ along the critical line as
\begin{equation} \label{scaling t}
t=\kappa x
\end{equation}
and for any fixed real number $r$, the variable $\zeta$ appearing in the $q$-Laplace transform \eqref{stationary q laplace 2} is taken of the form
\begin{equation}\label{scaling zeta}
\zeta = -q^{-\eta x+ \gamma x^{1/3} r }.
\end{equation}

In the following proposition we summarize  the effect of change of variables \eqref{scaling n_m} and of choices 
\eqref{scaling t}, \eqref{scaling zeta} on the statement of Corollary \ref{Corollary stationary q laplace}.
\begin{Prop} \label{proposition scaling}
Assume \cref{scaling n_m,scaling t,scaling zeta}, fix a real number $L$ such that
\begin{equation} \label{L r relation}
-L< r.
\end{equation}
and define the sequence
\begin{equation} \label{b coefficient}
\mathfrak{b}(\nu)=\begin{cases}
\varsigma^{\nu \gamma x^{1/3}}, \qquad & \text{if } \nu > - L,\\
\tau( n_\nu ), \qquad &\text{if } \nu \leq - L.
\end{cases}
\end{equation}
Define also
\begin{equation} \label{f tilde}
\tilde{f}(\nu) = f(n_\nu) = \frac{1}{1+q^{(\nu - r)\gamma x^{1/3}}},
\end{equation}
\begin{equation} \label{rescaled discrete Airy}
\tilde{A}(\nu, \theta) = \frac{\mathfrak{b}(\nu)}{ \mathfrak{b}(\theta) }\frac{\tau(m_\theta)}{ \tau(n_\nu) } A(n_\nu,m_\theta) =  \frac{\mathfrak{b}(\nu)}{ \mathfrak{b}(\theta) }  \int_D \frac{dw}{2 \pi \mathrm{i} w} \int_C \frac{dz }{2 \pi \mathrm{i}}\frac{z^{ \theta \gamma x^{ 1/3 }}}{ w^{ \nu \gamma x^{ 1/3 }}} \frac{e^{ x g(z) } }{ e^{  x g(w)} } \frac{(q \mathpzc{d}/w ; q)_\infty } {(q \mathpzc{d}/z ; q)_\infty} \frac{1}{z-w},
\end{equation}
\begin{equation*}
    \tilde{\Phi}_x^{(1)}(\nu) = d^{-\nu \gamma x^{1/3} -1} e^{ -x g(\mathpzc{d})}, \qquad \tilde{\Psi}_x^{(1)}(\nu)= d^{\nu \gamma x^{1/3}} e^{ x g(\mathpzc{d}) }, 
\end{equation*}
\begin{equation*}
\tilde{\Phi}_x^{(2)}(\nu) =   \int_{D_1} \frac{dw}{2 \pi \mathrm{i} w}  w^{ - \nu \gamma x^{ 1/3 }} e^{ -x g(w) }  \frac{(q \mathpzc{d}/w ; q)_\infty } {(q w/ \mathpzc{d} ; q)_\infty} \frac{1}{w-\mathpzc{d}}, \qquad \tilde{\Phi}_x(\nu)=\tilde{\Phi}_x^{(1)}(\nu) + \tilde{\Phi}_x^{(2)}(\nu),
\end{equation*}
\begin{equation*}
\tilde{\Psi}_x^{(2)}(\nu) =   \int_{C_1} \frac{dz}{2 \pi \mathrm{i} }  z^{ \nu \gamma x^{ 1/3 }} e^{ x g(z) }  \frac{(q z/ \mathpzc{d} ; q)_\infty } {(q  \mathpzc{d}/z ; q)_\infty} \frac{1}{z-\mathpzc{d}}, \qquad \tilde{\Psi}_x(\nu)=\tilde{\Psi}_x^{(1)}(\nu) + \tilde{\Psi}_x^{(2)}(\nu).
\end{equation*}
Then formula \eqref{stationary q laplace 2} for the $q$-Laplace transform still holds if we substitute, in the expression of $V_x$ \eqref{V}, $f, A, \Phi_x^{(i)}, \Psi_x^{(j)}, \Phi_x, \Psi_x$ with $\tilde{f}, \tilde{A}, \tilde{\Phi}_x^{(i)}, \tilde{\Psi}_x^{(j)}, \tilde{\Phi}_x, \tilde{\Psi}_x$ and we change the summation signs $\sum$ with $\sumDot$.
\end{Prop}
\begin{proof}
We can easily see that the tilde notation corresponds to applying to functions in \eqref{V} the change of variables \eqref{scaling n_m}. The Fredholm determinant $\det(\mathbf{1} - fA)_{l^2(\mathbb{Z})}$ is clearly not affected by the multiplication of $A$ with the factor $\frac{\mathfrak{b}(\nu)}{ \mathfrak{b}(\theta) }\frac{\tau(m_\theta)}{ \tau(n_\nu) }$, nor by the change of variables and it is therefore equal to $\det(\mathbf{1} - \tilde{f}\tilde{A})_{l^2(\dot{\mathbb{Z}})}$. Similar considerations are true also for the remaining functions in \eqref{V}.
\end{proof}
Expressions reported in the statement of Proposition \ref{proposition scaling} are indeed amenable to a rigorous asymptotic analysis when $x$ tends to infinity. Such limits are computed through a steep descent method using the fact, reported in Proposition \ref{prop g properties}, that $g$ has a double critical point in the vicinity of $\mathpzc{d}$. For the sake of a rigorous  procedure, in the next Definition we make some hypothesis on parameters $q,\Xi, \mathbf{S}$, that will hold true throughout the rest of the Section.

\begin{Def}[Conditions on parameters] \label{conditions on parameters}
Take $a,\sigma$ such that $a>\mathpzc{d}$ and $\sigma \in [0,1)$. Parameters $q,\mathpzc{d},\Xi, \mathbf{S}$ are assumed to satisfy \eqref{HS6VM parameters},\eqref{eq: placements Xi S} and they are spaced so that there exist $ R_a,R_\sigma,R_q$, with the properties that
\begin{equation}\label{bounds parameters asymptotics}
a \leq \xi_k s_k \leq a+R_a, \qquad \sigma \leq s_k^2 \leq \sigma + R_\sigma, \qquad \text{for all }k, \qquad 0 \leq q \leq R_q
\end{equation}
and
\begin{equation}\label{bound R_a}
a+R_a < \frac{2 a}{1+\sigma} < \mathpzc{d}/q.
\end{equation}
Numbers $R_a,R_\sigma,R_q$ are strictly positive, yet small in the sense given by Proposition \ref{prop contour D}.
\end{Def} 
The reason for such a restricted choice of parameters lies in the perturbative approach we used to prove Proposition \ref{prop contour D}. There, we showed the steep ascent property for integration contour $D$, in the case where $q=0, \xi_k s_k =a, s_k^2=\sigma$ for each $k \geq 2$. Subsequently, through a continuity argument we concluded that the same property must hold also when parameters are taken in suitably small neighborhoods of our original choices, hence \eqref{bounds parameters asymptotics}. Obtaining a similar kind of result when $q$ is taken far from 0, means constructing an explicit closed contour $D$ on which one would be able to show that the function $g$ assumes a global minimum at $\varsigma$. This is indeed possible in principle, but obtaining explicit bounds for parameters $q,\xi_k,s_k$ becomes prohibitive.

\begin{Remark}
Conditions stated in Definition \ref{conditions on parameters} are far from being optimal and they are essentially consequences of our choice for the representation of the integral kernels $K,\Phi_x,\Psi_x$ in \eqref{kernel},\eqref{Phi}, \eqref{Psi}. In particular the assumption $2a/(1+\sigma) < q^{-1} \mathpzc{d}$, reported in \eqref{bound R_a}, is used to ensure the exponential decay of rear tails of $f \Phi_x \Psi_x$ in Propositions \ref{prop Upsilon}, \ref{prop Upsilon 2}.
\end{Remark}
We will now start the computation of limiting expressions of quantities presented in Proposition \ref{proposition scaling} for $x \to \infty$. 
\begin{Prop} \label{prop discrete Airy}
We have
\begin{equation} \label{Tracy Widom limit}
\det \left( \mathbf{1} - \tilde{f}\tilde{A} \right)_{l^2(\dot{\mathbb{Z}})} = F_2(r) + \frac{1}{\gamma x^{ 1/3 }} R_x^{(1)}(r),
\end{equation}
where $F_2$ is defined in \eqref{tracy widom} and the error term $R_x^{(1)}$ satisfies the following properties
\begin{enumerate}
\item For each $r^* \in \mathbb{R}$, there exists $M_{r^*}>0$ such that, for all $x$,
\begin{equation} \label{remander bounded}
\left | R^{(1)}_x (r^*) \right | < M_{r^*};
\end{equation}
\item There exist $\epsilon>0$, such that, for all $r^* \in [r - \epsilon, r]$, we have
\begin{equation} \label{remander continuous}
\lim_{x \to \infty} \left( R^{(1)}_x (r^*) - R^{(1)}_x (r^* -  \frac{1}{\gamma x^{1/3}} ) \right) = 0,
\end{equation}
uniformly.
\end{enumerate}
\end{Prop}

The proof of this crucial fact is long and we present it in several steps.

\begin{Lemma}[Convergence on moderately large sets] \label{convergence on moderately large sets}
Let $\delta$ be a number in the interval $(0, 1/3)$. Then for $(\nu, \theta) \in [-L, x^{\delta/3}]^2$ we have \footnote{for motivation on the choice of sets $[-L, x^{\delta / 3}]$ see Remark
\ref{remark choice sets}.},
\begin{equation} \label{formula convergence on bounded sets}
\tilde{A}(\nu, \theta) = \frac{1}{\gamma x^{1/3}} K_{\emph{Airy}}(\nu, \theta) + \frac{1}{\gamma^2 x^{2/3}} Q(\nu, \theta) + \mathcal{O} \left( x^{2\delta/3-1}  \right)
\end{equation}
and the error term satisfies 
\begin{equation} \label{bound remainder}
x^{2/3}\mathcal{O}(x^{2 \delta /3 -1}) \xrightarrow[x \to \infty]{} 0,
\end{equation}
uniformly in the sequence of sets $(\nu, \theta) \in [-L ,  x^{\delta/3}]^2$. Moreover the exponential estimates, 
\begin{equation} \label{exponential bound kernel}
|K_{\emph{Airy}}(\nu, \theta)|, |Q(\nu, \theta)| < c_1 e^{-c_2(\nu + \theta)},
\end{equation}
hold for all $(\nu, \theta) \in [-L , x^{\delta/3}]^2$, for an opportune choice of positive constants $c_1, c_2$ which do not depend on $x$.
\end{Lemma}
\begin{proof}

The definition itself of scaling parameters is functional to perform a saddle point analysis. In particular we want to show that, when $\nu$ and $\theta$ are relatively small quantities, compared to $x^{1/3}$, the integrals in \eqref{rescaled discrete Airy} are dominated by the value of the integrands at the double critical point $\varsigma$.
To do so we suitably deform contours $C,D$ in such a way that, for $x$ large enough, the following properties hold:\footnote{for the sake of the uniform convergence over compact sets conditions 2,4 are not necessary, but we still state them as they will become useful later in Lemma \ref{tail}.
}
\begin{enumerate}
\item $\max_{z\in C} \mathfrak{Re}\{g(z)\} = g(\varsigma(1- \frac{1}{2 \gamma x ^{1/3}}))$;
\item $\max_{z \in C}|z| = \varsigma(1- \frac{1}{2 \gamma x ^{1/3}})$;
\item $\min_{w \in D} \mathfrak{Re}\{g(z)\} = g(\varsigma(1+ \frac{1}{2 \gamma x ^{1/3}}))$;
\item $\min_{w \in D} |w| = \varsigma(1+ \frac{1}{2 \gamma x ^{1/3}})$.
\end{enumerate}

The idea is to take paths like those depicted in Figure \ref{contours asymptotics}. Based on results of Appendix \ref{appendix contours}, we now construct the steep descent contour $C$. The same procedure can be applied to provide an exact expression for $D$ as well and therefore we will omit this in the discussion.

Fix an arbitrarily small positive number $\epsilon$ and consider $C$ to be the union of two curves $\tilde{C}_1$, $\tilde{C}_2$ such that
\begin{align}
\tilde{C}_1&=\partial \mathsf{D} \left( 0, \varsigma (1- \epsilon) \right) \cap \left\{ z\in \mathbb{C} |\ \mathfrak{Re}(z)\leq \frac{\varsigma}{4}( 3 + \sqrt{1 - 8 \epsilon + 4 \epsilon^2 } ) \right\},\\
\tilde{C}_2&= \left\{ \mathbbm{1}_{[0,\frac{2}{\gamma x^{1/3}}]}(|\rho|) \varsigma \left(  1 - \frac{4 + \gamma^2 x^{2/3} \rho^2}{ 8 \gamma x^{1/3} } + \mathrm{i} \frac{\sqrt{3}}{2} \rho \right) \right. \nonumber
\\ & \qquad \left. + \mathbbm{1}_{[\frac{2}{\gamma x^{1/3}},\infty)}(|\rho|) \varsigma \left(  1 - \frac{ |\rho|}{2} + \mathrm{i} \frac{\sqrt{3}}{2} \rho \right) :\ |\rho| \leq \frac{1}{2}(1-\sqrt{1 - 8 \epsilon + 4 \epsilon^2}) \right\},
\end{align}
where $\partial\mathsf{D}(c,R)$ indicates a circumference of center $c$ and radius $R$. To put it in simple terms $C$ is a circle of radius $\varsigma(1- \epsilon)$ up until it intersects for the first time (from the left) the two complex lines exiting from $\varsigma$ with slope $\pm \frac{2 \pi}{3}$ (as in Figure \ref{contours asymptotics}, b)). After $C$ meets these intersection points, denoted with $p_\pm$, it becomes $\tilde{C}_2$, a regular curve which coincides with such lines for a while and passes strictly to the left of $\varsigma$. 

We claim that the contribution of the integral in the $z$ variable in \eqref{rescaled discrete Airy} are given, up to an error which is exponentially small in $x$, by the integral along the contour $\tilde{C}_2$. To show this, we first notice that from Proposition \ref{prop contour C}, if $\epsilon$ is small enough we can assume that, along $\tilde{C}_1$ the real part of $g(z)$  is a decreasing function. Therefore, the contribution of the term $e^{xg(z)}$ can be estimated by its values at the extremal points of $\tilde{C}_1$,
\begin{equation*}
p_{\pm} = \frac{\varsigma}{4} \left( 3 + \sqrt{1 - 8 \epsilon + 4 \epsilon^2} \right) \pm \mathrm{i} \frac{\sqrt{3}}{4} \varsigma (1- \sqrt{1- 8 \epsilon + 4 \epsilon^2} ) \approx \varsigma \left( 1- \epsilon \pm \mathrm{i} \sqrt{3} \epsilon \right).
\end{equation*}
Let us evaluate the quantity $\mathfrak{Re}\{ g(p_{\pm}) \} - g(\varsigma)$ through a Taylor expansion. By using \eqref{derivatives g}, we have
\begin{equation*}
\mathfrak{Re}\{ g(p_{\pm}) \} - g(\varsigma) = \frac{ 8 g'''(\varsigma) \varsigma^3}{3!} \epsilon^3 + R(\epsilon) \epsilon^4,
\end{equation*}
where $R(\epsilon)$ is the Taylor remainder and it is a regular, bounded function in a neighborhood of zero. The factor $\varsigma^3 g'''(\varsigma)$ is strictly negative, as stated in Proposition \ref{prop g properties} and therefore we obtain the bound
\begin{equation*}
e^{x(g(z) - g(\varsigma))} \leq e^{-cx}, \qquad \text{ for each } z \in \tilde{C}_1,
\end{equation*}
which holds for some positive constant $c$. 

Through an analogous argument we can deform the $D$ contour too and separate it in an union of two curves $\tilde{D}_1$ and $\tilde{D}_2$ (see Figure \ref{contours asymptotics}). As for the $C$ contour case, we can take $\tilde{D}_2$ to be a curve that follows the two complex half lines $\{ \varsigma + e^{\pm \mathrm{i} \frac{\pi}{3}} \rho: \rho \geq 0 \}$ in a neighborhood of size $\epsilon$ of $\varsigma$ and that passes strictly to the right of $\varsigma$. For $\epsilon$ small enough, but still of order 1, the remaining contour $\tilde{D}_1$ can be chosen so that the contribution of the $w$ integral over $\tilde{D}_1$ to the kernel $\tilde{A}$ are exponentially small in $x$.  

We also remark that curves $\tilde{C}_2, \tilde{D}_2$ are kept at a distance of size $x^{-1/3}$ from $\varsigma$ (and hence from each other) due to the presence in the integral expression of $\tilde{A}$ of a sigularity at $z=w$.

\begin{figure}[t]
\centering

\begin{minipage}{.45 \textwidth}
\includegraphics{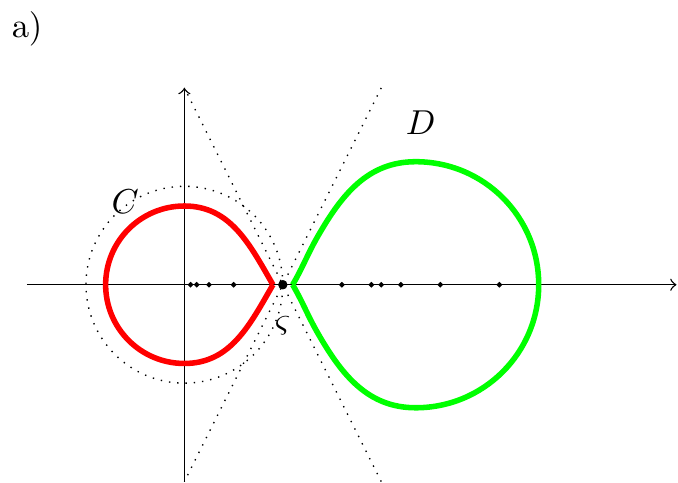}
\end{minipage}
\begin{minipage}{.45 \textwidth}
\includegraphics{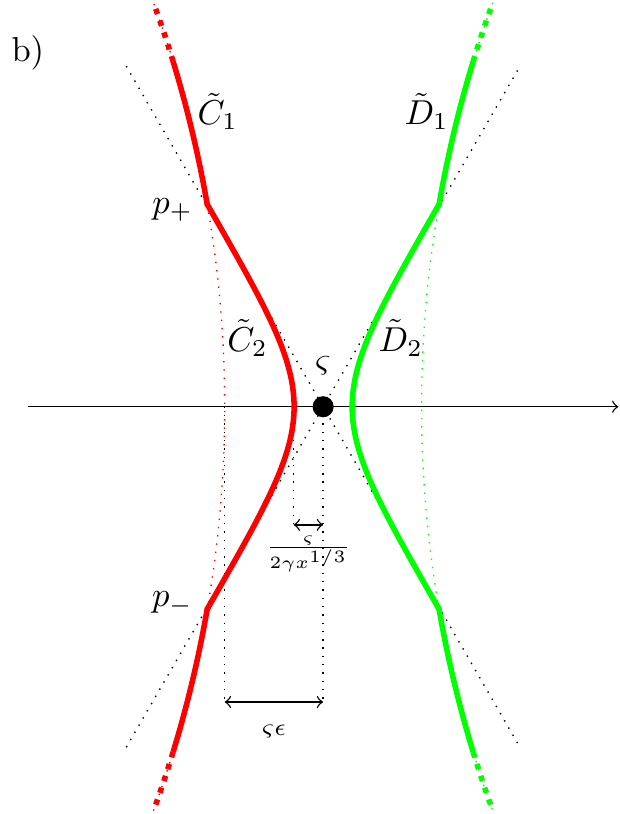}
\end{minipage}

\caption{\small a) Choices of integration contours in Lemma \ref{convergence on moderately large sets}. The red contour $C$ encircles the sigularities $\{q^k \mathpzc{d}\}_{k \geq 1}$, is contained inside a circle of radius $\varsigma(1-\epsilon)$ and joins the point $\varsigma( 1 - \frac{1}{2 \gamma x^{1/3}})$ with slope $\frac{\pi}{3}$ from above (resp. $-\frac{\pi}{3}$ from below). The green contour $D$ contains the singularities $\{ \xi_k s_k \}_{k=2, \dots, x}$ and forms at the point $\varsigma ( 1 + \frac{1}{2 \gamma x^{1/3}} )$ a cusp of width $\frac{2}{3}\pi$, symmetric to that of $C.$\\
b) A representation of contours $C$ and $D$ in the immediate vicinity of the critical point $\varsigma$.} \label{contours asymptotics}
\end{figure}

We can summarize discussion made so far expressing the kernel $\tilde{A}$ as
\begin{equation} \label{kernel A C2 D2}
\tilde{A}(\nu, \theta) = \frac{ \varsigma^{\nu \gamma x^{1/3}} }{ \varsigma^{\theta \gamma x^{1/3}} } \frac{1}{(2 \pi \mathrm{i})^2} \int_{\tilde{D}_2} \frac{dw}{w} \int_{\tilde{C}_2} dz \frac{z^{ \theta \gamma x^{ 1/3 } }}{ w^{ \nu \gamma x^{ 1/3 }}} \frac{\exp \{ x g(z) \} }{ \exp \{  x g(w)\} } \frac{(q \mathpzc{d}/w ; q)_\infty } {(q \mathpzc{d}/z ; q)_\infty} \frac{1}{z-w} + \mathcal{O}(e^{-cx}),
\end{equation}
where we notice that, with respect to \eqref{rescaled discrete Airy}, the integration contours have become $\tilde{C}_2$ and $\tilde{D}_2$ and the remainder is a quantity which decays as an exponential in $x$. We can now safely employ the saddle point method to give an estimate of the integral expression in \eqref{kernel A C2 D2}. The only significant contribution to the double integral \eqref{rescaled discrete Airy} is given when variables $z,w$ are separated from $\varsigma$ by a distance of order $x^{-1/3}$. For this reason we like to apply the change of variables
\begin{equation*}
z=\varsigma \left( 1- \frac{Z}{\gamma x^{ 1/3 }} \right), \qquad \qquad w=\varsigma \left( 1- \frac{W}{\gamma x^{ 1/3 }} \right),
\end{equation*}
and we write, through simple Taylor expansions, different terms of the integrand function in \eqref{kernel A C2 D2} as 
\begin{gather}
\frac{z^{\theta \gamma x^{1/3}}}{w^{\nu \gamma x^{1/3}}} = \frac{\varsigma^{\theta \gamma x^{1/3}}}{\varsigma^{\nu \gamma x^{1/3}}} \frac{e^{-\theta Z}}{e^{-\nu W}}\left[ 1 + \frac{1}{\gamma x^{1/3}}\left( \frac{\nu W^2}{2} - \frac{\theta Z^2}{2} \right) + \mathcal{O}\left( \frac{Z^3 \theta}{x^{2/3}}, \frac{W^3 \nu}{x^{2/3}}, \frac{Z^4 \theta^2}{x^{2/3}} , \frac{W^4 \nu^2}{x^{2/3}}\right) \right], \label{taylor expansion integrand 1} \\
\frac{e^{xg(z)}}{e^{ x g(w)} } = \frac{e^{Z^3/3}}{e^{W^3/3}} \left[ 1 + \frac{1}{\gamma x^{1/3}} \left( E_1 Z^4 - E_1 W^4 \right) + \mathcal{O}\left( \frac{Z^8}{x^{2/3}} ,\frac{W^8}{x^{2/3}} \right) \right], \label{taylor expansion integrand 2} \\
\frac{(q \mathpzc{d}/w ; q)_\infty }{(q \mathpzc{d}/z ; q)_\infty} = 1 + \frac{1}{\gamma x^{1/3}} ( E_2 Z - E_2 W ) + \mathcal{O}\left( \frac{Z^2}{x^{2/3}} , \frac{W^2}{x^{2/3}} \right).\label{taylor expansion integrand 3}
\end{gather}
In these expressions, coefficients $E_1,
E_2$, naturally possess exact expressions, which we do not report as they are irrelevant for the computations. 

Thanks to \eqref{kernel A C2 D2},\eqref{taylor expansion integrand 1},\eqref{taylor expansion integrand 2},\eqref{taylor expansion integrand 3} we obtain an expansion of $\tilde{A}$ in the infinitesimal quantity $1/(\gamma x^{1/3})$. Collecting together terms of order $1/(\gamma x^{1/3})$ and $1/(\gamma x^{1/3})^2$ we obtain
\begin{equation} \label{tilde A order 1}
\tilde{A}(\nu, \theta) = \frac{1}{\gamma x^{1/3}} \int_{e^{-\frac{2}{3}\pi \mathrm{i}} \infty}^{e^{\frac{2}{3}\pi \mathrm{i}} \infty}   \frac{dW}{2 \pi \mathrm{i}} \int_{e^{\frac{\pi}{3} \mathrm{i}} \infty}^{e^{-\frac{\pi}{3} \mathrm{i}} \infty} \frac{dZ}{2 \pi \mathrm{i}} \frac{e^{Z^3/3 - \theta Z}}{e^{W^3/3 - \nu W}} \frac{1}{W - Z} + \frac{1}{(\gamma x^{1/3})^2} Q(\nu, \theta) + \mathcal{O}(x^{2\delta/3 - 1}).
\end{equation}
with the kernel $Q$ being given by
\begin{equation} \label{kernel Q}
Q(\nu, \theta) = 
\int_{e^{-\frac{2}{3}\pi \mathrm{i}} \infty}^{e^{\frac{2}{3}\pi \mathrm{i}} \infty}   \frac{dW}{2 \pi \mathrm{i}} \int_{e^{\frac{\pi}{3} \mathrm{i}} \infty}^{e^{-\frac{\pi}{3} \mathrm{i}} \infty} \frac{ d Z}{2 \pi \mathrm{i}} \frac{e^{Z^3/3 - \theta Z}}{e^{W^3/3 - \nu W}} \left(  \frac{\nu W^2}{2} - \frac{\theta Z^2}{2} +  E_1 (Z^4 - W^4) +  E_2 (Z - W) \right) \frac{1}{W - Z}.
\end{equation}
By recognizing the expression of the Airy kernel \eqref{Airy kernel} in \eqref{tilde A order 1} we write $\tilde{A}$ as in \eqref{formula convergence on bounded sets}.

All we are left to do is to prove the exponential bound \eqref{exponential bound kernel} for $Q$, since the same type of estimate for $K_{\text{Airy}}$ follows from well known decay properties of the Airy functions \cite{abramowitz+stegun}. To do this consider the following parametrization of the integration variables 
\begin{equation} \label{change of variable epsilon}
Z = \frac{\tilde{b}}{2} + |\varrho_1| e^{-\sign(\varrho_1) \mathrm{i} \pi/3}, \qquad W = -\frac{\tilde{b}}{2} + |\varrho_2| e^{-\sign(\varrho_2) \mathrm{i} 2 \pi/3},
\end{equation}
for $\rho_1, \rho_2 \in \mathbb{R}$ and $\tilde{b}$ being a positive real number. Applying the substitution \eqref{change of variable epsilon} in \eqref{kernel Q}, we straightforwardly obtain an inequality like

\begin{equation}\label{kernel Q estimate}
\begin{split}
|Q(\nu, \theta)| < e^{-\frac{\tilde{b}}{2}(\theta + \nu)} \frac{e^{\tilde{b}^3/12}}{\tilde{b}} \int_0^\infty d\varrho_1 \int_0^\infty d\varrho_2 & \Big( |\theta| P_{\tilde{b}}(\rho_1) + |\nu| P_{\tilde{b}}(\rho_2)  + S_{\tilde{b}}(\varrho_1) +  S_{\tilde{b}}(\varrho_2) \Big)\\
& \times e^{-\frac{\tilde{b}}{4}(\rho_1^2 + \rho_2^2) -\left( \frac{\theta}{2} -\frac{\tilde{b}^2}{8} \right)\rho_1 -\left( \frac{\nu}{2} -\frac{\tilde{b}^2}{8} \right)\rho_2},
\end{split}
\end{equation}
where $P_{\tilde{b}}$ and $S_{\tilde{b}}$ are polynomials and by making use of elementary estimates on the integrals on the right hand side of \eqref{kernel Q estimate}, we can finally show \eqref{exponential bound kernel}.

The error term $\mathcal{O}(x^{2 \delta /3 -1})$ in \eqref{formula convergence on bounded sets} is obtained taking into account quantities
\begin{equation*}
\mathcal{O}\left( \frac{Z^3 \theta}{x^{2/3}}, \frac{W^3 \nu}{x^{2/3}}, \frac{Z^4 \theta^2}{x^{2/3}} , \frac{W^4 \nu^2}{x^{2/3}}\right) , \mathcal{O}\left( \frac{Z^8}{x^{2/3}} ,\frac{W^8}{x^{2/3}} \right), \mathcal{O}\left( \frac{Z^2}{x^{2/3}} , \frac{W^2}{x^{2/3}} \right),\mathcal{O}(e^{-cx})
\end{equation*}
from \eqref{kernel A C2 D2} and \eqref{taylor expansion integrand 1}, \eqref{taylor expansion integrand 2}, \eqref{taylor expansion integrand 3} in the saddle point integration. Due to the presence of the exponentially decaying term $e^{Z^{3}/3 - W^3/3 - Z \theta + W \nu}$ we can formulate bounds like \eqref{kernel Q estimate} for these remainders as well, to finally show \eqref{bound remainder}. This concludes our proof.
\end{proof}

\begin{Lemma}[Exponential decay of front tails] \label{decay tails}
Let $L'$ be an arbitrary large positive real numbers (possibly of order $x$ raised to some power). Then there exists $x_*$, such that for all $x>x_*$ the bound
\begin{equation} \label{tail}
\left| \gamma x^{1/3} \tilde{A}(\nu, \theta) \right| <  e^{-\nu - \theta}
\end{equation}
holds for each $(\nu, \theta) \in [-L, \infty)^2 \setminus [-L,L']^2$. 
\end{Lemma}
\begin{proof}
We use again suitable deformations of contours described in Lemma \ref{convergence on moderately large sets} to estimate, for large $x$, the contribution of the factor 
\begin{equation*}
\frac{z^{\theta \gamma x^{ 1/3 }}}{ w^{\nu \gamma x^{ 1/3 }} }
\end{equation*}
to the double integral \eqref{rescaled discrete Airy}. Let's first prove \eqref{tail} in the case $\theta \geq \nu$. When this is the case, we take the contour $D$ exactly as in Lemma \ref{convergence on moderately large sets} and we modify $C = \tilde{C}_1 \cup \tilde{C}_2$, where
\begin{equation*}
\begin{split}
&\tilde{C}_1= \partial \mathsf{D} \left( 0, \varsigma- \frac{2 \varsigma }{\gamma x^{1/3}} \right) \cap \{ z\in \mathbb{C} |\ \mathfrak{Re}(z)\leq \varsigma( 1 -\frac{3}{\gamma x^{ 1/3 }}) \},\\
&\tilde{C}_2=\varsigma(1-\frac{3}{\gamma x^{ 1/3 }}) + \mathrm{i} [-\varsigma \tilde{a} , \varsigma  \tilde{a}]
\end{split}
\end{equation*}
and $\tilde{a}$ is given by the intersections of the vertical complex line $\{ \varsigma(1-\frac{3}{\gamma x^{ 1/3 }}) + \mathrm{i} y |\ y\in \mathbb{R}  \}$ with the circle $\partial \mathsf{D} \left( 0, \varsigma- \frac{2 \varsigma}{\gamma x^{1/3}} \right)$. We can also write down its exact expression as 
\begin{equation*}
\tilde{a}= \sqrt{\frac{2 }{\gamma x^{1/3}}-\frac{5}{ \gamma^2 x^{2/3}}} \approx \sqrt[]{\frac{2}{\gamma}} \frac{1}{x^{1/6}} + \mathcal{O} (x^{-1/3}) .
\end{equation*}
From Proposition \ref{prop contour C}, $\partial \mathsf{D} \left( 0, \varsigma ( 1 - \frac{2 }{\gamma x^{1/3}} ) \right)$ is a steep descent contour for $\mathfrak{Re}(g)$ and we can assume that
\begin{equation*}
\max_{z \in C_1 } \mathfrak{Re} \{ g(z) \} =  \mathfrak{Re} \{ g( \varsigma ( 1 - \frac{3}{ \gamma x^{1/3}} + \mathrm{i} \tilde{a}  ) ) \}.
\end{equation*}
To evaluate the real part of the function $g$ on the complex segment $\tilde{C}_2$ we use the parametrization
\begin{equation} \label{z on C_2}
z=\varsigma \left( 1 - \frac{3}{\gamma x^{1/3}} + \mathrm{i} \frac{Z}{\gamma x^{1/3}} \right).
\end{equation}
In this case $Z$ is a real number ranging in an interval which, up to corrections of order $x^{-1/3}$ is $[-\sqrt{2 \gamma} x^{1/6} , \sqrt{2 \gamma} x^{1/6}]$. Expanding $g$ in Taylor series around $\varsigma$ and recalling \eqref{derivatives g}, we have 
\begin{equation} \label{expansion g Z}
\mathfrak{Re}\{ g(z) \} - g(\varsigma) = \frac{1}{ x } \frac{\varsigma^3 g'''(\varsigma)}{ 3! \gamma^3} (- 27 + 9 Z^2) + \frac{1}{x^{4/3}} \frac{\varsigma^4 g^{(4)}(\varsigma) }{ 4! \gamma^4 } Z^4 + \mathcal{O}\left( \frac{Z^2}{x^{4/3}} , \frac{Z^4}{x^{5/3}} , \frac{Z^6}{x^2} \right),
\end{equation}
where the presence of terms of order higher than three takes into account the fact that $Z$ can be of order $x^{1/6}$. When $Z/(\gamma x^{1/3})=\tilde{a}$, \eqref{expansion g Z} becomes
\begin{equation*}
\mathfrak{Re}\{ g(\varsigma(1-\frac{3}{\gamma x^{1/3}} +\mathrm{i} \tilde{a})) \} - g(\varsigma) = \frac{1}{\gamma^2 x^{2/3}} \left( 3 \varsigma^3 g'''(\varsigma) + \frac{1}{6} \varsigma^4 g^{(4)}(\varsigma) \right) + \mathcal{O}(x^{-1})
\end{equation*}
and the term on the right hand side of order $x^{-2/3}$ is negative. This can be shown either directly computing the derivatives of $g$ or simply recalling that the point $\varsigma(1-\frac{3}{\gamma x^{1/3}} +\mathrm{i} \tilde{a})$ lies on a steep descent contour. These calculations imply the estimate
\begin{equation} \label{bound tails 1}
|e^{x(g(z) - g(\varsigma) )}| \leq e^{-c x^{1/3}}, \qquad \text{ for each }z \in \tilde{C}_1, 
\end{equation}
for some positive constant $c$. On the other hand, when $z$ belongs to $\tilde{C}_2$, \eqref{expansion g Z} gives us that
\begin{equation} \label{bound tails 2}
|e^{x(g(z) - g(\varsigma) )}| \leq e^{9 - \tilde{c} Z^2 }, \qquad \text{for each } |Z| \leq \gamma \tilde{a} x^{1/3},
\end{equation}
for some other positive constant $\tilde{c}$. 

\begin{figure}[t]
\centering
\includegraphics{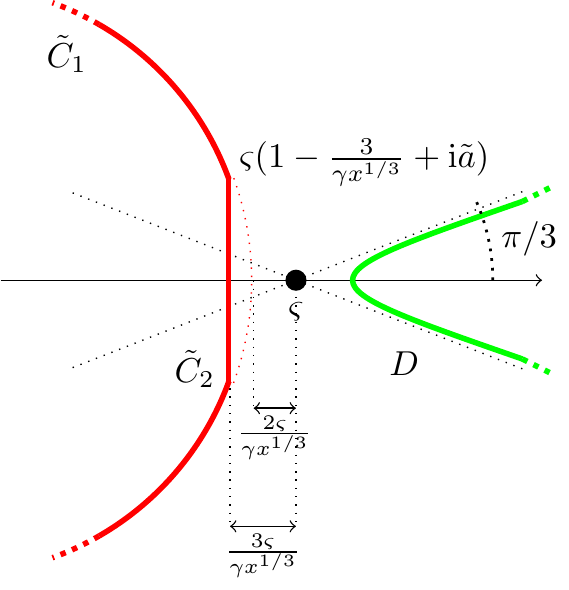}
\caption{\small Choice of integration contours of Lemma \ref{decay tails}. The red contour $C$ is the union of $\tilde{C}_1$, an arc of the circle of center 0 and radius $\varsigma(1- \frac{2}{\gamma x^{1/3}})$, and $\tilde{C}_2$, a vertical segment passing for the point $\varsigma(1 - \frac{3}{\gamma x^{1/3}})$ on the real axis. On the other hand $D$, in the vicinity of the critical point $\varsigma$, stays close to the lines exiting from $\varsigma$ with slope $\pm \frac{\pi}{3}$ (dotted lines).  } \label{contours asymptotics tails}
\end{figure}

To complete the list of preliminary estimates for terms depending on $z$ in the integral formula \eqref{rescaled discrete Airy} of the kernel $\tilde{A}$, we need to address the factor $z^{\theta \gamma x^{1/3}}$. First we notice that, since the contour $C$ lies inside the circle centered at $0$ with radius $\varsigma(1 - 2/(\gamma x^{1/3}))$, we have
\begin{equation} \label{bound tails 3}
\left| \frac{z}{\varsigma} \right|^{\theta \gamma x^{1/3}} \leq \exp \left\{\theta \gamma x^{1/3} \log \left(1 - \frac{
2}{\gamma x^{1/3}} \right) \right\} \leq  e^{-2 \theta}, \qquad \text{for each } z\in C,
\end{equation}
as a result of the simple inequlity $\log(1 + y) \leq y$, valid for all $y>-1$. Moreover, when $z$ is on $\tilde{C}_2$, using the parametrization \eqref{z on C_2}, we have
\begin{equation} \label{bound tails 4}
\left| \frac{z}{\varsigma} \right|^{\theta \gamma x^{1/3}} = \exp\left\{ \theta \gamma x^{1/3} \log \left| 1 - \frac{3}{\gamma x^{1/3}} + \frac{ \mathrm{i} Z }{\gamma x^{1/3}} \right|  \right\}  \leq \exp\left\{ - \theta \left( 3 - \frac{9 + Z^2}{2 \gamma x^{1/3}} \right) \right\}.
\end{equation}
To evaluate the kernel $\tilde{A}$ we also need to provide some estimates for quantities involving the variable $w$. The choice of contours $C,D$ implies that
\begin{equation}\label{bound tails 5}
\frac{1}{z-w} \leq \frac{\varsigma}{ \gamma x^{1/3}} \qquad \text{and} \qquad \left| \frac{(q \mathpzc{d}/w ; q)_\infty}{(q \mathpzc{d}/z ; q)_\infty} \right| \leq \Gamma_1,
\end{equation}
for some constant $\Gamma_1$. In addition, since $\nu > -L$ and $|w|> \varsigma$, combined with the fact that $D$ is steep ascent for the function $\mathfrak{Re}\{g\}$, as proved in Proposition \ref{prop contour D}, we have that
\begin{equation}\label{bound tails 6}
\left| \frac{\varsigma}{w} \right|^{\nu \gamma x^{1/3}} \exp\{ x ( g(\varsigma) - g(w) ) \} \leq \left| \frac{\varsigma}{w} \right|^{ - L \gamma x^{1/3}} \exp\{ x (g(\varsigma) - g(w)) \} \leq \Gamma_2,
\end{equation}
for some other constant $\Gamma_2$. Combining together inequalities \eqref{bound tails 4}, \eqref{bound tails 5}, \eqref{bound tails 6}, we can write
\begin{equation} \label{bound tails 7}
\begin{split}
|\tilde{A}(\nu, \theta)| &= \frac{\varsigma^{\nu \gamma x^{1/3}} }{\varsigma^{\theta \gamma x^{1/3}}} \left| \int_C \frac{dz}{2 \pi} \int_D \frac{dw}{2 \pi w} \frac{ z^{\theta \gamma x^{1/3}} }{ w^{ \nu \gamma x^{1/3}} } \frac{\exp\{ x g(z) \}}{ \exp\{ x g(w) \} } \frac{(q \mathpzc{d}/w ; q )_\infty }{ (q \mathpzc{d}/z ; q )_\infty } \frac{1}{z-w} \right|\\
&\leq \frac{\Gamma_1 \Gamma_2 l(D)}{(2 \pi)^2 \gamma x^{1/3}} \int_C \left| dz \left( \frac{z}{\varsigma} \right)^{\theta \gamma x^{1/3} } \exp\{ x(g(z) - g(\varsigma)) \}  \right|,
\end{split}
\end{equation}
where $l(D)$ is the length of the curve $D$. The integral over $C$ is naturally split into different contributions coming from contours $\tilde{C}_1$ and $\tilde{C}_2$. On $\tilde{C}_1$, utilizing \eqref{bound tails 1} and \eqref{bound tails 3} we have
\begin{equation} \label{bound tails 8}
\int_{\tilde{C}_1}\left| dz \left( \frac{z}{\varsigma} \right)^{\theta \gamma x^{1/3}} \exp\{ x ( g(z) - g(\varsigma) ) \}  \right| \leq e^{-2 \theta} e^{-c x^{-1/3}} l(\tilde{C}_1),
\end{equation}
whereas on $\tilde{C}_2$, from \eqref{bound tails 2}, \eqref{bound tails 4} we obtain
\begin{equation} \label{bound tails 9}
\int_{\tilde{C}_2} \left| dz \left( \frac{z}{\varsigma} \right)^{\theta \gamma x^{1/3}} \exp\{ x (g(z) - g(\varsigma)) \} \right| \leq \int_{-a \gamma x^{1/3}}^{a \gamma x^{1/3}} dZ \exp \left\{ - \theta \left( 3 - \frac{9 + Z^2}{2 \gamma x^{1/3}} \right) + 9 - \tilde{c} Z^2 \right\}.
\end{equation}
To estimate the integral on the right hand side of \eqref{bound tails 9}, set a large integer $N$ and split the integration segmant into $|Z|<N$ and $N<|Z|< \tilde{a} \gamma x^{1/3}$. When $|Z|<N$ the term $\frac{9 + Z^2}{2 \gamma x^{1/3}}$ is small and we can denote it with $\mathcal{O}(N^2/x^{1/3})$. On the other hand, when $N<|Z|< \tilde{a} \gamma x^{1/3}$, since $(3 - \frac{9 + Z^2}{2 \gamma x^{1/3}}) > 2$, the integrand becomes very small due to the presence of the exponential of $-\tilde{c}Z^2$. We can therefore write
\begin{equation}\label{bound tails 10}
\begin{split}
\text{rhs }\eqref{bound tails 9} &\leq e^{-2 \theta} \left( e^{- \theta ( 1- \mathcal{O}(N^2/x^{1/3}) )} \int_{-N}^N e^{9 - \tilde{c}Z^2} dZ + \int_{N<|Z|\tilde{a} \gamma x^{1/3}} dZ e^{9 - \tilde{c} Z^2 } \right)\\
&= e^{- 2 \theta} \left( e^{- \theta (1 - \mathcal{O}(N^2/x^{1/3}))} \Gamma_3 + \mathcal{O}(e^{- \tilde{c} N^2 }) \right),
\end{split}
\end{equation}
with $\Gamma_3$ being a constant coming from the integration of the exponential.

We can now plug \eqref{bound tails 8}, \eqref{bound tails 9}, \eqref{bound tails 10} into the right hand side of \eqref{bound tails 7} to finally obtain
\begin{equation} \label{bound tail 11}
|\tilde{A}(\nu, \theta)| < e^{- 2 \theta} \left( e^{- c x^{1/3}} l(C_1) + e^{-\theta (1 - \mathcal{O}(N^2/x^{1/3}))} \Gamma_3 + \mathcal{O}(e^{- \tilde{c} N^2 }) \right) \frac{\Gamma_1 \Gamma_2 l(D)}{ (2 \pi)^2 \gamma x^{1/3}}.
\end{equation}
The term inside the parentheses can be made smaller than $\frac{(2 \pi)^2}{\Gamma_1 \Gamma_2 l(D)}$ taking $x \gg 0$ and $L' \gg 0$ (remember $L'<\theta$), so that \eqref{bound tail 11} reduces to
\begin{equation*}
|\tilde{A}(\nu, \theta)| < e^{- 2 \theta} \frac{1}{\gamma x^{1/3}},
\end{equation*}
which implies \eqref{tail} since $-2 \theta < -\theta - \nu$.

The complementary case $\nu > \theta$ can be studied analogously, deforming the contour $D$, instead of $C$, symmetrically with respect to the critical point $\varsigma$. 
\end{proof}

Up to this point we estimated the kernel $\tilde{A}$ in a region where both $\theta$ and $\nu$ are bounded from below. When this is not the case the saddle point method cannot be applied any longer as the contribution to the integral \eqref{rescaled discrete Airy} of the term
$$
\frac{z^{\gamma x^{1/3}\theta}}{w^{\gamma x^{1/3} \nu}} 
$$
is no more negligible. In the following Lemma we show how to control the rear tails of $\tilde{A}$.
\begin{Lemma} \label{bounded discrete Kernel}
The kernel $\tilde{A}$ defines a trace class operator on $l^2(\dot{\mathbb{Z}})$ and it  satisfies, for each $x$, the bound
\begin{equation} \label{tilde tilde component}
|\tilde{A}(\nu, \theta)| \leq 1.
\end{equation}
\end{Lemma}
\begin{proof}
We have to provide some robust estimate only in the cases
\begin{equation*}
\nu \leq -L, \ \theta > -L \qquad \text{ and } \qquad \nu > -L,\ \theta \leq -L. 
\end{equation*}
The reason of this is that, when both $\nu$ and $\theta$ are bigger than $-L$ we can use results of Lemma \ref{convergence on moderately large sets} and Lemma \ref{decay tails}, whereas, if they are both smaller than $-L$, this lemma is essentially a change of variable of Proposition \ref{trace class}.

In both cases we provide arguments analogous to those used in Appendix \ref{appendix  bounds}, readapting them to a kernel expressed in a double integral form instead of as a sum of rank one operators.

To start assume that $\nu \leq -L, \ \theta > -L$. We have
\begin{equation*}
\tilde{A}(\nu, \theta) = \frac{\tau(n_\nu)}{\varsigma^{ \theta \gamma x^{\frac{1}{3}} }} \frac{1}{(2 \pi \mathrm{i})^2} \int_D \frac{dw}{w} \int_C dz \frac{z^{\theta \gamma x^{1/3}} }{w^{\nu \gamma x^{1/3}}} \frac{ e^{ x g(z) }}{ e^{ x g(w) } } \frac{(q \mathpzc{d}/w;q)_\infty }{(q \mathpzc{d}/z;q)_\infty } \frac{1}{z-w} ,
\end{equation*}
so that, the $z$ integral is easily evaluated through the saddle point techniques we already exploited above and the $w$ integral can be estimated using the residue theorem. We can then write
\begin{equation*}
|\tilde{ A}(\nu, \theta)| \leq \Gamma_x \tau(n_\nu) \max_{k=2, \dots ,x}\{ (\xi_k s_k)^{\nu \gamma x^{1/3}} \},
\end{equation*}
where $\Gamma_x$ is a constant with respect to $\nu$ depending on $x$ (and in fact possibly diverging in $x$). Definition \eqref{tau} of $\tau$, at this point, guarantees that $|\tilde{A}(\nu , \theta)|$
becomes geometrically small for $\nu \ll 0$.

Let's now set $\nu > -L,\ \theta \leq -L$. In this case the rescaled kernel becomes
\begin{equation} \label{tilde tilde 1}
\tilde{ A}(\nu, \theta) = \frac{ \varsigma^{ \nu \gamma x^{1/3 } } }{ \tau(m_\theta) } \frac{1}{(2 \pi \mathrm{i} )^2} \int_D \frac{dw}{w} \int_C dz \frac{z^{\theta \gamma x^{1/3}} }{w^{\nu \gamma x^{1/3}}} \frac{e^{ x g(z) }}{ e^{ x g(w) } } \frac{(q \mathpzc{d}/w;q)_\infty }{(q \mathpzc{d}/z;q)_\infty } \frac{1}{z-w}.
\end{equation}
We can move the integration contour $C$ to encircle $D$, by noticing that this does not create any additional contribution to the integral. For example, we can take $C$ to be a circle of radius $R$ slightly bigger than $\max_{k=2, \dots, x}\{ (\xi_k s_k) \}$. Estimating the $w$ integral in \eqref{tilde tilde 1} with a saddle point method we obtain
\begin{equation*}
|\tilde{ A}(\nu, \theta)| \leq \Gamma'_x \frac{1}{\tau(m_\theta)} R^{\theta \gamma x^{1/3}},
\end{equation*}
which is again a geometrically small quantity for fixed $x$ and $\theta \ll 0$. Also in this case the quantity $\Gamma'_x$ is a constant with respect to $\theta$ depending on $x$. This proves that $\tilde{A}$ is a trace class operator.

Following the same strategy used in the proof of Lemma \ref{lemma invertible} we can also show that
\begin{equation*}
\tilde{ A }^2= \tilde{ A },
\end{equation*}
which implies that the operator norm $\Vert \tilde{ A } \Vert = 1$ and hence \eqref{tilde tilde component}.
\end{proof}
We can now start the evaluation of the Fredholm determinant of the kernel $\tilde{f}\tilde{A}$.

\begin{Lemma} \label{lemma determinant rear tail}
Then there exist constants $c,x^*$ such that, for each $x>x^*$, we have
\begin{equation} \label{determinant rear tail}
\left| \det \left( \mathbf{1} - \tilde{f} \tilde{A }\right)_{l^2( \dot{\mathbb{Z}} )} - \det \left( \mathbf{1} - \tilde{f} \tilde{A }\right)_{l^2( \dot{\mathbb{Z}} \cap [-L,\infty ))} \right| < e^{-c L \gamma x^{1/3}}.
\end{equation}
\end{Lemma}
\begin{proof}
First we set constants $c_1, c_2$ such that the estimate
\begin{equation*}
\left| \sqrt[]{\tilde{f}(\nu) \tilde{f}(\theta)  } \tilde{A}(\nu, \theta) \right| < \begin{cases} 
\frac{1}{\gamma x^{1/3}} c_1 e^{-\theta - \nu}, &\qquad \text{if } \theta, \nu \in [-L, \infty),\\
\frac{1}{\gamma x^{1/3}} e^{c_2 (\min(\theta, -L) + \min(\nu, -L)) \gamma x^{1/3} }, &\qquad \text{else},
\end{cases}
\end{equation*}
holds, for each $x$ sufficiently large. This is always possible as a result of Lemmas \ref{convergence on moderately large sets}, \ref{decay tails}, \ref{bounded discrete Kernel} and from the fact that $\tilde{f}(\nu)$, given in \eqref{f tilde}, decays exponentially in $x^{1/3}$ when $\nu \ll 0$. In particular we can easily deduce the additional bound
$$
\left| \sqrt[]{\tilde{f}(\nu) \tilde{f}(\theta)  } \tilde{A}(\nu, \theta) \right| < c_3 \frac{1}{\gamma x^{1/3}} e^{-|\theta|-|\nu|},
$$
true for any $\theta, \nu$, for some constant $c_3$. We have
\begin{equation}\label{rear tail lhs}
\begin{split}
\text{lhs of \eqref{determinant rear tail}} &=
\left| \sum_{k \geq 1} \frac{(-1)^k}{k!} \sumDot_{(\nu_1, \dots, \nu_k) \notin [-L, \infty)^k} \det_{i,j=1}^k \left( \sqrt[]{\tilde{f}(\nu_i) \tilde{f}(\nu_j)  } \tilde{A}(\nu_i, \nu_j) \right) \right|\\
& \leq \sum_{k \geq 1} \frac{(-1)^k}{(k-1)!} \sumDot_{\nu_1 \leq -L} \sumDot_{ \nu_2, \dots, \nu_k } \left| \det_{i,j=1}^k \left( \sqrt[]{\tilde{f}(\nu_i) \tilde{f}(\nu_j)  } \tilde{A}(\nu_i, \nu_j) \right) \right|.
\end{split}
\end{equation}
Thanks to the Hadamard's inequality we can estimate the determinantal term in the sum as
\begin{equation*}
\begin{split}
\left| \det_{i,j=1}^k \left( \sqrt[]{\tilde{f}(\nu_i) \tilde{f}(\nu_j)  } \tilde{A}(\nu_i, \nu_j) \right) \right| &\leq \prod_{i=1}^k\left( \sum_{j=1}^k \tilde{f}(\nu_i) \tilde{f}(\nu_j) \left| \tilde{A}(\nu_i, \nu_j) \right|^2 \right)^{1/2} \\
&\leq \frac{k^{k/2}}{(\gamma x^{1/3})^k}e^{c_2 \nu_1 \gamma x^{1/3}} c_3^{k-1} \prod_{j=2}^k e^{-|\nu_j|},
\end{split}
\end{equation*}
so that, using this bound in \eqref{rear tail lhs}, we obtain our result.
\end{proof}

\begin{Lemma} \label{lemma determinant front tail is small}
Take constants $L,\delta$ such that $-L<r$ and $\delta \in (0, 1/3)$. Then there exist constants $C,x^*$ such that, for each $x>x^*$, we have
\begin{equation} \label{determinant front tail is small}
\left| \det \left( \mathbf{1} - \tilde{f} \tilde{A} \right)_{l^2( \dot{\mathbb{Z}} \cap [-L,\infty ))} - \det \left( \mathbf{1} - \tilde{f} \tilde{ A }\right)_{l^2( \dot{\mathbb{Z}} \cap [-L, x^{\delta/3} ))} \right| < C e^{ - x^{\delta/3}}.
\end{equation}
\end{Lemma}
\begin{proof}
The proof of \eqref{determinant front tail is small} makes use of the exponential bound \eqref{tail} choosing $L'=x^{\delta/3}$ and it is similar to that of Lemma \ref{lemma determinant rear tail}, therefore we omit it.
\end{proof}

\begin{Remark} \label{remark choice sets}
The statement of Proposition \ref{prop discrete Airy} not only tells us that 
$$
\det (\mathbf{1} - \tilde{f} \tilde{A}) \xrightarrow[x\to \infty]{} F_2(r),
$$
but also it gives us an estimate of the error depending on $x$ and this will be essential in the proof of Theorem \ref{theorem: baik rains limit H}. To measure such error term, namely $\frac{1}{\gamma x^{1/3}} R_x^{(1)}$ in \eqref{Tracy Widom limit}, we approximated the kernel $\tilde{f} \tilde{A}$ on $l^2(\dot{\mathbb{Z}}^2)$ with its truncated version defined only on $(\dot{\mathbb{Z}} \cap [-L , x^{\delta/3}])^2$. The choice of the supremum of the segment $[-L, x^{\delta / 3}]$ is actually very relevant and possibly differentiate our analysis of the Fredholm determinant from that of earlier works, such as \cite{BFS2007}. Had we considered the convergence of $\tilde{f}\tilde{A}$ only on compact sets like $[-L,L']$, with $L'$ being some finite constant, we would have ended up, in Lemma \ref{lemma determinant front tail is small} (replacing every $x^{\delta / 3}$ with $L'$), with a bound like
\begin{equation} \label{wrong estimate}
\text{\emph{lhs of }}\eqref{determinant front tail is small} < C e^{-L'}.
\end{equation}
This clearly would have not been enough for our purposes, as the right hand side of \eqref{wrong estimate} has no dependence on $x$ and in particular does not decay when $x$ becomes infinite.
\end{Remark}

\begin{proof}[Proof of Proposition \ref{prop discrete Airy}]
To prove this result we first use Lemma \ref{lemma determinant rear tail} and Lemma \ref{lemma determinant front tail is small} to restrict our attention to the Fredholm determinant of $\tilde{f} \tilde{A}$ in $l^2(\dot{\mathbb{Z}} \cap[-L,x^{\delta/3}] )$. The error we make while considering this restriction is exponentially small in $x$ and hence it is irrelevant when it comes to a decomposition like \eqref{Tracy Widom limit}. Using results of Lemma \ref{convergence on moderately large sets} we have
\begin{gather}
\tilde{A} (\nu, \theta) = \frac{1}{\gamma x^{1/3}}K_{\text{Airy}}(\nu, \theta) + \frac{1}{(\gamma x^{1/3})^2} Q(\nu, \theta) + \mathcal{O}(x^{2\delta/3-1}), \\
\tilde{f}(\nu) = \mathbbm{1}_{[r, \infty)}(\nu) + \Delta_r(\nu),
\end{gather}
where the term $\Delta_r$ is simply expressed as
\begin{equation*}
\Delta_r(\nu) = \begin{cases}
\frac{1}{1 + q^{(\nu - r)\gamma x^{1/3}}}, \qquad & \text{if } \nu < r,\\
\frac{ - q^{(\nu - r) \gamma x^{1/3} }}{1+ q^{(\nu - r)\gamma x^{1/3}}}, \qquad & \text{if } \nu \geq r.
\end{cases}
\end{equation*}
In order to give a sharp estimate of the difference
\begin{equation*}
\det \left( \mathbf{1} - \tilde{f} \tilde{A} \right)_{l^2(\dot{\mathbb{Z}} \cap[-L,x^{\delta/3}] )} - \det(\mathbf{1} - \mathbbm{1}_{[r, \infty)}K_{\text{Airy}})_{\mathcal{L}^2[-L, x^{\delta/3} )},
\end{equation*}
we plan to use the determinantal identity
\begin{equation} \label{determinant sum matrices identity}
\det(B^{(1)} + \cdots + B^{(N)}) = \sum_{ \substack{ \cup_{i=1}^N I_i = \{ 1, \dots, k \} \\ I_i \cap I_j = \emptyset \text{ if }i\neq j } } \det(B^{(I_1, \dots, I_N)}),
\end{equation}
holding for generic $k \times k$ matrices $B^{(1)}, \dots, B^{(N)}$. The matrix $B^{(I_1, \dots, I_N)}$ appearing in the right hand side of \eqref{determinant sum matrices identity} is constructed taking columns of the $B^{(j)}$'s according to the choice of the partition $I_1, \dots, I_N$ of the set $\{1, \dots, k\}$ as
\begin{equation*}
B^{(I_1, \dots, I_N)}_{i,j} = B^{(l)}_{i,j} \qquad \text{if } i\in I_l.
\end{equation*}
We set
\begin{equation*} 
	\begin{aligned}
&B^{(1)}_{i,j} = \mathbbm{1}_{[r, \infty)} (\nu_i) K_{\text{Airy}}(\nu_i, \nu_j), \qquad & B^{(2)}_{i,j} = \frac{1}{\gamma x^{1/3}} \mathbbm{1}_{[r, \infty)}(\nu_i) Q(\nu_i , \nu_j),\\
&B^{(3)}_{i,j} = \Delta_r(\nu_i) K_{\text{Airy}}(\nu_i, \nu_j), &B^{(4)}_{i,j} = \frac{1}{\gamma x^{1/3}} \Delta_r(\nu_i) Q(\nu_i, \nu_j),\\
& B^{(5)}_{i,j} = \mathcal{O} (x^{2(\delta-1)/3})
	\end{aligned}
\end{equation*}
and we write the Fredholm determinant of $\tilde{f} \tilde{ A}$ as
\begin{equation*}
\det \left( \mathbf{1} - \tilde{f} \tilde{A} \right)_{l^2(\dot{\mathbb{Z}} \cap[-L,x^{\delta/3}] )} = 1 + \sum_{k \geq 1} \frac{(-1)^k}{k!} \int_{-L}^{x^{\delta/3}} d\nu_1\cdots \int_{-L}^{x^{\delta/3}} d\nu_k \det_{i,j=1}^k( B^{(1)}_{i,j} + \cdots + B^{(5)}_{i,j}) + \mathcal{O}(1/x^{1/3}),
\end{equation*}
where the remainder term comes from approximating the sum $\sumDot$ with the integral sign and clearly does not play a role in the proof of properties \eqref{remander bounded},\eqref{remander continuous}. 
Using  \eqref{determinant sum matrices identity} we have
\begin{equation} \label{determinant sum 5 matrices}
\det(B^{(1)} + \cdots + B^{(5)}) = \sum_{ \substack{ \cup_{i=1}^5 I_i = \{ 1, \dots, k \} \\ I_i \cap I_j = \emptyset \text{ if }i\neq j } } \det(B^{(I_1, I_2, I_3, I_4, I_5)})
\end{equation}
and we aim to estimate which term of \eqref{determinant sum 5 matrices}, after an integration over $[-L, x^{\delta/3}]^k$, is of order 1 or of order $x^{-1/3}$. The Hadamard inequality provides for the bound
\begin{equation*}
\left| \det(B^{(I_1, I_2, I_3, I_4, I_5)}) \right| \leq \prod_{l=1}^5 \prod_{i \in I_l} \left( \sum_{j=1}^k |B^{(l)}_{i,j} |^2 \right)^{1/2},
\end{equation*}
while \eqref{exponential bound kernel} allows us to write
\begin{equation*}
	\begin{aligned}
&B^{(1)}_{i,j} < c_1 e^{-c_2 \nu_i}, \qquad & B^{(2)}_{i,j} < \frac{1}{\gamma x^{1/3}} c_1 e^{- c_2 \nu_i}, \qquad
&B^{(3)}_{i,j} < c_1 q^{|\nu_i - r|\gamma x^{1/3}},\\ &B^{(4)}_{i,j} < \frac{c_1 q^{|\nu_i - r|\gamma x^{1/3}}}{\gamma x^{1/3}}, \qquad
& B^{(5)}_{i,j} = \mathcal{O} (x^{2(\delta-1)/3}).
	\end{aligned}
\end{equation*}
Integrating the generic term of \eqref{determinant sum 5 matrices} we obtain
\begin{equation} \label{bounds matrices B}
\int_{-L}^{x^{\delta/3}} d\nu_1\cdots \int_{-L}^{x^{\delta/3}} d\nu_k \left| \det(B^{(I_1, I_2, I_3, I_4, I_5)}) \right|=k^{k/2} \mathcal{O}\left(x^{-|I_2|/3 - |I_3|/3 - 2|I_4|/3 + |I_5|(\delta - 2/3)} \right),
\end{equation}
where the exponents $-|I_3|/3$ and $-2|I_4|/3$ appear due to the fact that the function $q^{|\nu -r|\gamma x^{1/3}}$ is exponentially small in $x$ outside of a neighborhood of size $x^{-1/3}$ of $r$. We can now easily deduce the expansion
\begin{equation*}
\begin{split}
\det \left( \mathbf{1} - \tilde{f} \tilde{A} \right)_{l^2( \dot{\mathbb{Z}} \cap [-L, x^{\delta/3} ])} &=  
\det( \mathbf{1} - \mathbbm{1}_{[r, \infty)} K_{\text{Airy}} )_{\mathcal{L}^2([-L,x^{\delta/3}])}\\
&+ \sum_{k \geq 1} \frac{(-1)^k}{k!} \int_{-L}^{x^{\delta/3}} d\nu_1\cdots \int_{-L}^{x^{\delta/3}} d\nu_k \sum_{ \substack{ I_1 \cup I_2 \cup I_3 =\{1, \dots ,k\} \\ I_i\cap I_j = \emptyset \text{ if } i\neq j \\ |I_1| = k-1 } }\det_{i,j=1}^k( B^{(I_1,I_2,I_3,\emptyset,\emptyset)}_{i,j} )\\
&+ \smallO(x^{-1/3}).
\end{split}
\end{equation*}
This last equality is equivalent to \eqref{Tracy Widom limit}, where the term $(\gamma x^{1/3})^{-1}R_x^{(1)}$ is essentially given by the second addend in the right hand side. Using \eqref{bounds matrices B} we can also prove the boundedness property \eqref{remander bounded}, due to the generality of $r$.

Now that we have an expression for the remainder term $R_x^{(1)}$, we can also show its "continuity" with respect to the parameter $r$ expressed in \eqref{remander continuous}. To do this we take $\epsilon$ small enough so that $r - \epsilon \gg -L$ and we estimate of the difference
\begin{equation} \label{difference remainder}
\int_{-L}^{x^{\delta/3}} d\nu_1\cdots \int_{-L}^{x^{\delta/3}} d\nu_k  \det_{i,j=1}^k ( B^{(I_1,I_2,I_3,\emptyset,\emptyset)}_{i,j}(r^*) ) - \det_{i,j=1}^k ( B^{(I_1,I_2,I_3,\emptyset,\emptyset)}_{i,j}( r^* - \frac{1}{\gamma x^{1/3}})  ),
\end{equation}
for $r^* \in [r - \epsilon, r]$, in case exactly one between $|I_2|$ and $|I_3|$ is equal to 1 and the other one is 0. We point out that, with a little abuse of notation, we highlighted the dependence on $r$ of the matrices $B^{(I_1, I_2,I_3,\emptyset,\emptyset)}$.

Defining matrices $C^{(1)}, C^{(2)}$ as
\begin{equation*}
C^{(1)} = B^{(I_1,I_2,I_3,\emptyset,\emptyset)}(r ), \qquad C^{(2)}= B^{(I_1,I_2,I_3,\emptyset,\emptyset)}_{i,j}(r +\frac{1}{\gamma x^{1/3}}) - B^{(I_1,I_2,I_3)}_{i,j}(r),
\end{equation*}
we can write the integrand of \eqref{difference remainder} as
\begin{equation*}
\det(C^{(1)}) - \det(C^{(1)} -C^{(2)} ) = \sum_{ \substack{ J_1 \cup J_2 =\{1, \dots, k\} \\ J_1 \cap J_2 = \emptyset \\ J_2 \neq \emptyset } } \det ( C^{(J_1, J_2)} ).
\end{equation*}
Using again the Hadamard inequality we can estimate each term $\det(C^{(J_1, J_2)})$ and integrating over $\nu_1, \dots \nu_k$ we find that, due to the presence of columns of $C^{(2)}$, the generic expression \eqref{difference remainder} converges to 0 as $x$ grows and this is uniformly for $r^* \in [r - \epsilon, r]$.
\end{proof}

Let's now see what is the asymptotic behavior of remaining terms of \eqref{V}.

\begin{Prop} \label{proposition limit V 2}
Recall choices \eqref{scaling t}, \eqref{scaling zeta}. Then we have
\begin{equation}\label{asymptotic dummy limit}
\begin{split}
\frac{1}{\gamma x^{1/3}} \left( t a_0(\mathpzc{d}) - \polygamma_0(1/\zeta) - 2 \polygamma_0(q) + \polygamma_0(q \zeta) - x h_0(\mathpzc{d} ) \right)
= r-\varpi^2 + \frac{1}{\gamma x^{1/3}}R^{(2)}_x(r).
\end{split}
\end{equation}
The error term $R_x^{(2)}$ satisfies the following properties
\begin{enumerate}
\item for each $r^* \in \mathbb{R}$, there exists $M_{r^*}>0$ such that, for all $x$,
\begin{equation} \label{R^2 bounded}
\left | R^{(2)}_x (r^*) \right | < M_{r^*};
\end{equation}
\item there exist $\epsilon >0$, such that, foe all $r^* \in [r- \epsilon, r]$ we have
\begin{equation} \label{continuity R^2}
\lim_{x \to \infty} \left( R^{(2)}_x (r^*) - R^{(2)}_x (r^* - \frac{1}{\gamma x^{1/3}}) \right) = 0.
\end{equation}
uniformly.
\end{enumerate}

\end{Prop}
\begin{proof}
First we see that the term
$$
\polygamma_0(1/\zeta) + 2 \polygamma_0 (q)=\sum_{n \geq 0} \left( \frac{q^n/\zeta}{1-q^n/\zeta} + 2 \frac{q^{n+1}}{1-q^{n+1}} \right) 
$$
plays no role in the limit as it is a bounded quantity in $x$ for each fixed $r$.

Less trivial is to calculate the limiting form of $\polygamma_0( q \zeta)$, which is a summation like
$$
\sum_{n \geq 0} \frac{q^{n+1-X}}{1+q^{n+1-X}},
$$
for $X$ being large. The kicker here is understanding that the main contribution to the sum is given by terms where $m$ runs between $0$ and $2\ceil{X}$ \footnote{here $\ceil{}$ is the ceiling function}. Coupling the $(k-1)$th and the $(2 \ceil{X} - k -1)$th addends and using the simple inequality
\begin{equation*}
1 - \frac{1 - q^2}{1 + q^2 + q^{k - X} + q^{X-k+2} } \leq \frac{1}{1+q^{X-k}} + \frac{1}{1+q^{X-2\ceil{X}+k}} \leq 1,
\end{equation*}
we see that
\begin{equation*}
\sum_{n \geq  0} \frac{1}{1+q^{X-n-1} }= X + \mathcal{O}(1).
\end{equation*}
We are interested in the case when $X=\eta x - \gamma x^{1/3} r$, so that, plugging this result into \eqref{asymptotic dummy limit} we are left to calculate
\begin{equation*}
\begin{split}
&\lim_{x \rightarrow \infty} \frac{1}{\gamma x^{1/3}} \left(  \kappa x a_0(\mathpzc{d}) - x h_0(\mathpzc{d}) - \eta x + r \gamma x^{1/3}  \right)\\
&= r + \lim_{x \rightarrow \infty}\frac{x^{2/3}}{\gamma} \mathpzc{d} g'( \mathpzc{d} ),
\end{split}
\end{equation*}
which gives \eqref{asymptotic dummy limit} and \eqref{R^2 bounded} after expanding $g'$ around its critical point $\varsigma$ as
\begin{equation*}
g'(\mathpzc{d}) \approx \frac{1}{2}g'''(\varsigma)( \mathpzc{d} - \varsigma)^2=\frac{\varsigma^2}{2 \gamma^2}g'''(\varsigma) \frac{\varpi^2}{x^{2/3}} + \smallO(x^{-2/3}).
\end{equation*}
This procedure also proves the boundedness of the remainder $R^{(3)}_x$ due to the generality of $r$.

Result \eqref{continuity R^2} follows from expression \eqref{asymptotic dummy limit}. We have
\begin{equation*}
R^{(2)}_x(r^*) = -\gamma x^{1/3} r^* - \sum_{n \geq 0} \left( \frac{q^n/\zeta}{ 1 - q^n / \zeta } - \frac{q^{n+1} \zeta}{1 - q^{n+1} \zeta} \right) + (\text{terms independent of $r^*$}), 
\end{equation*}
where $\zeta = -q^{-\eta x + \gamma x^{1/3} r^*}$. In this way the difference $R^{(2)}_x(r^*) - R^{(2)}_x(r^*-1/(\gamma x^{1/3}))$ becomes
\begin{equation*}
-1 + \frac{q^{\eta x - \gamma x^{1/3} r^*}}{1 + q^{\eta x - \gamma x^{1/3} r^*}} + \frac{q^{-\eta x + \gamma x^{1/3} r^*}}{ 1 + q^{-\eta x + \gamma x^{1/3} r^*} },
\end{equation*}
which converges to zero exponentially as $x$ goes to infinity.
\end{proof}
Lastly we state the convergence result for terms $\tilde{\Phi}_x^{(i)},\tilde{\Phi}_x, \tilde{\Psi}_x^{(j)},\tilde{\Psi}_x$. The procedure closely traces what was done in Lemmas \ref{convergence on moderately large sets}, \ref{decay tails}.
\begin{Prop} \label{prop Upsilon}
We have
\begin{equation} \label{upsilon convergence}
\frac{ \mathpzc{d} }{\gamma x^{1/3}} \sumDot_{\nu \in \dot{\mathbb{Z}} } \sum_{\substack{ i,j =1,2 \\ (i,j)\neq (1,1)}} \tilde{f}(\nu) \tilde{\Phi}_x^{(i)}(\nu) \tilde{\Psi}_x^{(j)}(\nu) = \sum_{\substack{ i,j =1,2 \\ (i,j)\neq (1,1)}} \int_r^{\infty} \Upsilon_{-\varpi}^{(i)}(\nu) \Upsilon_{\varpi}^{(j)}(\nu) d \nu + \frac{1}{\gamma x^{1/3}}R^{(3)}_x(r),
\end{equation}
where functions $\Upsilon^{(1)}, \Upsilon^{(2)}$ are defined in \eqref{Upsilon12} and the error term $R_x^{(3)}$ satisfies the following properties
\begin{enumerate}
\item for each $r^* \in \mathbb{R}$ there exists $M_{r^*}>0$ such that, for all $x$,
\begin{equation} \label{R^3 bounded}
\left | R^{(3)}_x (r^*) \right | < M_{r^*};
\end{equation}
\item there exists $\epsilon > 0$, such that, for all $r^* \in [r - \epsilon,r]$ we have
\begin{equation} \label{R^3 continuous}
\lim_{x \to \infty} \left( R^{(3)}_x (r^* + \frac{1}{\gamma x^{1/3}}) - R^{(3)}_x (r^*) \right) = 0.
\end{equation}
uniformly.
\end{enumerate}
\end{Prop}
\begin{proof}
By making use of the saddle point method it is easy, at this stage, to obtain a convergence result as
\begin{equation}\label{upsilon convergence 1}
\mathpzc{d} \sum_{\substack{ i,j =1,2 \\ (i,j)\neq (1,1)}} \tilde{f}(\nu) \tilde{\Phi}_x^{(i)}(\nu) \tilde{\Psi}_x^{(j)}(\nu) \xrightarrow[x \rightarrow \infty]{} \sum_{\substack{ i,j =1,2 \\ (i,j)\neq (1,1)}} \mathbbm{1}_{[r, \infty)}(\nu) \Upsilon_{-\varpi}^{(i)}(\nu) \Upsilon_\varpi^{(j)}(\nu) 
\end{equation}
and to estimate the error term depending on $x$ and $r$.
This holds for $\nu$ in relatively large sets of the form $[-L, x^{\delta/3}]$ for some fixed $L>0$ and $\delta \in (0, 1/3)$. Also, assuming a suitably strong decay of tails of summands in the left hand side of \eqref{upsilon convergence 1}, this easily leads to an expansion of type \eqref{upsilon convergence}. 

Using suitable deformations of contours in the integral expressions of $\tilde{\Phi}_x^{(2)}, \tilde{\Psi}_x^{(2)}$, such as those seen in Lemma
\ref{decay tails}, one can also establish an exponential type decay for the front tail ($\nu \gg 0$) of \eqref{upsilon convergence 1}.

The exponential decay we have in the left hand side of \eqref{upsilon convergence 1}, when $\nu$ goes to $- \infty$ is slightly different from what seen previously and in particular, here we make use of the hypothesis $\frac{2 a }{1 + \sigma} < q^{-1} \mathpzc{d}$ stated in \eqref{bound R_a}. We evaluate separately each one of the three summands in the left hand side of \eqref{upsilon convergence}, when $(i,j)$ is either equal to $(1,2),(2,1)$ or $(2,2)$.

We start with the $(i,j)=(1,2)$ term. From expressions reported in Proposition \ref{proposition scaling}, we write
\begin{equation}\label{upsilon convergence 2}
\mathpzc{d} \tilde{f}(\nu) \tilde{\Phi}_x^{(1)}(\nu) \tilde{\Psi}_x^{(2)}(\nu) = \frac{1}{1 + q^{\gamma x^{1/3} (\nu - r) }} \int_{C_1} \frac{dz}{2 \pi \mathrm{i} z} \left( \frac{z}{\mathpzc{d}} \right)^{\gamma x^{1/3} \nu} e^{x(g(z) - g(\mathpzc{d}))} \frac{(q z / \mathpzc{d} ; q)_\infty}{(q;q)_\infty}.
\end{equation}
We take $C_1$ to be a circle of center in 0 and radius $\varsigma (1 - h/(\gamma x^{1/3}))$, where $h$ is chosen so that $|z|<\mathpzc{d}$ for all $z$ in $C_1$ (e.g. take $h >\varpi$). When $x$ is large enough, such $C_1$ is a steep descent contour for $\mathfrak{Re}\{g\}$, as proven in Proposition \ref{prop contour C}. This, along with the fact that $\varsigma$ is a double critical point for $g$ allows us to state the bound
\begin{equation} \label{upsilon convergence 3}
|\exp\{ x ( g(z) - g(\mathpzc{d}) ) \}| < \text{const} \qquad \text{for all } z \in C_1.
\end{equation}
Moreover, the choice of $C_1$ also allows us to write
\begin{equation}\label{upsilon convergence 4}
\begin{split}
\left| \frac{z}{\mathpzc{d}} \right|^{\gamma x^{1/3} \nu} &= \exp\left\{ \gamma x^{1/3} \nu \left( \log \left( 1 - \frac{h}{\gamma x^{1/3}} \right) - \log \left( 1 - \frac{w}{\gamma x^{1/3}} \right)  \right) + \smallO(1) \right\}\\
&\leq \exp \left\{ \nu ( \varpi - h ) + \smallO (1) \right\},
\end{split}
\end{equation}
having used the simple logarithmic inequality $\frac{y}{1+y} \leq \log(1+y) \leq y$, valid for all $y>-1$. Despite the right hand side of \eqref{upsilon convergence 4} is a quantity which diverges exponentially when $\nu \to -\infty$, its contribution is easily balanced by the term $1/(1 + q^{\gamma x^{1/3} (\nu - r) } )$ in \eqref{upsilon convergence 2}, which, for $\nu < -L$, decays as $q^{-(1-r/L)\gamma x^{1/3} \nu }$. Following \eqref{upsilon convergence 3}, \eqref{upsilon convergence 4} we come to the estimate 
\begin{equation*}
|\eqref{upsilon convergence 2}| < \text{const} \frac{l(C_1)}{2\pi} e^{(h-\varpi)\nu } q^{-(1-r/L)\gamma x^{1/3} \nu } 
\end{equation*}
where in the right hand side the constant term also includes a trivial bound for the factor $\frac{(q z / \mathpzc{d} ; q)_\infty}{(q;q)_\infty}$. This is enough to show that for $L$ large enough, we have
\begin{equation} \label{upsilon convergence 5}
|\eqref{upsilon convergence 2}| < c_1 e^{ c_2 \gamma x^{1/3} \nu}, \qquad \text{for all } \nu < -L,  
\end{equation}
where $c_1$ and $c_2$ are two suitably chosen positive constants.

We now want to establish a type of bound similar to \eqref{upsilon convergence 5} for the term $(i,j)=(2,1)$ of the left hand side of \eqref{upsilon convergence}. Again, from \eqref{Phi 2}, \eqref{Psi 1} we write
\begin{equation} \label{upsilon convergence 6}
\mathpzc{d} \tilde{f}(\nu) \tilde{\Phi}_x^{(2)}(\nu) \tilde{\Psi}_x^{(1)}(\nu) = \frac{1}{1 + q^{\gamma x^{1/3} (\nu - r)}} \int_{D_1} \frac{dw}{2 \pi \mathrm{i} w} \left( \frac{\mathpzc{d}}{w} \right)^{\gamma x^{1/3} \nu} e^{ x ( g(\mathpzc{d}) -g(w) ) }
 \frac{(q \mathpzc{d} / w ; q)_\infty}{(q w / \mathpzc{d} ; q)_\infty} \frac{\mathpzc{d}}{w - \mathpzc{d}}.
\end{equation}
As a contour $D_1$ we can simply take the contour $D$ described in Proposition \ref{prop contour D}. Since we can always deform the integration contour in a neighborhood of size $ x^{-1/3}$ of $\varsigma$, without loss of generality, we assume that $\mathpzc{d}$ lies strictly at the left of $D_1$. With this choice, we know that $D_1$ is a steep ascent contour for $\mathfrak{Re}\{g\}$ and this, along with the fact that $\varsigma$ is  double critical point for $g$ implies the bound
\begin{equation*}
|\exp\{ x( g(\mathpzc{d}) - g(w) )\}| < \text{const} \qquad \text{for all } w \in D_1.
\end{equation*}
Another consequence of the choice of contour $D_1$ is that
\begin{equation*}
\max_{w \in D_1} |w| \leq \frac{2 a}{1+\sigma}, 
\end{equation*}
as reported in \eqref{inequality D}. This immediately gives us the estimate
\begin{equation*}
\left| \frac{\mathpzc{d}}{w} \right|^{\gamma x^{1/3} \nu} \leq \left| \frac{(1 +\sigma ) \mathpzc{d} }{2 a} \right|^{\gamma x^{1/3} \nu} \qquad \text{for all } w \in D_1,
\end{equation*}
since in this case $\nu$ is taken to be negative. In expression \eqref{upsilon convergence 6}, the contribution of the factor $\mathpzc{d}/(w-\mathpzc{d})$ is bounded, in absolute value, by a quantity of order $x^{1/3}$ and therefore we come to write
\begin{equation} \label{upsilon convergence 7}
|\eqref{upsilon convergence 6}| < \text{const} \frac{l(D_1)}{2 \pi} x^{1/3} \left| \frac{(1 + \sigma)}{2 a} q^{-(1 - r/L)}\mathpzc{d} \right|^{\gamma x^{1/3} \nu} .
\end{equation}
When $L$ is large enough, the assumption $2a/(1+\sigma) < q^{-1} \mathpzc{d}$ of \eqref{bound R_a} guarantees that the right hand side of \eqref{upsilon convergence 7} is bounded by an exponential function in $\gamma x^{1/3} \nu$ whenever $\nu<-L$ and this concludes our analysis of the rear tail of the term $(i,j)=(2,1)$.

To obtain the same type of result also for the case when $(i,j)=(2,2)$ one can reproduce, with minor adjustments, the same argument we used for $(i,j)=(2,1)$ and therefore we omit details on this part.

We have, at this point proved a bound for the summands in expression \eqref{upsilon convergence} of the form
\begin{equation*}
\left| \sum_{\substack{ i,j=1,2 \\ (i,j) \neq (1,1) }} \tilde{f}(\nu) \tilde{\Phi}_x^{(i)}(\nu) \tilde{\Psi}_x^{(j)}(\nu)  \right| < c_1 e^{c_2 \gamma x^{1/3} \nu} \qquad \text{for all }\nu<-L,
\end{equation*}
for suitably chosen positive constants $c_1, c_2$. This concludes our argument.
\end{proof}

\begin{Prop} \label{prop Upsilon 2}
We have
\begin{equation} \label{convergence third}
\frac{ \mathpzc{d} }{\gamma x^{1/3}} \sumDot_{\nu \in \dot{\mathbb{Z}}}(\tilde{f} \tilde{A} \tilde{\varrho} \tilde{f} \tilde{\Phi}_x)(\nu) \tilde{\Psi}_x(\nu) = \int_r^{\infty} d \nu \Upsilon_\varpi (\nu) \int_r^\infty d \lambda_1 \int_r^{\infty} d\lambda_2 \varrho_{\emph{Airy};r}(\nu, \lambda_1) K_{\emph{Airy}}(\lambda_1, \lambda_2) \Upsilon_{-\varpi}(\lambda_2) + \frac{1}{\gamma x^{1/3}}R^{(4)}_x(r),
\end{equation}
where the integral kernel $\varrho_{\emph{Airy};r}$ and the function $\Upsilon_\varpi$ were defined in \eqref{traci widom rho}, \eqref{Upsilon} and the error term $R_x^{(4)}$ satisfies the following properties
\begin{enumerate}
\item for each $r^*\in \mathbb{R}$ there exists $M_{r^*}>0$ such that, for all $x$
\begin{equation} \label{R4 bounded}
\left | R^{(4)}_x (r^*) \right | < M_{r^*};
\end{equation}
\item there exists $\epsilon>0$ such that, for all $r^* \in [r-\epsilon, r]$, we have
\begin{equation} \label{R4 continuous}
\lim_{x \to \infty} \left( R^{(4)}_x (r^*) - R^{(4)}_x ( r^* -\frac{1}{\gamma x^{1/3}} ) \right) = 0.
\end{equation}
uniformly.
\end{enumerate} 
\end{Prop}
\begin{proof}
First we expand the expression in the left hand side of \eqref{convergence third} as 
\begin{equation} \label{convergence third summation}
\frac{ \mathpzc{d} }{\gamma x^{1/3}} \sumDot_{\nu , \lambda_1, \lambda_2\in \dot{\mathbb{Z} } } \tilde{f}(\nu) \tilde{A}(\nu,\lambda_1) \tilde{\varrho}(\lambda_1,\lambda_2) \tilde{f}(\lambda_2) \tilde{\Phi}_x(\lambda_2) \tilde{\Psi}_x(\nu).
\end{equation}
We can split the summation \eqref{convergence third summation} as
\begin{equation*}
\frac{ \mathpzc{d} }{\gamma x^{1/3}} \left( \sumDot_{(\nu , \lambda_1, \lambda_2)\in [-L, x^{\delta/3}]^3 } + \sumDot_{\nu , \lambda_1, \lambda_2\notin [-L, x^{\delta/3}]^3} \right) \left( \tilde{f}(\nu) \tilde{A}(\nu , \lambda_1) \tilde{\varrho}(\lambda_1, \lambda_2) \tilde{f}(\lambda_2) \tilde{\Phi}_x(\lambda_2) \tilde{\Psi}_x(\nu) \right).
\end{equation*}

Using estimates already encountered in the proofs of Proposition \ref{prop discrete Airy}, \ref{prop Upsilon}, we know that the contribution of summation where indices do not belong to $[-L,x^{\delta/3}]^3$ is exponentially small in some power of $x$. On the other hand, when all $\nu, \lambda_1, \lambda_2$ belong to $[-L,x^{\delta/3}]$, we can safely employ the saddle point method to estimate the summand terms in and obtain their expansion in power of $x^{-1/3}$, as done in \eqref{formula convergence on bounded sets} for $\tilde{A}$. This would ultimately lead to the convergence result \eqref{convergence third} and to a verification of properties \eqref{R4 bounded}, \eqref{R4 continuous} for the remainder term. The procedure is analogous to what explained throughout the rest of the section and therefore we do not describe its details any further. 
\end{proof}
\begin{proof}[Proof of Theorem \ref{theorem: baik rains limit H}]
Using a rather elementary argument, detailed in Secion 5 of \cite{FerrariVeto2015TWLimit}, it is possible to show that proving \eqref{eq: baik rains limit intro} is equivalent to showing that
\begin{equation*}
\lim_{x \rightarrow \infty} \mathbb{E}_{\text{HS}(\mathpzc{d},\mathpzc{d})} \left ( \frac{1}{(-q^{ \mathcal{H}(x, \kappa x) - \eta x + r \gamma x^{1/3}} ;q)_\infty} \right ) = F_\varpi (r).
\end{equation*}
To do so we use formula \eqref{stationary q laplace 2} to express the $q$-Laplace transform on the left hand side. We want to evaluate
\begin{equation*} \label{main theorem 1}
\lim_{x \to \infty} \frac{1}{(q;q)_\infty} \sum_{k\geq 0} \frac{(-1)^k q^{\binom{k+1}{2}}}{(q;q)_k} \left( V_x(\zeta q^{-k})- V_x(\zeta q^{-k -1})\right), 
\end{equation*}
and to do so we aim to bring the limit inside the summation symbol. We start by fixing a small number $\epsilon$ and we split the summation in \eqref{main theorem 1} into two different contributions. One comes from the sum over $k$ ranging in the region $[0, \epsilon \gamma x^{1/3}]$ and the other is given by $k$ in $(\epsilon \gamma x^{1/3}, \infty)$. For each of these terms we can use different estimates. 

We start with the latter, that is we take $k > \epsilon \gamma x^{1/3}$. A general inequality that can be deduced from the definition of $V_x=\lim_{v \to \mathpzc{d} } V_{x;v,\mathpzc{d}} $ and from Theorem \ref{theorem: q laplace shifted introduction}, in case $\zeta$ is a negative number, is
\begin{equation*}
\begin{split}
V_x(\zeta q^{-k}) &= \lim_{v \uparrow \mathpzc{d} } \frac{1}{1- v/ \mathpzc{d} } \mathbb{E}_{\text{HS}(v, \mathpzc{d})\otimes \mathsf{m}} \left ( \frac{1}{(\zeta q^{-k} q^{\mathcal{H}-\mathsf{m} };q)_\infty} \right )\\
&\leq \lim_{v \uparrow \mathpzc{d} } \frac{1}{1- v/ \mathpzc{d} } \mathbb{E}_{\text{HS}(v, \mathpzc{d})\otimes \mathsf{m}} \left ( \frac{1}{(\zeta q^{-k^*} q^{\mathcal{H} -\mathsf{m}};q)_\infty} \right ) = V_x(\zeta q^{-k^*}),
\end{split}
\end{equation*}
which holds for every $k^* < k$. By taking $k^*= \epsilon \gamma x^{1/3}$ and $r^*=r-\epsilon$ we obtain the estimate
\begin{equation}\label{bounded convergence 2}
\begin{split}
q^{\binom{k+1}{2}}  \left(  V_x(\zeta q^{-k})- V_x(\zeta q^{-k -1}) \right ) &\leq 2 q^{\binom{k+1}{2}} V_x(\zeta q^{(r^*-r)\gamma x^{1/3}})\\
&=2q^{\binom{k+1}{2}} \left( \gamma x^{1/3} \chi_\varpi(r^*) + \sum_{i=1}^4 S^{(i)}(r^*) R_x^{(i)}(r^*)  + \mathcal{O}(x^{-1/3}) \right).
\end{split}
\end{equation}
In the right hand side of \eqref{bounded convergence 2} we used results of Propositions \ref{prop discrete Airy}, \ref{proposition limit V 2}, \ref{prop Upsilon}, \ref{prop Upsilon 2} to provide the approximate expression of $V_x$. Function $\chi_\varpi$ was defined in \eqref{baik rains chi} and terms $S^{(i)}$'s are explicit, bounded functions which for convenience we do not report explicitly. We can therefore write
\begin{equation} \label{main theorem 2}
\frac{1}{(q;q)_\infty} \sum_{k > \epsilon \gamma x^{1/3}} \frac{(-1)^k q^{\binom{k+1}{2}}}{(q;q)_k} \left( V_x(\zeta q^{-k})- V_x(\zeta q^{-k -1})\right) = \smallO (e^{-c x^{2/3}}), 
\end{equation}
for some positive constant $c$, since, from \eqref{bounded convergence 2} we see that the right hand side is a quantity exponentially small in $x^{2/3}$, due to the presence of the term $q^{\binom{k+1}{2}}$.

We now consider the contribution of the summation in \eqref{main theorem 1}, when the index $k$ is less than $\epsilon \gamma x^{1/3}$. Once again, using results of Propositions \ref{prop discrete Airy}, \ref{proposition limit V 2}, \ref{prop Upsilon}, \ref{prop Upsilon 2} we have
\begin{equation}\label{bounded convergence 1}
\begin{split}
 V_x(\zeta q^{-k})- V_x(\zeta q^{-k -1})&=\gamma x^{1/3} \left ( \chi_\varpi (r - \frac{k}{\gamma x^{1/3}}) - \chi_\varpi (r - \frac{k+1}{\gamma x^{1/3}}) \right)\\
 &+\sum_{i=1}^4 S^{(i)} (r - \frac{k}{\gamma x^{1/3}})  \left( R^{(i)}_x(r - \frac{k}{\gamma x^{1/3}}) - R^{(i)}_x(r - \frac{k+1}{\gamma x^{1/3}}) \right)\\
 &+\mathcal{O}(x^{-1/3}),
\end{split}
\end{equation}
which immediately implies 
\begin{equation} \label{main theorem 3}
\left| V_x(\zeta q^{-k}) - V_x (\zeta q^{-k-1}) \right| \leq \text{const},
\end{equation}
after expanding $\chi_\varpi$ around $r - \frac{k}{\gamma x^{1/3}}$. 

We can finally evaluate the limit \eqref{main theorem 1}. Using the bound \eqref{main theorem 2}, we write
\begin{equation} \label{main theorem 4}
\begin{split}
\eqref{main theorem 1} &= \lim_{x \to \infty} \left( \sum_{k=0}^{\epsilon \gamma x^{1/3}} + \sum_{k > \epsilon \gamma x^{1/3}} \right) \frac{(-1)^k q^{\binom{k+1}{2}}}{(q;q)_\infty (q;q)_k} (V_x(\zeta q^{-k}) - V_x (\zeta q^{-k-1})) \\
&= \lim_{x \to \infty} \sum_{k \geq 0}  \frac{(-1)^k q^{\binom{k+1}{2}}}{(q;q)_\infty (q;q)_k} (V_x(\zeta q^{-k}) - V_x (\zeta q^{-k-1})) \mathbbm{1}_{[0,\epsilon \gamma x^{1/3}]}(k)
\end{split}
\end{equation}
and following estimate \eqref{main theorem 3}, we can employ the bounded convergence theorem to exchange the limit and summation symbols in the right hand side of \eqref{main theorem 4}. Here the pointwise convergence
\begin{equation*}
\lim_{x \to \infty} V_x(\zeta q^{-k}) - V_x (\zeta q^{-k-1}) = \frac{\partial}{\partial r} \chi_\varpi (r)
\end{equation*}
can be established through the expansion \eqref{bounded convergence 1}, using the fact that the difference between remainder terms $R_x^{(i)}$'s converges to zero, as reported in Propositions \ref{prop discrete Airy}, \ref{proposition limit V 2}, \ref{prop Upsilon}, \ref{prop Upsilon 2}. We can therefore write
\begin{equation*}
\eqref{main theorem 1} = \sum_{k \geq 0} \frac{(-1)^k q^{\binom{k+1}{2}}}{(q;q)_\infty (q;q)_k} \frac{\partial}{\partial r} \chi_\varpi (r) = \frac{\partial}{\partial r} \chi_\varpi (r),
\end{equation*}
which concludes the proof.
\end{proof}

\section{Specializations of the Higher Spin Six Vertex Model} \label{section specializations}
In this section we take a look at the most relevant degenerations of the Higher Spin Six Vertex Model. Letting parameters vary and considering different scalings we can study models which could be discrete or continuous both in time or space. 
\subsection{\texorpdfstring{Stationary $q$-Hahn particle process}.}
First we will consider the $q$-Hahn TASEP, a space-time discrete particle process introduced in \cite{PovolotskyQHahn} as a dual counterpart of a general chipping model solvable by coordinate Bethe Ansatz. As a consequence of exact  results obtained in Section \ref{section matching and fredholm} we will establish here determinantal formulas describing the position of a tagged particle for the model in the stationary regime and under certain assumptions on parameters we establish Baik-Rains fluctuations. 

The $q$-Hahn TASEP is a three parameters dependent simple exclusion process where particles, at each time step, move in a predetermined direction with jumps distributed according to a $q$-deformed Beta binomial law. This means that, recording the position of particles in the lattice at a specific time $t$ in a strictly decreasing sequence $\mathbf{y}(t)=\{ y_k(t) \}_{k \in \mathbb{Z}}$, then after a time unit, $\mathbf{y}(t)$ is updated to a new sequence
$$
\mathbf{y}(t+1) = \{ y_k(t) + \mathsf{J}_k^{t+1} \}_{k \in \mathbb{Z}},
$$
where the values of jumps $\mathsf{J}_k^{t+1}$ are chosen with probabilities $\mathbb{P}(\mathsf{J}_k^{t+1} = j | \mathbf{y}(t))$, given by
\begin{equation} \label{beta binomial}
\varphi_{q,\mu, \nu}(j\ |g_k) = \mu^{j} \frac{(\nu / \mu ; q)_j (\mu;q)_{g_k-j}}{(\nu;q)_{g_k}} \frac{(q;q)_{g_k}}{(q;q)_j (q;q)_{g_k-j}} \qquad \text{for }j=0, \dots, g_k
\end{equation}
and $g_k=y_{k-1}(t)- y_k(t)-1$ is the gap between the $(k-1)$-th and the $k$-th particle. The fact that, provided 
\begin{equation*}
0\leq \nu <\mu<1 \qquad \text{and} \qquad 0\leq q<1, 
\end{equation*}
$\varphi_{q, \mu, \nu}$ is a probability distribution is is consequence of the $q$-Gauss summation \eqref{q Gauss sum}.
\begin{figure}[b]
\centering
\includegraphics[scale=0.9]{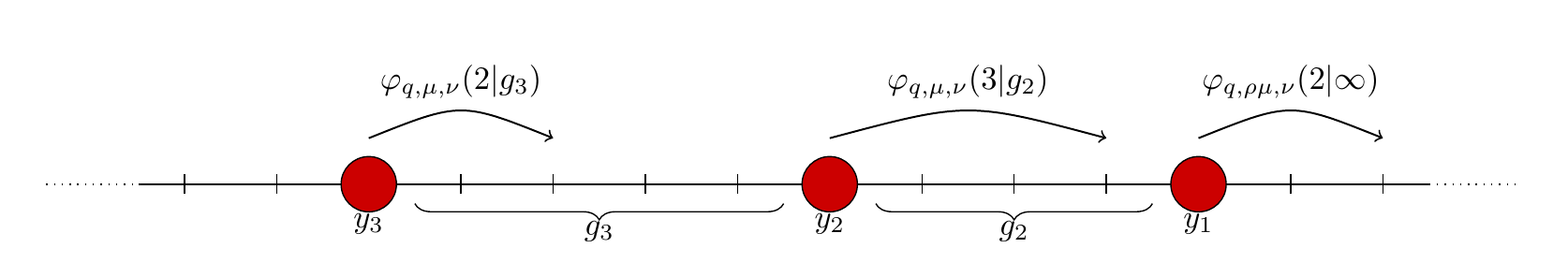}
\caption{\small A visualization of the dynamics of the $q$-Hahn TASEP.}
\end{figure}
The case when the system possesses a rightmost particle, say the one labeled with 1, can be considered ideally placing particles with labels 0, -1, -2,\dots infinitely far away. Here, when we are interested in the evolution of $y_1(t), \dots, y_x(t),$ we can reduce to study a model with only $x$ particles. This is due to the fact that the dynamics of particles $x+1, x+2, \dots$ cannot influence the motion of the ones to their right. In this case, Bethe Ansatz techniques are available (see \cite{Corwin2015qHahn}, \cite{BCPS2015SpectalTheory}) and, for the special initial conditions 
\begin{equation*}
y_k(0)= -k \qquad \text{a.s.} \qquad \text{for } k = 1, \dots, x 
\end{equation*}
the distribution of the single particle $y_x(t)$ exhibits a determinantal structure. This particular property was used in \cite{Veto15qHahn} to establish Tracy-Widom fluctuations for the integrated 
current.

Our goal is to study a different class of initial conditions, where particles fill, with different densities, expressed in terms of two parameters $0 < \mathpzc{d}_-,\mathpzc{d}_+ <1$, the regions respectively at the left and at the right of the origin. More specifically, these are given setting
\begin{equation} \label{q hahn 2 densities}
\begin{split}
&y_1(0)=-1 \qquad \text{a.s.},\\
&y_{k-1}(0) - y_k(0) - 1 \sim \begin{cases}
q\text{NB}(\nu, \mathpzc{d}_+) \qquad \text{if }k \leq 1\\
q\text{NB}(\nu, \mathpzc{d}_-) \qquad \text{if }k > 1,
\end{cases}
\end{split}
\end{equation}
and we refer to these as \emph{double sided $q$-negative binomial initial conditions}. In words, at time $t=0$, consecutive particles occupying the negative half line $\mathbb{Z}_{ \leq -1}$ (those with labels greater or equal than 1) are spaced with $q$-negative binomial distribution of parameters $(\nu, \mathpzc{d}_-)$ and those in the portion of the lattice $\mathbb{Z}_{\geq -1}$ are spaced with $q$-negative binomial law of parameters $(\nu,\mathpzc{d}_+)$. An important particular case of initial conditions \eqref{q hahn 2 densities} is given setting 
\begin{equation} \label{q hahn equal densities}
\mathpzc{d}_- = \mathpzc{d}_+ =\mathpzc{d}.
\end{equation}
As proven in Proposition \ref{prop q hahn stationary}, with this particular choice, the dynamics of the $q$-Hahn TASEP preserves the distribution of gaps between consecutive particles and therefore \eqref{q hahn 2 densities}, \eqref{q hahn equal densities} are regarded as \emph{stationary initial conditions}.

The reason behind the exact solvability of the model with initial conditions \eqref{q hahn 2 densities} is that the study of evolution of coordinates $y_1(t), \dots, y_x(t)$ (with $x \geq 1$), can be reduced to the study of the same quantities in a system including only finitely many particles. Indeed the presence of infinitely many particles (spaced with distribution $q$NB$(\nu, \mathpzc{d}_+)$) at the right of the first one can be mimicked by simply slowing down $y_1$ by a quantity depending on $\mathpzc{d}_+$. This is a consequence of the fact that the dynamics of the $q$-Hahn TASEP preserves the spacing between $y_1, y_0, y_{-1}, \dots$ and of the simple identity
\begin{equation} \label{q binomial identity q hahn}
\mathbb{P}(\mathsf{J}_1^t = j) = \sum_{g \geq 0} \varphi_{q, \mu, \nu}(j |g) \pi_g 
= \varphi_{q,\mu \mathpzc{d} , \nu \mathpzc{d}} (j| \infty),
\end{equation}
where
\begin{equation} \label{densities q hahn}
\pi_M = \mathpzc{d}^M \frac{(\nu;q)_M}{(q;q)_M} \frac{(\mathpzc{d} , q)_\infty}{(\nu \mathpzc{d} ; q)_\infty }.
\end{equation}
This can be proven expanding terms $\varphi_{q, \mu, \nu}(j|g),\pi_g$ and using the $q$-binomial theorem \eqref{q binomial summation}. Equality \eqref{q binomial identity q hahn} shows that, when $y_0(t) - y_1(t) -1$ is distributed according to $q$NB$(\nu, \mathpzc{d}_+)$ at each time, then the effective distribution of jumps of $y_1$ is given by $\varphi_{q,\mu \mathpzc{d}_+ , \nu \mathpzc{d}_+} (\bullet| \infty)$.

We now come to prove the claim that, with initial conditions \eqref{q hahn 2 densities}, \eqref{q hahn equal densities}, the dynamics preserves the distribution of gaps.
\begin{Prop} \label{prop q hahn stationary}
Let $\mathbf{y}(t)=\{ y_k(t) \}_{k \in \mathbb{Z}}$ be the array of positions of particles of a $q$-Hahn TASEP having initial conditions
\begin{equation*}
y_{k-1}(0) - y_k(0) - 1 \sim q\text{NB}(\nu, \mathpzc{d}), \qquad \text{for all } k \in \mathbb{Z}, 
\end{equation*}
and $\mathpzc{d}$ is a fixed parameter in the interval $(0,1)$. Then, for each $t \geq 0$
\begin{equation*}
y_{k-1}(t) - y_k(t) - 1 \sim q\text{NB}(\nu, \mathpzc{d}), \qquad \text{for all } k \in \mathbb{Z}. 
\end{equation*}
\end{Prop}
We start stating a simple summation identity
\begin{Lemma}
For any complex numbers $a,b,c$ and integer $M\geq0$, we have
\begin{equation} \label{summation a b c}
\sum_{k\geq 0}a^k \frac{(b;q)_k}{(q;q)_k}\frac{(c;q)_{M-k}}{(q;q)_{M-k}}=\frac{1}{2 \pi \mathrm{i}}\oint_{C_0}\frac{(zab;q)_\infty}{(za;q)_\infty}\frac{(zc;q)_\infty}{(z;q)_\infty}\frac{dz}{z^{M+1}},
\end{equation}
where $C_0$ is a sufficiently small contour encircling $0$ and no other poles.
\end{Lemma}
\begin{proof}
For $z$ sufficiently close to 0 we define the functions
\begin{equation*}
F(z)=\sum_{l\geq 0}(za)^l\frac{(b;q)_l}{(q;q)_l}=\frac{(zab;q)_\infty}{(za;q)_\infty}, \qquad G(z)=\sum_{l \geq 0}z^l\frac{(c;q)_l}{(q;q)_l}
\end{equation*}
and we see that their product can be written as 
\begin{equation*}
F(z)G(z)=\sum_{M \geq 0}\Big( \sum_{k\geq 0} a^k \frac{(b;q)_k}{(q;q)_k} \frac{(c;q)_{M-k}}{(q;q)_{M-k}} \Big) z^M,
\end{equation*}
so that 
\begin{equation*}
\sum_{k\geq 0} a^k \frac{(b;q)_k}{(q;q)_k} \frac{(c;q)_{M-k}}{(q;q)_{M-k}}=\frac{1}{2 \pi \mathrm{i}} \oint_{C_0}F(z)G(z)\frac{dz}{z^{M+1}}
\end{equation*}
and we have our result.
\end{proof}
\begin{proof}[Proof of Proposition \ref{prop q hahn stationary}]
We proceed with a checking style argument. We will show that
\begin{equation*}
\pi_M = \sum_{k \geq 0} \pi_k \sum_{l \geq 0} \mathbb{P}(\mathsf{J}_{x-1}^{t+1}=M-k+l)\varphi_{q,\mu,\nu}(l|\ k),
\end{equation*}
or, equivalently, expanding all terms, that
\begin{equation} \label{stationary qhahn equality}
\mathpzc{d}^M \frac{(\nu,q)_M}{(q;q)_M}
=\sum_{k \geq 0} \mathpzc{d}^k \sum_{l \geq 0} (\mu \mathpzc{d})^{M-k+l} \frac{(\nu / \mu ;q)_{M-k+l}}{(q;q)_{M-k+l}} \frac{(\mathpzc{d} \mu;q)_\infty}{(\mathpzc{d} \nu ; q)_\infty} \mu^l \frac{(\nu/\mu;q)_l (\mu;q)_{k-l}}{(q;q)_l (q;q)_{k-l}},
\end{equation}
where we made use of a summation like \eqref{q binomial identity q hahn} to express the probability of the $(x-1)$-th particle making a jump of $M-k+l$ steps. In the right hand side of \eqref{stationary qhahn equality} we can exchange the summation order noticing that the sum in the $l$ index is nontrivial only for $l\leq k$. Therefore this can be written as
\begin{equation*}
\begin{split}
& \frac{(\mathpzc{d} \mu;q)_\infty}{( \mathpzc{d} \nu ; q)_\infty} \sum_{l \geq 0} \mu^l \mathpzc{d}^l \frac{(\nu/ \mu;q)_l}{(q;q)_l} \sum_{k \geq l} \mathpzc{d}^{k-l}(\mu \mathpzc{d})^{M-(k-l)} \frac{(\nu/ \mu; q)_{M-(k-l)}(\mu;q)_{k-l}}{(q;q)_{M-(k-l)}(q;q)_{k-l}}\\
&=\mu^M \mathpzc{d}^M \sum_{k'\geq 0}\mu^{-k'} \frac{(\nu/ \mu;q)_{M-k'}(\mu;q)_{k'}}{(q;q)_{M-k'}(q;q)_{k'}}.
\end{split}
\end{equation*}
The summation can be evaluated with \eqref{summation a b c} setting $a=1/\mu$, $b=\mu$ and $c=\nu/\mu$. We get
\begin{equation*}
\sum_{k'\geq 0}\mu^{-k'} \frac{(\nu/ \mu;q)_{M-k'}(\mu;q)_{k'}}{(q;q)_{M-k'}(q;q)_{k'}}=\frac{1}{2 \pi \mathrm{i} }\oint_{C_0} \frac{(z\nu/\mu;q)_\infty}{(z/ \mu;q)_\infty} \frac{dz}{z^{M+1}}=\mu^{-M}\frac{(\nu;q)_M}{(q;q)_M},
\end{equation*}
which combined with the previous identities completes the proof.
\end{proof}

The fact that \eqref{q hahn 2 densities}, \eqref{q hahn equal densities} constitute a family of translation invariant initial conditions was originally argued in \cite{Corwin2015qHahn}. There the author speculated the stationarity property starting from the fact that they are an infinite volume analog of the factorized steady state measures of the $q$-Hahn zero range process in the ring geometry \cite{PovolotskyQHahn}. Our proof is of some interest as it is elementary, in the sense that it only makes use of notable $q$-binomial identities.

Although the $q$-Hahn TASEP was introduced in \cite{PovolotskyQHahn} with no reference to stochastic vertex models, it is indeed possible to obtain it as a degeneration of the Higher Spin Six Vertex Model, as it was observed first in \cite{CorwinPetrov2016HSVM}. The natural way to construct a simple exclusion process from the Higher Spin Six Vertex Model is to interpret the vertical axis as a time direction and to read the number of paths vertically crossing vertices as the evolution of gaps between consecutive particles. More specifically, given occupation random variables $\mathsf{j}_1^1,\mathsf{j}_1^2, \dots, \mathsf{j}_1^t$ and $\mathsf{m}_2^t,\mathsf{m}_3^t, \dots, \mathsf{m}_x^t$ defined in \eqref{m rv}, \eqref{j rv}, we construct a configuration of particles $\mathbf{y}(t) = \{ y_k(t) \}_{k \geq 1}$ such that
\begin{equation*}
y_1(t) = -1 + \mathsf{j}_1^1 + \cdots + \mathsf{j}_1^t \qquad \text{and} \qquad y_{k-1}(t) - y_k(t) -1 = \mathsf{m}_k^t \qquad \text{for all } k\geq 2. 
\end{equation*}
In this way, horizontal occupation numbers $\mathsf{j}_1^t, \mathsf{j}_2^t, \dots$ are interpreted as jumping distances of particles during the update at time $t$ and the Markov operator describing the stochastic dynamics is given in general by the transfer operator $\mathfrak{X}_{u_t}^{(J)}$, as in \eqref{fused transfer matrix}. Although the exact form of the fused weights $\mathsf{L}^{(J)}$, reported in \eqref{fused weights L}, appearing in the definition of $\mathfrak{X}_{u_t}^{(J)}$, looks rather complicated it is possible to degenerate it and match it with an instance of the $q$ deformed beta binomial distribution \eqref{beta binomial}. This fact was first observed in \cite{Borodin2017SymmetricFunctions} and in our notation, Proposition 6.7 of the same article implies that
\begin{equation} \label{borodin limit fused weights}
\mathsf{L}_{u_t\xi_k,s_k}^{(J)} (i_1, j_1 |\ i_2, j_2) \xrightarrow[\substack{ s_k = s \\ \xi_k = 1/s \\ u_t = s^2 }]{} \mathbbm{1}_{i_1 + j_1 = i_2 + j_2 } \mathbbm{1}_{j_2 \leq i_1} \varphi_{q, q^J s^2, s^2} (j_2| i_1). 
\end{equation}
Expression \eqref{borodin limit fused weights} suggests us the right specialization to turn the transfer operator $\mathfrak{X}^{(J)}_{u_t}$ into the Markov generator of the $q$-Hahn TASEP. On the other hand, thanks to arguments carried in Section \ref{subsection integrable initial}, we also know how to employ analytic continuation techniques to describe the probability distribution of the model for certain random initial conditions, which indeed would correspond to \eqref{q hahn 2 densities}. 

We like to summarize this discussion concerning the matching between $q$-Hahn particle processes and Higher Spin Six Vertex Model in the following

\begin{Prop} \label{prop probability q hahn}
Consider the (non stochastic) Higher Spin Six Vertex Model on $\Lambda_{0,-1}$ with boundary conditions 
\begin{equation} \label{q hahn strange bc}
\mathsf{m}_{x}^{-1} = 0 \quad \text{a.s.}, \quad \text{for all }x \geq 1, \qquad \mathsf{j}_0^0 = K \quad \text{a.s}, \qquad \mathsf{j}_0^t = J \quad \text{a.s,} \quad \text{for all }t \geq 1,
\end{equation}
transfer operators $\mathfrak{X}_{q/\mathpzc{d}_-}^{(K)}, \mathfrak{X}_{s^2}^{(J)}, \mathfrak{X}_{s^2}^{(J)}, \dots$ and parameters
\begin{equation} \label{parameters q hahn}
\Xi=(\xi_1, s^{-1}, s^{-1}, \dots), \qquad \mathbf{S}= ( s_1, s, s, \dots).
\end{equation}
Then, for each $l$, the signed measure $\mathbb{P}(\mathcal{H}(x,t)-\mathsf{m}_1^0 = l)$ is an analytic function of $\mu = q^J s^2$ and $\wp = q^{-K}$. Moreover, setting
\begin{gather}
s_1 = 1/N, \qquad \xi_1 = \mathpzc{d_+} N, \qquad \text{and taking the limit }N \to \infty, \label{q hahn first s xi}\\
0<\mathpzc{d}_- < \mathpzc{d}_+ <1, \qquad s^2 = \nu, \qquad 0 \leq \nu < \mu < 1, \qquad \wp = 0, 
\end{gather}
we obtain
\begin{equation*}
\mathbb{P}(\mathcal{H}(x,t) -\mathsf{m}_1^0 = l) = \mathbb{P}_{q \text{\emph{H}}(\mathpzc{d}_-,\mathpzc{d}_+) \otimes \mathsf{m} } (y_x(t) + x - \mathsf{m} = l),
\end{equation*}
for each $l \in \mathbb{Z}$, $x \geq 1$ and $t \geq 0$. In the last equality, both sides are probability measures, $\mathbb{P}_{q \text{\emph{H}}(\mathpzc{d}_-,\mathpzc{d}_+) \otimes \mathsf{m} }$ refers to a product measure of a $q$-Hahn TASEP with initial conditions \eqref{q hahn 2 densities} and of a $q$Poisson($\mathpzc{d}_-/\mathpzc{d}_+$) random variable $\mathsf{m}$ (independent of $y_x$). 
\end{Prop}
\begin{proof}
We start considering expression \eqref{analytical extension probability}, which is stated for a model with boundary conditions \eqref{q hahn strange bc}, $J=1$, transfer operator $\mathfrak{X}_{q/v}^{(K)},\mathfrak{X}_{u_1},\mathfrak{X}_{u_1},\dots$ and generic parameters $u_t, \xi_x , s_x$. The fusion of rows procedure, as explained in Section \ref{subsection fused dynamics}, allows us to substitute $\mathfrak{X}$ with $\mathfrak{X}^{(J)}$ and it simply consists in specializing spectral parameters in geometric progressions of ratio $q$. We therefore operate the substitution 
\begin{equation} \label{parameters q hahn u}
(u_{Jm+1}, u_{Jm +2}, \dots, u_{J(m+1)}) \to (s^2, q s^2, q^2 s^2, \dots, q^{J-1} s^2 ) \qquad \text{for }m \geq 0
\end{equation}
and, as a result, in \eqref{analytical extension probability}, we change the factor $\tilde{\Pi}$ in the integrand into
\begin{equation*}
\tilde{\Pi} (\mathbf{z} ; \Xi^{-1} \mathbf{S}, \hat{\mathbf{U}}) \to \prod_{j = 1}^x \left( \prod_{i=2}^x (z_j s^2 ;q)_\infty^{-1} \left( \frac{(s^2 z_j ; q)_\infty}{(q^J s^2 z_j ; q)_\infty}  \right)^t \right).
\end{equation*}
As long as the quantity $q^J s^2$ is smaller than 1 in absolute value, no new pole is created for the integration in $z_1, \dots, z_x$ on the torus $\mathbb{T}^x$ and therefore we can analytically prolong $\mathbb{P}( \mathcal{H}(x,t) - \mathsf{m}_1^0 = \bullet )$ to the region $s^2 = \nu, \mu = q^J s^2, 0\leq\nu<\mu<1$. Choice of parameters \eqref{q hahn first s xi} implies that $\mathsf{j}_1^t \sim \varphi_{q, \mathpzc{d}_+ \mu , \mathpzc{d}_+ \nu }(\bullet | \infty)$ as it was observed in \eqref{limit LJ} and together with conditions on $\nu, \mu$, it turns the transfer operator $\mathfrak{X}_{s^2}^{(J)}$ into a Markov generator describing a $q$-Hahn TASEP where the rightmost particle is slower of a factor  $\mathpzc{d}_+$ compared to the others. This is a basic consequence of \eqref{borodin limit fused weights}. 

The analytic continuation in parameter $\wp=q^{-K}$ is treated as in Proposition \ref{analytical extension probability proposition}. As a result of choice $\wp=0$, random variables $\mathsf{m}_1^0, \mathsf{m}_2^0, \dots$ become independently distributed as
\begin{equation*}
\mathsf{m}_1^0 \sim q\text{Poi}(\mathpzc{d}_- / \mathpzc{d}_+ ), \qquad \text{and} \qquad \mathsf{m}_x^0 \sim q\text{NB}(\nu, \mathpzc{d}_- ).
\end{equation*}
This passage is explained more extensively in Section \ref{subsection integrable initial}. Recalling the definition of $\mathcal{H}$, given in \eqref{height H bar}, interpreting $\mathsf{m}_k^t$ as the gap between the $(k-1)$-th and the $k$-th particle and $\mathsf{j}_k^t$ as the jumps made by the $k$-th particle during the update at time $t$ we realize that
\begin{equation*}
y_x(t) + x - \mathsf{m} \stackrel{\mathcal{D}}{=} \mathcal{H}(x,t) - \mathsf{m}_1^0 \qquad \text{for all }x \geq 1,t \geq 0,
\end{equation*}
where the equality holds in distribution. So far $y_x$ is the position of the $x$-th particle of a $q$-Hahn TASEP with a slower particle, but as a consequence of identity \eqref{q binomial identity q hahn}, this is equivalent, in distribution, to the position of the $x$-th particle in a model with infinitely many particles at the right of $y_1$ spaced with $q$-negative binomial distribution of parameters $(\nu, \mathpzc{d}_+)$. This concludes the proof.
\end{proof}
We come now to state our main results on the double sided $q$-negative binomial $q$-Hahn TASEP.
\begin{Prop}
For $\mathpzc{d}_- < \mathpzc{d}_+$, we have
\begin{equation}\label{first fredholm determinant qHahn}
\mathbb{E}_{q \text{\emph{H}} (\mathpzc{d}_-,\mathpzc{d}_+) \otimes \mathsf{m}} \left( \frac{1}{(\zeta q^{y_x(t) + x - \mathsf{m} };q)_{\infty}} \right)  = \det(\mathbf{1}-fK_{q\text{\emph{H}}})_{l^2(\mathbb{Z})},
\end{equation}
where
\begin{gather}
f(n)= \frac{1}{1-q^{n}/\zeta},\\
K_{q\text{\emph{H}}}(n,m)= A_{q\text{\emph{H}}}(n,m) + (\mathpzc{d}_+ - \mathpzc{d}_-) \Phi_{q\text{\emph{H}},x}(m)\Psi_{q\text{\emph{H}},x}(n),\label{kernel qHahn}\\
A_{q\text{\emph{H}}}(n,m)=\frac{\tau(n)}{\tau(m)} \int_D \frac{dw}{2 \pi \mathrm{i}} \int_C \frac{dz}{2 \pi \mathrm{i}} \frac{z^m}{w^{n+1}} \frac{F(z)}{F(w)} \frac{1}{z-w}, \label{A qHahn}
\\
\Phi_{q\text{\emph{H}},x}(n)=\tau(n)\int_D \frac{dw}{2 \pi \mathrm{i}}\frac{1}{w^{n+2}}\frac{1}{ (\mathpzc{d}_+/w ; q)_\infty } \frac{ 1 }{F(w)}, \label{Phi qHahn}\\
\Psi_{q\text{\emph{H}},x}(n)=\frac{1}{\tau(n)} \int_C \frac{dz}{2 \pi \mathrm{i}} z^{n} \frac{1}{z-\mathpzc{d}_-} (q z /\mathpzc{d}_+;q)_\infty F(z) \label{Psi qHahn}.
\end{gather}
The contour $D$ encircles 1,$\mathpzc{d}_+$ and no other singularity, whereas $C$ contains 0 and $q^{k}\mathpzc{d}_-$, for any $k$ in $\mathbb{Z}_{\geq 0}$. Moreover, $\tau(n)$ is taken to be
\begin{equation} \label{tau qHahn}
\tau(n)=\begin{cases}
    b^n,       & \quad \text{if } n \geq 0\\
    c^n,  & \quad \text{if } n < 0,\\
  \end{cases}
\end{equation}
with $\mathpzc{d}_- < b < \mathpzc{d}_+  < 1 < c < 1 / \nu $, and 
\begin{equation} \label{F qHahn}
F(z)= \left( \frac{(\nu z;q)_\infty}{(\mu z;q)_\infty } \right)^t  \left( \frac{( z;q)_\infty}{(\nu z;q)_\infty} \right)^{x-1} \frac{1}{(q \mathpzc{d}_-/z ;q)_\infty}.
\end{equation}
Finally, $\mathsf{m}$ is a $q$Poi($\mathpzc{d}_-/\mathpzc{d}_+$) random variable independent of $y_x$.
\end{Prop}
\begin{proof}
We only need to specialize results of Theorem \ref{theorem: q laplace shifted introduction} to the same choice of parameters adopted in Proposition \ref{prop probability q hahn} and ultimately to perform the analytic continuation in parameter $\mu=q^J s^2$.
\end{proof}
It is now safe to apply techniques developed for the Higher Spin Six Vertex Model to describe asymptotic fluctuations of the position of a tagged particle in the stationary $q$-Hahn TASEP. In order to fix the parameters describing the scaling of the $q$-Hahn TASEP we introduce the families of functions
\begin{equation*}
\mathpzc{a}_{-1}(z; p, \tilde{p}) = \log \left( \frac{(z p ; q)_\infty}{(z \tilde{p} ; q)_\infty} \right) \qquad \text{and} \qquad \mathpzc{a}_{k+1}(z;p,\tilde{p}) = z \frac{d}{dz} \mathpzc{a}_k (z;p,\tilde{p}),
\end{equation*}
for all $k \geq 0$. When $k\geq 0$, $\mathpzc{a}_k(z;p,\tilde{p})$ is expressed in terms of $q$-polygamma like functions \eqref{digamma psi} as $\polygamma_k(\tilde{p} z) -\polygamma_k(p z)$. 
\begin{Def}[Scalings for the stationary $q$-Hahn TASEP] \label{scaling qHahn} 
For numbers $\mathpzc{d} \in (0,1)$ and $\varpi \in \mathbb{R}$, we set
\begin{equation*}
    \gamma_{q\emph{H}}= - \frac{1}{2^{1/3}} \left( \frac{\mathpzc{a}_1(\mathpzc{d};\nu,1)}{\mathpzc{a}_1(\mathpzc{d};\nu,\mu)} \mathpzc{a}_2(\mathpzc{d};\nu,\mu) - \mathpzc{a}_2(\mathpzc{d};\nu,1)    \right)^{1/3},
\end{equation*}
\begin{equation*}
    \kappa_{q\emph{H};\varpi} = \frac{\mathpzc{a}_1(\mathpzc{d};\nu,1)}{\mathpzc{a}_1(\mathpzc{d};\nu,\mu)} + \frac{\mathpzc{a}_2(\mathpzc{d};\nu,1) \mathpzc{a}_1(\mathpzc{d};\nu,\mu) - \mathpzc{a}_1(\mathpzc{d};\nu,1) \mathpzc{a}_2(\mathpzc{d};\nu,\mu) }{\mathpzc{a}_1(\mathpzc{d};\nu,\mu)^2} \frac{\varpi}{\gamma_{q\emph{H}} x^{1/3}},
\end{equation*}
\begin{equation*}
    \eta_{q\emph{H};\varpi} = \kappa_{q\emph{H};\varpi} \mathpzc{a}_0(\mathpzc{d};\nu,\mu) -\mathpzc{a}_0(\mathpzc{d};\nu,1) + \frac{\mathpzc{a}_2(\mathpzc{d};\nu,1) \mathpzc{a}_1(\mathpzc{d};\nu,\mu) - \mathpzc{a}_1(\mathpzc{d};\nu,1) \mathpzc{a}_2(\mathpzc{d};\nu,\mu) }{\mathpzc{a}_1(\mathpzc{d};\nu,\mu)} \frac{\varpi^2}{\gamma_{q\emph{H}}^2 x^{2/3}}
\end{equation*}
\end{Def}
By means of quantities $\kappa_{q\text{H}}, \eta_{q\text{H}}, \gamma_{q\text{H}}$ we are now going to confirm the KPZ-scaling conjecture for the stationary $q$-Hahn TASEP, result that in our notation reads 
\begin{equation*} 
\frac{y_x(\kappa_{q\text{H};\varpi } x) - (\eta_{q\text{H},\varpi} -1 ) x }{\gamma_{q\text{H}} x^{1/3}}\xrightarrow[x \to \infty]{\mathcal{D}} F_\varpi,
\end{equation*}
where $F_{\varpi}$ is the Baik-Rains distribution introduced in Definition \ref{Baik-Rains definition}. For the sake of a rigorous procedure in the asymptotics we need to establish technical conditions on parameters defining the model. As for the Higher Spin Six Vertex Model, the main technical issue is to guarantee the existence of steep descent/ascent contours $C/D$ for the real part of a function $g_{q \text{H}}$, which in this case is given by
\begin{equation}\label{g qHahn}
    g_{q \text{H}}(z) = - \eta_{q\text{H},\varpi} \log(z) + \kappa_{q\text{H},\varpi} \mathpzc{a}_{-1}(z;\nu, \mu) - \mathpzc{a}_{-1}(z;\nu, 1).
\end{equation}
Function $g_{q \text{H}}$ possesses a double critical point $\varsigma$ in a neighborhood of order $\varpi/x^{1/3}$ of $\mathpzc{d}$ and the construction of contours $C,D$ enables us to perform a saddle point analysis to evaluate the Fredholm determinant of the kernel $K_{q\text{H}}$. The expression of the critical point $\varsigma$ can be given explicitely and it is identical to \eqref{varsigma} once we substitute $\gamma=\gamma_{q\text{H}}, a_k = \mathpzc{a}_k(\mathpzc{d}, \nu, \mu)$ and $h_k = \mathpzc{a}_k(\mathpzc{d}, \nu, 1)$.

We find that, in the $q$-Hahn TASEP case, the analysis of $g_{q \text{H}}$ slighlty differs from that of the homologous function $g$ for the general Higher Spin Six Vertex Model. In particular, the problem of the existence of steep contours was already considered in \cite{Veto15qHahn} and we can take advantage of results obtained by the author in the same paper, which we summarize in the following Proposition. 
\begin{Prop}[\cite{Veto15qHahn}, Prop. 6.2, 6.3] \label{prop steep q hahn}
Define the curves 
\begin{equation} \label{contours Veto}
C=\left \{ \varsigma e^{\mathrm{i} \vartheta}| \ \vartheta \in [0, 2 \pi) \right \}, \qquad D=\left \{ 1 - ( 1 - \varsigma ) e^{\mathrm{i} \vartheta}| \ \vartheta \in [0, 2 \pi) \right \}.
\end{equation}
Then, assuming
\begin{equation} \label{technical condition q hahn}
0 \leq q \leq \nu < \mu \leq 1/2,
\end{equation}
we have, for $x$ large enough
\begin{enumerate}
\item $\mathfrak{Re}\{g_{q \text{\emph{H}}}\}$ assumes, on the contour $C$, a unique global maximum in $\varsigma$;
\item$\mathfrak{Re}\{g_{q \text{\emph{H}}}\}$ assumes, on the contour $D$, a unique global minimum in $\varsigma$.
\end{enumerate}
\end{Prop}

\begin{figure}[ht]
\centering
  \includegraphics[clip, trim=2.5cm 7.5cm 2.5cm 7.5cm,scale=0.6]{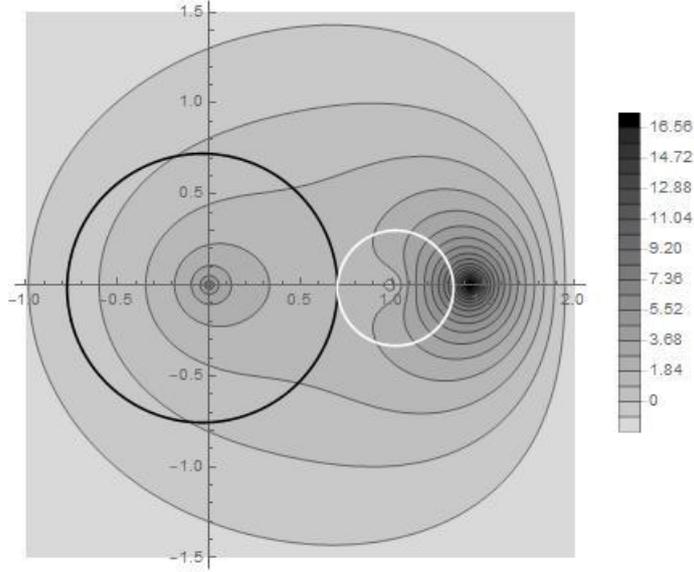}
\caption{\small In the picture we see the contour plot
$\mathfrak{Re}\{g (z)\}$ for the particular choice of parameters $\varsigma=0.7, \mu=0.7, \nu=0.4, q=0.3, x=20$. Here $z$ lies the complex rectangle $ [-1,2] + \mathrm{i} [-1.5 ,1.5] $ and, as the legend shows, to darker shades correspond greater values of $\mathfrak{Re}\{g (z)\}$. The black and the white circle are respectively $C$ and $D$ and they intersect in the critical point $\varsigma$.} \label{plot g qHahn}
\end{figure}

Condition \eqref{technical condition q hahn} appears to be technical, as it could be argued through simple numerical tests. As an example we report in Figure \ref{plot g qHahn} the plot of the real part of $g_{q\text{H}}$ for a choice of parameters $q, \nu, \mu, \mathpzc{d}$ not included in \eqref{technical condition q hahn} and from where it appears evident that, also in that case steep contours $C,D$ can be constructed. We do not attempt here to loosen hypothesis on Proposition \ref{prop steep q hahn} and we simply use such results to adapt our asymptotic analysis of the $q$-Laplace transform in the stationary $q$-Hahn TASEP setting.

We come to the following
\begin{theorem}
Consider the $q$-Hahn TASEP with parameters $q, \nu, \mu$ as in \eqref{technical condition q hahn} and 
stationary initial conditions, where $\mathpzc{d}$ satisfies
\begin{equation} \label{qhahn thm hypothesis}
\mathpzc{d} > \frac{2 q}{ 1 + q}.
\end{equation}
Then we have
\begin{equation} \label{baik rains limit qhahn}
\lim_{x \rightarrow \infty} \mathbb{P}_{q\emph{H}(\mathpzc{d}, \mathpzc{d}) }\left( \frac{y_x(\kappa_{q\emph{H};\varpi } x) - (\eta_{q\emph{H},\varpi} -1 ) x }{\gamma_{q\emph{H}} x^{1/3}} > -r \right) = F_\varpi(r).
\end{equation}
\end{theorem}
\begin{proof}
We see that from the Fredholm determinant identity \eqref{first fredholm determinant qHahn}, employing the same procedure detailed in the proofs of Theorem \ref{theorem q laplace double sided}, we can decouple the quantity $y_x$ from the random shift $\mathsf{m}$. This leads us to an exact expression for the $q$-Laplace transform $\mathbb{E}_{q \text{H}(\mathpzc{d} , \mathpzc{d})}\left( (\zeta q^{y_x(t) + x} ; q)_\infty^{-1} \right)$ as
\begin{equation} \label{q laplace qhahn}
\frac{1}{(q;q)_\infty} \sum_{k \geq 0} \frac{(-1)^k q^{\binom{k}{2}+k} }{(q;q)_k} \left( V_x(\zeta q^{-k}) - V_x(\zeta q^{-k-1}) \right),  
\end{equation}
where the function $V_x$ takes the form
\begin{equation} \label{V q hahn}
\begin{split}
V_x(\zeta) = \det (\mathbf{1} - f A_{q \text{H}}) \Bigg(& - \polygamma_0(1/\zeta) + \polygamma_0(q\zeta) - 2 \polygamma_0(q) - x \mathpzc{a}_0(\mathpzc{d};\nu,1)  + t \mathpzc{a}_0(\mathpzc{d};\nu,\mu)  \\
& - \mathpzc{d} \sum_{\substack{i,j=1,2\\ (i,j)\neq (1,1)}} \sum_{n \in \mathbb{Z}} f(n)\Phi_{q \text{H},x}^{(i)}(n) \Psi_{q \text{H},x}^{(j)}(n) - \mathpzc{d}  \sum_{n \in \mathbb{Z}}(fA_{q \text{H}} \varrho_{q \text{H}} f\Phi_{q \text{H},x})(n) \Psi_{q \text{H},x}(n) \Bigg).
\end{split}
\end{equation}
Here $\varrho_{q \text{H}}=(\mathbf{1} - f A_{q \text{H}})^{-1}$ and terms $\Phi_{q \text{H},x}^{(i)},\Psi_{q \text{H},x}^{(j)}$ are obtained from $\Phi_{q \text{H},x}, \Psi_{q \text{H},x}$ separating the contribution of pole $\mathpzc{d}$ on integral expressions \eqref{Phi qHahn}, \eqref{Psi qHahn}, analogously to \cref{Phi 1,Phi 2,Psi 1,Psi 2}.

Result \eqref{baik rains limit qhahn} now follows evaluating the large $x$ limit of \eqref{q laplace qhahn}, \eqref{V q hahn} after setting
\begin{equation*}
\zeta = -q^{-\eta_{q\text{H};\varpi} x + \gamma_{q\text{H}} x^{ 1/3 }r}, \qquad t= \kappa_{q\text{H};\varpi} x
\end{equation*}
and this can be done through methods developed throughout Section \ref{section time asymptotics}. We want to remark that the main tool used to compute the asymptotic behavior of the $q$-Laplace transform is the saddle point method, applied to the complex integral expression of the kernels $A_{q \text{H}}, \Phi_{q \text{H},x} \otimes \Psi_{q \text{H},x}$. This procedure is rigorously justified by the statement of Proposition \ref{prop steep q hahn} which guarantees the existence of steep integration contours $C,D$.

The additional hypothesis \eqref{qhahn thm hypothesis}, we made on the density parameter $\mathpzc{d}$ is analogous to condition $2a/(1+\sigma) < q^{-1} \mathpzc{d}$ stated in \eqref{bound R_a} for the Higher Spin Six Vertex Model. In particular \eqref{qhahn thm hypothesis} implies that, for $x$ large enough,
\begin{equation*}
\frac{q}{ \mathpzc{d} } \times \max_{w \in D} |w| < 1,
\end{equation*}
where the contour $D$ is defined in \eqref{contours Veto}. This fact can be used to establish the exponential decay of rear tails of terms
\begin{equation*}
- \sum_{\substack{i,j=1,2\\ (i,j)\neq (1,1)}}  f(n)\Phi_{q \text{H},x}^{(i)}(n) \Psi_{q \text{H},x}^{(j)}(n) -  (fA_{q \text{H}} \varrho_{q \text{H}} f\Phi_{q \text{H},x})(n) \Psi_{q \text{H},x}(n),
\end{equation*}
in the expression of $V_x$  analogously to what is explained in the proof of Proposition \ref{prop Upsilon}. 
\end{proof}

\subsection{Continuous time processes} \label{section half continuous}
There are mainly two possible scalings giving rise to meaningful continuous time versions of the Higher Spin Six Vertex Model (here we only treat the unfused model, therefore $J=1$) and the aim of this paragraph is to briefly define and make a few comments on them.

Possibly the naivest way to proceed is to simply scale the spectral parameter $u= -\varepsilon$ along with the discrete time $t= \ceil{ \varepsilon^{-1} \mathfrak{t}}$ and then let $\varepsilon$ go to zero. In this limit the vertex weights $\mathsf{L}_{\xi_x u,s_x}$ become, up to order $\varepsilon$, as shown in Table \ref{weights table half continuous}.
\begin{table}[h!]
  \centering
  \begin{tabular}{l||c|c|c|c}
     & \includegraphics[]{images/fig_VertexConf_0_g_0_g.pdf}  &  \includegraphics[]{images/fig_VertexConf_0_g_1_gm1.pdf} & 
     \includegraphics[]{images/fig_VertexConf_1_g_1_g.pdf}& 
     \includegraphics[]{images/fig_VertexConf_1_g_0_gp1.pdf}\\
    \hhline{-----}
    $\mathsf{L}_{\xi_x \varepsilon, s_x}$ & $1-\xi_x s_x (1-q^g) \varepsilon$ & $\xi_x s_x (1-q^g) \varepsilon $ & $s_x^2 q^g + \xi_x s_x (1- s_x^2 q^g)\varepsilon$ & $1-s_x^2 q^g - \xi_x s_x (1- s_x^2 q^g)\varepsilon$\\
  \end{tabular}
  \caption{\small{Weights $\mathsf{L}_{\xi_x \varepsilon,s_x}(i_1, j_1 |\  i_2, j_2)$ in the continuous time scaling $u=\varepsilon \sim 0 $.}}
  \label{weights table half continuous}
\end{table}
Using standard arguments one can rigorously show the convergence to a Markov process $\mathfrak{X}^{\text{hc}}$ which evolves according to the following rules
\begin{itemize}
\item paths move on the quadrant $\mathbb{Z}_{\geq 2} \times \mathbb{R}_{\geq 0}$, where at each discrete $x$-coordinate is associated a Poisson clock with rate $\xi_x s_x (1-q^{\#\{\text{paths travelling the $x$-th lane}\}})$
\item each path travels vertically with unitary speed, possibly temporarily sharing with others the same route. When the clock at the generic position $x$ rings, one of the paths occupying this lane is immediately diverted to its right and placed at the random location $x+k$ with probability
\begin{equation} \label{hc1}
s_{x+1}^2 \cdots s_{x+k-1}^2 q^{\mathfrak{h}(x+1) - \mathfrak{h}(x+k)} (1-s_{x+k}^2 q^{\mathfrak{h}(x+k) - \mathfrak{h}(x+k+1)})
\end{equation}
and from there it continues its upward movement.
If at the moment the $x$-th clock rings no path is occupying position $x$, nothing happens.
\item paths randomly emanate from the boundary $\{1\} \times \mathbb{R}_{\geq 0}$ with exponential law
\begin{equation} \label{hc2}
\mathbb{P}(\text{a path is generated from the segment }\{1\} \times [t, t+ \varepsilon ] ) = \xi_1 s_1 \varepsilon + o(\varepsilon).
\end{equation}
If a generation happens at ordinate $t$ the path travels horizontally to the random location $(k,t)$ with probability
\begin{equation} \label{hc3}
s_{2}^2 \cdots s_{k-1}^2 q^{\mathfrak{h}(2) - \mathfrak{h}(k)} (1-s_{k}^2 q^{\mathfrak{h}(k) - \mathfrak{h}(k+1)})
\end{equation}
and subsequently proceeds turning upward.
\end{itemize}
In our description we assumed, as before, the definition of the height function $\mathfrak{h}(x)$ at a specific ordinate $t$ to be the number of paths strictly to the right of $x-1$.

A possible relevant degenerations of this model is the $q$-TASEP. This is obtained setting $s_k^2=0$ at each location while keeping the $\xi_k s_k$'s finite positive quantities, to be interpreted as speeds of particles.

Other than the procedure we just described, one can possibly consider the ASEP scaling of the Stochastic Six Vertex Model. In this case we set
\begin{equation} \label{ASEP scaling}
\begin{gathered}
s_1^2=0, \quad -\xi_1s_1 > 0,\qquad \qquad \xi_x = 1, \quad s_x=q^{-\frac{1}{2}} \quad \text{for } x>1,\\
u=q^{-\frac{1}{2}} \left( 1+(1-q)\varepsilon \right), \qquad \qquad t=\ceil{ \varepsilon^{-1} \mathfrak{t} },
\end{gathered}
\end{equation}
while, at the same time we shift the position $x$ to $x+t$. With the choice \eqref{ASEP scaling}, when $\varepsilon$ becomes small, we see that paths tend to have diagonal trajectories and the displacements from these diagonals have to be read as the movement of particles in an ASEP dynamics. The coefficient $-\xi_1 s_1$, which previously determined the rate at which paths entered the system now has to be interpreted as a density parameter. More specifically, the initial conditions given by this specializations are half-Bernoulli, in the sense that they describe an ASEP having, at time $\mathfrak{t}=0$, the positive half line empty and each remaining location independently filled with a particle with probability
\begin{equation*}
\frac{ -\xi_1 s_1 q^{-\frac{1}{2}} }{ 1 -\xi_1 s_1 q^{-\frac{1}{2}}  }.
\end{equation*}
The asymmetry here is governed by $q$ and one interprets the height function of the Higher Spin Six Vertex Model as the integrated current of particles through a specific location.

By making use of analytic continuation techniques as those considered above, one can extend these initial conditions to the so called double sided Bernoulli initial conditions, where particles fill locations also in the positive half line independently with Bernoulli law. In this setting, in \cite{Aggarwal2016FluctuationsASEP}, the author was able to study asymptotic properties of the integrated current $\mathfrak{J}$ of the stationary ASEP. As one could expect, also in this case determinantal structures were found considering the $q$-Laplace transform
\begin{equation} \label{q Laplace ASEP}
\left\langle \frac{1}{(\zeta q^{\mathfrak{J}};q)_\infty} \right\rangle.
\end{equation}
An interesting observation is that determinantal expressions for the $q$-Laplace transform \eqref{q Laplace ASEP} obtained in \cite{Aggarwal2016FluctuationsASEP} are similar yet different from the ones we would get employing elliptic determinantal techniques utilized in \cite{q-TASEPtheta} and in this paper. In a future work \cite{ImaMucSasASEPqTASEP} we plan to shed light on relationships between these two different determinantal structures and there we will provide a more detailed analysis of the Stochastic Six Vertex Model, which therefore is here omitted.
\subsection{Inhomogeneous Exponential Jump Model}
This continuous time/continuous space degeneration of the ($J=1$) Higher Spin Six Vertex Model was recently introduced in \cite{BorodinPetrov2017ExpoJump}, where authors were able to study asymptotics and phase transitions of the model with step initial conditions. Here we will apply our results to take into account its stationary state.

The emergence of a continuous space structure in the Higher Spin Six Vertex Model can be recovered considering a particular scaling of the Markov process $\mathfrak{X}^{\text{hc}}$ defined in the Section \ref{section half continuous}. For its description we need the following
\begin{Def} \label{parameters exponential jump}
In this Section we denote with $\mathbf{B} \subset \mathbb{R}_{\geq 0}$ (set of roadblocks) a fixed discrete set, with no accumulation point and for any arbitrary small positive number $\varepsilon$ we set
\begin{equation*}
\mathbf{B}^\varepsilon = \{ \floor{\varepsilon^{-1} b} | b \in \mathbf{B} \}.
\end{equation*}
To the set $\mathbf{B}$ we associate a weight function
\begin{equation*}
\mathsf{p}:\mathbf{B} \rightarrow[0,1].
\end{equation*}
Moreover we set $\mathsf{v}, \mathsf{k}$ to be positive functions and we refer to them respectively as speed and jumping distance function.   
\end{Def}
In light of Definition \ref{parameters exponential jump} we now specialize the continuous time process $\mathfrak{X}^{\text{hc}}$ setting
\begin{equation} \label{scaling exponential jump}
\begin{gathered}
s_i^2= \begin{cases}
e^{-\varepsilon \mathsf{k}(i\varepsilon)}, \qquad &\text{ if } i\in \mathbb{Z}_{\geq 2} \setminus \mathbf{B}^\varepsilon \\
\mathsf{p}(b), \qquad & \text{ if } i \in \mathbb{Z}_{\geq 2} \cap \mathbf{B}^\varepsilon
\end{cases}, \\
\xi_i s_i =\mathsf{v}(i \varepsilon) \qquad \text{for } i \geq 2, \qquad \qquad x= \ceil{ \epsilon^{-1}\mathbf{x} }.
\end{gathered}
\end{equation}
When $\varepsilon$ goes to zero the half continuous Higher Spin Six Vertex Model $\mathfrak{X}^{\text{hc}}$ converges to a process $\mathfrak{X}^{\text{EJ}}$ which we are yet to describe. To do so we need to degenerate expressions \eqref{hc1},\eqref{hc2}, \eqref{hc3} according to the scaling detailed in \ref{scaling exponential jump} and take the limit $\varepsilon\to 0$. In this case we make use of the zero range process language, where the paths at location $(x,t)$ are interpreted as a stack of particles at location $x$ and time $t$. As particles randomly move on $\mathbb{R}_{>0}$, we describe the process through the quantity
\begin{equation*}
    \mathfrak{H}(\mathbf{x},\mathfrak{t})= -\#\left\{
    \parbox{35mm}{\raggedright particle in the interval $(0,\mathbf{x}]$ at time 0} \right\} + \# \left\{
    \parbox{45mm}{\raggedright particle moving to the right of $\mathbf{x}$ during the interval $(0,\mathfrak{t}]$} \right\},
\end{equation*}
that is clearly the analogous of the height function $\mathcal{H}$ \eqref{height correct}.

Given a locally finite configuration of stacks of particles on $\mathbb{R}_{> 0}$, they evolve according to $\mathfrak{X}^{\text{EJ}}$ as follows:
\begin{itemize}
\item at any location $y$ hosting a stack of particles, independently of the rest of the system, a Poisson clock rings with rate $\mathsf{v}(y)(1-q^{\# \{ \text{number of particles at $y$} \}} )$. As the clock rings, exactly one particle of the stack becomes active 
\item an active particle at location $y$ and time $\mathfrak{t}$ performs a random jump to its right of length $\Delta y$ taken with law
\begin{equation*}
\mathbb{P}\left( \Delta y \geq l | \text{ the jump started at $y$} \right) = e^{-\mathsf{K}(y,y+l)} q^{\mathfrak{H}(y_+,\mathfrak{t})-\mathfrak{H}(y+l,\mathfrak{t})} \prod_{b\in \mathbf{B}:y<b<y+l} \textsf{p}(b).
\end{equation*}
Here $\mathsf{K}(y, y+l)=\int_y^{y+l}\mathsf{k}(t)dt$ and the difference $\mathfrak{H}(y_+,\mathfrak{t})-\mathfrak{H}(y+l,\mathfrak{t})$ is the number of particles lying within the interval $(y, y+l]$ at time $\mathfrak{t}$.
\item active particles are injected at position $y=0$ according to a Poisson process with intensity $\mathsf{v}(0)$.
\end{itemize}
The mechanism is clear. When a particle decides to jump, it chooses a distance $\Delta y$ with exponential distribution and, as it flies to reach the targeted destination, it might get captured by a stack of other particles with probability $1-q^{\# \{ \text{ particles in the stack} \} }$ or blocked by a roadblock $b$ with probability  $1-\mathsf{p}(b)$ (see Figure \ref{figure Expo Jump}).

\begin{figure}[ht]
\centering
\includegraphics[scale=1.2]{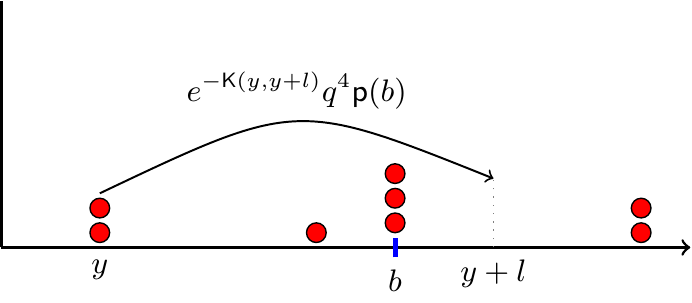}

\caption{\small A graphic visualization of the dynamics of the Exponential Jump Model. In the example exactly one particle at location $y$ is attempting a jump of length greater than $l$. The probability of this event involves the presence of other particles in the term $q^4$ and the presence of a roadblock at location $b$ in the term $\mathsf{p}(b)$ other than the exponential factor $e^{-\mathsf{K}(y,y+l)}$. } \label{figure Expo Jump}
\end{figure}
To discuss continuous degenerations of the Higher Spin Six Vertex Model with random boundary conditions considered above we give the following
\begin{Def} \label{marked Poisson process}
Assume we have positive piecewise continuous function $\mathfrak{L}$ possessing left and right limit at each point, a family of probability distributions $\{\mathbf{\varphi}_y\}_{y \in \mathbb{R}}$ on $\mathbb{Z}_{ \geq 1}$ and an open subset $\Omega \subset \mathbb{R}$. We define the marked Poisson process $\mathfrak{m}_{ \mathfrak{L}, \bm{\varphi} }$ as the process which picks a set of points $\{ y_r \}_r$ on $\Omega$ according to an inhomogeneous Poisson process with rate given by $\mathfrak{L}$ and assigns to each one of the $y_r$'s, independently, a mark chosen with law $\bm{\varphi}_{y_r}$.
\end{Def}
The definition of the marked Poisson process comes in handy when we take the scaling form of the double sided $q$-negative binomial Higher Spin Six Vertex Model.
The basic limit
\begin{equation*}
\rho^M\frac{(e^{-\varepsilon \alpha };q)_M}{(q;q)_M} \frac{(\rho;q)_\infty}{( e^{-\varepsilon \alpha } \rho ;q)_\infty} \xrightarrow[ \varepsilon \sim 0 ]{}  \left( \delta_{0,M} +\varepsilon \alpha \frac{\rho^M}{1-q^M} \right)\left( 1 + \varepsilon \alpha \polygamma_0(\rho) \right)^{-1} +\smallO(\varepsilon),
\end{equation*}
implies that, with the scaling \eqref{scaling exponential jump}, the half continuous Higher Spin Six Vertex model with $q$-NB($s_i^2,v/(\xi_i s_i)$) entries at location $i$ in the horizontal boundary becomes the Exponential Jump Model with initial conditions described as:
\begin{itemize}
\item on $\mathbb{R}_{>0} \setminus \mathbf{B}$ places stacks of particles  according to an inhomogeneous marked Poisson process $\mathfrak{m}_{\mathfrak{L},\bm{ \varphi}}$, where
\begin{equation*}
\bm{\varphi}_y (k) =  \frac{(  v /\mathsf{v}(y) )^k}{1-q^k}\polygamma_0 ( v / \mathsf{v}(y) )^{-1} \qquad \text{and} \qquad \mathfrak{L}(y)=\mathsf{k}(y) \polygamma_0 ( v / \mathsf{v}(y)); 
\end{equation*}
\item on each $b \in \mathbf{B}$ places a stack of $M_b$ particles with probability
\begin{equation} \label{probability roadblocks}
( v / \mathsf{v}(b) )^{M_b} \frac{(\mathsf{p}(b);q)_{M_b}}{(q;q)_{M_b}} \frac{( v / \mathsf{v}(b);q)_\infty}{( \mathsf{p}(b) v/\mathsf{v}(b) ; q )_\infty}.
\end{equation}
\end{itemize}
We refer this process with the symbol $\mathcal{P}(v,\mathsf{v},\mathsf{k},\mathsf{p})$. We report a simple property of a general Marked Poisson process on the line.
\begin{Prop}
Consider a marked Poisson process $\mathfrak{m}_{\mathfrak{L}, \bm{\varphi}}$, as in Definition \ref{marked Poisson process} and consider the random variable $\mathcal{M}(a,b)$ to be the sum of marks contained within the interval $(a,b)$. Then, we have
\begin{equation} \label{generating marked poisson}
\mathbb{E}\left( z^{\mathcal{M}(a,b)} \right) = \exp\left\{ \int_a^b \mathfrak{L}(y) \left( \mathbb{E}_{\bm{\varphi}_y}(z^\mathcal{M}) -1 \right) dy  \right\},
\end{equation}
where $\mathbb{E}_{\bm{\varphi}_y}(z^\mathcal{M})$ is the generating function of the random variable $\mathcal{M}$ counting the marks at the generic location $y$.
\end{Prop}
\begin{proof}
First we see that, from the definition itself of the inhomogeneous marked Poisson process we can write the probability distribution of $\mathcal{M}(a,b)$ as
\begin{equation*}
\mathbb{P}(\mathcal{M}(a,b)=N) = e^{-\int_a^b \mathfrak{L}(y)dy} \sum_{k\geq 0} \sum_{ \substack{\mu \vdash N \\ l(\mu) = k } } \sum_{\nu \sim \mu} \int_{a\leq y_1 < y_2 < \cdots <y_k \leq b } \prod_{j=1}^k \mathfrak{L}(y_j) \varphi_{y_j}( \eta_j) dy_j
\end{equation*}
where the notation $\eta \sim \mu$ means that $\eta$ is a permutation of the partition $\mu$. The generating function can now be evaluated as
\begin{equation*}
\begin{split}
\mathbb{E}\left( z^{\mathcal{M}(a,b)} \right) &=e^{-\int_a^b \mathfrak{L}(y)dy} \sum_{k\geq 0} \sum_{ \mu: l(\mu) = k } \sum_{\nu \sim \mu} \int_{a\leq y_1 < y_2 < \cdots <y_k \leq b } \prod_{j=1}^k \mathfrak{L}(y_j) z^{\eta_j} \varphi_{y_j}( \eta_j) dy_j\\
&=e^{-\int_a^b \mathfrak{L}(y)dy} \sum_{k\geq 0} \frac{1}{k!} \left( \sum_{\eta \geq 1} \int_a^b \mathfrak{L}(y) z^\eta \varphi_y( \eta ) dy \right)^k,
\end{split}
\end{equation*}
which gives \eqref{generating marked poisson}.
\end{proof}
Below we report a proof that the process $\mathcal{P}(v,\mathsf{v}, \mathsf{k}, \mathsf{p})$ indeed admit the stationary measure as a particular case.
\begin{Prop}
Assume that $\mathsf{v}(0) < \mathsf{v}(y)$ for all $y>0$. Then, the process $\mathcal{P}(\mathsf{v}(0),\mathsf{v}, \mathsf{k}, \mathsf{p} )$ is stationary for the Exponential Jump Model.
\end{Prop}
\begin{proof}
The proof of this fact, in the homogeneous case, was already given in \cite{BorodinPetrov2017ExpoJump}. One could simply regard it as a continuous space modification of the argument we used in section \ref{About the initial conditions}. Nonetheless it might still be interesting to explicitly work out the calculations in this particular case as well.

We aim to prove that at any location $L$ the process which counts particles jumping from the region $[0,L]$ to $(L, \infty)$ is a Poisson process with rate $\mathsf{v}(0)$, and hence the current is constant and the density is stationary.

Since the set $\mathbf{B}$ has no accumlation points we can write $\mathbf{B} \cap [0,L] = \{b_1, \dots, b_n\}$, for some finite $n$ and subsequently we partition $[0,L]$ as a disjoint union of intervals
\begin{equation*}
[0,L]= I_0 \cup I_1 \cup \cdots \cup I_n,
\end{equation*}
where $I_0=[0,b_1], I_1=(b_1, b_2], \dots, I_n=(b_n,L]$.
In the infinitesimal time interval $(0, \Delta \mathfrak{t})$ we assume that, up to the terms quadratically small in $\Delta \mathfrak{t}$, we can write
\begin{equation} \label{sum prob expo jump}
\mathbb{P}
    \left\{
    \parbox{30mm}{\raggedright a particle crosses $L$ during $(0, \Delta \mathfrak{t})$}
    \right\}=\sum_{k=0}^n \mathbb{P}
    \left\{
    \parbox{41mm}{\raggedright a particle crosses $L$ during $(0, \Delta \mathfrak{t})$ jumping from $I_k$}
    \right\} + o(\Delta \mathfrak{t}).
\end{equation}
Now, let's consider the single term in the summation in the right hand side of \eqref{sum prob expo jump} and after cleverly using the definition of the model we can easily see that
\begin{equation} \label{calculation prob expo jump}
\begin{split}
\mathbb{P}
    \left\{
    \parbox{41mm}{\raggedright a particle crosses $L$ during $(0, \Delta \mathfrak{t})$ jumping from $I_k$}
    \right\} =& \mathsf{v}(0) \Delta \mathfrak{t} \int_{b_k}^{b_{k+1}} dy \frac{\mathsf{k}(y)}{1-\mathsf{v}(0) / \mathsf{v}(y) } e^{-\mathsf{K}(y,L)} \mathbb{E}(q^{\mathfrak{H}(y,0)-\mathfrak{H}(L,0)})\prod_{j=k+1}^n \mathsf{p}(b_j)\\
    &+ \mathsf{v}(0) \Delta \mathfrak{t} \left( 1- \frac{\mathsf{p}(b_{k+1})(1-\mathsf{v}(0)/\mathsf{v}(b_{k+1}))}{1-\mathsf{p}(b_{k+1})\mathsf{v}(0)/ \mathsf{v}(b_{k+1})} \right) \mathbb{E}(q^{\mathfrak{H}(y,0)-\mathfrak{H}(L,0)})\prod_{j=k+2}^n \mathsf{p}(b_j).
\end{split}
\end{equation}
At this point we can split the difference $\mathfrak{H}(y,0)- \mathfrak{H}(L,0)$ in two independent contribution: one coming from the marked Poisson process which we baptize as $\mathcal{M}(y,L)$ and the other coming form particles encountered at roadblocks which we call $ \mathcal{M}_{\mathbf{B}}(y,L) $. Using expression \eqref{generating marked poisson} and $q$ summation identities concerning the measure \eqref{probability roadblocks}, we obtain
\begin{equation*}
\begin{split}
\mathbb{E}(q^{\mathcal{H}(y,0) - \mathcal{H}(L,0)}) &= \mathbb{E}(q^{\mathcal{M}(y,L)}) \mathbb{E}(q^{ \mathcal{M}_{\mathbf{B}}(y,L) } )\\
&=\exp\left\{ - \int_y^L \mathsf{k}(w) \frac{ \mathsf{v} (0) / \mathsf{v}(w) }{ 1- \mathsf{v} (0) / \mathsf{v}(w) } dw \right\} \prod_{b:y<b<L}\frac{1- \mathsf{v} (0) / \mathsf{v}(b_j) }{ 1- \bm{p}(b_j) \mathsf{v} (0) / \mathsf{v}(b_j) }.
\end{split}
\end{equation*}
Substituting this last identity in the left hand side of \eqref{calculation prob expo jump} we get
\begin{equation*}
\begin{split}
\mathsf{v}(0) \Delta \mathfrak{t} & \left[ \prod_{j=k+2}^n \frac{ \mathsf{p}(b_j) (1- \mathsf{v} (0) / \mathsf{v}(b_j)) }{ 1- \mathsf{p}(b_j) \mathsf{v} (0) / \mathsf{v}(b_j) } \exp\left\{ - \int_{b_{k+1}}^L \mathsf{k}(w) \frac{ \mathsf{v} (0) / \mathsf{v}(w) }{ 1- \mathsf{v} (0) / \mathsf{v}(w) } dw \right\} \right.\\
& \qquad - \left. \prod_{j=k+1}^n \frac{ \mathsf{p}(b_j) (1- \mathsf{v} (0) / \mathsf{v}(b_j)) }{ 1- \mathsf{p}(b_j) \mathsf{v} (0) / \mathsf{v}(b_j) } \exp\left\{ - \int_{b_k}^L \mathsf{k}(w) \frac{ \mathsf{v} (0) / \mathsf{v}(w) }{ 1- \mathsf{v} (0) / \mathsf{v}(w) } dw \right\} \right],
\end{split}
\end{equation*}
from which we deduce that the sum on the right hand side of \eqref{sum prob expo jump} telescopes to $\mathsf{v}(0) \Delta \mathfrak{t}$.
\end{proof}

We now state a result analogous to that of Theorem \ref{theorem: q laplace shifted introduction} in order to characterize the distribution of  $\mathfrak{H}$ in the Exponential Jump Model with initial conditions given by $\mathcal{P}(v, \mathsf{v},\mathsf{k},\mathsf{p})$. In order to apply techniques developed throughout Section \ref{section matching and fredholm} we will assume that the speed function $\mathsf{v}$ is of the form
\begin{equation} \label{velocity condition}
\mathsf{v}(y)=
\begin{cases}
\mathsf{v}_0, \qquad &\text{if } y=0\\
1, \qquad &\text{if } y>0,
\end{cases}
\end{equation}
where $\mathsf{v_0}<1$ and that the system presents no roadblocks. This means that the spatial inhomogeneity is all encoded in the jumping distance function $\mathsf{k}$, on which we do not make any particular assumption. We will refer to an Exponential Jump Model with initial conditions $\mathcal{P}(v,\mathsf{v},\mathsf{k}, \mathsf{p}=0)$ with $\mathsf{v}$ as in \eqref{velocity condition} with the shorthand EJ($v, \mathsf{v_0}; \mathsf{k} $).
\begin{Prop} \label{prop fredholm det expo jump}
For $0<v<\mathsf{v}_0 <  1 $, we have
\begin{equation} \label{fredholm determinant expo jump}
\mathbb{E}_{\emph{EJ}(v,\mathsf{v_0};\mathsf{k}) \otimes \mathsf{m}} \left( \frac{1}{(\zeta q^{\mathfrak{H}(\mathbf{\emph{x}}, \mathfrak{t}) - \mathsf{m}};q)_{\infty}} \right) = \det(\mathbf{1}-fK_{\text{\emph{EJ}}})_{l^2(\mathbb{Z})},
\end{equation}
where
\begin{equation*}
f(n)= \frac{1}{1-q^{n}/\zeta},    
\end{equation*}
\begin{equation*}
K_{\text{\emph{EJ}}}(n,m)= A_{\text{\emph{EJ}}}(n,m)  + (\mathsf{v}_0 - v) \Phi_{\text{\emph{EJ}},\mathbf{\emph{x}}}(m)\Psi_{\text{\emph{EJ}},\mathbf{\emph{x}}}(n),
\end{equation*}
\begin{equation*}
A_{\text{\emph{EJ}}}(n,m)=  \int_D \frac{dw}{2\pi \mathrm{i}} \int_C \frac{dz}{2\pi \mathrm{i}} \frac{z^{m}}{w^{n+1}} \frac{ \exp\left\{\mathfrak{t}z- \polygamma_0(z) \int_0^{\mathbf{\emph{x}}} \mathsf{k}(y) dy \right\} }{ \exp\left\{\mathfrak{t}w- \polygamma_0(w) \int_0^{\mathbf{\emph{x}}} \mathsf{k}(y) dy \right\} }  \frac{(q v/w;q)_\infty}{(q v/z;q)_\infty} \frac{1}{z-w},
\end{equation*}
\begin{equation*}
\Phi_{\text{\emph{EJ}},\mathbf{\emph{x}}}(n)= \int_D \frac{dw}{2\pi \mathrm{i}} \frac{1}{w^{n+1}} \exp\left\{-\mathfrak{t}w + \polygamma_0\left( w\right) \int_0^{\mathbf{\emph{x}}} \mathsf{k}(y) dy \right\}  \frac{(q v/w;q)_\infty}{( q w/\mathsf{v}_0;q)_\infty} \frac{1}{w-\mathsf{v}_0},
\end{equation*}
\begin{equation*}
\Psi_{\text{\emph{EJ}},\mathbf{\emph{x}}}(n)= \int_C \frac{dw}{2\pi \mathrm{i}} z^n \exp\left\{\mathfrak{t}z- \polygamma_0\left( z \right) \int_0^{\mathbf{\emph{x}}} \mathsf{k}(y) dy \right\} \frac{ ( q z/\mathsf{v}_0;q)_\infty }{ (q v/z;q)_\infty } \frac{1}{z-v}.
\end{equation*}
The contour $D$ encircles $\mathsf{v}_0, 1$ and no other singularity, whereas $C$ contains 0 and $q^{k}v$, for any $k$ in $\mathbb{Z}_{\geq 0}$. Finally $\mathsf{m}$ is a $q$Poisson random variable with parameter $v/\mathsf{v}_0$ independent of the particle process.
\end{Prop}
As usual the way to obtain formulas useful for the analysis of the stationary state of the exponential jump model is to set $v=\mathsf{v}_0$. This degenerates the quantity $\mathfrak{H} - \mathsf{m}$ and subsequently the right hand side of \eqref{fredholm determinant expo jump}. Therefore we might proceed with removing the dependence on the independent random quantity $\mathsf{m}$ with an argument equal to that of Lemma \ref{lemma shifting argument}, obtaining an expression as
\begin{equation} \label{stationary q laplace expo jump}
\mathbb{E}_{\text{EJ}(v,\mathsf{v}_0;\mathsf{k})}
\left ( \frac{1}{(\zeta q^{\mathfrak{H}(\text{x}, \mathfrak{t})};q)_\infty} \right ) = \frac{1}{(v/\mathsf{v}_0;q)_\infty}\sum_{k\geq 0}\frac{(-1)^kq^{\binom{k}{2}}}{(q;q)_k} \left( \frac{v}{\mathsf{v}_0} \right)^k \mathbb{E}_{\text{EJ}(v,\mathsf{v}_0;\mathsf{k}) \otimes \mathsf{m} }
\left ( \frac{1}{(\zeta q^{\mathfrak{H}(\text{x},\mathfrak{t})-\mathsf{m}-k};q)_\infty} \right ).
\end{equation}
Both left and right hand side of \eqref{stationary q laplace expo jump} can be proven to be analytic functions of both $\mathsf{v}_0$ and $v$ in a neighborhood of $v=\mathsf{v}_0$ which unlocks the mechanisms developed in Sections \ref{section regularizations}, \ref{section time asymptotics} to study asymptotics. In this case, instead of considering large time/space asymptotics, we let the jumping parameter $\mathsf{k}$ grow along with the time. This corresponds to watching the system evolve with particles moving at a slow speed for long period of time. When this is the case, the scaling we adopt is
\begin{equation*}
    \mathfrak{t} = \kappa_{\emph{EJ};\varpi} \EuScript{T} \qquad \text{and} \qquad \mathsf{k}(y)= \mathpzc{k}(y) \EuScript{T},
\end{equation*}
where $\kappa_{\emph{EJ};\varpi}$, along with other scaling parameters is fixed in the following
\begin{Def}[Scaling parameters Exponential Jump Model] \label{scaling expo jump} 
Set $0<\mathsf{v}_0<1$ and $\mathpzc{K}(\emph{x})=\int_0^\mathbf{x} \mathpzc{k}(y) dy$. Then, we set
\begin{equation*}
\gamma_{\emph{EJ}} =  \frac{1}{2^{1/3}} \left( ( \polygamma_3 ( \mathsf{v}_0 ) - \polygamma_2 ( \mathsf{v}_0 ) ) \mathpzc{K}(\emph{x}) \right)^{1/3},
\end{equation*}
\begin{equation*}
\kappa_{\emph{EJ};\varpi} = \frac{1}{\mathsf{v}_0} \polygamma_2 ( \mathsf{v}_0) \mathpzc{K}(\emph{x}) +  \frac{1}{\mathsf{v}_0} (\polygamma_3 ( \mathsf{v}_0) -\polygamma_2 ( \mathsf{v}_0) ) \mathpzc{K}(\emph{x}) \frac{\varpi}{\gamma x^{1/3}},
\end{equation*}
\begin{equation*}
\eta_{\emph{EJ};\varpi} = \kappa_{\emph{EJ};\varpi} \mathsf{v}_0 - \polygamma_1(\mathsf{v}_0) \mathpzc{K}(\emph{x}) + (\polygamma_3 ( \mathsf{v}_0) -\polygamma_2 ( \mathsf{v}_0) ) \mathpzc{K}(\emph{x}) \frac{\varpi^2}{\gamma^2 x^{2/3}}.
\end{equation*}
\end{Def}
As a last result we can establish Baik-Rains fluctuations of $\mathfrak{H}(\text{x},\eta_{\text{EJ};\varpi} \EuScript{T})$ around $\eta_{\text{EJ};\varpi} \EuScript{T}$.
\begin{theorem}
Consider the stationary state of the inhomogeneous exponential jump model and let $q$ be in a sufficiently small neighborhood of zero. Then, we have
\begin{equation*}
\lim_{\EuScript{T} \rightarrow \infty} \mathbb{P}_{\emph{EJ}(\mathsf{v_0},\mathsf{v_0}, \EuScript{T} \mathpzc{k})}\left( \frac{ \mathfrak{H}(\emph{x},\kappa_{\emph{EJ};\varpi} \EuScript{T}) - \eta_{\emph{EJ};\varpi} \EuScript{T} }{ \gamma_{\emph{EJ}} \EuScript{T}^{1/3}  } > - r \right) =F_\varpi(r).
\end{equation*}
\end{theorem}

\appendix

\section{\texorpdfstring{Preliminaries on $q$-deformed quantities}.} \label{appendix qfunctions}

Along the course of the paper we largely made use of $q$-deformed quantities, such as $q$-Pochhammer symbols and $q$-hypergeometric series. 
The reader might consider these as fairly common and established notions, but, for the sake of completeness, we still like to dedicate this appendix to recall their definitions.

Assuming $q$ is a parameter in the interval $[0,1)$, we define the $q$-Pochhammer symbol

\begin{equation}\label{q Pochhammer}
(z;q)_n = \begin{cases}
(1-z)(1-zq) \cdots (1-zq^{n-1}), \qquad \qquad &\text{if } n \in \mathbb{Z}_{>0},\\
1, &\text{if } n=0,\\
(1-zq^{n})^{-1} (1-zq^{n-1})^{-1} \cdots (1-zq)^{-1}, & \text{if } n \in \mathbb{Z}_{<0},
\end{cases}
\end{equation}
for every meaningful $z \in \mathbb{C}$. When $n$ is positive, the $q$-pochammer symbol (A.1) is a polynomial in $z$ and it admits the expansion
\begin{equation} \label{q Pochammer expansion}
(z;q)_n=\sum_{k=0}^n (-z)^k q^{\binom{k}{2}} \binom{n}{k}_q,
\end{equation}
where we introduced the $q$-binomial
\begin{equation}
\binom{n}{k}_q = \frac{(q;q)_n}{(q;q)_k (q;q)_{n-k}}.
\end{equation}
When we let the integer $n$ grow to $+ \infty$, we see that the product in the left hand side of \eqref{q Pochhammer} is convergent and hence we can define

\begin{equation} \label{q pochammer infinite}
(z;q)_{\infty} = \prod_{j \geq 0} (1- z q^{j}).
\end{equation}
An important result concerning $q$-Pochhammer symbols is the summation identity
\begin{equation} \label{q binomial summation}
\sum_{k \geq 0} \frac{(a;q)_k}{(q;q)_k}z^k = \frac{(za;q)_\infty}{(z;q)_\infty} \qquad \qquad \text{for } a\in \mathbb{C}, |z|<1,
\end{equation}
which can be found in \cite{IsmailClassicalAndQuantum}, Theorem 12.2.5. and it is usually called $q$-binomial theorem.  A slightly more general version of summation \eqref{q binomial summation} is the so called $q$-Gauss summation (\cite{IsmailClassicalAndQuantum}, Theorem 12.2.4)
\begin{equation} \label{q Gauss sum}
   \sum_{n \geq 0} \left( \frac{c}{a b} \right)^n \frac{(a,b;q)_n}{(c,q;q)_n} = \frac{(c/a, c/b;q)_\infty}{(c, c/(ab);q)_\infty}, \qquad \text{for } |c/(ab)|<1, \qquad \text{or }b\in q^{\mathbb{Z}_{<0}} . 
\end{equation}

The $q$-hypergeometric series
\begin{equation} \label{q hypergeometric series}
\setlength\arraycolsep{1pt}
{}_{r+1} \phi_{r} \left(\begin{matrix}&a_1,& a_2, & \dots, & a_{r+1} &\\&b_1,&b_2, &\cdots ,& b_{r}&\end{matrix} \Big| q, z \right) = \sum_{k \geq 0} \frac{ (a_1 , \dots , a_{r+1};q)_k }{ (b_1, \dots , b_{r},q;q)_k }z^k,
\end{equation}
is defined for generic parameters $a_1, \dots, a_{r+1} \in \mathbb{C}$, $b_1, \dots, b_{r} \in \mathbb{C} \setminus q^{\mathbb{Z}_{<0}}$ and $|z|<1$. In the case when at least one of the $a_j$ is of the form $q^{-k}$, for some non-negative integer $k$, the $q$-hypergeometric series \eqref{q hypergeometric series} becomes a finite sum and its definition holds also for more general complex numbers $z$. The regularized terminating $q$-hypergeometric function is also defined as 
\begin{equation} \label{regularized q hypergeometric series}
\setlength\arraycolsep{1pt}
{}_{r+1} \bar{\phi}_{r} \left(\begin{matrix}&q^{-n},& a_1, & \dots, & a_r &\\&b_1,&b_2, &\cdots ,& b_r &\end{matrix} \Big| q, z \right) = \sum_{k = 0}^n z^k \frac{(q^{-n};q)_k}{(q;q)_k} \prod_{j=1}^r (a_j;q)_k (q^k b_j ;q)_{n-k}.
\end{equation}
In Section \ref{About the initial conditions} we used the $q$-analog of the Chu-Vandermonde identity (\cite{IsmailClassicalAndQuantum}, (12.2.17)) that we report as
\begin{equation}
    {}_{2} \bar{\phi}_{1} \left(\begin{matrix}q^{-n}, a\\c \end{matrix} \Big| q, q \right) = \frac{(c/a;q)_n}{(c;q)_n}a^n.
\end{equation}
In the paper we also made use of functions $\polygamma_j$, defined as
\begin{equation} \label{digamma psi}
\polygamma_j(z) = \sum_{k\geq 1} \frac{k^j z^k}{1-q^k}.
\end{equation}
They are related to the more classical $q$-polygamma function \cite{weissteinQPolygamma}
\begin{equation*}
\bm{\psi}_q(\theta) = - \log (1-q) + \log(q) \sum_{n \geq 0} \frac{q^{n+\theta}}{1- q^{n + \theta}},
\end{equation*}
since 
\begin{equation} \label{real polygamma}
\polygamma_0(z) = \frac{1}{\log(q)}\left[ \log(1-q) + \bm{\psi}_q\left( \frac{\log(z)}{\log(q)} \right)  \right]
\end{equation}
and 
\begin{equation}
\frac{d}{dz} \polygamma_j(z) = \frac{1}{z} \polygamma_{j+1}(z).
\end{equation}
The inverse of the infinite $q$-Pochhammer symbol \eqref{q pochammer infinite} is ofter called $q$-exponential and through it one can define a $q$-deformed notion of the common Laplace transform. For a given $f \in \ell^1(\mathbb{Z})$ the function 
\begin{equation} \label{q-Laplace Appendix}
\tilde{f}(\zeta) = \sum_{n \in \mathbb{Z}} \frac{f(n)}{(q^n \zeta;q)_\infty} \qquad \text{for }\zeta \in \mathbb{C} \setminus q^{\mathbb{Z}}
\end{equation}
is the $q$-Laplace transform of $f$. As for the usual Laplace transform, the operation $f \mapsto \tilde{f}$ admits an inverse. This is discussed, for example, in \cite{q-TASEPtheta} and we do not report the exact form of the inverse $q$-Laplace transform as we do not explicitly make use of it during this paper.
\section{\texorpdfstring{Bounds for $\phi_l, \psi_l, \Phi_x, \Psi_x$}.}\label{appendix  bounds}
We collect here some useful bounds for the quantities $\phi_l, \psi_l, \Phi_x, \Psi_x$ defined in \cref{phi,psi,Phi,Psi}.
Terms $\Phi_x,\Psi_x$ can be further decomposed as 
\begin{gather*}
\Phi_x(n) = \Phi_x^{(1)}(n) + \Phi_x^{(2)}(n),\\
\Psi_x(n) = \Psi_x^{(1)}(n) + \Psi_x^{(2)}(n),
\end{gather*}
obtained separating from the integration \eqref{Phi} (resp. \eqref{Psi}) the contribution of pole $w= \mathpzc{d}$ (resp. $z= v$) from that of other poles. Their exact expression was given in \cref{Phi 1,Phi 2,Psi 1,Psi 2}.
\begin{Prop} \label{useful inequalities}
Let $v<\mathpzc{d}$. Then, for all fixed $x$, there exist constants $\Gamma_1, \Gamma_2, \Gamma_3, \Gamma_4>0$, such that
\begin{equation}
|\phi_l(n)|,|\psi_l(n)|,|\Phi_x(n)|,|\Phi^{(2)}_x(n)| < \Gamma_1 e^{-\Gamma_2 |n|} \qquad \text{for all }n \in \mathbb{Z} \label{geometric bound phi}
\end{equation}
and
\begin{equation}
|\Psi_x(n)| < 
\begin{cases}
\Gamma_1 e^{-\Gamma_2 |n|} \qquad \text{if }n \in \mathbb{Z}_{<0}\\
\Gamma_3 e^{-\Gamma_4 |n|} \qquad \text{if }n \in \mathbb{Z}_{\geq 0}. \label{inequality Psi_x}
\end{cases}
\end{equation}
Moreover $\Gamma_1,\Gamma_2$ can be chosen so that their relative bounds also hold for $v,\mathpzc{d}$ in the region \eqref{condition v d analytical continuation} (in this case the parameter $b$ appearing in the definition of $\tau$ \eqref{tau} satisfies $qv< b<\mathpzc{d}$).
\end{Prop}
\begin{proof}
We start with the terms $\phi_l,\Phi_x(n),\Phi_x^{(2)}$. Evaluating the complex integrals as sums of residues it is straightforward to get the inequalities
\begin{equation*}
|\phi_l(n)|,|\Phi_x(n)|,|\Phi_x^{(2)}(n)|\leq \text{const.} \tau(n) \left( \frac{\mathbbm{1}_{n\geq 0}}{|\mathpzc{d}|^{|n|}} + \mathbbm{1}_{n<0} \max_{i\geq 2}(|\xi_i s_i|)^{|n|}\right)
\leq \Gamma_1 e^{-\Gamma_2 |n|},
\end{equation*}
for some constants $\Gamma_1, \Gamma_2$ depending on the integrand functions but not on $n$. 

To obtain a similar bound for the term $\psi_l(n)$ we distinguish two cases. When $n$ is positive we take the contour $C$ to be a circle of radius $r_+$ so that $qv < r_+ < b$. On the other hand, when $n$ is negative we take $C$ to be a circle of radius $r_-$ strictly bigger than $c$, not containing any of the numbers $\xi_i/s_i$ (we remark that the definition itself of $\tau(n)$ and of numbers $b,c$ is tailor-made for these conditions to be possible). With this choices we easily get 
\begin{equation*}
|\psi_l(n)| \leq \text{const.}\frac{1}{\tau(n)} \left( \mathbbm{1}_{n\geq 0} r_+^{|n|} + \mathbbm{1}_{n<0} \frac{1}{r_-^{|n|}} \right)
\leq \Gamma_1 e^{-\Gamma_2 |n|}.
\end{equation*}

An argument equivalent to that used for $\psi_l$ can be carried to show \eqref{inequality Psi_x}. The only difference here is that the radius $r_+$ has to be chosen so that $v < r_+ < b$ and hence we cannot extend this bound to the region $ qv <b \leq v $. 
\end{proof}

\begin{Prop} \label{useful inequalities 2}
Let $v$ satisfy \eqref{condition v} and $v<\mathpzc{d}$ or possibly \eqref{condition v d analytical continuation}. Then, for each $x$, there exist constants $\Gamma_1, \Gamma_2>0$ such that
\begin{equation} \label{inequality phi Psi}
|\phi_l (n) \Psi_x (n) |,|\Phi_x^{(2)}(n) \Psi_x (n) | < \Gamma_1 e^{-\Gamma_2 |n|}.
\end{equation}

\end{Prop}
\begin{proof}
From Proposition \ref{useful inequalities} we see that we only have to prove \eqref{inequality phi Psi} for positive $n$'s. When this is the case we see directly from the integral expression \eqref{Psi} and \eqref{phi} that we can bound both $|\phi_l(n) \Psi_x(n)|$ and $|\Phi^{(2)}_x(n) \Psi_x(n)|$ with some quantity proportional to
\begin{equation} \label{dummy estimate}
\frac{|v + \epsilon|^n}{\min_{k \geq 2} |\xi_k s_k|^n },
\end{equation}
by simply taking the $C$ contour as a circle of radius $v + \epsilon$, for $\epsilon$ being sufficiently small. Due to the condition 
$$
v < \min_{k \geq 2}|\xi_k s_k|,
$$
we see that $\epsilon$ can be chosen so that \eqref{dummy estimate} decays to zero and this completes the proof.
\end{proof}
\begin{Prop} \label{useful inequalities 3}
Let $v, \mathpzc{d}$ satisfy \eqref{condition v d analytical continuation}. Then, for each fixed $x$, there exist constants $\Gamma_1, Gamma_2>0$ such that
\begin{equation*}
|f(n) \Phi_x^{(1)}(n) \Psi_x^{(2)} (n) | < \Gamma_1 e^{-\Gamma_2 |n|}.
\end{equation*}
\end{Prop}
\begin{proof}
We use the integral expression \eqref{Psi 2}. When $n$ is positive we take the integration contour $C_1$ to be a circle of radius $qv + \epsilon$. A bound we can easily obtain is
\begin{equation*}
|f(n)\Phi_x^{(1)}(n) \Psi_x^{(2)}(n)|< \text{const.} \frac{(qv + \epsilon)^n}{\mathpzc{d}^n}  \qquad \text{for } n\geq 0.
\end{equation*}
On the other hand, when $n$ is negative we chose the contour $C_1$ as a circle of radius $v - \epsilon$ to get a bound like
\begin{equation*}
|f(n)\Phi_x^{(1)}(n) \Psi_x^{(2)}(n)|< \text{const.} \frac{1}{1-q^n/\zeta} \frac{(v - \epsilon)^n}{\mathpzc{d}^n} \qquad \text{for } n<0.
\end{equation*}
In both cases condition \eqref{condition v d analytical continuation} allows us to select $\epsilon$ small enough to guarantee exponential decay in $|n|$. 
\end{proof}

\section{Construction of contours} \label{appendix contours}
Here we discuss the construction of the steep descent contour $C$ and that of the steep ascent contour $D$ which were used in the asymptotic analysis of the Stationary Higher Spin Six Vertex Model in section \ref{section time asymptotics}. 

\begin{Prop}\label{prop contour C}
Consider fixed real numbers 
\begin{equation*}
0<v<\varsigma, \qquad 0<q<1
\end{equation*}
and assume that
\begin{equation*}
\varsigma < \inf_{k\geq 2}\{ \xi_k s_k \} \leq \sup_{k\geq 2}\{ \xi_k s_k \} < \infty \qquad \text{and} \qquad 0 \leq s_k^2 <1, \qquad \text{for all } k\geq 2.
\end{equation*}
Take also a number $\rho < \varsigma$ and define the contour
\begin{equation*}
C_{\rho} = \{ \rho e^{\mathrm{i} \vartheta} | \vartheta \in [0,2 \pi) \}.
\end{equation*}
Then, for $\rho$ sufficiently close to $\varsigma$ we have 
\begin{equation} \label{derivative real g}
\frac{d}{d\vartheta} \mathfrak{Re}\{g(\rho e^{\mathrm{i} \vartheta})\} < 0 \qquad \text{for } 0<\vartheta<\pi,
\end{equation}
where $g$ is given in \eqref{g}.
\end{Prop}
\begin{Remark}
The result of Proposition \ref{prop contour C} implies that $C_\rho$ is a steep descent contour for $\mathfrak{Re}(g)$ and in particular 
\begin{enumerate}
\item $\max_{z \in C_\rho} \mathfrak{Re}\{g(z)\} = g(\rho)$;
\item $\max_{z \in C_\rho} |z| = \rho$.
\end{enumerate}
This easily follows from \eqref{derivative real g} and from the fact that $g(\overline{z}) = \overline{g(z)}$, which implies that $\mathfrak{Re}(g)$ is symmetric with respect to the real axis.
\end{Remark}
\begin{proof}
Evaluating the derivative we have
\begin{equation} \label{derivative real part C}
\begin{split}
\frac{d}{d\vartheta} \mathfrak{Re}\{g(\rho e^{\mathrm{i} \vartheta})\} =& \sin \vartheta \kappa \left( \sum_{i=0}^{J-1} \frac{ q^i u \rho }{1 + q^{2i} u^2 \rho^2 - 2 q^i u \rho \cos \vartheta} \right)\\
&+ \sin \vartheta \frac{1}{x} \sum_{k=2}^x \sum_{j\geq 0} \left( \frac{\frac{q^j \rho}{\xi_k s_k} }{ 1 + \left( \frac{q^j \rho}{\xi_k s_k} \right)^2 -2\frac{q^j \rho}{\xi_k s_k} \cos\vartheta }  - \frac{\frac{q^j s_k^2 \rho}{\xi_k s_k} }{ 1 + \left( \frac{q^j s_k^2 \rho}{\xi_k s_k} \right)^2 -2\frac{q^j s_k^2 \rho}{\xi_k s_k} \cos\vartheta }  \right).
\end{split}
\end{equation}
Each term 
$$
\frac{ q^i u \rho }{1 + q^{2i} u^2 \rho^2 - 2 q^i u \rho \cos \vartheta},
$$
has a maximum in $\vartheta = 0$ due to the fact that $u$ and $\rho$ have opposite sign, and so does each single one of the summands in the double summation in \eqref{derivative real part C},
since the generic function
$$
\frac{ a }{ 1 + a^2 -2 a \cos\vartheta }  - \frac{ a \sigma}{ 1 + a^2 \sigma^2 -2 a \sigma \cos\vartheta }
$$
is decreasing in $0<\vartheta < \pi$, provided that $0<a, \sigma<1$.
Now, if $\rho$ is taken sufficiently close to the critical point $\varsigma$, in a neighborhood of $\vartheta=0$, the derivative of $\mathfrak{Re}\{g(\rho e^{\mathrm{i} \vartheta})\}$ is negative by construction and, thanks to considerations we just made, it stays negative along the whole half circle. 
\end{proof}

The construction of an explicit steepest ascent contour $D$ for a general choice of parameters $q, \Xi, \mathbf{S}$ becomes more complicated. Therefore we use the next Proposition both to exhibit a contour in a rather simple setting and to implicitly deduce conditions on $q, \Xi, \mathbf{S}$ under which our arguments of Section \ref{section time asymptotics} are perfectly well posed.  

\begin{Prop} \label{prop contour D}
For each choice of
\begin{equation*}
0<\varsigma <a, \qquad u<0, \qquad 0<\sigma <1,
\end{equation*}
there exist constants $R_a, R_\sigma, R_q>0$, such that for each choice of parameters $\{ \xi_k \}_{k\geq 2}, \{ s_k \}_{k \geq 2}$, $q$ satisfying
\begin{equation*}
\varsigma<\inf_{k\geq 2}\{\xi_k s_k\}, \qquad |\xi_k s_k - a|<R_a, \qquad |s_k^2 - \sigma|<R_\sigma, \qquad q<R_q,
\end{equation*}
we are able to construct a complex contour $D$ encircling the set $\{\xi_k s_k\}_{k \geq 2}$, for which
\begin{enumerate}
\item $\min_{z \in D} \mathfrak{Re}\{g(z)\} = g(\varsigma)$;
\item $\min_{z \in D} |z| = \varsigma$,
\end{enumerate}
where $g$ is given in \eqref{g}.
\end{Prop}
\begin{proof}
To show this result we essentially make use of a continuity argument. We start studying the case when 
\begin{equation*}
q=0, \qquad \xi_k, s_k=a \qquad s_k^2=\sigma \qquad \text{for all } k \geq 2.
\end{equation*}
With this choice of parameters the function $g$ becomes 
\begin{equation} \label{g q=0}
g(z)=- \eta \log(z)+ \kappa \log( 1- u z) + \log\left( \frac{a-z}{a- \sigma z} \right) + \mathcal{O}(x^{-1}),
\end{equation}
where we can neglect the contribution of the $\mathcal{O}(x^{-1})$ term as we are interested in this result only in the limiting case of $x \rightarrow \infty$. We define the contour $D$ to be the level curve
\begin{equation*}
D=\left\{ z:\ \mathfrak{Re}\left\{ \log\left( \frac{a-z}{a- \sigma z} \right) \right\} = \log\left( \frac{a-\varsigma}{a- \sigma \varsigma} \right) \right\},
\end{equation*}
which is a circle and admit the parametrization
\begin{equation*}
\left\{ \varsigma + \rho + \rho e^{\mathrm{i} \vartheta} |\ \vartheta \in [0, 2 \pi) \right\},
\end{equation*}
with the radius $\rho$ being
\begin{equation*}
\rho= \frac{a^2 - a\varsigma - a \varsigma \sigma + \varsigma^2 \sigma}{a+ a\sigma -2 \varsigma \sigma}.
\end{equation*}
We also report that the leftmost and rightmost extremes of the contour $D$ are respectively $\varsigma$ and $\varsigma + 2 \rho$ and one can easily find that the latter satisfies the inequality 
\begin{equation} \label{inequality D}
\varsigma + 2 \rho \leq   \frac{2 a}{1 + \sigma}.
\end{equation}
Along the curve $D$ we are able to calculate
\begin{equation*}
\frac{d}{d \vartheta} \mathfrak{Re} \left\{ g(\varsigma + \rho + \rho e^{\mathrm{i} \vartheta}) \right\}
\end{equation*}
and to analytically show that its only critical points are $\theta \in \mathbb{Z} \pi$. More specifically, substituting in \eqref{g q=0} the correct expressions of coefficients $\eta, \kappa$ given in \eqref{eq: k0 eta0 gamma0}
\begin{gather}
\eta=\frac{a \varsigma^2 (1- \sigma) (a- a^2 u + a \sigma-2 \varsigma \sigma + \varsigma^2 u \sigma) }{(a- \varsigma)^2 (a- \varsigma \sigma)^2},\\
\kappa=- \frac{a(1-\varsigma u)^2 (1-\sigma)(a^2 - \varsigma^2 \sigma)}{u (a- \varsigma)^2 (a-\varsigma \sigma)^2},
\end{gather}
we get
\begin{equation*}
\frac{d}{d \vartheta} \mathfrak{Re} \left\{ g(\varsigma + \rho + \rho e^{\mathrm{i} \vartheta}) \right\}= \sin\vartheta ( 1+ \cos \vartheta) \frac{1}{P(\cos \vartheta)}.
\end{equation*}
In the last expression $P$ is a polynomial of degree two in the argument and we see that zeros are only achieved on the real axis for $\theta=k\pi$ for $k \in \mathbb{Z}$. We can at this point readily verify that, along $D$ the real part of $g$ assumes a minimum at $z=\varsigma$ and a maximum at $z=\varsigma + 2 \rho$ as the function
\begin{equation*}
-\eta \log(y) + \kappa \log(1-u y)
\end{equation*}
is increasing for $y>\varsigma$, and one can check this by direct inspection of its first derivative, by making use of expressions \eqref{eq: k0 eta0 gamma0} for $\eta$ and $\kappa$.

We can now use the fact that $g$ is continuous in the parameters $\Xi, \mathbf{S}, q$ for $z$ belonging to $D$ and the fact that, by construction, it will always have a critical point in $z=\varsigma$, to state the existence of neighborhoods respectively of $a, \sigma$ and 0 in which every choice of $\xi_k s_k , s^2_k$ and $q$ will preserve the steepest ascent properties 1 and 2.
\end{proof}
\printbibliography
\end{document}